\documentclass[14pt, aps, preprint]{revtex4}
\usepackage{amsmath, amsfonts, amssymb, amsthm}
\usepackage{hyperref, graphicx, subfig, float, mathtools}
\newtheorem{theorem}{Theorem}
\newtheorem{claim}{Claim}
\newcommand{\erf}{\mathrm{erf}}

\hypersetup{
  pdftitle={Simplifying plasma balls and black holes with nonlinear diffusion},
  pdfauthor={Connor Behan},
  pdfkeywords={Theoretical Physics Quantum Field String Theory Holography AdS CFT Correspondence Nonlinear PDE Stochastic Thesis}
}

\begin{document}

\title{Simplifying plasma balls and black holes with nonlinear diffusion}
\author{Connor Behan}
\affiliation{Department of Physics and Astronomy, University of British Columbia, Vancouver BC, V6T 1Z1, Canada}
\date{\today}

\newpage
\begin{abstract}
In the Master's thesis of the author, we investigate certain aspects of gravitational physics that emerge from
stochastic toy models of holographic gauge theories. We begin by reviewing field theory thermodynamics, black hole
thermodynamics and how the AdS / CFT correspondence provides a link between the two. We then study a nonlinear evolution
equation for the energy density that was derived last year from a random walk governed by the density of states.
When one dimension is non-compact, a variety of field theories produce long lived plasma balls that are dual to black
holes. This is due to a trapping phenomenon associated with the Hagedorn density of states. With the help of numerical
and mathematical results, we show that problems arise when two or more dimensions are non-compact. A natural extension
of our model involves a system of partial differential equations for both energy and momentum. Our second model is shown
to have some desired, but also some undesired properties, such as a potential disagreement with hydrodynamics.
\end{abstract}
\maketitle

\newpage
\begin{acknowledgments}
Firstly, I would like to thank my supervisor, Mark Van Raamsdonk, for involving me in his research and lending his
strong intuition whenever it was needed. Apart from proofreading, he works hard to ensure that students have projects
matching their interests. Gordon Semenoff agreed to proofread this thesis as well. I thank him, not only for that, but
for delivering the lectures that first taught me string theory. I wish to thank my collaborators Klaus Larjo,
Nima Lashkari and Brian Swingle who discussed many of the problems that came up in our work, and possessed the skill
and motivation to solve them. I am grateful to many fellow students, especially Michael McDermott, Fernando Nogueira and
Jared Stang, for sharing their progress and taking an interest in mine. In the first year of this work, I was supported
by the Nation Science and Engineering Research Council of Canada.

Lastly, it is a pleasure to thank my parents for all of their love and support. This
thesis is dedicated to my cousin Greg, whose wedding I missed while pursuing this degree.
\end{acknowledgments}
\newpage
\tableofcontents
\newpage

\section{Introduction}
This thesis deals with a recently proposed toy model for the dynamics of energy distributions in thermal field
theories. These include the conformal field theories and deformations of them that have gravity duals according to the
AdS / CFT correspondence \cite{maldacena}. As argued in \cite{mvr}, our model suggests that certain important aspects of
gravitational physics emerge for thermodynamic reasons. From this perspective, it is related to the ideas of entropic
gravity in \cite{jacobson, verlinde1, lashkari, faulkner}.

\begin{figure}[h]
\includegraphics[scale=0.7]{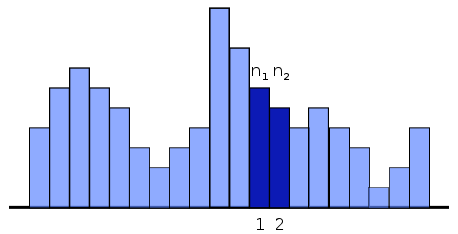}
\caption{By focusing on two adjacent sites in this energy distribution, one may check whether the system favours a
homogeneous or an inhomogeneous state. The properties of this entropic evolution are determined by the local density
of states.}
\label{hopping}
\end{figure}
Before deriving the equations of our model, it is helpful to consider Figure \ref{hopping} --- energy quanta that
randomly hop between sites in a line. To each site, we ascribe a density of states $\rho(n)$ counting the number of
ways for it to have $n$ units of energy. The growth of this function lets us determine which scenario is more likely:
site 1 giving a quantum to site 2 or site 2 giving a quantum to site 1. Asking this question is equivalent to comparing
the sizes of $\rho(n_1 - 1) \rho(n_2 + 1)$ and $\rho(n_1 + 1) \rho(n_2 - 1)$. Positing a form
$\log \rho(n) \propto n^{\alpha}$, we see that site 1 is most likely to give up energy when $\alpha < 1$ and site 2
is most likely to give up energy when $\alpha > 1$. Thus, we see that this random walk leads to diffusion when the
density of states is log-concave and clustering when the density of states is log-convex. In the diffusion case
\textit{e.g.}, a uniform energy distribution is the inevitable final state, even when the microscopic physics are
completely reversible. Special attention is paid to the Hagedorn phase $\alpha = 1$ which is almost completely static.

Even though the essence of our model is this simple statement, it takes the form of a nonlinear partial differential
equation that accepts a $\rho(E)$ function as input. A ubiquitous density of states, which we derive using the
AdS / CFT correspondence, consists of four phases. One of the narrow phases is omitted throughout this thesis for
simplicity. The three that are left consist of a diffusive phase at high energies, a Hagedorn phase at intermediate
energies and another diffusive phase at low energies. Roughly speaking, these respectively correspond to a black hole
forming, living for a long time and ultimately evaporating away. Less ambitiously, we may say that they correspond to
balls of plasma in a purely field theoretic setting \cite{aharony}. We derive rigorous bounds on the decay times for
these objects in our model and compare them to the hadronization times in \cite{aharony}. We find that our times are
longer in one dimension and much shorter in higher dimensions.

To address these problems, a second model is proposed that treats momentum as another quantity that moves
stochastically through a lattice. Since the second model is much more complicated, the discussion of its properties
remains at a speculative level. Even though evolution equations for energy
and momentum sound similar to the spirit of hydrodynamics, we compare our equations to the hydro equations and only
find agreement in the crudest approximation. Despite taking the form of classical PDEs, we hasten to emphasize that our
models include quantum effects when functions like the density of states are chosen appropriately.

This thesis begins with theoretical background in Chapter 2. This chapter focuses on the tools needed to derive
thermodynamic quantities via the AdS / CFT correspondence and contains some lengthy derivations. The main model is
derived afterward in Chapter 3. In Chapter 4, various results from the mathematical literature on nonlinear diffusion
equations are applied to our PDE and used to derive the time scales for black hole evaporation. The suspicious features
of our results are discussed in this chapter as well. Chapter 5 introduces numerical methods that are suitable for our
PDE and uses them to check most of our results. The method chosen for most problems is the implicit Crank-Nicolson
approach. Chapters 6 and 7 contain the newer results that were derived after
\cite{mvr} appeared. Their focus is the extension of our model that includes momentum. Just as our first model depends
on a density of states, our second model depends on a momentum restricted density of states. An expression for this
quantity is derived that allows a small amount of numerical work to be done. Code forming the basis for all of our
simulations is presented in the appendix.

\section{Aspects of holography}
Of all the conjectures that have been made about quantum gravity, the one that has had the largest impact so far is the
AdS / CFT correspondence proposed by Juan Maldacena \cite{maldacena}. Known by various other names like holography or
gauge-gravity duality, it states that string theory in anti-de Sitter space is equivalent to a conformally invariant
quantum field theory living on the boundary of that space. Questions about string theory can therefore be recast in the
language of quantum field theory without gravity. Deriving the evidence for the AdS / CFT correspondence would exceed
the scope of this thesis \cite{magoo}. Instead, we will explore certain dynamical processes that
can be best understood with the correspondence. The effect that will demand most of our attention is
black hole evaporation. Hawking's derivation of black hole evaporation is one of the most successful uses of quantum 
field theory in curved spacetime and any eventual theory of quantum gravity is expected to account for it. Many studies
of Hawking radiation have been done using string theory and the AdS / CFT correspondence in
particular \cite{lt1, lt2, horowitz, marolf}.

Naturally, the first such studies focused on the original version of the correspondence in which the background is
$AdS_5 \times \mathbb{S}^5$ \cite{maldacena}. If one writes the six-dimensional Euclidean Dirac matrices as
\begin{equation}
\Gamma_i = \left [
\begin{tabular}{cc}
$0$ & $\bar{C}_i$ \\
$C_i$ & $0$
\end{tabular}
\right ] \; i \in \{1, \dots, 6\} \; , \nonumber
\end{equation}
the conformal field theory is specified by the Lagrangian \cite{polchinski}
\begin{eqnarray}
L &=& -\frac{1}{2g^2_{\mathrm{YM}}} \int_{\mathbb{R}^3} \frac{1}{2} \mathrm{Tr} \left ( F_{\mu\nu} F^{\mu\nu} + 2D_{\mu} \phi_i D^{\mu} \phi^i - [\phi_i, \phi_j][\phi^i, \phi^j] \right ) \nonumber \\
&& + \mathrm{Tr} \left ( \bar{\lambda} \sigma^{\mu} D_{\mu} \lambda - C_i \lambda [\phi^i, \lambda] - \bar{C}_i \bar{\lambda} [\phi^i, \bar{\lambda}] \right ) \textup{d}x \; . \label{sym}
\end{eqnarray}
Typically the gauge group is $U(N)$ or $SU(N)$ meaning that the scalars, spinors and vectors that show up are really
$N \times N$ matrices consisting of those types of fields. This is called the $\mathcal{N} = 4$ Super Yang-Mills theory
or sometimes the field theory of $D3$-branes. A less than encouraging fact about string theory is that
$AdS_5 \times \mathbb{S}^5$ is far from the only background we need to consider. There is really a whole landscape of
vacua whose boundary field theories may look very different. Indeeed CFT duals have been proposed for
$AdS_4 \times \mathbb{CP}^3$ \cite{abjm}, $AdS_3 \times \mathbb{S}^3 \times \mathbb{T}^4$ \cite{lunin},
$AdS_3 \times \mathbb{S}^3 \times \mathbb{S}^3 \times \mathbb{S}^1$ \cite{tong1} and many others.

Calculations involving these theories are difficult. Even showing that (\ref{sym}) has conformal symmetry is not
trivial. Something that allows us to explore Hawking's process from the holographic viewpoint without choosing a
specific Lagrangian is the intimate connection between black holes and thermodynamics.

\subsection{Thermodynamics}
A number of different field theories have the same thermodynamic potentials. A useful example of this appears in a
conformal field theory. Neglecting the Casimir effect, energy and entropy are both extensive so they must be
proportional to the volume. A conformal theory has no intrinsic scale so the only dimensionful quantity that can
multiply this volume is the temperature. This leads to the expressions $E \propto VT^{d + 1}$ and $S \propto VT^d$.
Substituting them into eachother yields
\begin{equation}
S \propto V^{\frac{1}{d + 1}} E^{\frac{d}{d + 1}} \; . \label{cft-entropy}
\end{equation}
The density of states will turn out to play a fundamental role in our model so we will sometimes exponentiate this
expression.

In the calculations that follow we will see some situations in which this formula for the entropy does not hold. In
general, the rule is that (\ref{cft-entropy}) becomes true for non-conformal theories if the energy is much larger
than any other scale. Different low energy behaviours can be introduced if one compactifies a CFT like (\ref{sym})
on a sphere.

\subsubsection{In free field theory}
An exercise done in \cite{behan} is finding the partition function of a free field theory. Starting with the fact that
$E = p$ in a massless theory, $Z^*(\textbf{p}) = 1 + e^{-\beta p}$ is the contribution of a single fermionic mode and
$Z(\textbf{p}) = 1 + e^{-\beta p} + e^{-2\beta p} + \dots = \left ( 1 - e^{-\beta p} \right )^{-1}$ is the contribution
of a single bosonic mode. Using $s^*$ and $s$ for the number of internal states, the partition function is given by
\begin{equation}
Z = \prod_{\textbf{p}} Z(\textbf{p})^s Z^*(\textbf{p})^{s^*} \; . \nonumber
\end{equation}
If we take the log, the product turns into a sum and if we take the momentum spectrum to be continuous, the sum turns
into an integral. Remembering the integration measure for momentum space, we have
\begin{eqnarray}
\log Z &\approx& \int_{\mathbb{R}^d} s^* \log \left ( 1 + e^{-\beta p} \right ) - s \log \left ( 1 - e^{-\beta p} \right ) \frac{V \textup{d}\textbf{p}}{(2\pi)^d} \nonumber \\
&=& \frac{d \omega_d V}{(2\pi)^d} \int_0^{\infty} s^* \log \left ( 1 + e^{-\beta p} \right ) - s \log \left ( 1 - e^{-\beta p} \right ) p^{d - 1} \textup{d}p \nonumber \\
&=& \frac{d! \omega_d V}{(2\pi \beta)^d} \left [ s\zeta(d + 1) + s^*\zeta^*(d + 1) \right ] \nonumber \\
&\equiv& \frac{AV}{\beta^d} \; . \label{partfunc}
\end{eqnarray}
Here, $\zeta(\sigma) = \sum_{n = 1}^{\infty} \frac{1}{n^{\sigma}}$ is the Riemann zeta function,
$\zeta^*(\sigma) = \sum_{n = 1}^{\infty} \frac{(-1)^{n - 1}}{n^{\sigma}}$ is the alternating zeta function and
$\omega_d$ is the volume of a unit ball in $\mathbb{R}^d$. We may now use
$S = \frac{\textup{d}}{\textup{d}T} (T \log Z)$ and $E = T^2 \frac{\textup{d}}{\textup{d}T} (\log Z)$ to show that
(\ref{cft-entropy}) holds with a proportionality constant of
$\left [ \frac{(d + 1)^{d + 1}}{d^d} A \right ]^{\frac{1}{d + 1}}$.

If one is interested in the density of states, the exponential of this entropy is certainly the first term in $\rho(E)$.
However, there are an infinite number of other terms that come from the differences between the canonical and
microcanonical ensembles. The second term is a standard result that comes from treating $Z(\beta)$ as the Laplace
transform of $\rho(E)$. Performing a saddle point approximation,
\begin{eqnarray}
\rho(E) &=& \frac{1}{2\pi} \int_{-\infty}^{\infty} Z(i\beta) e^{i\beta E} \textup{d}\beta \nonumber \\
&\equiv& \frac{1}{2\pi} \int_{-\infty}^{\infty} e^{f(i\beta)} \textup{d}\beta \nonumber \\
&\sim& \frac{1}{2\pi} \int_{-\infty}^{\infty} e^{f(i\beta_0) - \frac{\beta^2}{2}f^{\prime\prime}(i\beta_0)} \textup{d}\beta \nonumber \\
&=& \frac{1}{\sqrt{2\pi f^{\prime\prime}(i\beta_0)}} e^{f(i\beta_0)} \; . \nonumber
\end{eqnarray}
The higher asymptotic terms cannot be found in the same way because the integral of
$e^{f(i\beta_0) - \frac{\beta^2}{2}f^{\prime\prime}(i\beta_0) - i\frac{\beta^3}{6}f^{\prime\prime\prime}(i\beta_0)}$
has no closed form solution. Instead, powers of $\beta$ after the first two must be Taylor expanded again so that
the above becomes
\begin{eqnarray}
\rho(E) &\sim& \frac{1}{2\pi} \int_{-\infty}^{\infty} e^{f(i\beta_0) - \frac{\beta^2}{2}f^{\prime\prime}(i\beta_0)} \left ( 1 - i\frac{\beta^3}{6}f^{\prime\prime\prime}(i\beta_0) - \frac{\beta^6}{72}f^{\prime\prime\prime}(i\beta_0)^2 + \dots \right ) \nonumber \\
&& \left ( 1 + \frac{\beta^4}{24}f^{\prime\prime\prime\prime}(i\beta_0) + \frac{\beta^8}{1152}f^{\prime\prime\prime\prime}(i\beta_0)^2 + \dots \right ) \dots \textup{d}\beta \; . \nonumber
\end{eqnarray}
These calculations require us to consider an ever-growing number of ways in which a power of $\beta$ can be made.
Nevertheless, this method is still practical for finding the third term in $\rho(E)$ and the resulting expression is
\begin{eqnarray}
\rho(E) &\sim& \frac{1}{\sqrt{2\pi}} \left ( \frac{d}{(d + 1)^{d + 1}} AVE^{-d-2} \right )^{\frac{1}{2(d + 1)}} \nonumber \\
&& \exp \left [ \left ( \frac{(d + 1)^{d + 1}}{d^d} AVE^d \right )^{\frac{1}{d + 1}} - \frac{(d + 2)(2d + 1)}{24(d + 1)} \left ( dAVE^d \right )^{-\frac{1}{d + 1}} \right ] \; . \label{free-dos}
\end{eqnarray}

In \cite{behan}, (\ref{free-dos}) is found in a different way. The inverse Laplace transform of $Z(\beta)$ is found
exactly via a Hankel contour but as a Taylor series, not an asymptotic series. From this series
\begin{equation}
\rho(E) \sim \sum_{j = 1}^{\infty} \frac{(AV)^j E^{dj - 1}}{j!(dj - 1)!} \; , \label{free-dos-series}
\end{equation}
the first three asymptotic terms are picked off. An advantage of this is that (\ref{free-dos-series}) can be compared
to a recent expression for the $d = 1$ density of states due to Loran, Sheikh-Jabbari and Vincon \cite{loran}:
\begin{equation}
\rho(E) \sim \frac{\pi c V}{3} \frac{I_1 \left ( \sqrt{2\pi c V E / 3} \right )}{\sqrt{2\pi c V E / 3}} \; . \label{free-dos-bessel}
\end{equation}
Neither is a generalization of the other because $d$ is arbitrary in (\ref{free-dos-series}) and the interactions are
arbitrary in (\ref{free-dos-bessel}).

In the partition function we have constructed, the $\beta$ variable is conjugate to the energy. There are also
conjugate variables associated with each momentum direction. Something special that we can do in $1 + 1$ dimensions
is combine these into a complex number. Let $p$ be a positive momentum. If there are $N$ excitations
of $p$ and $\tilde{N}$ excitations of $-p$, this state has an energy of $p(N + \tilde{N})$ and a momentum of
$p(N - \tilde{N})$. Therefore, generalized partition functions we can write down are:
\begin{eqnarray}
Z^*(p) &=& \sum_{N = 0}^1 \sum_{\tilde{N} = 0}^1 e^{-pV \tau_2 (N + \tilde{N}) + pVi \tau_1 (N - \tilde{N})} = \left ( 1 + e^{-pV (\tau_2 - i \tau_1)} \right ) \left ( 1 + e^{-pV (\tau_2 + i \tau_1)} \right ) \nonumber \\
Z(p) &=& \sum_{N = 0}^{\infty} \sum_{\tilde{N} = 0}^{\infty} e^{-pV \tau_2 (N + \tilde{N}) + pVi \tau_1 (N - \tilde{N})} = \left ( 1 - e^{-pV (\tau_2 - i \tau_1)} \right )^{-1} \left ( 1 - e^{-pV (\tau_2 + i \tau_1)} \right )^{-1} \; . \nonumber
\end{eqnarray}
Taking the product of $Z(p)^s Z^*(p)^{s^*}$ over all positive momenta, we have
\begin{eqnarray}
\log Z &\approx& \int_0^{\infty} s^* \left [ \log \left ( 1 + e^{-pV (\tau_2 - i \tau_1)}\right ) + \log \left ( 1 + e^{-pV (\tau_2 + i \tau_1)}\right ) \right ] \nonumber \\
&& - s \left [ \log \left ( 1 - e^{-pV (\tau_2 - i \tau_1)}\right ) + \log \left ( 1 - e^{-pV (\tau_2 + i \tau_1)}\right ) \right ] \frac{V \textup{d}p}{2\pi} \nonumber \\
&=& \frac{1}{2\pi} \left ( \frac{1}{\tau_2 - i\tau_1} + \frac{1}{\tau_2 + i\tau_1} \right ) \left [ s\zeta(2) + s^*\zeta^*(2) \right ] \nonumber \\
&=& -\frac{\Im \left ( \frac{1}{\tau} \right )}{\pi} \left [ s\zeta(2) + s^*\zeta^*(2) \right ] \nonumber \\
&=& -\frac{\pi \Im \left ( \frac{1}{\tau} \right )}{6} \left ( s + \frac{s^*}{2} \right ) \; . \label{gen-partfunc}
\end{eqnarray}
The dimensionless number $\tau = \tau_1 + i\tau_2$ is called the \textit{modular parameter}. If
$\tau \equiv \frac{i\beta}{V}$, (\ref{gen-partfunc}) becomes the regular partition function (\ref{partfunc}). The
quantity $c \equiv s + \frac{s^*}{2}$ appearing in (\ref{gen-partfunc}) is central charge that we would use in
(\ref{free-dos-bessel}) if we wanted to apply it to a free theory.

\subsubsection{In string theory}
The worldsheet theory of a string can be regarded as a conformal field theory in $1 + 1$ dimensions. However,
$S \propto \sqrt{E}$ would not be correct for a macroscopic observer who has different notions of energy and
dimensionality. The worldsheet Lagrangian for a supersymmetric string theory in flat space is
\begin{eqnarray}
L &=& -\frac{1}{4\pi \alpha^{\prime}} \int_0^{2\pi} \partial_a X^{\mu} \partial^a X_{\mu} - i \bar{\Psi}^{\mu} \gamma^a \partial_a \Psi_{\mu} \textup{d}\sigma^1 \nonumber \\
&=& -\frac{1}{2\pi \alpha^{\prime}} \int_0^{2\pi} 2\partial X^{\mu} \tilde{\partial} X_{\mu} - i\psi^{\mu} \tilde{\partial} \psi_{\mu} - i\tilde{\psi}^{\mu} \partial \tilde{\psi}_{\mu} \textup{d}\sigma^1 \; . \label{string-action}
\end{eqnarray}
In the second form we have split each Dirac spinor field into two Majorana spinor fields. We have also written
derivatives with respect to $\sigma_{\pm} = \sigma^0 \pm \sigma^1$ as $\partial$ and $\tilde{\partial}$. What makes
this different from a usual quantum field theory is that the $D$ scalar fields $X^{\mu}$ can be interpreted as
positions in a $D$-dimensional \textit{target space}. The worldsheet energy comes from
$\sigma^a \mapsto \sigma^a + \delta \sigma^a$ but the energy we should use for counting states is the conserved
quantity associated with $X^{\mu} \mapsto X^{\mu} + \delta X^{\mu}$. The worldsheet has $\mathrm{SO}(1, 1)$ Lorentz
symmetry regardless of how many fields there are, but the $\mathrm{SO}(D - 1, 1)$ Lorentz symmetry of the target space
is more sensitive. To survive quantization it requires that $D = 10$ \cite{polchinski}. If we had left the fermions out
of (\ref{string-action}) to construct bosonic string theory, the same calculation would tell us that $D = 26$.

To calculate the free energy of a gas of strings, we will begin in the same way as before.
\begin{eqnarray}
F &=& \frac{1}{\beta} \int_{\mathbb{R}^{D - 1}} \sum_m \log \left ( 1 - e^{-\beta \sqrt{p^2 + m^2}} \right ) - \sum_{m^*} \log \left ( 1 + e^{-\beta \sqrt{p^2 + m^{*2}}} \right ) \frac{V \textup{d}\textbf{p}}{(2\pi)^{D - 1}} \nonumber \\
&=& -\frac{1}{\beta} \sum_{n = 1}^{\infty} \int_{\mathbb{R}^{D - 1}} \sum_m \frac{1}{n} e^{-\beta n \sqrt{p^2 + m^2}} - \sum_{m^*} \frac{(-1)^n}{n} e^{-\beta n \sqrt{p^2 + m^{*}}} \frac{V \textup{d}\textbf{p}}{(2\pi)^{D - 1}} \nonumber
\end{eqnarray}
This expression has a sum over the masses of bosons and a sum over the masses of fermions. To arrive at
(\ref{partfunc}), we set these masses to zero and replaced the sums by degeneracy factors. This was valid because the
masses became negligible in the high temperature limit. The high temperature limit of a string theory is different
because it supports arbitrarily large masses. To continue, we will use the trick
\begin{equation}
\int_0^{\infty} e^{-\frac{a^2s}{2} - \frac{b^2}{2s}} \frac{\textup{d}s}{\sqrt{s}} = \frac{\sqrt{2\pi}}{a} e^{-ab} \nonumber
\end{equation}
to rewrite the free energy density.
\begin{eqnarray}
\frac{F}{V} &=& -\sum_{n = 1}^{\infty} \int_0^{\infty} \int_{\mathbb{R}^{D - 1}} \sum_m e^{-\frac{\beta^2 n^2 s}{2} - \frac{p^2 + m^2}{2s}} - (-1)^n \sum_{m^*} e^{-\frac{\beta^2 n^2 s}{2} - \frac{p^2 + m^{*2}}{2s}} \frac{\textup{d}\textbf{p}}{(2\pi)^{D - 1}} \frac{\textup{d}s}{\sqrt{2\pi s}} \nonumber \\
&=& -\sum_{n = 1}^{\infty} \int_0^{\infty} e^{-\frac{\beta^2 n^2}{4\pi \alpha^{\prime} \tau_2}} \left [ \sum_m e^{-\pi \tau_2 \alpha^{\prime} m^2} - (-1)^n \sum_{m^*} e^{-\pi \tau_2 \alpha^{\prime} m^{*2}} \right ] (4 \pi^2 \alpha^{\prime} \tau_2)^{-\frac{D}{2}} \frac{\textup{d} \tau_2}{\tau_2} \label{partfunc-sofar}
\end{eqnarray}
Above, we have made the substitution $\tau_2 = \frac{1}{2\pi \alpha^{\prime} s}$. To proceed further, we need to know
the mass spectrum of our theory.

For concreteness we will work in Type II which is a theory of closed strings. This is natural because evidence of the
AdS / CFT correspondence was first discovered with Type IIB string theory \cite{maldacena}. Very little would change if
we used Type I or heterotic strings. The mode expansions for the scalar fields are identical to the ones that describe
the closed bosonic string:
\begin{equation}
X^{\mu}(\sigma^0, \sigma^1) = x^{\mu} + \alpha^{\prime} p^{\mu} \frac{\sigma_+ + \sigma_-}{2} + i \sqrt{\frac{\alpha^{\prime}}{2}} \sum_{n \neq 0} \frac{1}{n} \left ( \tilde{\alpha}^{\mu}_n e^{-in\sigma_+} + \alpha^{\mu}_n e^{-in\sigma_-} \right ) \; . \nonumber
\end{equation}
For the closed superstring, the left and right movers ($\psi$ and $\tilde{\psi}$) are independent and have the mode
expansions
\begin{eqnarray}
\psi^{\mu}(\sigma^0, \sigma^1) &=& \sqrt{\alpha^{\prime}} \sum_{r \in \mathbb{Z} + v} b^{\mu}_r e^{ir\sigma_+} \nonumber \\
\tilde{\psi}^{\mu}(\sigma^0, \sigma^1) &=& \sqrt{\alpha^{\prime}} \sum_{r \in \mathbb{Z} + v} \tilde{b}^{\mu}_r e^{-ir\sigma_-} \; . \nonumber
\end{eqnarray}
Since fermions can have two different types of boundary conditions, the parameter
$v \in \left \{ 0, \frac{1}{2} \right \}$ denotes which one we are using. For Ramond fermions, which are periodic,
$v = 0$. For Neveu-Schwarz fermions which are antiperiodic, $v = \frac{1}{2}$. The creation and anhilation
operators above obey the relations \cite{bbs}
\begin{eqnarray}
[\alpha^{\mu}_m, \alpha^{\nu}_n] = &m \eta^{\mu\nu} \delta_{m + n, 0}& = [\tilde{\alpha}^{\mu}_m, \tilde{\alpha}^{\nu}_n] \nonumber \\
\{ b^{\mu}_r, b^{\nu}_s \} = &\eta^{\mu\nu} \delta_{r + s, 0}& = \{ \tilde{b}^{\mu}_r, \tilde{b}^{\nu}_s \} \; . \nonumber
\end{eqnarray}

To build up the spectrum from this, we need to consider gauge symmetries. The action (\ref{string-action}) came from
a more general action in which the worldsheet metric was dynamical. Choosing $g_{ab} = \eta_{ab}$ restricts the
physical Hilbert space to only those states which are anhilated by the \textit{Virasoro generators}:
\begin{eqnarray}
L_m &=& \frac{1}{2} \sum_{n \in \mathbb{Z}} : \alpha_{m - n} \cdot \alpha_n : + \frac{1}{4} \sum_{r \in \mathbb{Z} + v} (2r - m) : b_{m - r} \cdot b_r : - a \delta_{m, 0} \nonumber \\
G_r &=& \sum_{n \in \mathbb{Z}} \alpha_n \cdot b_{r - n} \; . \nonumber
\end{eqnarray}
Analogous expressions hold for $\tilde{L}_m$ and $\tilde{G}_r$. Like $D$, the normal ordering constant $a = v$ is an
anomaly that can be fixed by demanding Lorentz invariance \cite{bbs}. We may use the relativistic dispersion relation,
the mode expansions and the Virasoro generators to write down a formula for the mass operator.
\begin{eqnarray}
m^2 &=& -p_{\mu}p^{\mu} \nonumber \\
&=& -\frac{2}{\alpha^{\prime}} \alpha_0 \cdot \alpha_0 \nonumber \\
&=& -\frac{2}{\alpha^{\prime}} \left [ L_0 - \sum_{n > 0} \alpha_{-n} \cdot \alpha_n - \sum_{r > v} r b_{-r} \cdot b_r + a \right ] \nonumber \\
&=& \frac{2}{\alpha^{\prime}} \left [ \sum_{n > 0} \alpha_{-n} \cdot \alpha_n + \sum_{r > v} r b_{-r} \cdot b_r - a \right ] \label{mass-operator}
\end{eqnarray}
We must have $m^2 = \tilde{m}^2$. This translates into a condition known as \textit{level matching} requiring every
state to have the same number of left and right moving excitations. Despite accounting for a gauge symmetry in this
way, the action (\ref{string-action}) still has some gauge symmetry left. A common technique for dealing with this
redundancy is fixing the lightcone gauge. This essentially means that any Lorentz index $\mu$ running from $0$ to
$D - 1$ becomes a regular index $i$ running from $1$ to $D - 2$ \cite{bbs}.

We now have everything we need to derive the massless spectrum of Type II string theory. In typical examples of a Fock
space, the ground state is unique. It is a singlet with respect to any symmetry group of interest and denoted most
often by $\left | 0 \right >$. This is not the case for the superstring. For Ramond fermions, the operators $b^i_0$
commute with $m^2$ meaning that many states have zero mass. This degenerate ground state in fact transforms as a
spinor in ten dimensions. Moreover it can be split into two chiralities $\left | + \right >$ and $\left | - \right >$.
This is different from four dimensions which would make the split into Weyl spinors inconsistent with the split
into Majorana spinors that we have already performed \cite{polchinski}. For Neveu-Schwarz fermions, the lowest lying
state has negative $m^2$. However, one of the advantages of the superstring is that it allows us to avoid this tachyon
and start at the massless states. These are also degenerate and are denoted by $b^i_{-\frac{1}{2}} \left | 0 \right >$.
Even though $b$ is an anticommuting operator for the worldsheet, the $i$ index here makes this a vector particle
in the target space. We have shown that massless R states are spacetime fermions while massless NS states are
spacetime bosons. The choice between R and NS can be made for the left and right movers separately. This means that
Type II string theories have four sectors \cite{bbs}.
\begin{align}
\begin{tabular}{|c|c|c|c|}
\hline
& $(\tilde{v}, v)$ & Type IIA & Type IIB \\
\hline
R-R & $(0, 0)$ & $\left | - \right > \otimes \left | + \right >$ & $\left | + \right > \otimes \left | + \right >$ \\
NS-NS & $\left ( \frac{1}{2}, \frac{1}{2} \right )$ & $\tilde{b}^i_{-\frac{1}{2}} \left | 0 \right > \otimes b^i_{-\frac{1}{2}} \left | 0 \right >$ & $\tilde{b}^i_{-\frac{1}{2}} \left | 0 \right > \otimes b^i_{-\frac{1}{2}} \left | 0 \right >$ \\
NS-R & $\left ( \frac{1}{2}, 0 \right )$ & $\tilde{b}^i_{-\frac{1}{2}} \left | 0 \right > \otimes \left | + \right > \;\;\;\;\;\;$ & $\tilde{b}^i_{-\frac{1}{2}} \left | 0 \right > \otimes \left | + \right > \;\;\;\;\;\;$ \\
R-NS & $\left ( 0, \frac{1}{2} \right )$ & $\;\;\;\;\;\; \left | - \right > \otimes b^i_{-\frac{1}{2}} \left | 0 \right >$ & $\;\;\;\;\;\; \left | - \right > \otimes b^i_{-\frac{1}{2}} \left | 0 \right >$ \\
\hline
\end{tabular}
\label{spectrum-table}
\end{align}
Each sector is 64-fold degenerate.

Our expression for the free energy density has terms like $e^{-\pi \tau_2 \alpha^{\prime} m^2}$ summed over masses.
These sums look like familiar partition functions if we substitute (\ref{mass-operator}) in for $m^2$. Going back to
(\ref{partfunc-sofar}), it is almost correct to replace the term in square brackets with
$64 Z_{0, 0}(\tau_2) + 64 Z_{\frac{1}{2}, \frac{1}{2}}(\tau_2) - 128(-1)^n Z_{0, \frac{1}{2}}(\tau_2)$. In the notation
being used
\begin{eqnarray}
Z_{\tilde{v}, v} &=& e^{2(v + \tilde{v})\pi \tau_2} \left [ \sum_{\left \{ N^i_n, \tilde{N}^i_n \right \}} \exp \left ( -2\pi \tau_2 \sum_{i = 1}^8 \sum_{n = 1}^{\infty} n \left ( N^i_n + \tilde{N}^i_n \right ) \right ) \right ] \nonumber \\
&& \left [ \sum_{\left \{ M^i_r \right \}} \exp \left ( -2\pi \tau_2 \sum_{i = 1}^8 \sum_{r = v + 1}^{\infty} rM^i_r \right ) \right ] \left [ \sum_{\left \{ \tilde{M}^i_r \right \}} \exp \left ( -2\pi \tau_2 \sum_{i = 1}^8 \sum_{r = \tilde{v} + 1}^{\infty} r\tilde{M}^i_r \right ) \right ] \; , \label{partfunc-schematic}
\end{eqnarray}
$N^i_n$ and $\tilde{N}^i_n$ are bosonic occupation numbers while $M^i_r$ and $\tilde{M}^i_r$ are fermionic occupation
numbers. Accounting for level matching is the one correction that needs to be made. This can be done by inserting a
Kronecker delta
\begin{equation}
\delta_{L,R} = \int_{-\frac{1}{2}}^{\frac{1}{2}} e^{2\pi i \tau_1 (L - R)} \textup{d}\tau_1 \nonumber
\end{equation}
where $L = \sum_{i, n} n \tilde{N}^i_n + \sum_{i, r} r \tilde{M}^i_r$ and
$R = \sum_{i, n} n N^i_n + \sum_{i, r} r M^i_r$ are the left and right excitations respectively. Multiplying this by
(\ref{partfunc-schematic}), we see that the quantity being integrated is nothing but the generalized partition function
for 8 fermions and 8 bosons. Rewriting (\ref{partfunc-sofar}),
\begin{eqnarray}
\frac{F}{V} &=& -\sum_{n = 1}^{\infty} \int_0^{\infty} \int_{-\frac{1}{2}}^{\frac{1}{2}} e^{-\frac{\beta^2 n^2}{4\pi \alpha^{\prime} \tau_2}} 64 Z_{\mathrm{B}}(\tau)^8 \left [ Z_{\mathrm{R}}(\tau)^8 + Z_{\mathrm{NS}}(\tau)^8 - 2(-1)^n \sqrt{Z_{\mathrm{R}}(\tau) Z_{\mathrm{NS}}(\tau)}^8 \right ] \nonumber \\
&& (4 \pi^2 \alpha^{\prime} \tau_2)^{-5} \frac{\textup{d} \tau_1 \textup{d} \tau_2}{\tau_2} \; . \label{string-free-energy}
\end{eqnarray}
In terms of our old notation, $Z_{\mathrm{B}}(\tau) = Z(\tau)$ while $Z_{\mathrm{R}}(\tau)$ and $Z_{\mathrm{NS}}(\tau)$
approach $Z^*(\tau)$ in the small $\tau$ limit.

Our goal is to investigate the high temperature limit of (\ref{string-free-energy}). Since this corresponds to
$\beta \rightarrow 0$, the integral is dominated by the $n = 1$ term of the sum and the small $\tau$ limits of the
worldsheet partition functions. A curious fact about string theories is that at a high enough temperature, called the
Hagedorn temperature $T_{\mathrm{H}}$, the free energy density diverges. We will solve for $\beta_{\mathrm{H}}$. This
can be done by looking at any one of the four terms in the integrand of (\ref{string-free-energy}). Substituting the
generalized partition function (\ref{gen-partfunc}), the function we are integrating is
\begin{equation}
64 e^{-\frac{\beta^2}{4\pi \alpha^{\prime} \tau_2}} e^{-2\pi \Im \left ( \frac{1}{\tau} \right )} \frac{(4 \pi^2 \alpha^{\prime} \tau_2)^{-5}}{\tau_2} \; . \nonumber
\end{equation}
The value of $\beta_{\mathrm{H}}$ is reached when the overall exponent is zero. This means
\begin{eqnarray}
\beta_{\mathrm{H}}^2 &=& 8\pi^2 \alpha^{\prime} \frac{\tau_2^2}{\tau_1^2 + \tau_2^2} \nonumber \\
&\approx& 8\pi^2 \alpha^{\prime} \; . \label{hagedorn-temperature}
\end{eqnarray}
The partition function for a system first diverges when the density of states becomes exponential and the
decay of the Boltzmann factor can no longer overpower such growth. This is equivalent to saying
that $S \propto E$. The proportionality constant can be read off from (\ref{hagedorn-temperature}) because
$\beta = \frac{\textup{d} S}{\textup{d} E}$ must give $\beta_{\mathrm{H}}$. The result of this, in contrast to
(\ref{cft-entropy}) is:
\begin{equation}
S = 2\pi \sqrt{2 \alpha^{\prime}} E \; . \label{string-entropy}
\end{equation}

\subsection{Black holes}
It is clear how entropy arises in the field theories we have discussed. If we only know the energy $E$ of a field,
the corresponding ensemble of particles can be in any one of $\rho(E)$ microstates contributing to our lack of
knowledge about the system. Microstates of this form do not appear to be present for black holes. Classically, one can
learn everyting about a black hole from just three numbers: mass, charge and angular momentum. The discovery that black
holes have entropy as well, has led to some of the deepest results in theoretical physics \cite{wald}.

\subsubsection{Useful metrics}
Black hole metrics in arbitrary dimension have seen increasing interest since the discovery of the AdS / CFT
correspondence \cite{emparan}. Most authors take the Einstein equations to be fundamental so that they read
\begin{equation}
\mathcal{R}_{\mu\nu} - \frac{1}{2} \mathcal{R} g_{\mu\nu} + \Lambda g_{\mu\nu} = 8\pi G T_{\mu\nu} \label{einstein}
\end{equation}
regardless of how many dimensions there are. The same cannot be said of the Newtonian limit. If (\ref{einstein})
describes the full theory of gravity in $d + 1$ dimensions, one can show that dimension dependent prefactors
necessarily appear in the Poisson equation \cite{emparan}:
\begin{eqnarray}
\Delta \Phi &=& 8\pi G \frac{d - 2}{d - 1} \rho \nonumber \\
&=& 8\pi G \frac{d - 2}{d - 1} M \delta(0) \; . \label{newtonian-limit}
\end{eqnarray}
The second form above specializes to a point mass of $M$.

The neutral, irrotational black hole in arbitrary dimension is called the Schwarzschild-Tangherlini solution.
\begin{equation}
ds^2 = -\left ( 1 - \frac{\mu}{r^{d - 2}} \right ) dt^2 + \left ( 1 - \frac{\mu}{r^{d - 2}} \right )^{-1} dr^2 + r^2 d\Omega_{d - 1}^2 \label{schwarzschild}
\end{equation}
We will take this opportunity to review some of the basic properties that black hole metrics should have.
\begin{enumerate}
\item
It can be checked that (\ref{schwarzschild}) solves Einstein's equation with no cosmological constant and no
stress-energy tensor. In fact, it is the unique spherically symmetric and time independent solution. The requirement
that it be time independent is redundant if $d = 3$.
\item
It is clear that (\ref{schwarzschild}) is asymptotically flat. Therefore, the notion of ``escaping'' from a potential
well in this metric is well defined.
\item
It is also clear that beyond a certain radius, one can no longer escape. Far away from the origin, $t$ is timelike and
$r$ is spacelike but this reverses when $r$ falls below $\mu^{\frac{1}{d - 2}}$. On the inside of this
\textit{event horizon}, the flow of time in (\ref{schwarzschild}) is such that an object is inexorably drawn toward
the centre.
\item
This implies that the mass generating the event horizon has been compressed to a point. This, along with the fact that
metrics for equal point masses should be indistinguishable, tells us that (\ref{schwarzschild})
represents the most efficient packing of said mass into a sphere of radius $\mu^{\frac{1}{d - 2}}$.
\item
The mass may be computed by taking the Newtonian limit. For those unfamiliar with the ADM procedure \cite{adm}, we will
consider a test particle far away from the origin, moving radially outward. If the motion is non-relativistic, the
timelike geodesic condition becomes $g_{tt} \dot{t}^2 \approx -1$ or
$\dot{t}^2 \approx \left ( 1 - \frac{\mu}{r^{d - 2}} \right )^{-1}$. Substituting this into the geodesic equation,
\begin{eqnarray}
\ddot{r} &\approx& -\Gamma^r_{tt} \dot{t}^2 \nonumber \\
&\approx& \frac{1}{2} \partial_r g_{tt} \frac{g^{rr}}{1 - \frac{\mu}{r^{d - 2}}} \nonumber \\
&=& -\nabla \left ( \frac{\mu}{2r^{d - 2}} \right ) \; . \nonumber
\end{eqnarray}
The function inside the gradient should be a solution to (\ref{newtonian-limit}). Recalling the Green's function for
the $d$-dimensional Laplacian, this is only true if $\mu = \frac{16 \pi GM}{d(d - 1) \omega_d}$.
\end{enumerate}

A further generalization of interest to us is the metric for a black hole that is asymptotically AdS. Anti-de Sitter
space is the maximally symmetric solution to Einstein's equations when they have a negative cosmological constant. A
cosmological constant defines a length scale for the spacetime with
\begin{equation}
\Lambda \equiv -\frac{d(d - 1)}{2L^2} \label{cosmological-constant}
\end{equation}
by convention. Because of this, AdS posesses a conformal boundary. Its radial co-ordinate is infinite but an observer
is able to reach it in a finite amount of proper time. The solution is
\begin{equation}
ds^2 = -\left ( 1 + \frac{r^2}{L^2} \right ) dt^2 + \left ( 1 + \frac{r^2}{L^2} \right )^{-1} dr^2 + r^2 d\Omega_{d - 1}^2 \label{ads}
\end{equation}
and a black hole metric asymptotic to this is
\begin{equation}
ds^2 = -\left ( 1 + \frac{r^2}{L^2} - \frac{\mu}{r^{d - 2}} \right ) dt^2 + \left ( 1 + \frac{r^2}{L^2} - \frac{\mu}{r^{d - 2}} \right )^{-1} dr^2 + r^2 d\Omega_{d - 1}^2 \; . \label{schwarzschild-ads}
\end{equation}
Uniqueness of (\ref{schwarzschild-ads}) in the same sense as (\ref{schwarzschild}) is suspected but not known.
The horizon radius is a solution to $\frac{\mu}{r_0^{d - 2}} = 1 + \frac{r_0^2}{L^2}$ with the outermost one being
the point of no return. The mass is again given by $\mu = \frac{16 \pi GM}{d(d - 1) \omega_d}$. To see this,
a non-relativistic radial trajectory satisfies
\begin{eqnarray}
\ddot{r} &\approx& -\Gamma^r_{tt} \dot{t}^2 \nonumber \\
&\approx& \frac{1}{2} \partial_r g_{tt} \nonumber \\
&=& -\nabla \left ( \frac{\mu}{2r^{d - 2}} \right ) - \frac{r}{L^2} \; . \nonumber
\end{eqnarray}
While this lacks the sophistocation of methods like \cite{ashtekar}, we obtain the right answer if we simply subtract
the acceleration that a particle would have in pure AdS.

It will be convenient to rewrite this metric in Eddington-Finkelstein co-ordinates. This can be done using either
the retarded time or advanced time, which take the form
\begin{eqnarray}
u &=& t - r^* \nonumber \\
v &=& t + r^* \nonumber
\end{eqnarray}
respectively. Differentiating these to arrive at
\begin{eqnarray}
\dot{u} &=& \dot{t} - \frac{\textup{d}r^*}{\textup{d}r} \dot{r} \nonumber \\
\dot{v} &=& \dot{t} + \frac{\textup{d}r^*}{\textup{d}r} \dot{r} \; , \nonumber
\end{eqnarray}
we see that $r^*$ should be chosen so that its radial derivative is the factor relating $\dot{t}$ and
$\dot{r}$. For a null geodesic,
\begin{equation}
\dot{t} = \pm \left ( 1 + \frac{r^2}{L^2} - \frac{\mu}{r^{d - 2}} \right )^{-1} \dot{r} \nonumber
\end{equation}
where the positive sign corresponds to an outgoing particle and the negative sign corresponds to an ingoing particle.
We now have
\begin{equation}
\frac{\textup{d}r^*}{\textup{d}r} = \left ( 1 + \frac{r^2}{L^2} - \frac{\mu}{r^{d - 2}} \right )^{-1} \; . \nonumber
\end{equation}
In certain cases, this can be integrated to give $r^*$ explicitly. However, this is not needed for replacing $dt$.
Performing the change of variables,
\begin{eqnarray}
ds^2 &=& -\left ( 1 + \frac{r^2}{L^2} - \frac{\mu}{r^{d - 2}} \right ) du^2 - 2 du dr + r^2 d\Omega_{d - 1}^2 \nonumber \\
ds^2 &=& -\left ( 1 + \frac{r^2}{L^2} - \frac{\mu}{r^{d - 2}} \right ) dv^2 + 2 dv dr + r^2 d\Omega_{d - 1}^2 \label{ed-fink-ads}
\end{eqnarray}
are the desired metrics.

\subsubsection{Hawking radiation}
To derive the relation between entropy and the area of a black hole, we will follow Hawking's original paper
\cite{hawking} as well as the clarifications in \cite{parker, thompson}. What we will see is that two observers --- one
observing spacetime before a black hole has formed, the other after --- will have different definitions of the quantum
vacuum. The Klein-Gordon equation for a field in curved spacetime is
\begin{eqnarray}
\nabla_{\mu} \nabla^{\mu} \phi - m^2 \phi = 0 \nonumber \\
\frac{1}{\sqrt{-g}} \partial_{\mu} \left ( \sqrt{-g} g^{\mu\nu} \partial_{\nu} \phi \right ) - m^2 \phi = 0 \; . \label{kg-equation}
\end{eqnarray}
For any two solutions $\phi_1$ and $\phi_2$, the Klein-Gordon inner product
\begin{equation}
\left < \phi_1, \phi_2 \right > = \int_S \left [ \phi^*_1 \nabla^{\mu} \phi_2 - \phi_2 \nabla^{\mu} \phi^*_1 \right ] \textup{d}\Sigma_{\mu} \label{kg-inner-product}
\end{equation}
will be conserved. The notation above suggests a scalar field, but this does not have to be the case. Fields with
multiple components like vectors and spinors satisfy the Klein-Gordon equation componentwise. To indicate that Hawking
radiation is a mixture of all types of particles, we will write creation and anhilation operators as
$a^{\mathcal{I}\dagger}$ and $a^{\mathcal{I}}$ where $\mathcal{I}$ is an index set. One simplification we will make,
however, is that the fields are massless.

We may write a basis of solutions to (\ref{kg-equation}) as $f_i$ and choose them to be orthonormal with respect to
(\ref{kg-inner-product}). If we do this, the field operator takes the form
\begin{equation}
\phi^{\mathcal{I}} = \sum_i f_i a^{\mathcal{I}}_i + f^*_i a^{\mathcal{I}\dagger}_i \; . \nonumber
\end{equation}
In other words, positive frequency modes multiply anhilation operators while negative frequency modes multiply creation
operators. We will let these $f_i$ represent any particles that can be seen before a black hole forms. Since there are
no such particles, the past observer will see the vacuum state $\left | 0 \right >_a$ defined as the state that is
anhilated by all $a^{\mathcal{I}}_i$. The future observer sees a different metric and in particular a different time
component of the metric. This means he will have a different definition of positive and negative frequency. Writing
\begin{equation}
\phi^{\mathcal{I}} = \sum_i g_i b^{\mathcal{I}}_i + g^*_i b^{\mathcal{I}\dagger}_i \; , \nonumber
\end{equation}
each $g_i$ representing a particle in the black hole spacetime should be expressible as a linear combination
of the $f_i$. A positive frequency $g_i$ may therefore include a contribution from a negative frequency $f_i$ and
\textit{vice versa}. If so, the $b^{\mathcal{I}}_i$ will not anhilate the $a^{\mathcal{I}}_i$ vacuum and the
$a^{\mathcal{I}}_i$ will not anhilate the $b^{\mathcal{I}}_i$ vacuum. This discrepancy between
$\left | 0 \right >_a$ and $\left | 0 \right >_b$ means that the future observer will see radiation precisely because
the past observer did not.

It is not correct to say that the only modes of $\phi^{\mathcal{I}}$ are $f_i$ waves that the past observer can see
and $g_i$ waves that the future observer can see. There are also $h_i$ waves in the future that cannot be seen because
they are behind the event horizon of the black hole. We would have to consider these if we wanted to write the past
modes as linear combinations of the future modes. As it happens, we will only need to write the future modes as linear
combinations of the past modes. Converting
\begin{eqnarray}
g_i &=& \sum_j \alpha_{ij} f_j + \beta_{ij} f^*_j \nonumber \\
h_i &=& \sum_j \sigma_{ij} f_j + \tau_{ij} f^*_j \nonumber
\end{eqnarray}
into a set of relations between operators, we arrive at the so-called Bogoliubov transformation:
\begin{eqnarray}
b^{\mathcal{I}}_i &=& \sum_j \alpha^*_{ij} a^{\mathcal{I}}_j - \beta^*_{ij} a^{\mathcal{I}\dagger}_j \nonumber \\
c^{\mathcal{I}}_i &=& \sum_j \sigma^*_{ij} a^{\mathcal{I}}_j - \tau^*_{ij} a^{\mathcal{I}\dagger}_j \; . \nonumber
\end{eqnarray}
This tells us that the number of particles detected as belonging to $g_i$ in the future is given by
\begin{equation}
n_i = {}_a \left < 0 \right | b^{\mathcal{I}\dagger}_i b^{\mathcal{I}}_i \left | 0 \right >_a = \sum_j \left | \beta_{ij} \right |^2 \; . \label{emission}
\end{equation}
Assuming that we are dealing with bosons, we also have
\begin{equation}
1 = {}_a \left < 0 \right | \left [ b^{\mathcal{I}}_i, b^{\mathcal{I}\dagger}_i \right ] \left | 0 \right >_a = \sum_j \left | \alpha_{ij} \right |^2 - \left | \beta_{ij} \right |^2 \; . \label{absorbtion}
\end{equation}
We will now be more explicit about what the modes are so that we may plug them into the inner product and find the
$\alpha_{ij}$ and $\beta_{ij}$ coefficients.

Spherical waves are convenient choices, but it is important not to use the expressions for flat space spherical waves
when we are really in a curved space. By construction, outgoing null geodesics are lines of constant $u$ while ingoing
null geodesics are lines of constant $v$. Therefore, the advanced and retarded times should be used in place of
$t \pm r$ giving us
\begin{eqnarray}
f_{\omega, l_1, \dots, l_{d-1}} (v, \theta_1, \dots, \theta_{d - 1}) = \frac{e^{i \omega v}}{\sqrt{(\omega r)^{d - 1} d \omega_d}} Y_{l_1, \dots, l_{d - 1}} (\theta_1, \dots, \theta_{d - 1}) \nonumber \\
g_{\omega, l_1, \dots, l_{d-1}} (u, \theta_1, \dots, \theta_{d - 1}) = \frac{e^{i \omega u}}{\sqrt{(\omega r)^{d - 1} d \omega_d}} Y_{l_1, \dots, l_{d - 1}} (\theta_1, \dots, \theta_{d - 1}) \label{spherical-waves}
\end{eqnarray}
as approximate solutions for large $r$. We could similarly consider
$f_{\omega, l_1, \dots, l_{d-1}} (u, \theta_1, \dots, \theta_{d - 1})$ and
$g_{\omega, l_1, \dots, l_{d-1}} (v, \theta_1, \dots, \theta_{d - 1})$ but these would affect the result very little.
The interesting effects come from waves that switch from ingoing to outgoing while the black hole is forming. By this,
we mean that waves of constant $v$ travel toward the collapsing mass at $r = 0$. As long as an event horizon has not
formed yet, such waves may emerge from the other side and start moving away with constant $u$. A natural question to
ask is \textit{which} constant $u$? That is, what will $u$ be in terms of the $v$ that the wave used to have? This is
the key question that must be answered before we can take an inner product and derive Hawking's result. The difficulty
in relating these is explained in Figure \ref{affine-parameter}.
\begin{figure}[h]
\centering
\subfloat[][Before]{\includegraphics[scale=0.4]{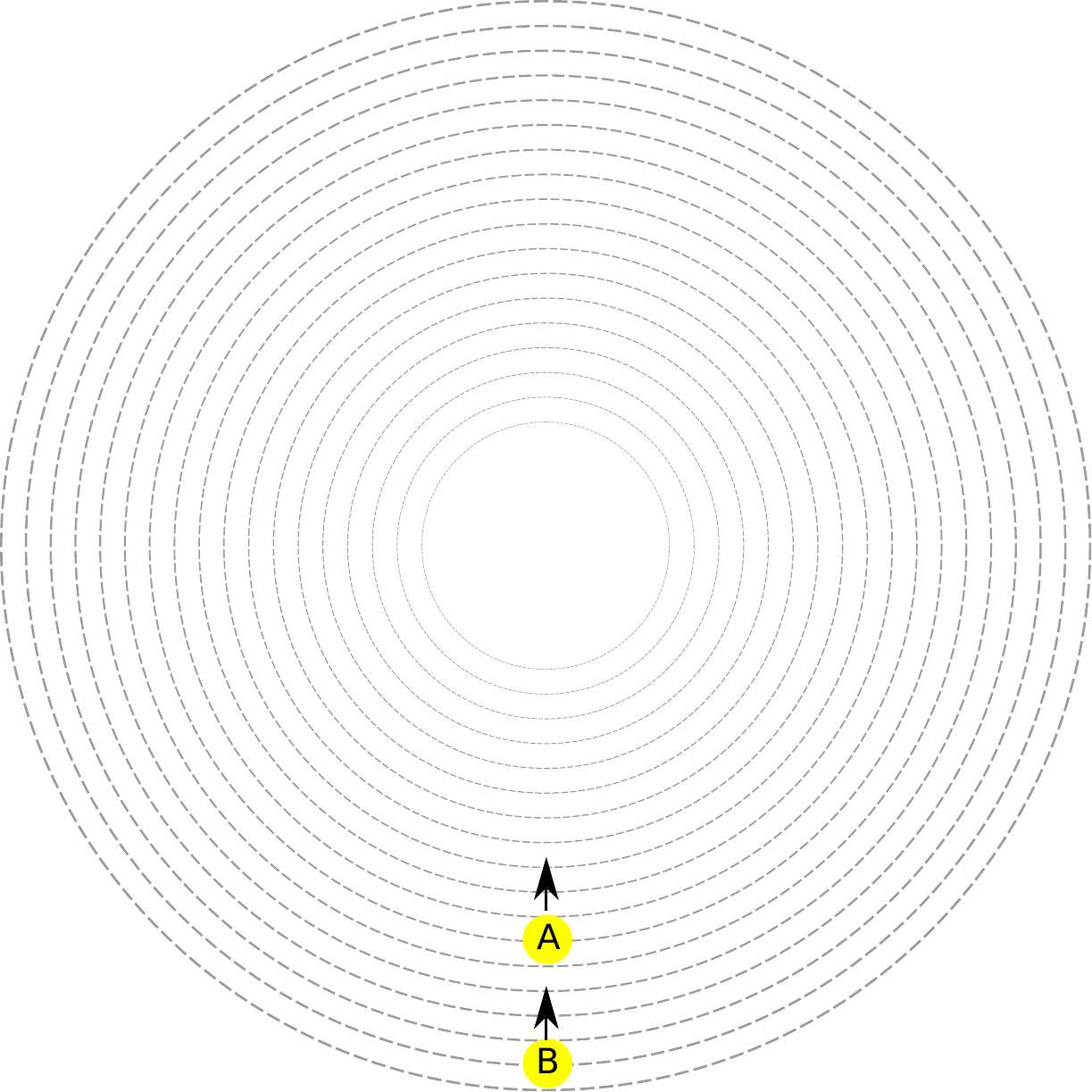}} \qquad
\subfloat[][After]{\includegraphics[scale=0.4]{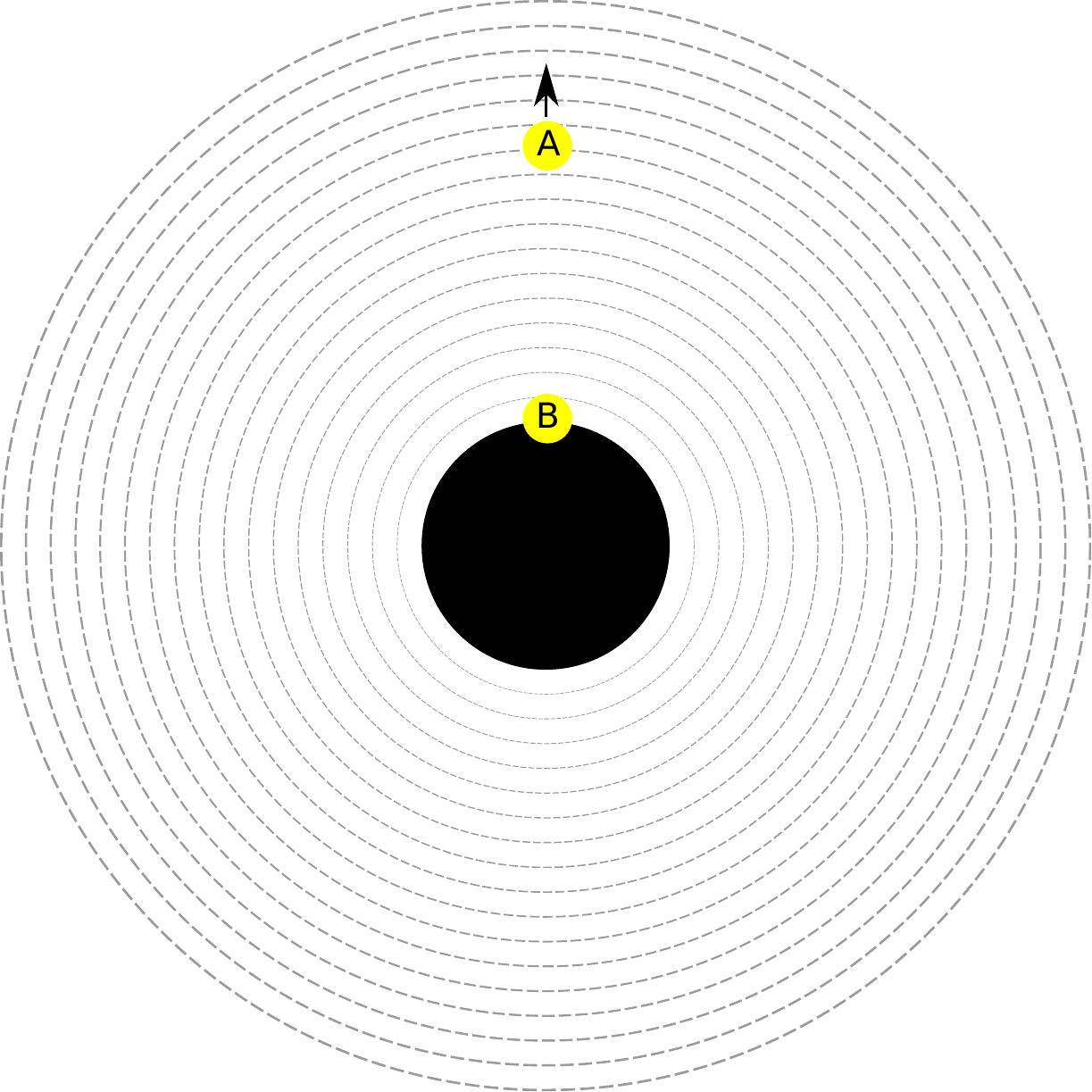}}
\caption{Photon A leaves a light source and starts heading toward $r = 0$. After a certain time interval, photon B does
the same. If the spacetime is Minkowski, the distance between the two photons will not change. Conversely, if a
black hole at $r = 0$ forms at just the right time, photon A will escape but photon B will stay trapped at the event
horizon forever.}
\label{affine-parameter}
\end{figure}
A better way to compare objects A and B is to imagine that A throws a ball C backwards until it is caught by B. The
proper time for C to travel should be equal along all stages of the journey. Instead of proper time, we will
use a difference of affine parameters which is appropriate for signals travelling at the speed of light.

Waves hoping to escape the black hole must start off with a $v$ smaller than the one posessed by photon B in Figure
\ref{affine-parameter}. We will call this largest advanced time $v_0$. If the spacetime is Minkowski, long before the
black hole has formed, $v = t + r$ and affine parametrizations are
\begin{eqnarray}
t(\lambda) &=& t(0) + \frac{\lambda}{2} \nonumber \\
r_{\mathrm{A}}(\lambda) &=& r_{\mathrm{A}}(0) - \frac{\lambda}{2} \nonumber \\
r_{\mathrm{B}}(\lambda) &=& r_{\mathrm{B}}(0) - \frac{\lambda}{2} \nonumber \\
r_{\mathrm{C}}(\lambda) &=& r_{\mathrm{A}}(0) + \frac{\lambda}{2} \; . \nonumber
\end{eqnarray}
By the time C makes it to B, it will have the same position and time co-ordinate as B so it must have the same $v$ as
B. Therefore subtracting the advanced times corresponding to A and C, we have
\begin{equation}
v - v_0 = r_{\mathrm{A}}(\lambda) - r_{\mathrm{C}}(\lambda) = -\lambda \; . \nonumber
\end{equation}
We may therefore call $v_0 - v$ an affine parameter that vanishes for the wave that stays at the event horizon. This
means that after the horizon at $r_0$ has formed, a wave's radial co-ordinate must look like $r = r_0 - \lambda$. We
will substitute this into the retarded time for C noting that C is not a wave of constant $u$ because it travels
backwards from A back to B.
\begin{eqnarray}
\dot{u} &=& \dot{t} - \frac{\textup{d}r^*}{\textup{d}r} \dot{r} \nonumber \\
&=& 2 \left ( 1 + \frac{r^2}{L^2} - \frac{\mu}{r^{d - 2}} \right )^{-1} \nonumber \\
&=& 2 \left ( 1 + \frac{r^2}{L^2} - \left ( \frac{r_0}{r} \right )^{d - 2} \left ( 1 + \frac{r_0^2}{L^2} \right ) \right )^{-1} \nonumber \\
&=& 2 \left ( 1 + \frac{(r_0 - \lambda)^2}{L^2} - \left ( \frac{r_0}{r_0 - \lambda} \right )^{d - 2} \left ( 1 + \frac{r_0^2}{L^2} \right ) \right )^{-1} \nonumber \\
\end{eqnarray}
This does not have a closed form integral, but the interesting effects come from waves that are close to the horizon.
Keeping only the lowest order in $\lambda$,
\begin{eqnarray}
\dot{u} &\approx& -\frac{2r_0}{\lambda} \frac{L^2}{(d - 2)L^2 + dr_0^2} \nonumber \\
u &\approx& -2r_0 \frac{L^2}{(d - 2)L^2 + dr_0^2} \log \left ( \frac{\lambda}{C} \right ) \nonumber \\
&=& -2r_0 \frac{L^2}{(d - 2)L^2 + dr_0^2} \log \left ( \frac{v_0 - v}{C} \right ) \; . \nonumber
\end{eqnarray}

The equation relating $u$ to $v$ has now been found, so we may substitute (\ref{spherical-waves}) into
(\ref{kg-inner-product}) for a surface whose normal derivative is $\partial^r$. We will abbreviate the
$l_1, \dots, l_{d - 1}$ dependence as $l$ and the $\theta_1, \dots, \theta_{d - 1}$ dependence as $\theta$.
\begin{eqnarray}
\alpha_{\omega, l}^{\omega^{\prime}, l^{\prime}} &=& \left < g_{\omega, l}, f_{\omega^{\prime}, l^{\prime}} \right > \nonumber \\
&=& i \int_{\mathbb{S}^{d - 1}} \int_{-\infty}^{v_0} f^*_{\omega^{\prime}, l^{\prime}}(v, \theta) \partial^r g_{\omega, l}(v, \theta) - g_{\omega, l}(v, \theta) \partial^r f^*_{\omega^{\prime}, l^{\prime}}(v, \theta) \textup{d}v \textup{d}\Omega \nonumber \\
&=& i \int_{\mathbb{S}^{d - 1}} \int_{-\infty}^{v_0} f^*_{\omega^{\prime}, l^{\prime}}(v, \theta) \left [ \partial_v + \left ( 1 + \frac{r^2}{L^2} - \frac{\mu}{r^{d - 2}} \right ) \partial_r \right ] g_{\omega, l}(v, \theta) \nonumber \\
&& - g_{\omega, l}(v, \theta) \left [ \partial_v + \left ( 1 + \frac{r^2}{L^2} - \frac{\mu}{r^{d - 2}} \right ) \partial_r \right ] f^*_{\omega^{\prime}, l^{\prime}}(v, \theta) \textup{d}v \textup{d}\Omega \nonumber \\
&=& \frac{i}{d \omega_{d - 1} \sqrt{\omega \omega^{\prime}}^{d - 1}} \int_{\mathbb{S}^{d - 1}} Y_{l}(\theta) Y^*_{l^{\prime}}(\theta) \textup{d}\Omega \nonumber \\
&& e^{i\omega^{\prime}v} \partial_v e^{2i \omega r_0 \frac{L^2}{(d - 2)L^2 + dr_0^2} \log \left ( \frac{v_0 - v}{C} \right )} - e^{2i \omega r_0 \frac{L^2}{(d - 2)L^2 + dr_0^2} \log \left ( \frac{v_0 - v}{C} \right )} \partial_v e^{i \omega^{\prime} v} \textup{d}v \nonumber \\
&=& \frac{\delta_{l_1}^{l_1^{\prime}} \delta_{l_2}^{l_2^{\prime}} \dots \delta_{l_{d - 1}}^{l_{d - 1}^{\prime}}}{d \omega_{d - 1} \sqrt{\omega \omega^{\prime}}^{d - 1}} \int_{-\infty}^{v_0} \left ( \frac{2 \omega r_0}{v_0 - v} \frac{L^2}{(d - 2)L^2 + dr_0^2} + \omega^{\prime} \right ) e^{2i \omega r_0 \frac{L^2}{(d - 2)L^2 + dr_0^2} \log \left ( \frac{v_0 - v}{C} \right )} e^{i \omega^{\prime} v} \textup{d}v \nonumber
\end{eqnarray}
The other Bogoliubov coefficient is found similarly. The only difference is that when integrating two spherical
harmonics without a complex conjugate, we have to use the identity
$Y_{l_1, l_2, \dots, l_{d - 1}} = (-1)^{l_1} Y^*_{-l_1, l_2, \dots, l_{d - 1}}$.
\begin{eqnarray}
\beta_{\omega, l}^{\omega^{\prime}, l^{\prime}} &=& \left < g_{\omega, l}, f^*_{\omega^{\prime}, l^{\prime}} \right > \nonumber \\
&=& \frac{(-1)^{l_1}\delta_{-l_1}^{l_1^{\prime}} \delta_{l_2}^{l_2^{\prime}} \dots \delta_{l_{d - 1}}^{l_{d - 1}^{\prime}}}{d \omega_{d - 1} \sqrt{\omega \omega^{\prime}}^{d - 1}} \int_{-\infty}^{v_0} \left ( \frac{2 \omega r_0}{v_0 - v} \frac{L^2}{(d - 2)L^2 + dr_0^2} - \omega^{\prime} \right ) e^{2i \omega r_0 \frac{L^2}{(d - 2)L^2 + dr_0^2} \log \left ( \frac{v_0 - v}{C} \right )} e^{-i \omega^{\prime} v} \textup{d}v \nonumber
\end{eqnarray}
The integrands above have a branch cut on the real axis because of the $\log \left ( \frac{v_0 - v}{C} \right )$. In
order to manipulate them with complex analysis, it is convenient to displace them with $\pm i \epsilon$. The
$\alpha$ and $\beta$ integrals become
\begin{eqnarray}
\int_{-\infty}^{v_0} \left ( \frac{2 \omega r_0}{v_0 - v} \frac{L^2}{(d - 2)L^2 + dr_0^2} + \omega^{\prime} \right ) e^{2i \omega r_0 \frac{L^2}{(d - 2)L^2 + dr_0^2} \log \left ( \frac{v_0 - v}{C} + i \epsilon \right )} e^{i \omega^{\prime} v} \textup{d}v \label{alpha-integral} \\
\int_{-\infty}^{v_0} \left ( \frac{2 \omega r_0}{v_0 - v} \frac{L^2}{(d - 2)L^2 + dr_0^2} - \omega^{\prime} \right ) e^{2i \omega r_0 \frac{L^2}{(d - 2)L^2 + dr_0^2} \log \left ( \frac{v_0 - v}{C} - i \epsilon \right )} e^{-i \omega^{\prime} v} \textup{d}v \label{beta-integral}
\end{eqnarray}
respectively. The signs for $\pm i \epsilon$ above are dictated by our requirement that $e^{\pm i \omega^{\prime} v}$
vanish at infinity.

\begin{figure}[h]
\includegraphics[scale=0.45]{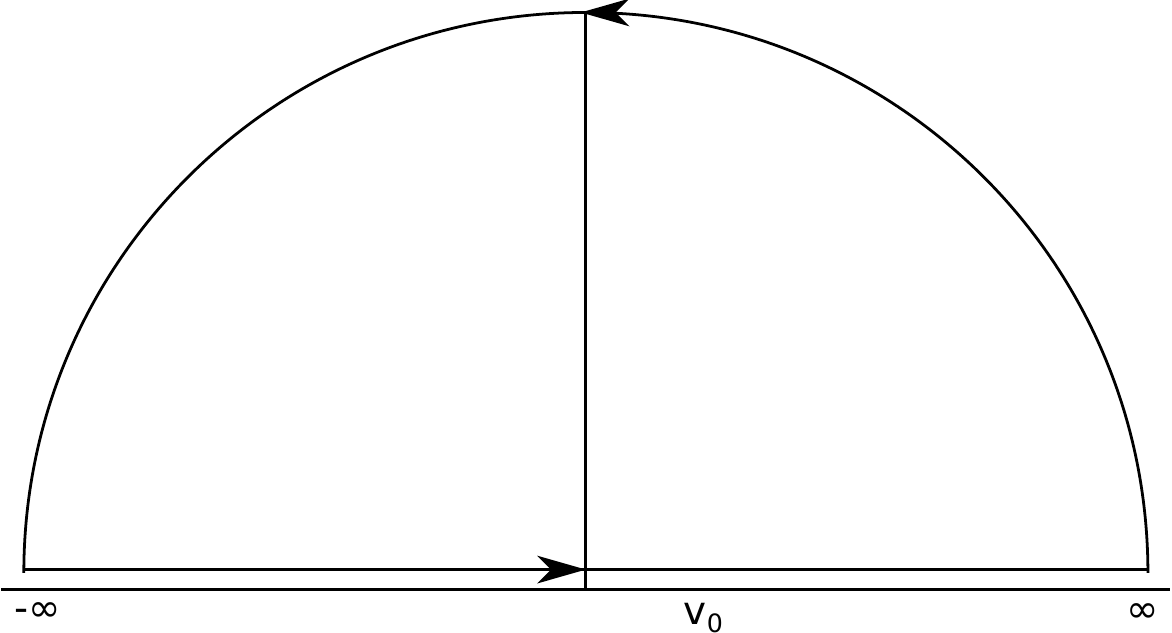}
\caption{For the function with $+i \epsilon$, it is not difficult to show that the semi-circular arc has no
contribution to the integral. Therefore the integral between $-\infty$ and $\infty$ is zero.}
\label{hawking-contour}
\end{figure}
The integral in (\ref{alpha-integral}) vanishes if the domain is the contour in Figure \ref{hawking-contour}. We may
therefore split it up as follows.
\begin{eqnarray}
&& \int_{-\infty}^{v_0} \left ( \frac{2 \omega r_0}{v_0 - v} \frac{L^2}{(d - 2)L^2 + dr_0^2} + \omega^{\prime} \right ) e^{2i \omega r_0 \frac{L^2}{(d - 2)L^2 + dr_0^2} \log \left ( \frac{v_0 - v}{C} + i \epsilon \right )} e^{i \omega^{\prime} v} \textup{d}v \nonumber \\
&=& -\int_{v_0}^{\infty} \left ( \frac{2 \omega r_0}{v_0 - v} \frac{L^2}{(d - 2)L^2 + dr_0^2} + \omega^{\prime} \right ) e^{2i \omega r_0 \frac{L^2}{(d - 2)L^2 + dr_0^2} \log \left ( \frac{v_0 - v}{C} + i \epsilon \right )} e^{i \omega^{\prime} v} \textup{d}v \nonumber \\
&=& -e^{i \omega^{\prime} v_0} \int_{-\infty}^0 \left ( \frac{2 \omega r_0}{v^{\prime}} \frac{L^2}{(d - 2)L^2 + dr_0^2} + \omega^{\prime} \right ) e^{2i \omega r_0 \frac{L^2}{(d - 2)L^2 + dr_0^2} \log \left ( - \left ( \frac{-v^{\prime}}{C} - i \epsilon \right ) \right )} e^{i \omega^{\prime} v^{\prime}} \textup{d}v^{\prime} \nonumber
\end{eqnarray}
In the second step, we made the substitution $v^{\prime} = v_0 - v$. This integral is written with the understanding
that it should be evaluated with a contour in the lower half plane. We must therefore write
$\log(-A) = \log(A) - i\pi$ instead of $\log(-A) = \log(A) + i\pi$ to avoid crossing the branch cut. We will perform
this step and then make another substitution $v^{\prime} = v - v_0$.
\begin{eqnarray}
&& \int_{-\infty}^{v_0} \left ( \frac{2 \omega r_0}{v_0 - v} \frac{L^2}{(d - 2)L^2 + dr_0^2} + \omega^{\prime} \right ) e^{2i \omega r_0 \frac{L^2}{(d - 2)L^2 + dr_0^2} \log \left ( \frac{v_0 - v}{C} + i \epsilon \right )} e^{i \omega^{\prime} v} \textup{d}v \nonumber \\
&=& -e^{i \omega^{\prime} v_0 + 2\pi \omega r_0 \frac{L^2}{(d - 2)L^2 + dr_0^2}} \int_{-\infty}^0 \left ( \frac{2 \omega r_0}{v^{\prime}} \frac{L^2}{(d - 2)L^2 + dr_0^2} + \omega^{\prime} \right ) e^{2i \omega r_0 \frac{L^2}{(d - 2)L^2 + dr_0^2} \log \left ( \frac{-v^{\prime}}{C} - i \epsilon \right )} e^{i \omega^{\prime} v^{\prime}} \textup{d}v^{\prime} \nonumber \\
&=& e^{2i \omega^{\prime} v_0 + 2\pi \omega r_0 \frac{L^2}{(d - 2)L^2 + dr_0^2}} \int_{-\infty}^{v_0} \left ( \frac{2 \omega r_0}{v_0 - v} \frac{L^2}{(d - 2)L^2 + dr_0^2} - \omega^{\prime} \right ) e^{2i \omega r_0 \frac{L^2}{(d - 2)L^2 + dr_0^2} \log \left ( \frac{v_0 - v}{C} - i \epsilon \right )} e^{i \omega^{\prime} v} \textup{d}v \nonumber
\end{eqnarray}
By manipulating (\ref{alpha-integral}), we have turned it into a multiple of (\ref{beta-integral}). This implies the
relation:
\begin{equation}
\left | \alpha_{\omega, l_1 \dots, l_{d - 1}}^{\omega^{\prime}, l_1^{\prime}, \dots, l_{d - 1}^{\prime}} \right | = e^{2\pi \omega r_0 \frac{L^2}{(d - 2)L^2 + dr_0^2}} \left | \beta_{\omega, l_1 \dots, l_{d - 1}}^{\omega^{\prime}, l_1^{\prime}, \dots, l_{d - 1}^{\prime}} \right | \; . \label{alpha-beta-ratio}
\end{equation}

Going back to (\ref{emission}) and (\ref{absorbtion}), the integral
$\sum_{l_1^{\prime}, \dots, l_{d - 1}^{\prime}} \int_0^{\infty} \left | \alpha_{\omega, l_1 \dots, l_{d - 1}}^{\omega^{\prime}, l_1^{\prime}, \dots, l_{d - 1}^{\prime}} \right |^2 - \left | \beta_{\omega, l_1 \dots, l_{d - 1}}^{\omega^{\prime}, l_1^{\prime}, \dots, l_{d - 1}^{\prime}} \right |^2 \textup{d}\omega^{\prime}$
describes what the black hole will absorb. The integral
$\sum_{l_1^{\prime}, \dots, l_{d - 1}^{\prime}} \int_0^{\infty} \left | \beta_{\omega, l_1 \dots, l_{d - 1}}^{\omega^{\prime}, l_1^{\prime}, \dots, l_{d - 1}^{\prime}} \right |^2 \textup{d}\omega^{\prime}$
describes what the black hole will emit. Without worrying about normalization, (\ref{alpha-beta-ratio}) tells us that
the ratio between a mode's absorbtion and emission cross sections is
\begin{equation}
n_{\omega} = \left ( e^{4\pi \omega r_0 \frac{L^2}{(d - 2)L^2 + dr_0^2}} - 1 \right )^{-1} \; . \label{bose-einstein-factor}
\end{equation}
This is precisely the Bose-Einstein thermal factor for a blackbody at temperature
\begin{equation}
T_{\mathrm{BH}} = \frac{1}{4\pi r_0} \left ( d - 2 + d \frac{r_0^2}{L^2} \right ) \label{bh-temperature}
\end{equation}
Had we used an anticommutator in (\ref{absorbtion}), we would have seen the Fermi-Dirac factor for the same
temperature. The famous Bekenstein-Hawking entropy, $S = \frac{A}{4G}$, clearly follows from this if $L$ is large.
For a general $L$, we will use the fact that $\mu = r_0^{d - 2} + \frac{r_0^d}{L^2}$ to write
\begin{equation}
\textup{d}M = \frac{d(d-1)\omega_d}{16\pi G} \textup{d}\mu = \frac{d(d-1)\omega_d}{16\pi G} \left [ (d - 2) r_0^{d - 3} + \frac{d}{L^2} r_0^{d - 1} \right ] \textup{d}r_0 \; . \nonumber
\end{equation}
Then integrating,
\begin{eqnarray}
S_{\mathrm{BH}} &=& \int \frac{\textup{d}E}{T_{\mathrm{BH}}} \nonumber \\
&=& \int \frac{\textup{d}M}{T_{\mathrm{BH}}} \nonumber \\
&=& \frac{d(d-1)\omega_d}{4G} \int r_0^{d - 2} \textup{d}r_0 \nonumber \\
&=& \frac{d\omega_d}{4G} r_0^{d - 1} \nonumber \\
&=& \frac{A}{4G} \label{bh-entropy}
\end{eqnarray}
and we see that the entropy formula is exactly the same in $AdS_{d + 1}$.

Hints that the area of a black hole somehow describes an entropy were already known in 1973 when Bekenstein proposed
the proportionality with a coefficient ``close to'' $\frac{\log 2}{8\pi}$ \cite{bekenstein1}. Apart from improving the
coefficient to $\frac{1}{4}$, Hawking's 1975 paper established that (\ref{bh-entropy}) is the genuine entropy of a
thermal spectrum \cite{hawking}. Modern techniques can derive (\ref{bh-entropy}) much more quickly but at the cost of
once again obscuring the nature of this entropy \cite{page, witten1}. In 1981, Bekenstein noticed that
(\ref{bh-entropy}) is more than just the entropy of a black hole. It is an upper bound on the entropy that \textit{any}
system occupying the same volume can have \cite{bekenstein2}. The argument, which was strong motivation for the
AdS / CFT correspondence \cite{magoo}, is remarkably simple. Suppose that a non-black hole system fills a ball of
radius $r_0$ and has more entropy than $\frac{A}{4G}$. Its mass must be less than that of a black hole with horizon
radius $r_0$ and therefore, it can be turned into said black hole through the addition of mass. Such a procedure would
give the system an entropy of $\frac{A}{4G}$ later on, violating the second law of thermodynamics. Incidentally, two
major open problems in physics are related to the evaporation of black holes. A featureless object described uniquely
by mass, charge and angular momentum should not have entropy and yet we have seen that it contains more entropy than
anything else. While their exact nature remains unknown, some methods for elucidating black hole microstates are
provided by string theory \cite{strominger}. A more serious problem is the \textit{black hole information paradox}.
This is concerned with the fact that a bath of radiation cannot contain information about the formation of a black
hole. If a thermal state is all that a black hole leaves behind after it evaporates, one effectively has a pure state
evolving into a mixed state which is a violation of unitarity. String theoretic resolutions to this have been proposed
as well but are, at the time of writing, much more speculative \cite{braunstein, amps}.

\subsection{Strong coupling}
The AdS / CFT correspondence delivers on a 1974 promise to make strongly coupled $U(N)$ and $SU(N)$ gauge theories
more tractable when $N$ is large \cite{thooft}. More precisely, the quantum gravity theory in the bulk that is dual to
a CFT becomes increasingly classical as we take $N \rightarrow \infty$ with $\lambda = g^2_{\mathrm{YM}} N$ fixed. This
is known as the large $N$, planar or `tHooft limit. The `tHooft coupling $\lambda$ which only needs to be fixed, is the
parameter that would have to be small for the usual Feynman diagram expansion to be valid. When a series of Feynman
diagrams is written down using powers of $\frac{1}{N}$ instead of $g_{\mathrm{YM}}$, the expansion looks very similar
to that of a closed string theory with coupling $g_{\mathrm{s}}$. For this reason, the identification
\begin{equation}
g_{\mathrm{s}} = \frac{\lambda}{4\pi N} \label{gs-identification}
\end{equation}
appears in the duality \cite{magoo}. This small string coupling allows a perturbative calculation to be done in AdS
when the field theory on the boundary is strongly coupled. Some of the coupling strengths not covered by this limit
(\textit{e.g.} large $g_{\mathrm{YM}}$ and $N$) can be explored with the help of string dualities. For instance, a
weak-strong symmetry known as S-duality is often associated with Type IIB string theory \cite{hull}. If IIB in AdS is
equivalent to SYM on the boundary, this statement implies that (\ref{sym}) is invariant under
$g_{\mathrm{YM}} \mapsto \frac{1}{g_{\mathrm{YM}}}$. Although such a result could have been discovered through
holography, it was discovered earlier using some of the same evidence that led to the correspondence
\cite{tseytlin, green}. In addition, we should note that even if one believes AdS / CFT, the idea of S-duality holding
for Type IIB string theory is also a conjecture \cite{hull}.

Using classical gravity to approximate field theories in the strong coupling regime has become the most widely explored
aspect of the AdS / CFT correspondence \cite{csaki, sachdev}. For this application, questions about whether string
theory is realized in nature, are irrelevant. We will go through an example of this duality, whereby the interacting
spectrum of (\ref{sym}) can be understood through our seemingly unrelated calculations regarding free field theories
and black holes.

\subsubsection{Gauge theory phases}
When compactified on a sphere, there are at least four interesting phases posessed by Super Yang-Mills. The
transitions in and out of these phases are gradual, as they must be for a theory with finitely many fields. However,
the transitions may become sharp in the strict $N \rightarrow \infty$ limit. Following our pattern above, we will
give an expression for the entropy of each phase in order of increasing energy.

The first thing we need to know is that in the original version of the correspondence, the bulk geometry is
$AdS_5 \times \mathbb{S}^5$ where the length scale of $AdS_5$ and the radius of $\mathbb{S}^5$ are equal
\cite{maldacena}. This radius, which we will call $L$ is given by the duality prescription as
\begin{equation}
L^4 = 4\pi g_{\mathrm{s}} \alpha^{\prime 2} N \; . \label{l4-identification}
\end{equation}
There is also a radius for the $\mathbb{S}^3$ of the the field theory, which we will call $R$. It is natural to
compare the dimensionless energies of the string theory $EL$ to ``some multiple'' of the dimensionless energies of the
field theory $ER$. By studying the Klein-Gordon equation (\ref{kg-equation}) for a graviton propagating in
$AdS_5 \times \mathbb{S}^5$, one may show that the second-lowest energy it can have is
\begin{equation}
E_1 = \frac{1}{L} \; . \label{first-excited-graviton}
\end{equation}
While this would not be the case for a general gauge theory, a highly supersymmetric theory like SYM has some excited
state energies that can be computed without the correspondence. This is a harder calculation but an analysis of
chiral primary operators \cite{magoo, witten2} tells us that $E_1 = \frac{1}{R}$ on the field theory side. The multiple
in question is therefore 1 and we will be able to replace $EL$ with $ER$ in what follows.

The low energy behaviour of the bulk is described by a free gas of strings in their worldsheet ground states.
Refering to (\ref{spectrum-table}), this is a gas of 128 bosons and 128 fermions --- essentially
gravitons and their superpartners. The free field theory result (\ref{partfunc}) includes a volume $V$, which is
only well defined if there is a clear separation between space and time, \textit{i.e.}
\begin{equation}
ds^2 = -dt^2 + g_{ij} dx^i dx^j \nonumber
\end{equation}
where $g_{ij}$ is a Riemanian metric. Since the line element
\begin{equation}
ds^2 = -\left ( 1 + \frac{r^2}{L^2} \right ) dt^2 + \left ( 1 + \frac{r^2}{L^2} \right )^{-1} dr^2 + r^2 d\Omega_3^2 + L^2 d\Omega_5^2 \label{ads5-s5}
\end{equation}
is not in this form, the brute force calculation of the partition function would have to start with the Klein-Gordon
equation. Solving the relevant Klein-Gordon equation is certainly a useful exercise. In addition to energy eigenvalues
like (\ref{first-excited-graviton}), it would allow us to derive a bound on the mass that any particle in $AdS$ must
satisfy \cite{bf1, bf2}. However, there is another method that can tell us the appropriate $V$ more quickly.
Performing a conformal transformation on (\ref{ads5-s5}), we may turn it into
\begin{equation}
ds^2 = -dt^2 + \left ( 1 + \frac{r^2}{L^2} \right )^{-2} dr^2 + \left ( 1 + \frac{r^2}{L^2} \right )^{-1} \left ( r^2 d\Omega_3^2 + L^2 d\Omega_5^2 \right ) \; . \label{rescaled-ads5-s5}
\end{equation}
The massless version of (\ref{kg-equation}) (called the minimally coupled Klein-Gordon equation) is not
invariant under such a rescaling, but the \textit{conformally coupled} Klein-Gordon equation is. This equation is
\begin{equation}
\nabla_{\mu} \nabla^{\mu} \phi + \frac{d - 1}{4d} \mathcal{R} \phi = 0 \label{kgr-equation}
\end{equation}
where $\mathcal{R}$ is the Ricci scalar. The volume associated with the (\ref{rescaled-ads5-s5}) metric would
therefore appear in the partition function for this graviton gas in the conformally coupled case.
\begin{eqnarray}
V &=& \int_0^{\infty} \iint_{\mathbb{S}^3 \times \mathbb{S}^5} \left ( 1 + \frac{r^2}{L^2} \right )^{-5} L^5 r^3 \textup{d}\Omega_3 \textup{d}\Omega_5 \textup{d}r \nonumber \\
&=& 15 \omega_3 \omega _5 L^5 \int_0^{\infty} \left ( 1 + \frac{r^2}{L^2} \right )^{-5} r^3 \textup{d}r \nonumber \\
&=& \frac{15}{24} \omega_3 \omega_5 L^9 \nonumber
\end{eqnarray}
Because the difference between minimal and conformal coupling only appears beyond the leading order thermodynamics
\cite{behan, magoo}, it is sufficient to substitute $d = 9$, $s = s^* = 128$ and the $V$ above. This yields
\begin{eqnarray}
S_1 &=& \left ( \frac{10^{10}}{9^9} \frac{15}{24} \frac{9! \omega_3 \omega_5 \omega_9}{(2\pi)^9} 128 \left ( \zeta(10) + \zeta^*(10) \right ) E^9 L^9 \right )^{\frac{1}{10}} \nonumber \\
&=& 10 \left ( 9! \omega_3 \omega_5 \omega_9 \frac{15}{24} \frac{1023}{4} \zeta(10) \right )^{\frac{1}{10}} \left ( \frac{EL}{18\pi} \right )^{\frac{9}{10}} \nonumber \\
&=& 10 \left ( \frac{2728}{9355} \pi^8 \right )^{\frac{1}{10}} \left ( \frac{ER}{9} \right )^{\frac{9}{10}} \label{sym-entropy1}
\end{eqnarray}
as the entropy of the lowest energy phase.

As energy increases, the infinite tower of worldsheet vibrations becomes important and the entropy enters the
Hagedorn regime. Starting with (\ref{string-entropy}), we just need to use the (\ref{gs-identification}) and
(\ref{l4-identification}) identifications to get the entropy for the second phase in terms of gauge theory parameters:
\begin{eqnarray}
S_2 &=& 2\pi \sqrt{2 \alpha^{\prime}} E \nonumber \\
&=& 2\pi \left ( \frac{L^4}{\pi g_s N} \right )^{\frac{1}{4}} E \nonumber \\
&=& 2\pi \left ( \frac{4}{\lambda} \right )^{\frac{1}{4}} ER \; . \label{sym-entropy2}
\end{eqnarray}

As with any proper string theory, the $AdS_5 \times \mathbb{S}^5$ background is not static. It receives a
backreaction from stringy states that becomes more significant as the energy increases. The highest energy phases of
Super-Yang Mills will therefore involve Newton's constant $G$. The formula
\begin{equation}
G_5 L^5 = 8\pi^3 g_{\mathrm{s}}^2 \alpha^{\prime 4} \label{g5-identification}
\end{equation}
is the last piece of the correspondence that we need \cite{burgess}. We mentioned previously that Einstein's
equations (and Newton's constant in particular) should be the same in all dimensions. Thus, it may seem strange to
refer to a five-dimensional gravitational constant $G_5$. The explanation is that $G_5$ is not the gravitational
constant at all, but rather an illusion created by the presence of compact dimensions. The true $G$ appears as a
prefactor in the Einstein-Hilbert action
\begin{equation}
\mathcal{S} = \frac{1}{16\pi G} \int_{AdS_5 \times \mathbb{S}^5} (\mathcal{R} - 2\Lambda) \sqrt{-g} \textup{d}x \; . \nonumber
\end{equation}
If the radius of the sphere is small enough, a macroscopic observer only sees $AdS_5$. Integrating out the
$\mathbb{S}^5$ and looking at the prefactor once again will tell us the relation between $G_5$ and $G$. Using the fact
that Ricci scalars add for direct product manifolds
\begin{eqnarray}
\mathcal{S} &=& \frac{1}{16\pi G} \int_{AdS_5 \times \mathbb{S}^5} (\mathcal{R}_{AdS_5} - 2\Lambda) \sqrt{-g} \textup{d}x + \frac{1}{16\pi G} \int_{AdS_5 \times \mathbb{S}^5} \mathcal{R}_{\mathbb{S}^5} \sqrt{-g} \textup{d}x \nonumber \\
\mathcal{S}_{\mathrm{eff}} &=& \frac{1}{16\pi G} \int_{\mathbb{S}^5} \sqrt{g_{\mathbb{S}^5}} \textup{d}x \int_{AdS_5} (\mathcal{R}_{AdS_5} - 2\Lambda) \sqrt{-g_{AdS_5}} \textup{d}x + \mathcal{S}_{\mathrm{shift}} \nonumber \\
&=& \frac{6 \omega_6 L^5}{16 \pi G} \int_{AdS_5} (\mathcal{R}_{AdS_5} - 2\Lambda) \sqrt{-g_{AdS_5}} \textup{d}x + \mathcal{S}_{\mathrm{shift}} \; . \nonumber
\end{eqnarray}
We have turned the action into an effective action by evaluating part of it. This makes it clear that
$G_5 = \frac{G}{6 \omega_6 L^5}$. The fact that these dimensionally reduced Newton constants are generally much
smaller than $G$ has led to the hypothesis that the apparent strength of gravity increases when the distance is very
small. Indeed, proponents of extra-dimension phenomenology have discussed the possibility of forming black holes
at the LHC \cite{banks, giddings, dimopoulos}.

With these constants in hand, we need to calculate the entropy associated with the geometry that develops in the
third phase. Since entropy increases with energy, it is only logical that our spacetime should eventually achieve the
geometry that has a monopoly on entropy --- that of a black hole. The mysterious microstates of this black hole can
be put in a one-to-one correspondence with the well defined microstates of the CFT. When the event horizon $r_0$ first
forms, it is smaller than the radius $L$. It is therefore a good approximation to describe it with the Schwarzschild
solution (\ref{schwarzschild}) involving all ten spacetime dimensions. Using the entropy formula (\ref{bh-entropy}),
\begin{eqnarray}
S_3 &=& \frac{9 \omega_9}{4G} r_0^8 \nonumber \\
&=& \frac{9 \omega_9}{4G} \mu^{\frac{8}{7}} \nonumber \\
&=& \frac{9}{4} \left ( \frac{G}{\omega_9} \right )^{\frac{1}{7}} \left ( \frac{2\pi E}{9} \right )^{\frac{8}{7}} \nonumber \\
&=& \frac{9}{4} \left ( 6 L^5 G_5 \frac{\omega_6}{\omega_9} \right )^{\frac{1}{7}} \left ( \frac{2\pi E}{9} \right )^{\frac{8}{7}} \nonumber \\
&=& \frac{9}{4} \left ( \frac{3\pi \omega_6}{N^2 \omega_9} \right )^{\frac{1}{7}} \left ( \frac{2\pi EL}{9} \right )^{\frac{8}{7}} \nonumber \\
&=& \frac{9}{4} \left ( \frac{1890}{N^2} \right )^{\frac{1}{7}} \left ( \frac{\pi ER}{9} \right )^{\frac{8}{7}} \; . \label{sym-entropy3}
\end{eqnarray}
As the black hole grows to a radius $r_0 \gg L$, the five small dimensions become negligible allowing us to use the
asymptotically AdS black hole (\ref{schwarzschild-ads}). This also makes it a good approximation to say
$\mu = r_0^2 \left ( 1 + \frac{r_0^2}{L^2} \right ) \approx \frac{r_0^4}{L^2}$. Inserting this into (\ref{bh-entropy}),
\begin{eqnarray}
S_4 &=& \frac{\omega_4}{G_5} r_0^3 \nonumber \\
&=& \frac{\omega_4}{G_5} \mu^{\frac{3}{4}} \nonumber \\
&=& \left ( \frac{\omega_4}{G_5} \right )^{\frac{1}{4}} \left ( \frac{4\pi EL^2}{3} \right )^{\frac{3}{4}} \nonumber \\
&=& \left ( \frac{2N^2 \omega_4}{\pi L^3} \right )^{\frac{1}{4}} \left ( \frac{4\pi EL^2}{3} \right )^{\frac{3}{4}} \nonumber \\
&=& \pi \sqrt{N} \left ( \frac{4}{3} ER \right )^{\frac{3}{4}} \; . \label{sym-entropy4}
\end{eqnarray}
This matches the behaviour that a conformal theory must have at high energies (\ref{cft-entropy}).

We have yet to give estimates for the energy ranges where these phases are valid. Prefactors for these energies would
be suspicious due to the gradual nature of the phase transitions. We will therefore only keep factors that may be
comparable to $N^2$. To determine when the Hagedorn phase becomes important, we should set $E$ to the mass of an
excited string. From (\ref{mass-operator}), we see that this is of order
$\frac{1}{\sqrt{\alpha^{\prime}}} \propto \frac{\lambda^{\frac{1}{4}}}{R}$.

Strings have a characteristic length and a black hole with this length as its horizon radius has a characteristic
energy. When the energy of a string gas exceeds this, it is expected to collapse to the small black hole that we
discussed before. Of course there are some non-black hole geometries having energies of this magnitude
(\textit{e.g.} a giant graviton \cite{mcgreevy}) but these are ``rare''. This is consistent with the
``heat death'' proposal in which a black hole is the inevitable final state of a system that evolves via thermal
fluctuations. An equivalent statement on the CFT side is that as the dimensions of gauge invariant operators increase,
the fraction of them that describe black holes approaches unity \cite{larjo}. The transition for this black hole
``probably forming'' can be found by checking when the Hagedorn entropy becomes comparable to the small black hole
entropy. Setting (\ref{sym-entropy2}) equal to (\ref{sym-entropy3}), this energy is of order
$\frac{N^2}{\lambda^{\frac{7}{4}} R}$.

Finally, the midpoint between the small black hole and the large black hole occurs when $r_0 = R$. Expressing the
event horizon radius in terms of the mass, $E \propto \frac{R^7}{G} \propto \frac{R^2}{G_5} \propto \frac{N^2}{R}$.
Putting this together we see that the entropy for strongly coupled SYM is given by
\begin{equation}
S(E) =
\begin{cases}
10 \left ( \frac{2728}{9355} \pi^8 \right )^{\frac{1}{10}} \left ( \frac{ER}{9} \right )^{\frac{9}{10}} & ER \ll \lambda^{\frac{1}{4}} \\
2\pi \left ( \frac{4}{\lambda} \right )^{\frac{1}{4}} ER & \lambda^{\frac{1}{4}} \ll ER \ll \lambda^{-\frac{7}{4}} N^2 \\
\frac{9}{4} \left ( \frac{1890}{N^2} \right )^{\frac{1}{7}} \left ( \frac{\pi ER}{9} \right )^{\frac{8}{7}} & \lambda^{-\frac{7}{4}} N^2 \ll ER \ll N^2 \\
\pi \sqrt{N} \left ( \frac{4}{3} ER \right )^{\frac{3}{4}} & N^2 \ll ER
\end{cases}
\; . \label{sym-entropy}
\end{equation}
Notice that if we were to find the entropy of free Super-Yang Mills by substituting $d = 3$ and $s = s^* = 8N^2$ in
(\ref{partfunc}), the result would be $\frac{4}{3} \pi \sqrt{N} \left ( ER \right )^{\frac{3}{4}}$. The entropies
differ by a factor of $\left ( \frac{4}{3} \right )^{\frac{1}{4}}$ or equivalently, the free energies differ by a
factor of $\frac{4}{3}$. Writing
\begin{equation}
F = -\frac{1}{6} h(\lambda) \pi^2 N^2 V T^4 \label{sym-free-energy4}
\end{equation}
with $h(0) = 1$ and $\lim_{\lambda \rightarrow \infty} h(\lambda) = \frac{3}{4}$, various authors have studied how
$h$ interpolates between these limits using curvature corrections on the string theory side \cite{gubser} and loop
diagrams on the field theory side \cite{fotopoulos, kim}. It was later found that interpolating between weakly coupled
and strongly coupled free energy is not as simple as multiplying by $h$. Corrections to (\ref{sym-free-energy4})
involving $T^2$ need to be multiplied by different functions of the `tHooft coupling \cite{burgess}.

\subsubsection{Plasma balls}
The microcanonical entropy of strongly coupled SYM on $\mathbb{S}^3$ (\ref{sym-entropy}) is a formula that we will use
repeatedly. Part of its derivation relied on the fact that the theory's dual description involved black holes radiating
a thermal spectrum. The goal of this thesis is to argue for the converse: an arbitrary field theory with an entropy
sufficiently similar to (\ref{sym-entropy}) exhibits dynamics that are indicative of black hole formation and
evaporation. There is a large class of field theory solutions, called plasma balls, that have been shown to be of this
type \cite{aharony}. Most studies of them are numerical \cite{aharony, figueras, bhardwaj} but at least one has been
constructed analytically \cite{milanesi}.

Consider the canonical phases of Super Yang-Mills found by fixing the temperature instead of the energy. We may
differentiate the entropy in (\ref{sym-entropy}) to plot $\beta$ as a function of $E$.
\begin{figure}[h]
\includegraphics[scale=0.45]{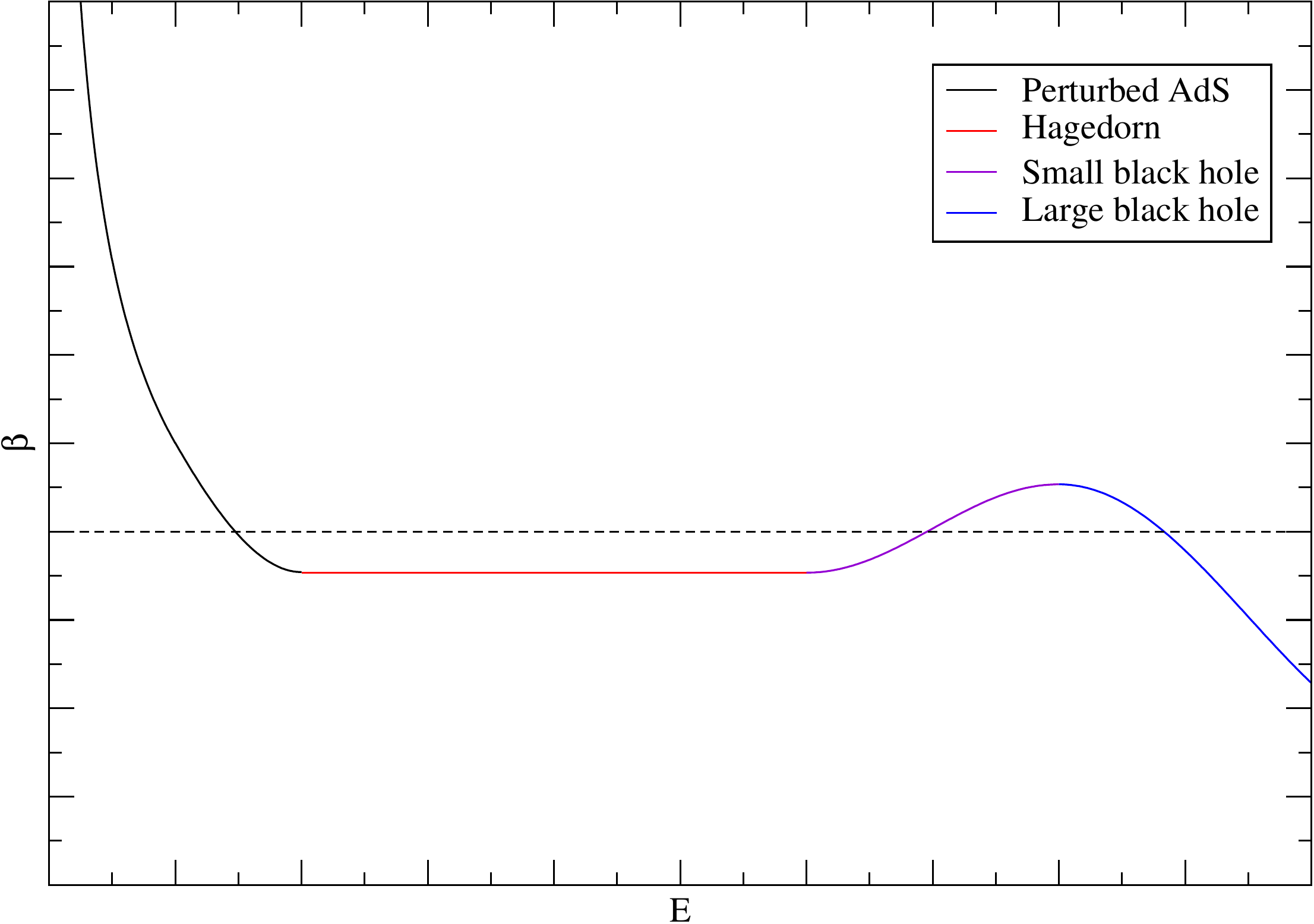}
\caption{For an inverse temperature like the one shown, the gravity side must choose the geometry that minimizes
the free energy.}
\label{hawking-page}
\end{figure}
At low temperatures, the system must be in the graviton gas phase. One could raise the temperature (lower the
dotted line) all the way to the Hagedorn temperature at which the canonical ensemble ceases to exist but the
most interesting situation occurs for an intermediate value where there is competition between three phases. From
(\ref{sym-entropy}), a straightforward calculation of the graviton gas free energy yields
\begin{equation}
F = -\frac{2728}{9355} \pi^8 R^9 T^{10} \label{sym-free-energy1}
\end{equation}
and we have already written the large black hole free energy (\ref{sym-free-energy4}). If we were to calculate the
small black hole free energy in the same way, we would find that it is positive, so (\ref{sym-free-energy1}) and
(\ref{sym-free-energy4}) are the only ones we need. Setting them equal, we find a first order phase transition at
\begin{equation}
T_{\mathrm{D}} = \frac{1}{R} \left ( \frac{9355 N^2}{16368 \pi^5} \right )^{\frac{1}{6}} \; . \label{dec-temperature}
\end{equation}
Were it not for the complication of the internal manifold $\mathbb{S}^5$, this would be the Hawking-Page transition
\cite{page} showing that a sufficiently large black hole in AdS can come to equilibrium with the radiation it emits.
Since the energy, entropy and temperature of a black hole are all known in terms of its event horizon radius $r_0$,
\begin{eqnarray}
F &=& E - TS \nonumber \\
&=& \frac{d \omega_d}{16 \pi G} r_0^{d - 2} \left ( 1 - \frac{r_0^2}{L^2} \right ) \; . \nonumber
\end{eqnarray}
When $r_0 < L$, this is minimized for an $r_0$ that rolls to zero. When $r_0 > L$, this is minimized for as large an
$r_0$ as possible. Substituting $r_0 = L$ into (\ref{bh-temperature}), we find
\begin{equation}
T_{\mathrm{HP}} = \frac{d - 1}{2\pi L} \; . \label{hp-temperature}
\end{equation}
This phenomenon on the field theory side has the interpretation of a deconfinement phase transition related to the
scale $R$. As $R \rightarrow \infty$, the temperature (\ref{dec-temperature}) vanishes and there is no confinement as
expected for a CFT in Minkowski space. Since the confining theory of greatest physical interest
(quantum chromodynamics) lives in infinite volume, it has little in common with Super Yang-Mills on $\mathbb{S}^3$.
A holographic study of QCD requires one to introduce a scale to SYM in a more drastic way.

This can be done by compactifying some but not all of the directions in a Minkowski CFT. Witten's model \cite{witten1}
\textit{e.g.} compactifies SYM on a Scherk-Schwarz circle --- $\mathbb{S}^1$ with antiperiodic fermions. Since the
other directions are extended, it is helpful to rewrite (\ref{ads}) so that they manifestly appear as Minkowski space.
Global AdS is large enough for this to be done several times yielding disjoint patches separated by co-ordinate
singularities. A given patch has the following metric known as Poincar\'e AdS
\begin{equation}
ds^2 = \frac{L^2}{z^2} \left [ -d\tau^2 + dz^2 + dx_i dx^i \right ] \label{poincare}
\end{equation}
where the boundary is located at $z = 0$. The transformation
\begin{eqnarray}
t &=& \arctan \left ( \frac{2L\tau}{L^2 + z^2 + x^2 - \tau^2} \right ) \nonumber \\
r &=& \frac{L}{z} \sqrt{z^2 + \tau^2 + \frac{1}{4L^2} \left ( L^2 + z^2 + x^2 - \tau^2 \right )^2} \label{global-to-poincare} \\
\sin \theta_i \dots \sin \theta_{i - 1} \cos \theta_i &=& \frac{x_i}{\sqrt{z^2 + \tau^2 + \frac{1}{4L^2} \left ( L^2 + z^2 + x^2 - \tau^2 \right )^2}} \nonumber \\
\end{eqnarray}
converts between the global and Poincar\'e metrics \cite{magoo}. Some sources assume $r \gg L$ before deriving
(\ref{poincare}) in order to write the simpler transformation $z = \frac{L^2}{r}$ \cite{schaposnik}. This gives the
false impression that (\ref{poincare}) is only approximately equal to a patch of AdS. Applying
(\ref{global-to-poincare}) to (\ref{schwarzschild-ads}) gives another form of the AdS black hole:
\begin{equation}
ds^2 = \frac{L^2}{z^2} \left [ -\left ( 1 - \frac{z^d}{z_0^d} \right ) d\tau^2 + \left ( 1 - \frac{z^d}{z_0^d} \right )^{-1} dz^2 + dx_i dx^i \right ] \; . \label{schwarzschild-poincare}
\end{equation}
If we compactify a spatial direction and let $j$ take on fewer values than $i$, the resulting metric is:
\begin{equation}
ds^2 = \frac{L^2}{z^2} \left [ -\left ( 1 - \frac{z^d}{z_0^d} \right ) d\tau^2 + \left ( 1 - \frac{z^d}{z_0^d} \right )^{-1} dz^2 + dx_j dx^j + d\theta^2 \right ] \; . \label{witten-model}
\end{equation}
An interesting procedure, that would not have worked for any of the previous metrics, is available to be used on
(\ref{witten-model}). We may Euclideanize, exchange the $\theta$ circle with the $\tau$ circle and switch back to a
Minkowskian signature. This yields a spacetime known as the AdS soliton without us having to solve Einstein's equations
again.
\begin{equation}
ds^2 = \frac{L^2}{z^2} \left [ -d\tau^2 + \left ( 1 - \frac{z^d}{z_0^d} \right )^{-1} dz^2 + dx_j dx^j + \left ( 1 - \frac{z^d}{z_0^d} \right ) d\theta^2 \right ] \label{ads-soliton}
\end{equation}
The deconfinement we saw earlier was a transition between two spacetimes that shared the same boundary: empty AdS and
the AdS black hole. We now have (\ref{witten-model}) and (\ref{ads-soliton}) competing for the same boundary. However,
the horizon position $z_0$ has a very different interpretation in (\ref{ads-soliton}) because it causes
the $\theta$ circle to shrink to zero size. A horizon that observers can safely cross changes the signature according
to $(-, +, +, \dots, +, +) \mapsto (+, -, +, \dots, +, +)$. On the other hand, if we allowed $z > z_0$, we would see
$(-, +, +, \dots, +, +) \mapsto (-, -, +, \dots, +, -)$. Because this is a Lorentzian theory, $z_0$ is simply a point
where the spacetime ends. This enduring scale, called the infrared wall, is what leads to a mass gap \cite{wiseman}.

\begin{figure}[h]
\centering
\subfloat[][Separate solutions]{\includegraphics[scale=0.4]{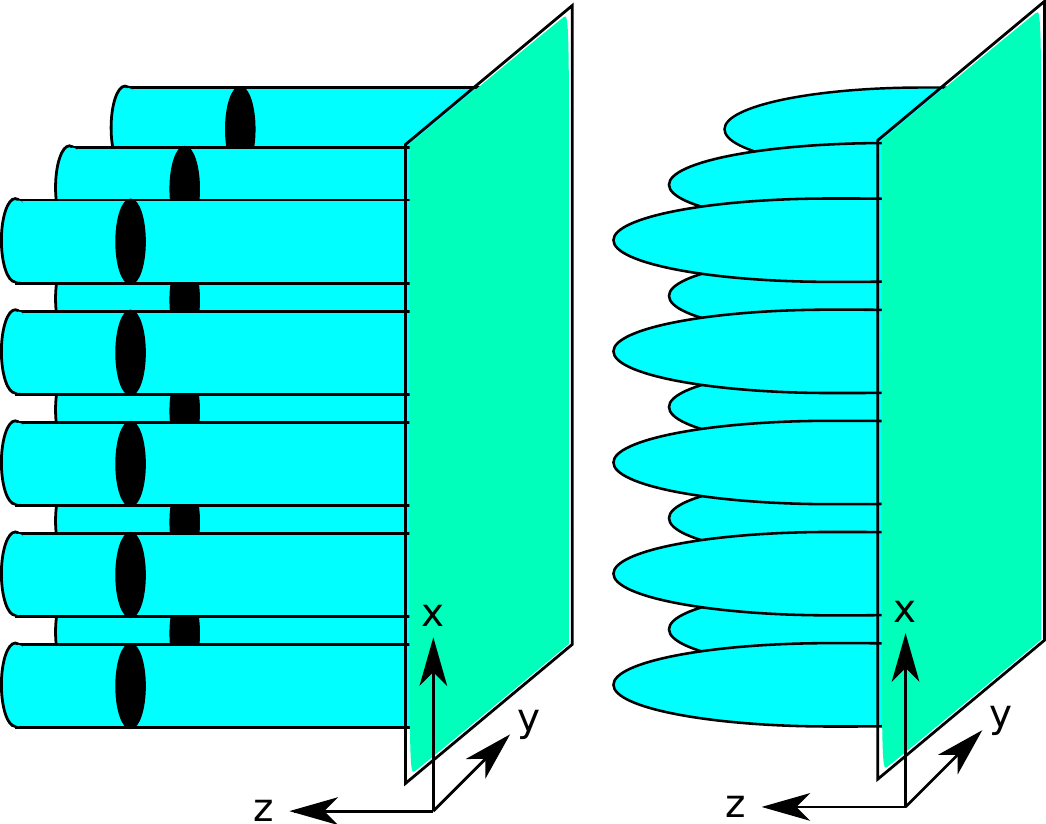}} \qquad
\subfloat[][Interpolation]{\includegraphics[scale=0.4]{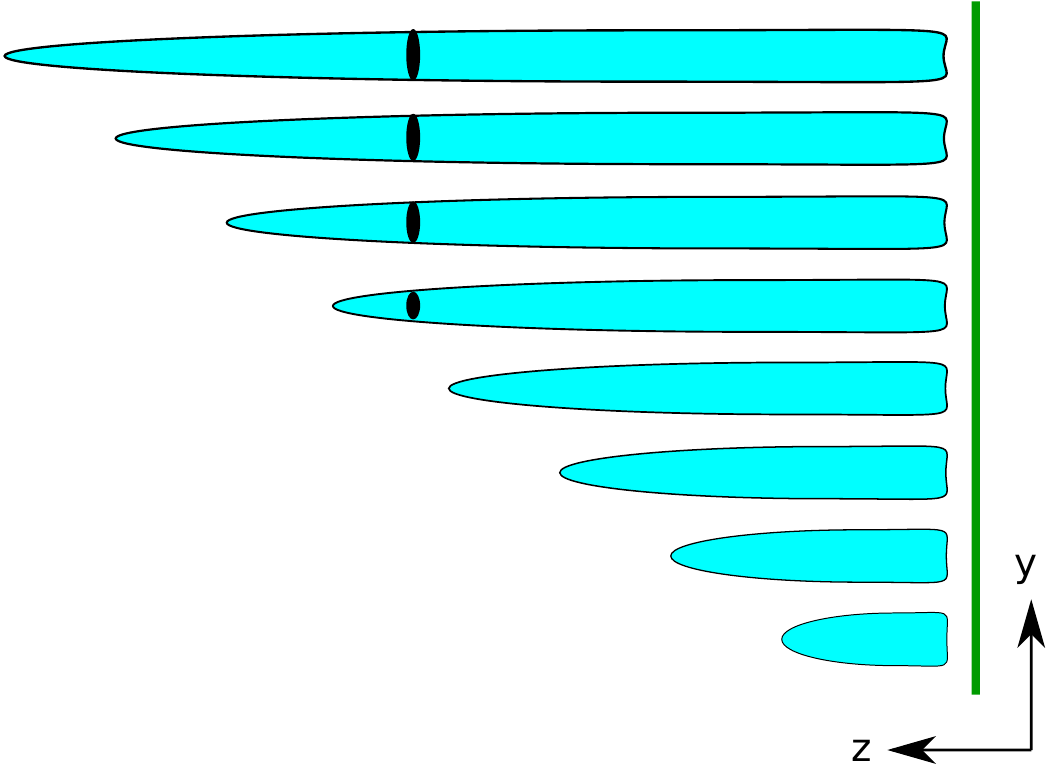}}
\caption{The green sheet is the boundary at $z = 0$. While Poincar\'e AdS looks like empty space ending on this
boundary, the black hole and soliton solutions are more interesting. Instead of enclosing a region of finite volume,
the horizon is a sheet at $z = z_0$ like that of a black brane. To remind us that there is a
Scherk-Schwarz circle at each point, we have drawn the horizon as a set of black circles inside cylinders that extend
infinitely far to the left. These cylinders become cigar shaped regions in the soliton solution that has the infrared
wall. On the right, we show an interpolation between these behaviours. For $y = \infty$, the IR wall does not exist.
For finite positive $y$ it is behind the horizon and for negative $y$ it is in front.}
\label{domain-wall}
\end{figure}
The Witten model with these two backgrounds is the typical arena for seeing plasma balls. These were conjectured
\cite{aharony} based on the observation that stable domain walls should exist between solutions like
(\ref{witten-model}) and (\ref{ads-soliton}). Roughly, such a domain wall is constructed by choosing a special
direction $y$ and making $z_0$ a function of $y$. Choosing this function appropriately, the bulk metric can be
made to look like the black hole at $y = \infty$ and the soliton at $y = -\infty$. This solution, which cannot be
found analytically, may look like the one in Figure \ref{domain-wall}. In order for it to be stable, the pressure of
the deconfined phase must be small enough to balance the domain wall tension at some temperature. Intuitive arguments
for this are given in \cite{aharony} with the final confirmation being numerical.
\begin{figure}[h]
\centering
\subfloat[][Domain wall]{\includegraphics[scale=0.4]{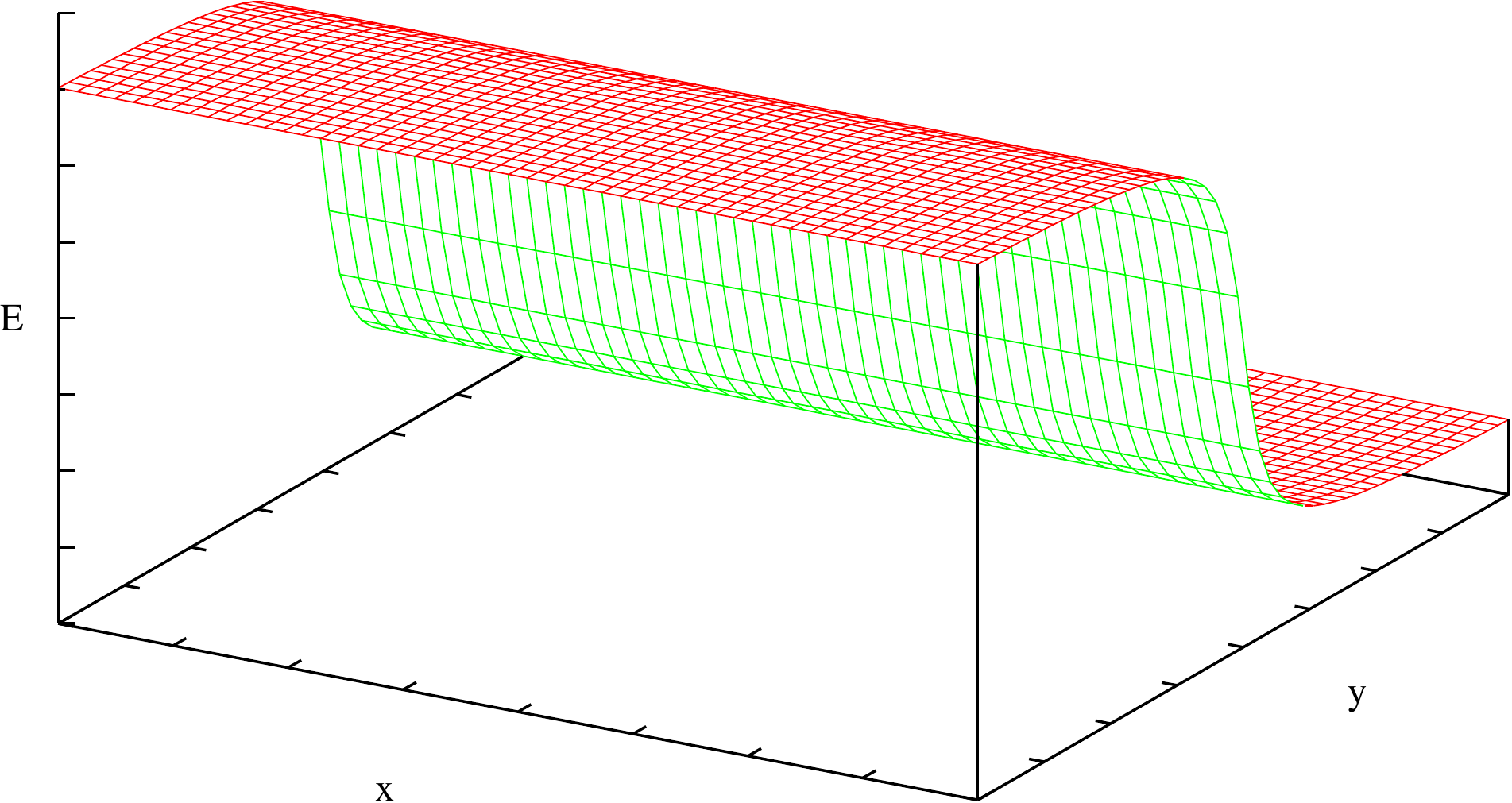}} \qquad
\subfloat[][Plasma ball]{\includegraphics[scale=0.4]{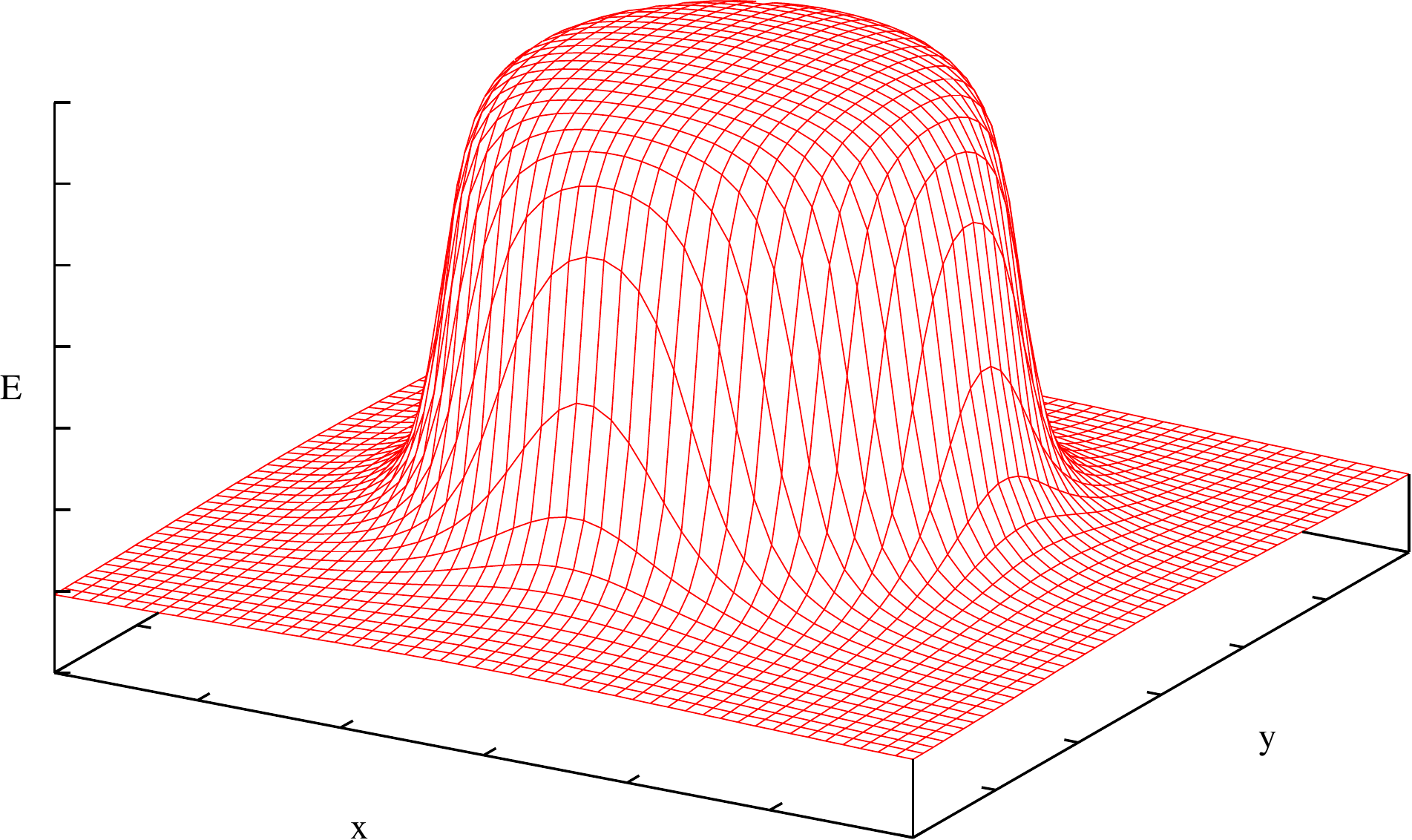}}
\caption{On the field theory side, the domain wall shows up as a sudden jump in the energy density. The minimum is
close to zero while the maximum is the expected energy at the deconfinement temperature. The existence of a plasma
ball is the additional assumption that we can make this domain wall ``wrap around''. The decay of the resulting
object might look like diffusion governed by the heat equation.}
\label{wrapping-domain-wall}
\end{figure}
This process can be repeated to further localize the domain wall. Instead of going from the confining vacuum at
$y = -\infty$ to the deconfined plasma at $y = \infty$, one may change the solution so that it goes from confined to
deconfined and back \cite{wiseman}. This can also be done using directions other than $y$ to make the area of the black
hole horizon finite. The black hole made in this way decays in a process that looks like some combination of shrinking
in $y$ and hitting the infrared wall in $z$. The field theory state dual to this black hole near the IR is called a
plasma ball. The dual decay process consists of hadrons leaving the ball and travelling outwards. Because they are
travelling into a confining vacuum, they must be color singlets, leading to a lifetime proportional to $N^2$
\cite{aharony}.

When deriving (\ref{sym-entropy}), the black holes we discussed were dual to energy eigenstates of SYM on
$\mathbb{S}^3$. In analogy with free theories (whose momentum eigenstates are completely delocalized), these states
have uniform energy density of order $\frac{N^2}{R^4}$ on the whole sphere. The situation is very different for plasma
balls. They have non-uniform energy densities like those in Figure \ref{wrapping-domain-wall} because their dual black
holes come from an interpolation of gravity saddle points. When we construct plasma balls (in a completely different
manner), we should keep in mind that small compact directions and infrared walls are likely to appear in the
corresponding geometries. As a check, it is interesting to see what goes wrong when trying to construct a plasma ball
for Super-Yang Mills on $\mathbb{S}^3$. In principle, one could prepare a state in the CFT that has Figure
\ref{wrapping-domain-wall}'s energy density at $t = 0$. Rather than a thermalization process dual to Hawking radiation,
this state's future is governed by the phenomenon of collective oscillations \cite{freivogel}. When quantized on a
sphere, the generators of conformal transformations obey the same commutation relations as traditional raising and
lowering operators. Rewriting them as $L^i_+$ and $L^i_-$, \cite{freivogel} constructed undamped oscillating states
by applying a function of them to a density matrix:
$\rho \mapsto \mathcal{N} g \left ( L^1_{\pm}, \dots, L^d_{\pm} \right ) \rho g^{\dagger} \left ( L^1_{\pm}, \dots, L^d_{\pm} \right )$.
In analogy with coherent states of the harmonic oscillator, explicit functions were given such as the simplest one:
\begin{equation}
g \left ( L^i_+, L^i_- \right ) = e^{\alpha L^i_+ + \beta L^i_-} \; . \label{coherent-state}
\end{equation}
When $\alpha = -\beta^*$, (\ref{coherent-state}) is unitary and $\mathcal{N} = 1$, but \cite{freivogel} gave
normalization constants for other $\alpha$ and $\beta$ as well. Crucially, the AdS isometry dual to
(\ref{coherent-state}) is no more complicated than a boost. This allows its effect on strongly coupled states to be
found with the AdS / CFT correspondence. The example considered in \cite{freivogel} starts with a three-dimensional
spacetime known as the BTZ black hole \cite{btz}:
\begin{equation}
ds^2 = -\left ( \frac{r^2}{R^2} - \frac{r_0^2}{R^2} \right ) dt^2 + \left ( \frac{r^2}{R^2} - \frac{r_0^2}{R^2} \right )^{-1} dr^2 + \frac{r^2}{R^2} dx^2 \; . \label{btz-ads}
\end{equation}
As before, the boundary theory dual to this has a uniform stress-energy tensor when compactified on the
$\frac{x}{R} \sim \frac{x}{R} + 2\pi$ sphere:
\begin{eqnarray}
T_{tt} &=& T_{xx} = \frac{r_0^2}{16\pi GR^3} \nonumber \\
T_{tx} &=& T_{xt} = 0 \; . \label{unboosted-tensor}
\end{eqnarray}
\begin{figure}[h]
\includegraphics[scale=0.45]{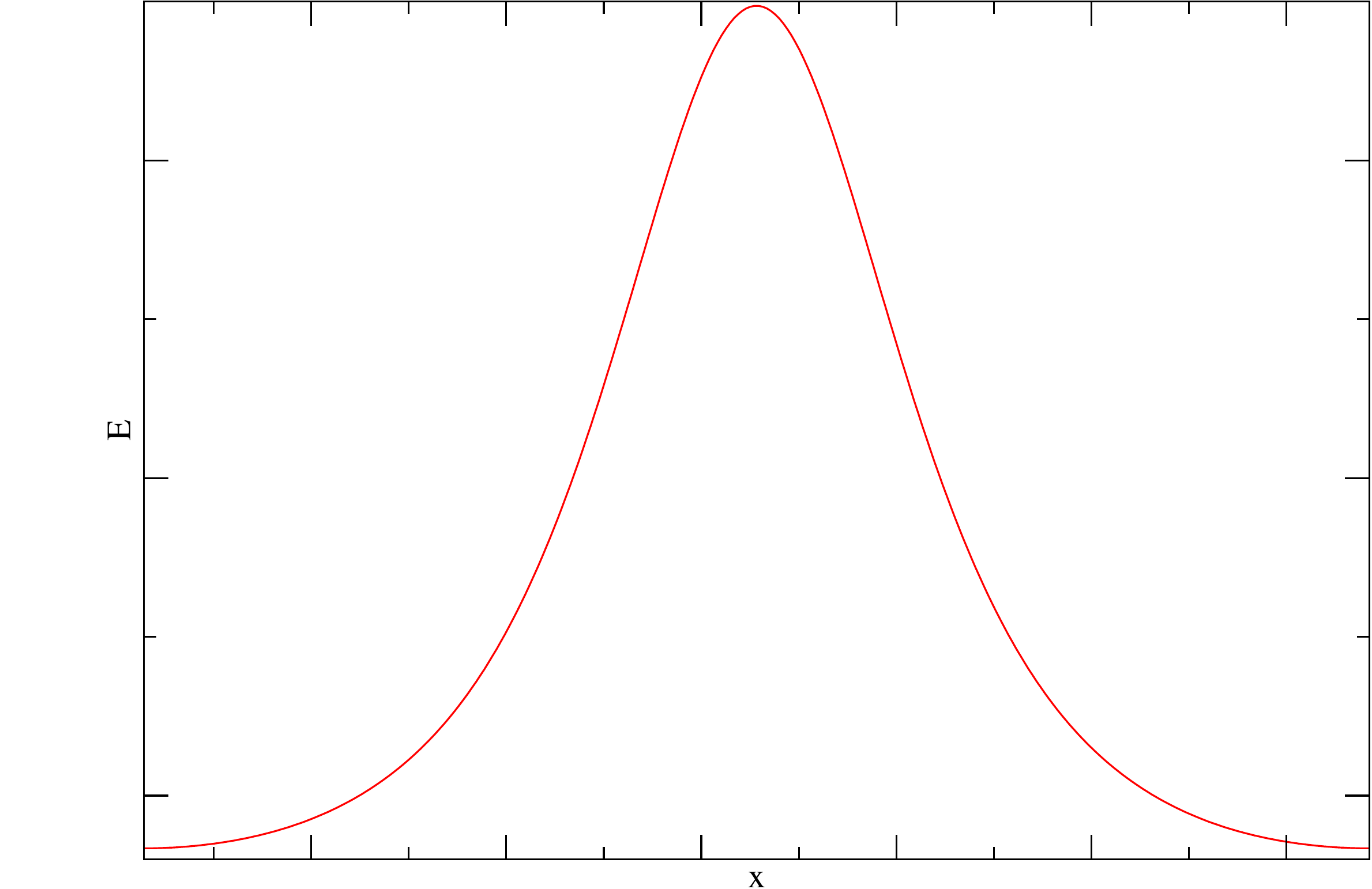}
\caption{This localized lump of energy is certainly something like the $t = 0$ slice of a plasma ball. We are claiming
that if it forms on a sphere, its evolution is too dynamic to be considered that of a plasma ball. Its motion cannot
be undone by a boost because this turns it back into the uniform profile of (\ref{unboosted-tensor}).}
\label{collective-oscillation}
\end{figure}
We could use $AdS_3$ / $CFT_2$ analogues of (\ref{l4-identification}) and (\ref{g5-identification}) to replace $r_0$
and $G$ with gauge theory expressions above. Because AdS has a boundary, boosting (\ref{btz-ads}) to a velocity of $v$
yields a black hole that oscillates about the origin indefinitely. The CFT state dual to this bouncing black hole has
the stress-energy tensor:
\begin{eqnarray}
T_{tt} &=& T_{xx} = \frac{1}{32\pi G R^3} \left [ \frac{(1 - v^2)(R^2 + r_0^2)}{\left ( v \cos \left ( \frac{t - x}{R} \right ) - 1 \right )^2} + \frac{(1 - v^2)(R^2 + r_0^2)}{\left ( v \cos \left ( \frac{t + x}{R} \right ) - 1 \right )^2} - 2R^2\right ] \nonumber \\
T_{tx} &=& T_{xt} = \frac{v(1 - v^2)(R^2 + r_0^2)}{8\pi G R^3} \left [ \frac{\sin \left ( \frac{t}{R} \right ) \sin \left ( \frac{x}{R} \right ) \left ( v \cos \left ( \frac{t}{R} \right ) \cos \left ( \frac{x}{R} \right ) - 1\right )}{\left ( v \cos \left ( \frac{t - x}{R} \right ) - 1 \right )^2 \left ( v \cos \left ( \frac{t + x}{R} \right ) - 1 \right )^2} \right ] \; . \label{boosted-tensor}
\end{eqnarray}
Thus we see that a valid CFT solution having Figure \ref{collective-oscillation} as its $t = 0$ energy density is not a
meta-stable ball at a fixed position, but a stable flow with an oscillating position \cite{freivogel}.

\section{Treating energy stochastically}
Our goal is to model certain features of an interacting field theory, without recourse to what the specific interaction
is. One way to accomplish this is to construct a model that is based on the system's density of states. A thermodynamic
quantity like this is easier to understand than the Hamiltonian because it is only an indication of the
\textit{spectrum} of the Hamiltonian. The result of our derivation will be an evolution equation for the energy density
at each point in space.

\subsection{Main equations}
To begin our analysis, we consider a cubic lattice of identical sites in $d$ dimensions and keep track of the number of
units of energy that can be found on each site. Following \cite{mvr}, we will write $\{ n_r \}$ to mean
$n_1$ units on site 1, $n_2$ units on site 2, $n_3$ units on site 3, \textit{etc}. Each configuration has a certain
probability of being realized. This probability is $P \left ( \{ N_r \}(t) = \{ n_r \} \right )$. Uppercase letters
have been used for random variables with the lowercase versions denoting specific values. However, we will often
shorten this to $P \left ( \{ n_r \} \right )$. In a stochastic process with continuous time, the probabilities as a
function of time obey the master equation \cite{ngvk}:
\begin{equation}
\frac{\partial P \left ( \{ n_r \} \right )}{\partial t} = \sum_{\{ n_r^{\prime} \}} \left [ P \left ( \{ n_r^{\prime} \} \right ) W_{\{ n_r^{\prime} \} \rightarrow \{ n_r \}} - P \left ( \{ n_r \} \right ) W_{\{ n_r \} \rightarrow \{ n_r^{\prime} \}} \right ] \; . \label{master}
\end{equation}
The $W$ quantities which determine the process are called the \textit{transition rates} and are defined by:
\begin{equation}
W_{\{ n_r \} \rightarrow \{ n_r^{\prime} \}} = \lim_{\tau \rightarrow 0} \frac{P \left ( \{ N_r \}(t + \tau) = \{ n_r^{\prime} \} | \{ N_r \}(t) = \{ n_r \} \right )}{\tau} \; . \nonumber
\end{equation}
To convert (\ref{master}) into something more concrete, we will make three physical assumptions: local energy
conservation, detailed balance and entropic dominance.

\subsubsection{Physical assumptions}
Inline with our first assumption, we declare that any transition which is nonlocal or does not conserve energy has a
$W$ value of zero. In the transition rates that are left, energy is transferred between two sites and those sites must
be nearest neighbours. Instead of listing all $\{ n_r^{\prime} \}$ configurations that can be reached from $\{ n_r \}$,
we may simply choose a pair of sites $\left < a, b \right >$ and a number $k$ to transfer between them. The master
equation therefore becomes
\begin{eqnarray}
\frac{\partial P \left ( \{ n_r \} \right )}{\partial t} &=& \sum_{\left < a, b \right >} \sum_{k \neq 0} \left [ P (\dots, n_a + k, n_b - k, \dots) W_{(n_a + k, n_b - k) \rightarrow (n_a, n_b)} \right. \nonumber \\
&& \left. - P (\dots, n_a, n_b, \dots) W_{(n_a, n_b) \rightarrow (n_a + k, n_b - k)} \right ] \; . \label{physical-master}
\end{eqnarray}
We will not work directly with probabilities, but rather the expectation of a particular site's energy:
\begin{equation}
\bar{n}_c \equiv \left < N_c \right > = \sum_{\{ n_r \}} P \left ( \{ n_r \} \right ) n_c \; . \label{expected-n}
\end{equation}
The next step is to differentiate (\ref{expected-n}) and substitute (\ref{physical-master}):
\begin{eqnarray}
\frac{\partial \bar{n}_c}{\partial t} &=& \sum_{\{ n_r \}} n_c \frac{\partial P \left ( \{ n_r \} \right )}{\partial t} \nonumber \\
&=& \sum_{\{ n_r \}} n_c \sum_{\left < a, b \right >} \sum_{k \neq 0} \left [ P (\dots, n_a + k, n_b - k, \dots) W_{(n_a + k, n_b - k) \rightarrow (n_a, n_b)} \right. \nonumber \\
&& \left. - P (\dots, n_a, n_b, \dots) W_{(n_a, n_b) \rightarrow (n_a + k, n_b - k)} \right ] \; . \nonumber
\end{eqnarray}
Every $P (\dots, n_a + k, n_b - k, \dots) W_{(n_a + k, n_b - k) \rightarrow (n_a, n_b)}$ in the sum is a
$P (\dots, n_a, n_b, \dots) W_{(n_a, n_b) \rightarrow (n_a + k, n_b - k)}$ for some other $\{ n_r \}$ and the
negative $k$ value. As long as $n_c$ does not appear, these cancel with the same coefficient. We may therefore
let $a = c$ and reindex.
\begin{eqnarray}
\frac{\partial \bar{n}_c}{\partial t} &=& \sum_{\left < b, c \right >} \sum_{k \neq 0} \sum_{\{ n_r \}} \left [ n_c P (\dots, n_b + k, n_c - k, \dots) W_{(n_b + k, n_c - k) \rightarrow (n_b, n_c)} \right. \nonumber \\
&& \left. - n_c P (\dots, n_b, n_c, \dots) W_{(n_b, n_c) \rightarrow (n_b - k, n_c + k)} \right ] \nonumber \\
&=& \sum_{\left < b, c \right >} \sum_{k \neq 0} \sum_{\{ n_r \}} \left [ (n_c + k) P (\dots, n_b, n_c, \dots) W_{(n_b, n_c) \rightarrow (n_b - k, n_c + k)} \right. \nonumber \\
&& \left. - n_c P (\dots, n_b, n_c, \dots) W_{(n_b, n_c) \rightarrow (n_b - k, n_c + k)} \right ] \nonumber \\
&=& \sum_{\left < b, c \right >} \sum_{k \neq 0} \sum_{\{ n_r \}} k P (\dots, n_b, n_c, \dots) W_{(n_b, n_c) \rightarrow (n_b - k, n_c + k)} \nonumber \\
&=& \left < \sum_{\left < b, c \right >} \sum_{k \neq 0} k W_{(n_b, n_c) \rightarrow (n_b - k, n_c + k)} \right > \nonumber \\
&\approx& \sum_{\left < b, c \right >} \sum_{k \neq 0} k W_{(\bar{n}_b, \bar{n}_c) \rightarrow (\bar{n}_b - k, \bar{n}_c + k)} \label{equation-for-general-w}
\end{eqnarray}
We have yet to show that it is safe to replace random variables by their expectations in the last step. Since
(\ref{equation-for-general-w}) is still quite general, further work is required to narrow down our choices for $W$.

The configurations $\{ n_r \}$ represent collections of several microstates. Introducing the function $\rho(n_r)$
giving the number of ways for site $r$ to have energy $n_r$, it is easy to count the number of ways in which our
configurations can be realized. There are $\prod_r \rho(n_r)$ microstates with the distribution $\{ n_r \}$. The
most familiar situation in statistical mechanics is that of thermal equilibrium. State $\mu$ in equilibrium is
achieved with probability $P(\mu) = \frac{1}{Z} e^{-\beta E}$. Since this is just one microstate, we should add up a
sufficient number of them to get
\begin{equation}
P \left ( \{ n_r \} \right ) = \frac{1}{Z} e^{-\beta E} \prod_r \rho(n_r) \; . \label{equilibrium-p}
\end{equation}
These are equilibrium probabilities so the master equation should vanish when they are inserted. This condition
provides a constraint on the possible transition rates but we will make use of a stronger condition; the principle of
detailed balance. Detailed balance, which holds for reversible Markov chains, states that all terms in
(\ref{physical-master}) should separately vanish in equilibrium instead of just the entire sum:
\begin{equation}
P (\dots, n_a + k, n_b - k, \dots) W_{(n_a + k, n_b - k) \rightarrow (n_a, n_b)} = P (\dots, n_a, n_b, \dots) W_{(n_a, n_b) \rightarrow (n_a + k, n_b - k)} \; . \label{detailed-balance}
\end{equation}
This principle was famously used by Einstein to predict spontaneous emission rates before quantum field theory had
been developed \cite{mccumber}. Substituting (\ref{equilibrium-p}) into (\ref{detailed-balance}), our condition becomes
\begin{equation}
\rho(n_a + k) \rho(n_b - k) W_{(n_a + k, n_b - k) \rightarrow (n_a, n_b)} = \rho(n_a) \rho(n_b) W_{(n_a, n_b) \rightarrow (n_a + k, n_b - k)} \; . \label{w-constraint}
\end{equation}
There are many solutions to this system of equations but some make more sense than others.

One solution to (\ref{w-constraint}) has $W_{(n_a + k, n_b - k) \rightarrow (n_a, n_b)}$ proportional to the number of
final states $\rho(n_a) \rho(n_b)$. This type of transition rate is the one most compatible with the ergodic principle.
When fluctuations are completely thermal, a higher number of final states should be the only thing favouring one
transition over another. We write
\begin{equation}
W_{(n_a, n_b) \rightarrow (n_a + k, n_b - k)} \propto C(n_a, n_b) \rho(n_a + k) \rho(n_b - k) \nonumber
\end{equation}
where there can be some additional dependence on $k$. Substituting this into (\ref{w-constraint}), we find
\begin{equation}
C(n_a, n_b) = C(n_a + k, n_b - k) \; . \nonumber
\end{equation}
Using this relation repeatedly, we may set $k$ equal to $n_a$ or $n_b$, telling us that $C$ is only a function of the
total energy. For this to be valid, any configuration must be reachable from the configuration obtained by having one
site shift all of its energy to a neighbour. This is the same as saying that there are no superselection sectors.
We now have
\begin{equation}
W_{(n_a, n_b) \rightarrow (n_a + k, n_b - k)} \propto C \left ( \frac{n_a + n_b}{2} \right ) \rho(n_a + k) \rho(n_b - k) \label{w-solution}
\end{equation}
where we have included the factor of $\frac{1}{2}$ for later convenience. We should note that
$W_{(n_a + k, n_b - k) \rightarrow (n_a, n_b)}$ is not proportional to the number of configurations that have $n_a$
on site $a$ and $n_b$ on site $b$. This would be $\rho (n_a) \rho (n_b) \prod_{r \notin \{a, b\}} \rho(n_r)$. A
transition rate involving all of these factors would be inconsistent with the nearest neighbour logic we have been
using. The only transition we have allowed is one in which site $b$ sends $k$ units of energy to site $a$. If a
transition rate compatible with the ergodic principle also depends linearly on $\rho(n_c)$ for some other site $c$,
this is not a transition rate for $(n_a, n_b) \rightarrow (n_a + k, n_b - k)$ but rather
$(n_a, n_b, n_c) \rightarrow (n_a + k, n_b - k, n_c)$. In other words, site $c$ has itself undergone a transition to
some other internal state that keeps the same energy $n_c$. This amounts to two transition happening in the same
timestep. Moreover, there is no way to tell that this is indeed a nearest neighbour transition paired with a
transition between internal states far away. It could have been site $b$ sending $k$ units of energy to site $c$
followed immediately by $c$ sending $k$ units of energy to site $a$, thus violating locality.

\subsubsection{The continuum limit}
The equation we wish to build on is
\begin{equation}
\frac{\partial n_c}{\partial t} = \sum_{\left < b, c \right >} \sum_{k \neq 0} k W_{(n_c, n_b) \rightarrow (n_c + k, n_b - k)} \label{non-continuum}
\end{equation}
where the transition rates are given by (\ref{w-solution}). So far, we have been assuming that the energy and the
spatial co-ordinate both vary by discrete amounts. One way to write this is to have site $c$ labelled by $\textbf{x}$
which means that site $b$ is $\textbf{x} \pm a\textbf{e}$ for some unit vector $\textbf{e}$. To consider a
continuous version of (\ref{non-continuum}), the lattice constant $a$ must approach zero. Additionally, the sum over
$k \neq 0$ must become a sum over $\pm \epsilon$ with $\epsilon$ also going to zero. Our formula (\ref{w-solution})
becomes
\begin{equation}
W_{(E(\textbf{x}), E(\textbf{x} + a\textbf{e})) \rightarrow (E(\textbf{x}) + \epsilon, E(\textbf{x} + a\textbf{e}) - \epsilon)} = C\left ( \frac{E(\textbf{x}) + E(\textbf{x} + a\textbf{e})}{2} \right ) \rho(E(\textbf{x}) + \epsilon) \rho(E(\textbf{x} + a\textbf{e}) - \epsilon) \label{transition-rates}
\end{equation}
where we use $E$ instead of $n$ to make it clear that we are talking about energy densities that are being incremented
continuously. One should keep in mind that $C$ has units of inverse time in order for $W$ to be a rate. Since our
differential equation for the energy density is now a function of the small parameters $\epsilon$ and $a$, a useful
approximation to it can be derived with a Taylor expansion.

Using (\ref{transition-rates}) in the continuous version of (\ref{non-continuum}), we have
\begin{eqnarray}
\frac{\partial E(\textbf{x})}{\partial t} &=& \epsilon \sum_{\textbf{e} \in \{ \pm \textbf{e}_1, \dots, \pm \textbf{e}_d \}} \left [ W_{(E(\textbf{x}), E(\textbf{x} + a\textbf{e})) \rightarrow (E(\textbf{x}) + \epsilon, E(\textbf{x} + a\textbf{e}) - \epsilon)} - W_{(E(\textbf{x}), E(\textbf{x} + a\textbf{e})) \rightarrow (E(\textbf{x}) - \epsilon, E(\textbf{x} + a\textbf{e}) + \epsilon)}\right ] \nonumber \\
&=& \epsilon \sum_{\textbf{e} \in \{ \pm \textbf{e}_1, \dots, \pm \textbf{e}_d \}} C\left ( \frac{E(\textbf{x}) + E(\textbf{x} + a\textbf{e})}{2} \right ) \left [ \rho(E(\textbf{x}) + \epsilon) \rho(E(\textbf{x} + a\textbf{e}) - \epsilon) \right. \nonumber \\
&& \left. - \rho(E(\textbf{x}) - \epsilon) \rho(E(\textbf{x} + a\textbf{e}) + \epsilon) \right ] \; . \label{continuum}
\end{eqnarray}
We will define
\begin{equation}
\frac{\partial E}{\partial t} = X(\epsilon, a) = \epsilon \tilde{X}(\epsilon, a) \nonumber
\end{equation}
in which case the relevant Taylor series becomes
\begin{equation}
\frac{\partial E}{\partial t} = \sum_{m = 0}^{\infty} \sum_{n = 0}^{\infty} \frac{\epsilon^m a^n}{m! n!} \left. \frac{\partial^{m + n} X}{\partial \epsilon^m \partial a^n} \right |_{\substack{a = 0 \\ \epsilon = 0}} \; . \label{2d-taylor}
\end{equation}
We can see from (\ref{continuum}) that $\tilde{X}(0, a) = 0 = \tilde{X}(\epsilon, 0)$, so any term in (\ref{2d-taylor})
that survives, must involve at least three derivatives of $X$: two with respect to $\epsilon$ and one with respect to
$a$. In fact, the number of derivatives we need to take is even higher. Differentiating something like
$\rho(E(\textbf{x} + a\textbf{e}))$ with respect to $a$ would contribute a
$\frac{\partial \rho}{\partial E} \partial_i E e_i$ term inside the sum. If we add up the $i^{\mathrm{th}}$ components
of $\textbf{e}$ where $\textbf{e}$ runs over all positive and negative standard basis vectors, the result is zero.
This means we need at least one more derivative with respect to $a$ and the approximation we seek is:
\begin{equation}
\frac{\partial E}{\partial t} = \frac{1}{4} \epsilon^2 a^2 \left. \frac{\partial^4 X}{\partial \epsilon^2 \partial a^2} \right |_{\substack{a = 0 \\ \epsilon = 0}} \; . \label{e-derivative}
\end{equation}
We will use the abbreviated notation $\rho_{+} = \rho(E(\textbf{x} + a\textbf{e}))$ and
$C_{+} = C \left ( \frac{E(\textbf{x}) + E(\textbf{x} + a\textbf{e})}{2} \right )$ which satisfy
\begin{eqnarray}
\frac{\partial \rho_{+}}{\partial a} &=& e_j \partial_j \rho_{+} \nonumber \\
\frac{\partial C_{+}}{\partial a} &=& \frac{1}{2} e_j \partial_j C_{+} \; . \label{convenient-relations}
\end{eqnarray}
Also, the derivatives with respect to $\epsilon$ are not calculated here but in the appendix. Picking up from where
the appendix leaves off,
\begin{eqnarray}
\left. \frac{\partial^4 X}{\partial \epsilon^2 \partial a^2} \right |_{\substack{a = 0 \\ \epsilon = 0}} &=& 2 \left. \frac{\partial^3 \tilde{X}}{\partial \epsilon \partial a^2} \right |_{\substack{a = 0 \\ \epsilon = 0}} \nonumber \\
&=& 4 \sum_{\textbf{e} \in \{ \pm \textbf{e}_1, \dots, \pm \textbf{e}_d \}} \frac{\partial^2}{\partial a^2} \left. \left [ C_{+} \left ( \rho_{+} \frac{\textup{d} \rho}{\textup{d} E} - \rho \frac{\textup{d} \rho_{+}}{\textup{d} E} \right ) \right ] \right |_{a = 0} \nonumber \\
&=& 4 \sum_{\textbf{e} \in \{ \pm \textbf{e}_1, \dots, \pm \textbf{e}_d \}} e_i \frac{\partial}{\partial a} \left. \left [ \frac{1}{2} \partial_i C_{+} \left ( \rho_{+} \frac{\textup{d} \rho}{\textup{d} E} - \rho \frac{\textup{d} \rho_{+}}{\textup{d} E} \right ) + C_{+} \left ( \partial_i \rho_{+} \frac{\textup{d} \rho}{\textup{d} E} - \rho \partial_i \frac{\textup{d} \rho_{+}}{\textup{d} E} \right ) \right ] \right |_{a = 0} \nonumber \\
&=& 4 \left [ C \left ( \partial_i \partial_j \rho \frac{\textup{d} \rho}{\textup{d} E} - \rho \partial_i \partial_j \frac{\textup{d} \rho}{\textup{d} E} \right ) -\partial_i C \left ( \rho \partial_j \frac{\textup{d} \rho}{\textup{d} E} - \partial_j \rho \frac{\textup{d} \rho}{\textup{d} E} \right ) \right ] \sum_{\textbf{e} \in \{ \pm \textbf{e}_1, \dots, \pm \textbf{e}_d \}} e_i e_j \nonumber \\
&=& 8 \left [ C \left ( \partial_i \partial_i \rho \frac{\textup{d} \rho}{\textup{d} E} - \rho \partial_i \partial_i \frac{\textup{d} \rho}{\textup{d} E} \right ) -\partial_i C \left ( \rho \partial_i \frac{\textup{d} \rho}{\textup{d} E} - \partial_i \rho \frac{\textup{d} \rho}{\textup{d} E} \right ) \right ] \; . \label{expression-before-log}
\end{eqnarray}
It is not immediately obvious but (\ref{expression-before-log}) simplifies to a more compact expression involing a
logarithm. If we expand
\begin{eqnarray}
-C \rho^2 \partial_i \frac{\textup{d} \log \rho}{\textup{d} E} &=& -C \rho^2 \partial_i \frac{\frac{\textup{d} \rho}{\textup{d} E}}{\rho} \nonumber \\
&=& C \rho^2 \frac{\partial_i \rho \frac{\textup{d} \rho}{\textup{d} E} - \rho \partial_i \frac{\textup{d} \rho}{\textup{d} E}}{\rho^2} \nonumber \\
&=& C \left ( \partial_i \rho \frac{\textup{d} \rho}{\textup{d} E} - \rho \partial_i \frac{\textup{d} \rho}{\textup{d} E} \right ) \; , \nonumber
\end{eqnarray}
we get something whose $\partial_i$ derivative is (\ref{expression-before-log}). This shows that
\begin{equation}
\frac{\partial E}{\partial t} = -2 \epsilon^2 a^2 \partial_i \left ( C(E) \rho^2(E) \partial_i \frac{\textup{d} \log \rho(E)}{\textup{d} E} \right ) \; . \label{hypercorrect-pde}
\end{equation}
Checking the dimensions of (\ref{hypercorrect-pde}), the left hand side is an energy density over a time. On the right
hand side, we have an energy density in the form of $\epsilon$ because the other $\epsilon$ cancels with the
$\textup{d}E$. We also have an inverse time because the function $C$ had inverse time units. The $a^2$ cancels with
the two spatial derivatives. From now on, we will drop unknown dimensionful parameters by absorbing them into the time.
The main differential equation of our model is
\begin{equation}
\frac{\partial E}{\partial t} = - \partial_i \left ( C \rho^2 \partial_i \frac{\textup{d} \log \rho}{\textup{d} E} \right ) \label{main-pde}
\end{equation}
in which $C$ is assumed to be a dimensionless function. Common choices for it will be $1$ and $\rho^{-2}$.

\subsection{Interesting features}
Our continuum limit equation has particularly nice things to say about a system with microcanonical phases like
(\ref{sym-entropy}). At least two of these phases only appear at energies that are large compared to the spatial volume.
After considering some insights in \cite{mvr} concerned with the case $\frac{\partial E}{\partial t} = 0$, we
will see that large energies are required to even trust the model at a basic level.

\subsubsection{Static situations}
A special role is played by the density of states whose logarithm is linear in $E$. This is the Hagedorn density of
states that we saw appearing in string theory and Super Yang-Mills: $\rho(E) \propto e^{\beta_{\mathrm{H}} E}$. In this
case $\frac{\textup{d} \log \rho(E)}{\textup{d}E} = \beta_{\mathrm{H}}$, a constant. The
$\partial_i$ acting on this constant will set the left hand side of (\ref{main-pde}) to zero. Under Hagedorn behaviour,
the energy distribution $E(\textbf{x}, t)$ does not change with time.

The gauge theories with holographic duals have a Hagedorn regime as well as other phases. As stated before, we expect
$\log \rho(E) \propto E^{\alpha}$ where we could have $\alpha < 1$, $\alpha > 1$ or $\alpha = 1$ depending on the
energy. Since the dynamics are frozen with a purely Hagedorn density of states, we expect changes in the energy to
take place very slowly if $\alpha = 1$ is the widest phase. The $\alpha = 1$ phase can equivalently be described as the
energy range for which the inverse temperature $\beta(E)$ is flat.

Using the fact that $\log \rho(E)$ is the microcanonical entropy, we can rewrite our main equation in terms of
$\beta(E)$ as well:
\begin{equation}
\frac{\textup{d} \log \rho(E)}{\textup{d} E} = \frac{\textup{d} S(E)}{\textup{d} E} = \beta(E) \nonumber
\end{equation}
and our equation becomes
\begin{equation}
\frac{\partial E}{\partial t} = - \partial_i \left ( C \rho^2 \partial_i \beta \right ) \; . \label{beta-pde}
\end{equation}
The phases can be characterized by whether $\beta$ is decreasing ($\alpha < 1$), increasing ($\alpha > 1$) or
neither ($\alpha = 1$).

With a Hagedorn density of states, $\frac{\partial E}{\partial t}$ vanishes for any energy distribution. Conversely
with a uniform energy distribution, $\frac{\partial E}{\partial t}$ vanishes for any density of states. This
equilibrium distribution may be stable or unstable depending on the phase we are in. We will decompose the energy as
\begin{equation}
E(\textbf{x}, t) = E_0 + \tilde{E}(\textbf{x}, t) \nonumber
\end{equation}
where $\tilde{E}$ is small, allowing us to keep only one power of it in the PDE (\ref{beta-pde}). First,
\begin{eqnarray}
\partial_i \beta(E) &\approx& \partial_i \left ( \beta(E_0) + \tilde{E} \frac{\textup{d} \beta(E_0)}{\textup{d}E} \right ) \nonumber \\
&=& \frac{\textup{d} \beta(E_0)}{\textup{d}E} \partial_i \tilde{E} \; . \nonumber
\end{eqnarray}
This expression with one power of $\tilde{E}$ is multiplied by $C(E) \rho^2(E)$. A first order expansion of
$C(E) \rho^2(E)$ would give an overall result that is second order in $\tilde{E}$ so we only expand it to zeroth order:
\begin{equation}
\frac{\partial \tilde{E}}{\partial t} \approx - C(E_0) \rho^2(E_0) \frac{\textup{d} \beta(E_0)}{\textup{d}E} \partial_i \partial_i \tilde{E} \; . \label{linear-beta-pde}
\end{equation}
This is either the heat equation or the reverse heat equation depending on whether the overall coefficient is negative
or positive. The sign of $\frac{\textup{d} \beta_0}{\textup{d}E}$ is what matters because $C$ and $\rho$ are positive
functions. Agreeing with our earlier intuition about the entropic dynamics of energy, we have the following three
cases:
\begin{itemize}
\item
$\frac{\textup{d}\beta}{\textup{d}E} < 0$ is a decreasing inverse temperature, a concave entropy and a log-concave
density of states. It leads to diffusion or inhomogeneities that decrease with time due to the heat equation.
\item
$\frac{\textup{d}\beta}{\textup{d}E} > 0$ is an increasing inverse temperature, a convex entropy and a log-convex
density of states. It leads to clustering or inhomogeneities that increase with time due to the reverse heat equation.
\item
$\frac{\textup{d}\beta}{\textup{d}E} = 0$ is a constant inverse temperature, a linear entropy and a Hagedorn density
of states. It leads to static behaviour.
\end{itemize}
Understanding the detailed properties of the diffusion and to a lesser extent the clustering caused by this PDE will be
the focus of the next chapter.

\subsubsection{Mean-field variances}
A loose end in this chapter has been the assumption that we may work only with expected values in
(\ref{equation-for-general-w}). In general, mean-field approximations may be used on quantities that have a small
variance. An energy with a small variance is also one of the desired features of our model. After all, the model is
an attempt at connecting the excitations of field theory degrees of freedom to Einstein gravity, something that is
completely deterministic.

If variances are initially small, we want to make sure that they grow slowly so that our model stays valid for a long
time. Just as we derived an expression for $\frac{\partial \bar{n}_c}{\partial t}$ from the master equation, we can
repeat the calculation for $\frac{\partial \bar{n^2}_c}{\partial t}$.
\begin{eqnarray}
\frac{\partial \bar{n^2}_c}{\partial t} &=& \sum_{\{ n_r \}} n^2_c \frac{\partial P \left ( \{ n_r \} \right )}{\partial t} \nonumber \\
&=& \sum_{\left < b, c \right >} \sum_{k \neq 0} \sum_{\{ n_r \}} \left [ n^2_c P (\dots, n_b + k, n_c - k, \dots) W_{(n_b + k, n_c - k) \rightarrow (n_b, n_c)} \right. \nonumber \\
&& \left. - n^2_c P (\dots, n_b, n_c, \dots) W_{(n_b, n_c) \rightarrow (n_b - k, n_c + k)} \right ] \nonumber \\
&=& \sum_{\left < b, c \right >} \sum_{k \neq 0} \sum_{\{ n_r \}} \left [ ((n_c + k)^2 - n^2_c) P (\dots, n_b, n_c, \dots) W_{(n_b, n_c) \rightarrow (n_b - k, n_c + k)} \right ] \nonumber \\
&=& \sum_{\left < b, c \right >} \sum_{k \neq 0} \sum_{\{ n_r \}} (2kn_c + k^2) P (\dots, n_b, n_c, \dots) W_{(n_b, n_c) \rightarrow (n_b - k, n_c + k)} \nonumber \\
&\approx& \sum_{\left < b, c \right >} \sum_{k \neq 0} (2kn_c + k^2) W_{(\bar{n}_b, \bar{n}_c) \rightarrow (\bar{n}_b - k, \bar{n}_c + k)} \label{v-equation-for-general-w}
\end{eqnarray}
If we combine (\ref{v-equation-for-general-w}) with (\ref{equation-for-general-w}),
\begin{eqnarray}
\sigma^2_c &=& \bar{n^2}_c - \bar{n}_c^2 \nonumber \\
\frac{\partial \sigma^2_c}{\partial t} &=& \frac{\partial \bar{n^2}_c}{\partial t} - 2\bar{n}_c \frac{\partial \bar{n}_c}{\partial t} \nonumber \\
&=& \sum_{\left < b, c \right >} \sum_{k \neq 0} k^2 W_{(\bar{n}_b, \bar{n}_c) \rightarrow (\bar{n}_b - k, \bar{n}_c + k)} \nonumber \\
&=& \sum_{\left < b, c \right >} \sum_{k \neq 0} k^2 C \left ( \frac{\bar{n}_b + \bar{n}_c}{2} \right ) \rho(\bar{n}_b - k) \rho(\bar{n}_c + k) \nonumber
\end{eqnarray}
where we have substituted the $W$ solution (\ref{w-solution}). This equation for the variance can be examined in the
continuum limit and the key is that we do not need as many orders as $\epsilon^2 a^2$ in the subsequent Taylor
expansion. The continuum limit is
\begin{eqnarray}
\frac{\partial \sigma^2(\textbf{x})}{\partial t} &=& \epsilon^2 \sum_{\textbf{e} \in \{ \pm \textbf{e}_1, \dots, \pm \textbf{e}_d \}} \left [ W_{(E(\textbf{x}), E(\textbf{x} + a\textbf{e})) \rightarrow (E(\textbf{x}) + \epsilon, E(\textbf{x} + a\textbf{e}) - \epsilon)} + W_{(E(\textbf{x}), E(\textbf{x} + a\textbf{e})) \rightarrow (E(\textbf{x}) - \epsilon, E(\textbf{x} + a\textbf{e}) + \epsilon)}\right ] \nonumber \\
&=& \epsilon^2 \sum_{\textbf{e} \in \{ \pm \textbf{e}_1, \dots, \pm \textbf{e}_d \}} C\left ( \frac{E(\textbf{x}) + E(\textbf{x} + a\textbf{e})}{2} \right ) \left [ \rho(E(\textbf{x}) + \epsilon) \rho(E(\textbf{x} + a\textbf{e}) - \epsilon) \right. \nonumber \\
&& \left. + \rho(E(\textbf{x}) - \epsilon) \rho(E(\textbf{x} + a\textbf{e}) + \epsilon) \right ] \label{v-continuum}
\end{eqnarray}
where the positive sign is due to the fact that we have $k^2$ instead of $k$. Again, define
\begin{equation}
\frac{\partial \sigma^2}{\partial t} = X(\epsilon, a) = \epsilon^2 \tilde{X}(\epsilon, a) \; . \nonumber
\end{equation}
When Taylor expanding $X$, we need at least two derivatives with respect to $\epsilon$ because of the
$\epsilon^2$ prefactor. This is all we need because $\tilde{X}(0, 0)$ is nonzero.
\begin{eqnarray}
\frac{\partial \sigma^2}{\partial t} &=& \frac{1}{2} \epsilon^2 \left. \frac{\partial^2 X}{\partial \epsilon^2} \right |_{\substack{a = 0 \\ \epsilon = 0}} \nonumber \\
&=& \epsilon^2 \tilde{X}(0, 0) \nonumber \\
&=& 2 \epsilon^2 \sum_{\textbf{e} \in \{ \pm \textbf{e}_1, \dots, \pm \textbf{e}_d \}} C \rho^2 \nonumber \\
&=& 4d \epsilon^2 C \rho^2 \label{variance-pde}
\end{eqnarray}
Unlike (\ref{hypercorrect-pde}), it is not the same unknown function that appears on the left and right hand sides of
(\ref{variance-pde}). One must first solve for $E(\textbf{x}, t)$ in order to solve for $\sigma(\textbf{x}, t)$.

Since the expression for $\frac{\partial E}{\partial t}$ has a small prefactor of $a^2$ while that for
$\frac{\partial \sigma^2}{\partial t}$ does not, it would seem that the rate of change of the variance is
parametrically larger, something we wished to avoid. However, it does not make sense to compare these quantities
directly. The energy has the same units as its standard deviation so we should compare
$\frac{\partial \sigma}{\partial t}$ to $\frac{\partial E}{\partial t}$ or more conveniently,
$\frac{\partial \sigma^2}{\partial t}$ to $\frac{\partial E^2}{\partial t}$.
\begin{equation}
\frac{\frac{\partial \sigma^2}{\partial t}}{\frac{\partial E^2}{\partial t}} = \frac{\frac{\partial \sigma^2}{\partial t}}{2E \frac{\partial E}{\partial t}} \label{need-high-energy}
\end{equation}
Factors of $\epsilon^2$ in the numerator and denominator cancel leaving $E a^2$ in the denominator. This tells us that
such a ratio of derivatives can indeed be small if the energy is large enough. In other words, a PDE like
(\ref{main-pde}) can be trusted to model high energy phenomena. Thinking about gravity, this includes the extreme
environments of black holes but not the everyday motion of test particles around them. On a more practical level, it
would be difficult to even write down the equation (\ref{main-pde}) if we were concerned with it holding for low
energies. For the field theories we are interested in, only asymptotic expressions are known for the density of states.
Even for situations in which the exact number of states is known for all energies, this $\rho(E)$ is not continuous.

\section{Nonlinear diffusion}
We now cover some of the properties of equations like (\ref{main-pde}) that are known analytically. The main
assumption we will use throughout this chapter is that $\beta$ is weakly decreasing, \textit{i.e.}
$\frac{\textup{d}\beta}{\textup{d}E} \leq 0$. This is a slight departure from the strongly coupled gauge theory result
as shown in Figure \ref{3-betas}.
\begin{figure}[h]
\centering
\subfloat[][Sharp phases]{\includegraphics[scale=0.2]{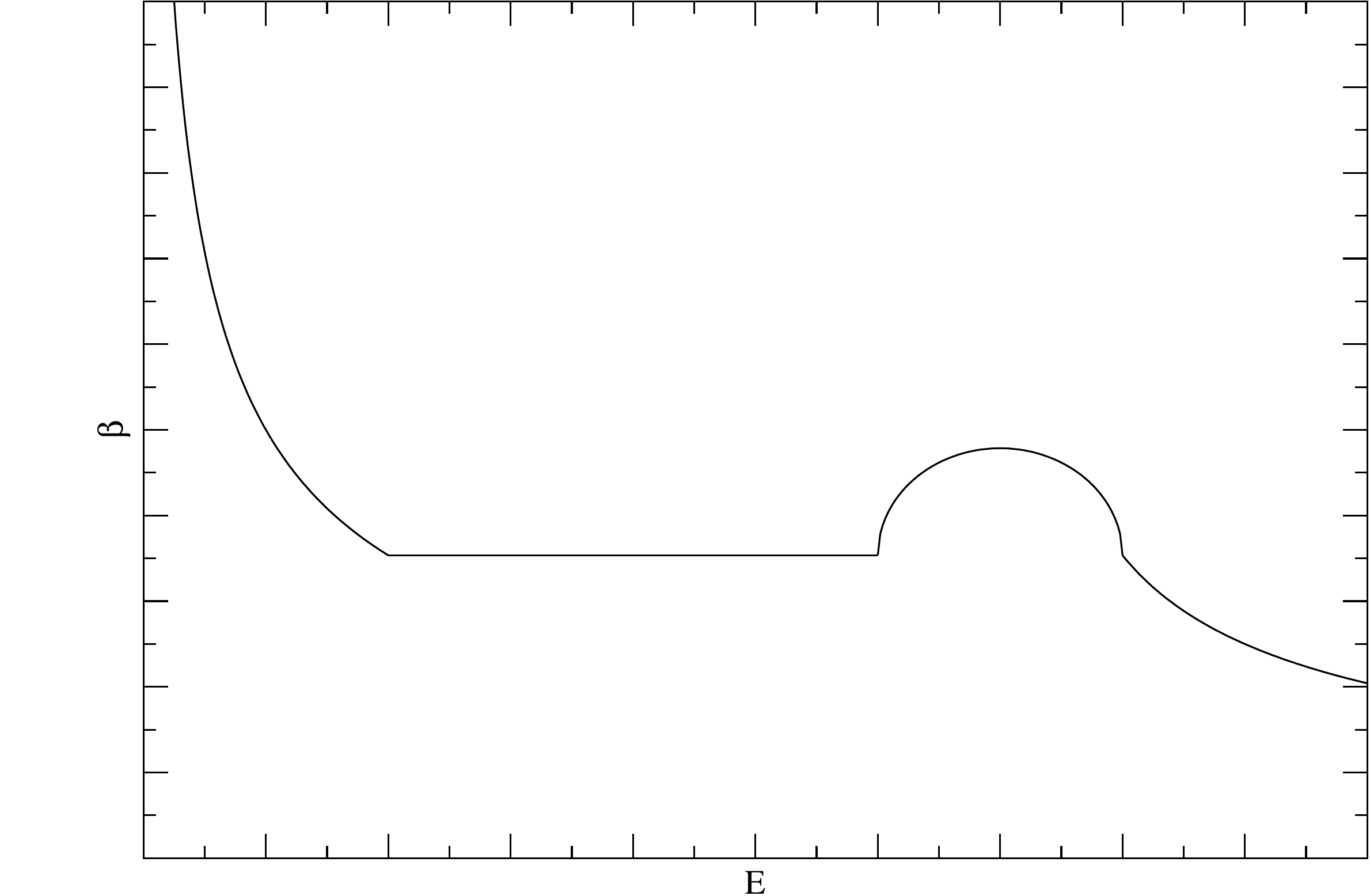}}
\subfloat[][Smooth phases]{\includegraphics[scale=0.2]{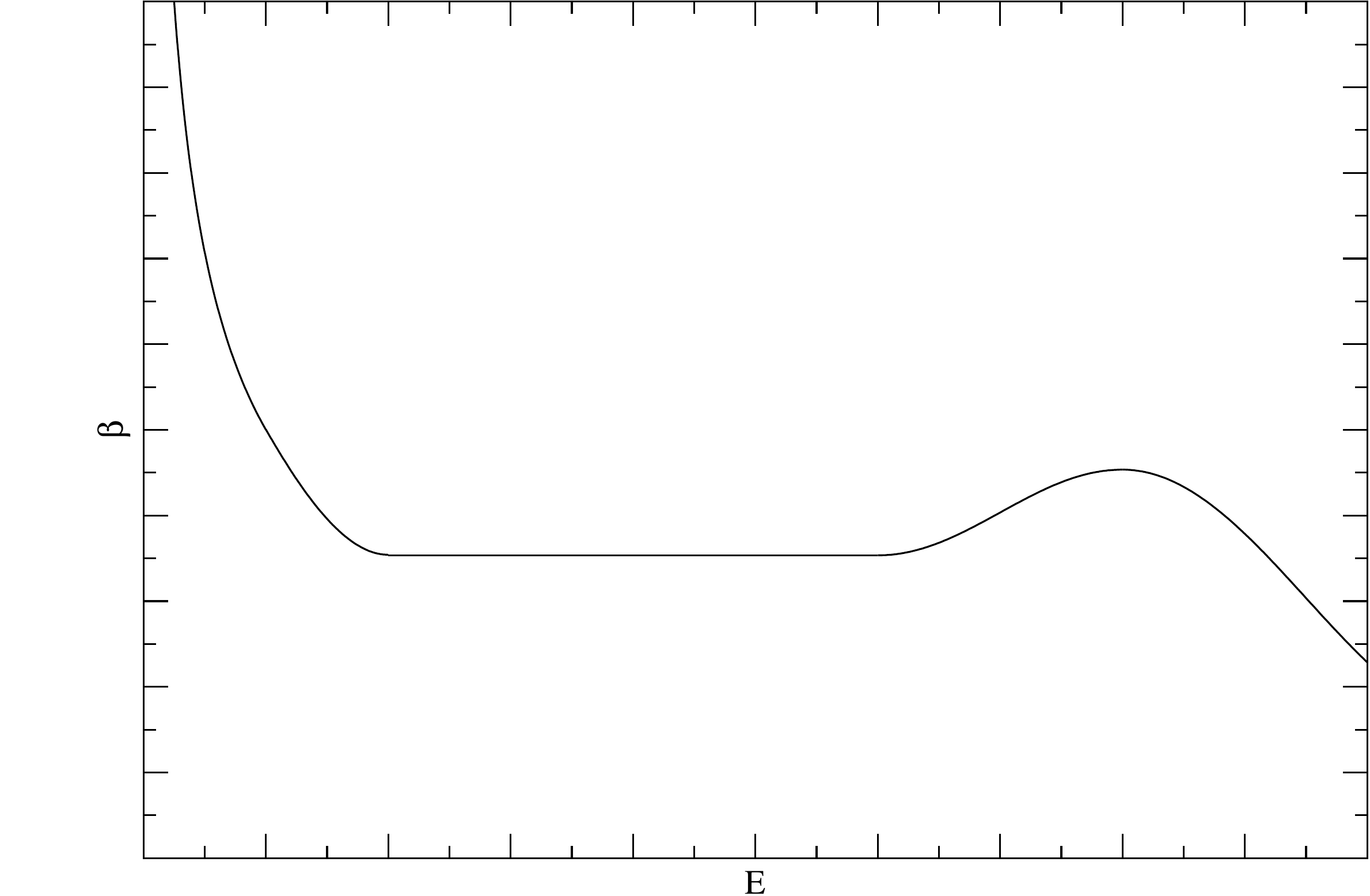}}
\subfloat[][Purely diffusive]{\includegraphics[scale=0.2]{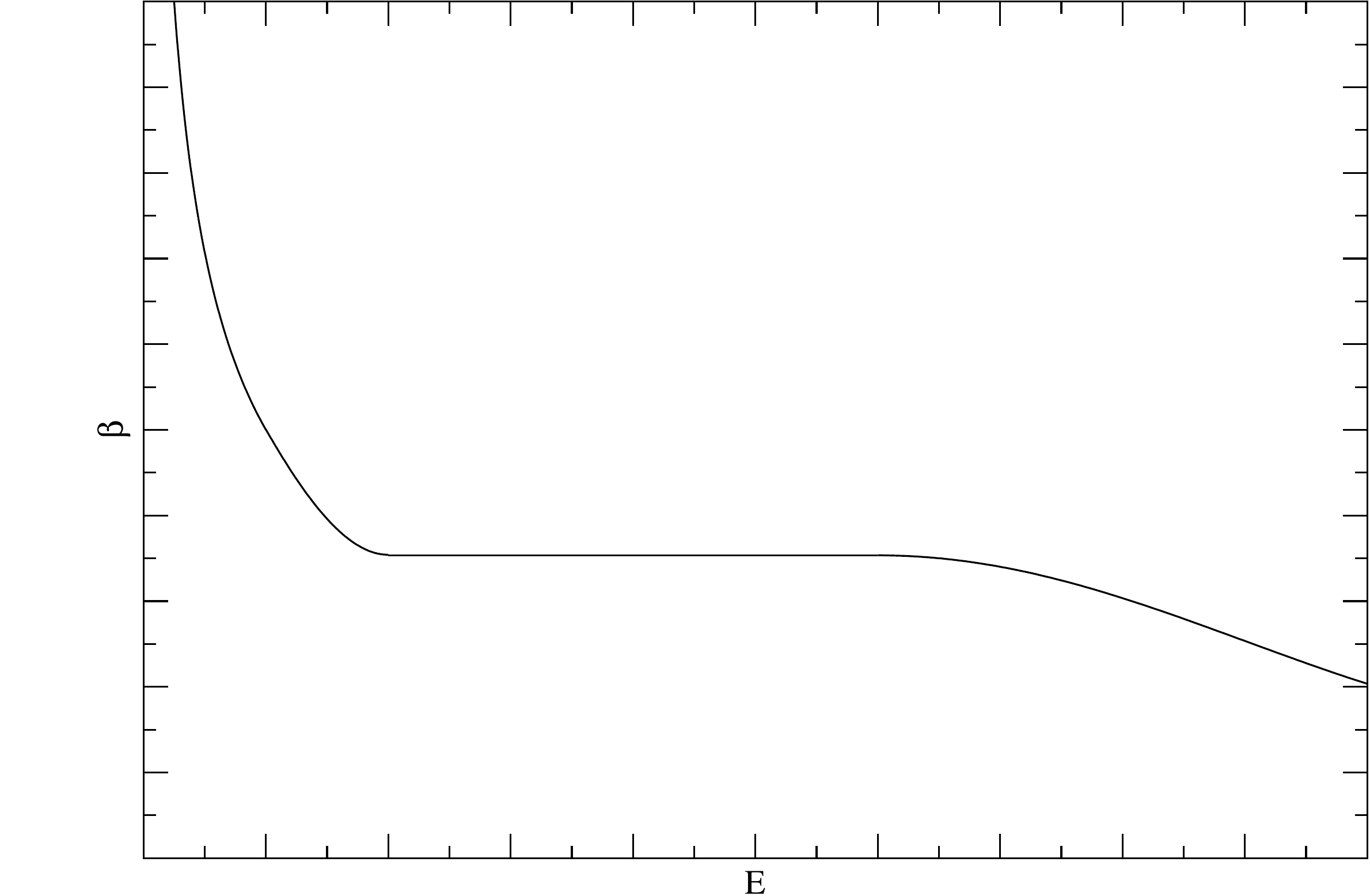}}
\caption{To emphasize the four phases expected in holographic gauge theories, we have drawn a piecewise
$\beta(E)$ function. It is more realistic to expect an approximation to this function that is differentiable
whenever $E > 0$. The function plotted in the middle figure is still difficult to work with because there is a small
range of energies for which it is increasing. Assuming that this phase is negligible is the best way to predict
the behaviour of our nonlinear PDE.}
\label{3-betas}
\end{figure}
If $C(E) = \rho^{-2}(E)$,
\begin{equation}
\frac{\partial E}{\partial t} = - \partial_i \left ( C(E) \rho^2(E) \partial_i \beta(E) \right ) \label{non-filtration-form}
\end{equation}
becomes
\begin{equation}
\frac{\partial E}{\partial t} = - \partial_i \partial_i \tilde{\beta}(E) \label{filtration-form}
\end{equation}
where $\tilde{\beta}(E) = \beta(E)$. We will in fact consider (\ref{filtration-form}) regardless of $C$. This is
because (\ref{non-filtration-form}) is always (\ref{filtration-form}) for some other $\tilde{\beta}$. Simply
define $\tilde{\beta}^{\prime} = C \rho^2 \beta^{\prime}$. Because $C$ and $\rho$ are positive functions,
$\tilde{\beta}$ is decreasing if and only if $\beta$ is. Therefore we will drop the tilde and write
\begin{equation}
\frac{\partial E}{\partial t} = - \Delta \beta(E) \label{filtration-equation}
\end{equation}
from now on. Of course the $\beta(E)$ in (\ref{filtration-equation}) no longer has to be of the form plotted in
Figure \ref{3-betas} but this will be unimportant for most of the results that follow.

\subsection{Basic properties on a bounded domain}
We will start by assuming that $x \in \Omega$ where $\Omega$ is an open, bounded domain in $\mathbb{R}^d$. This allows
the initial energy density $E(x, 0)$ to be integrable without decaying to zero. As shown in Figure
\ref{bounded-domain}, a potential problem with (\ref{filtration-equation}) is thus avoided because the low energies for
which $\beta$ diverges are not realized.
\begin{figure}[h]
\centering
\subfloat[][Energy distribution]{\includegraphics[scale=0.3]{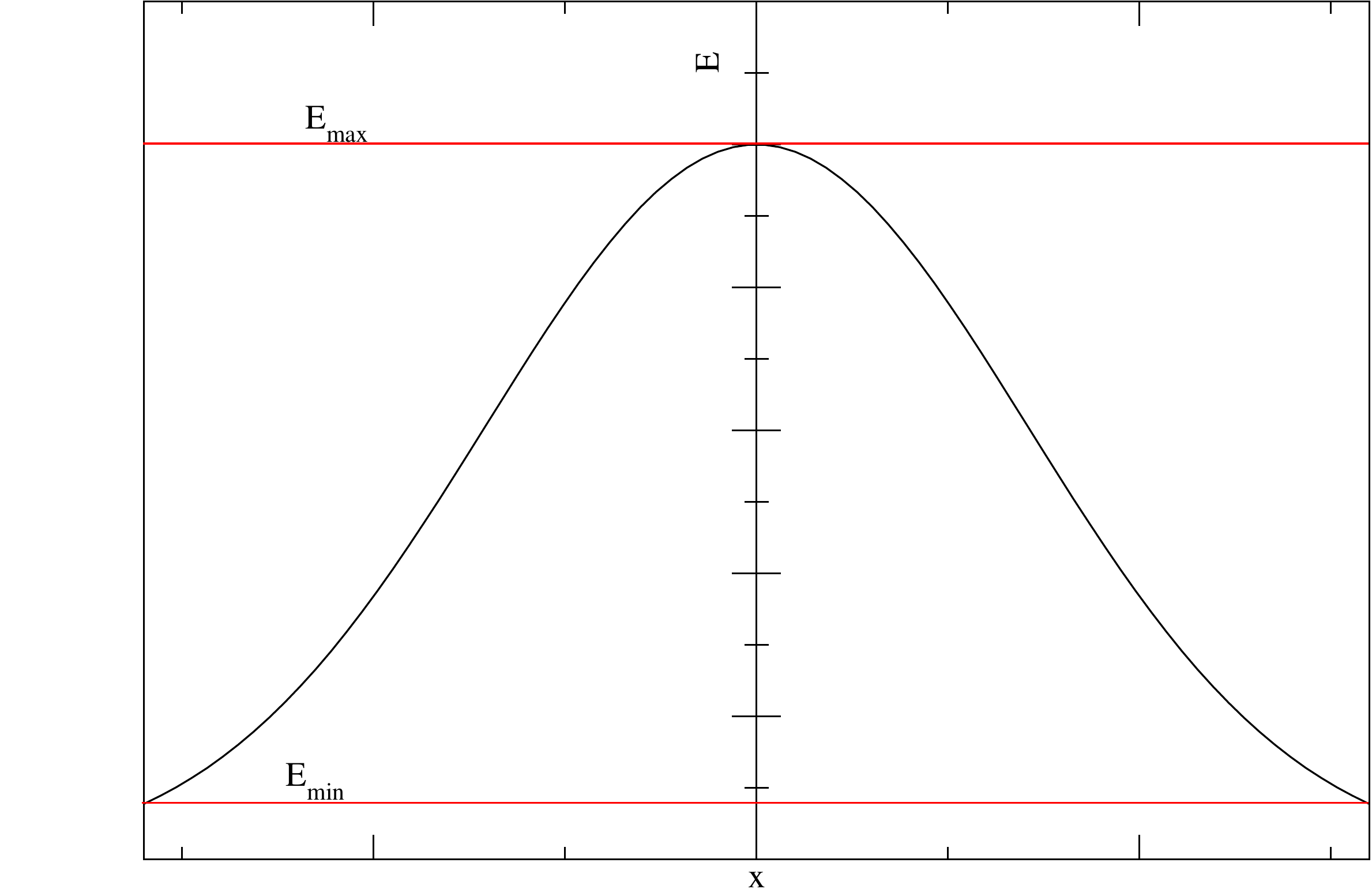}}
\subfloat[][Realized range of energies]{\includegraphics[scale=0.3]{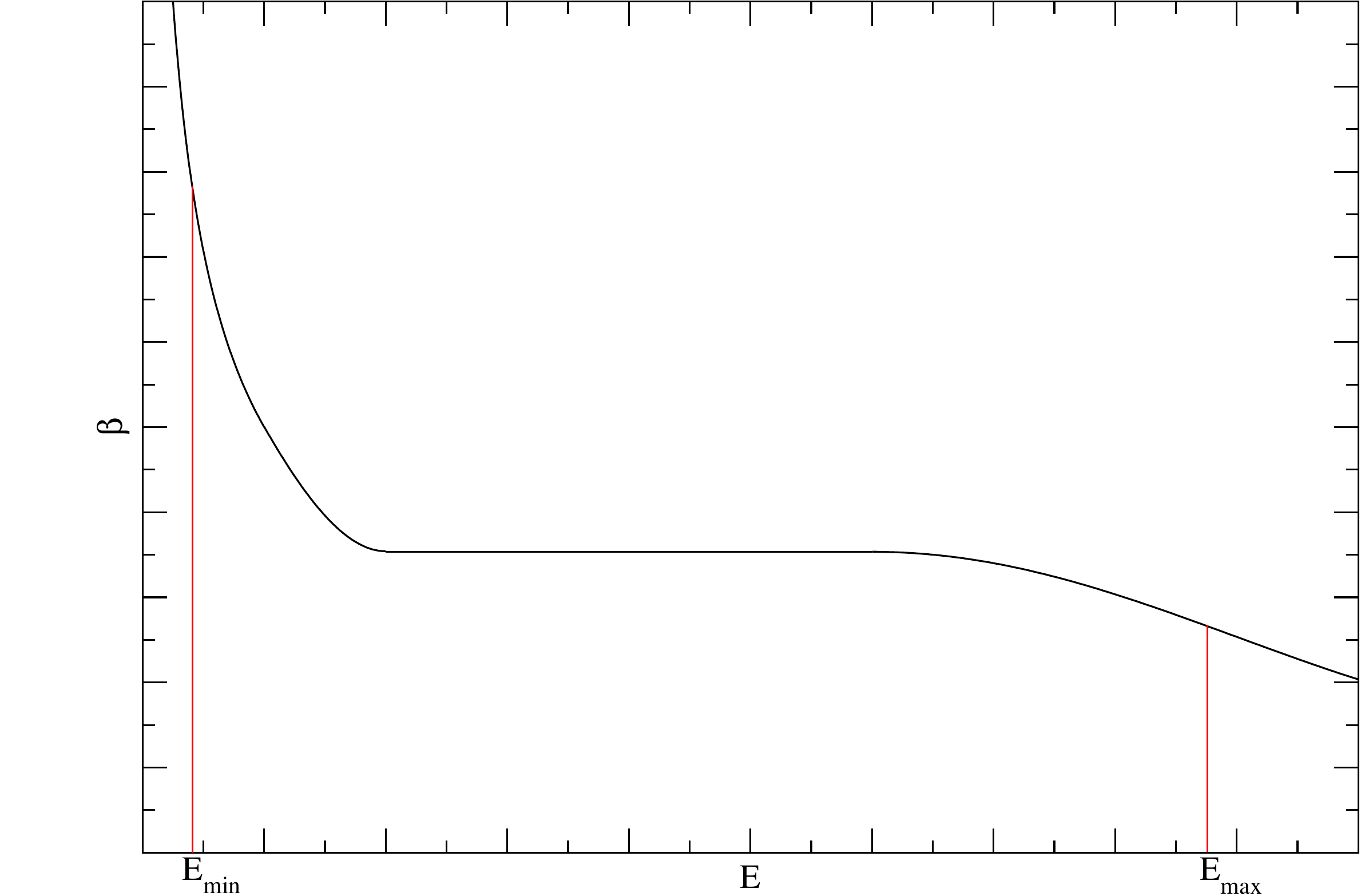}}
\caption{Resricting the size of the domain allows us to avoid the vanishingly small energies for which our PDE
no longer applies.}
\label{bounded-domain}
\end{figure}
The initial condition will be denoted $E_0(x)$ reflecting the fact that we choose an energy distribution to evolve
forward in time, \textit{i.e.} $E_0$ is an input to the Cauchy problem.

\subsubsection{Conservation of energy}
Energy conservation was one of the properties we demanded from the start. As with any Cauchy problem, whether or not
energy is conserved depends on the boundary conditions. As with the heat equation, Neumann boundary conditions
are the appropriate ones to consider. When discussing these mathematical results, we will use ``mass'' to refer to
the integral of $E(x, t)$ over space rather than a gap in the spectrum of a field theory.
\begin{theorem}
If $E$ solves
\begin{equation}
\begin{dcases}
\frac{\partial E}{\partial t} (x, t) = - \Delta \beta(E(x, t)) & (x, t) \in \Omega \times (0, \infty) \\
E(x, 0) = E_0(x) & x \in \Omega \\
\nabla E(x, t) \cdot n = 0 & (x, t) \in \partial \Omega \times (0, \infty)
\end{dcases} \nonumber
\end{equation}
then $M(t) = \int_{\Omega} E(x, t) \textup{d}x$ is constant.
\end{theorem}
\begin{proof}
We can show that the derivative of $M$ is zero using Green's first identity.
\begin{eqnarray}
\frac{\textup{d} M}{\textup{d} t} &=& \int_{\Omega} \frac{\partial E}{\partial t} \textup{d}x \nonumber \\
&=& \int_{\Omega} \Delta \beta(E) \textup{d}x \nonumber \\
&=& \int_{\Omega} \nabla \cdot \left ( \beta^{\prime}(E) \nabla E \right ) \textup{d}x \nonumber \\
&=& \int_{\Omega} \beta^{\prime}(E) \Delta E + \nabla \left ( \beta^{\prime}(E) \right ) \cdot \nabla E \textup{d}x \nonumber \\
&=& \int_{\partial \Omega} \beta^{\prime}(E) \left ( \nabla E \cdot n \right ) \textup{d}S_x \nonumber \\
&=& 0 \; . \nonumber
\end{eqnarray}
Dirichlet boundary conditions would not lead to conserved energy unless we finely tuned $\beta^{\prime}$ to be zero
on the boundary.
\end{proof}

\subsubsection{The maximum principle}
Perhaps the most ubiquitous tool in the study of elliptic and parabolic equations is the maximum principle. Although
it is sometimes introduced as a tool for studying linear equations, many nonlinear versions of it have appeared over
the years \cite{protter}. Our proof of a suitable maximum principle will be very similar to the one in \cite{tao}.
\begin{theorem}
Suppose that $\frac{\partial u}{\partial t} \geq - \Delta \beta(u)$ and
$\frac{\partial v}{\partial t} \leq - \Delta \beta(v)$ for a monotonically decreasing $\beta$. If
$u(x, 0) > v(x, 0)$ for all $x \in \Omega$, there does not exist a spacetime point
$(x_0, t_0) \in \Omega \times (0, \infty)$ for which $u(x_0, t_0) < v(x_0, t_0)$.
\end{theorem}
\begin{proof}
First suppose that $w \in C^2(\Omega \times (0, \infty))$ is initially non-negative but not always non-negative. Then
there must exist some point $(x_0, t_0) \in \Omega \times (0, \infty)$ such that $w(x_0, t_0) < 0$. We can let $x_0$
be the position of the minimum of $w(\cdot, t_0)$ so that $\nabla w(x_0, t_0) = 0$ and $\Delta w(x_0, t_0) \geq 0$
are also satisfied. If $t_0$ is the first time such a point occurs, $\frac{\partial w}{\partial t}(x_0, t_0) \leq 0$.

Now let $w(x, t) = e^{-At} \left [ \beta(v(x, t)) - \beta(u(x, t)) \right ]$ where $A$ is some positive constant.
Saying that $w$ is non-negative is the same as saying that $u \geq v$ because $\beta$ is monotone. Therefore if
$u \geq v$ initially but not always, the above says that there must be a point
$(x_0, t_0) \in \Omega \times (0, \infty)$ such that:
\begin{itemize}
\item $\beta(u(x_0, t_0)) > \beta(v(x_0, t_0))$
\item $\beta^{\prime}(u(x_0, t_0)) \nabla u(x_0, t_0) = \beta^{\prime}(v(x_0, t_0)) \nabla v(x_0, t_0)$
\item $-\Delta \beta(u(x_0, t_0)) \geq -\Delta \beta(v(x_0, t_0))$
\item $\beta^{\prime}(u(x_0, t_0)) \frac{\partial u}{\partial t}(x_0, t_0) - \beta^{\prime}(v(x_0, t_0)) \frac{\partial v}{\partial t}(x_0, t_0) \geq A \left [ \beta(u(x_0, t_0)) - \beta(v(x_0, t_0)) \right ]$
\end{itemize}
Since $u$ is a supersolution and $v$ is a subsolution, $\frac{\partial u}{\partial t} \geq -\Delta \beta(u)$ and
$\frac{\partial v}{\partial t} \leq -\Delta \beta(v)$ can be combined to give
$\frac{\partial}{\partial t}(u - v) + \Delta \left [ \beta(u) - \beta(v) \right ] \geq 0$. On the other hand, our four
conditions above can be combined into an inequality that contradicts this telling us that $u$ cannot drop below $v$ if
$u(x, 0) \geq v(x, 0)$.

First we use the fact that $\Delta \beta(u(x_0, t_0)) \leq \Delta \beta(v(x_0, t_0))$ to write
$\frac{\partial}{\partial t}(u(x_0, t_0) - v(x_0, t_0)) + \Delta \left [ \beta(u(x_0, t_0)) - \beta(v(x_0, t_0)) \right ] \leq \frac{\partial}{\partial t}(u(x_0, t_0) - v(x_0, t_0))$.
We must now show that the difference of the time derivatives is less than zero. To do this, we note that at least
one of the following must be true:
\begin{eqnarray}
\beta^{\prime}(u(x_0, t_0)) \frac{\partial u}{\partial t}(x_0, t_0) - \beta^{\prime}(v(x_0, t_0)) \frac{\partial v}{\partial t}(x_0, t_0) \leq \beta^{\prime}(u(x_0, t_0)) \frac{\partial}{\partial t}(u(x_0, t_0) - v(x_0, t_0)) \nonumber \\
\beta^{\prime}(u(x_0, t_0)) \frac{\partial u}{\partial t}(x_0, t_0) - \beta^{\prime}(v(x_0, t_0)) \frac{\partial v}{\partial t}(x_0, t_0) \leq \beta^{\prime}(v(x_0, t_0)) \frac{\partial}{\partial t}(u(x_0, t_0) - v(x_0, t_0)) \; . \nonumber
\end{eqnarray}
The first one is true if $\frac{\partial v}{\partial t}(x_0, t_0)$ and
$\beta^{\prime}(v(x_0, t_0)) - \beta^{\prime}(u(x_0, t_0))$ have the same sign. The second one is true if they have
different signs. Use $\beta^{\prime}_0$ to denote whatever prefactor appears in the correct statement, either
$\beta^{\prime}(u(x_0, t_0))$ or $\beta^{\prime}(v(x_0, t_0))$. The key is that this is a negative constant. The
last inequality in our list of four can turn this into
\begin{eqnarray}
\beta^{\prime}_0 \frac{\partial}{\partial t} (u(x_0, t_0) - v(x_0, t_0)) &\geq& \beta^{\prime}(u(x_0, t_0)) \frac{\partial u}{\partial t}(x_0, t_0) - \beta^{\prime}(v(x_0, t_0)) \frac{\partial v}{\partial t}(x_0, t_0) \nonumber \\
&\geq& A \left [ \beta(u(x_0, t_0)) - \beta(v(x_0, t_0)) \right ] \; . \nonumber
\end{eqnarray}
Dividing through by $\beta^{\prime}_0$, we get
\begin{equation}
\frac{\partial}{\partial t} (u(x_0, t_0) - v(x_0, t_0)) \leq \frac{A}{\beta^{\prime}_0} \left [ \beta(u(x_0, t_0)) - \beta(v(x_0, t_0)) \right ] \; . \nonumber
\end{equation}
Using the first inequality in the list of four, we see that this is less than zero, completing the proof.
\end{proof}
Even though this could still be generalized further \cite{arena}, it is already more general than the maximum
principle. To get the maximum principle, let $v$ be a solution (a special case of a subsolution) and $u$ be
$\max \{ v(x, 0) : x \in \bar{\Omega} \}$. Since $u$ is a constant, it is also a solution and therefore a supersolution.
The theorem above now tells us that $u$ must continue to upper bound $v$ at all later times which means the maximum
of $v$ decreases with time. As always, this is equivalent to the minimum principle which states that the minimum
increases with time. We have already seen in (\ref{linear-beta-pde}) that our PDE turns into the heat equation when
the energy distribution is close to uniform. The maximum principle tells us that our PDE is similar to the heat
equation in some respects for general energy distributions as well.

In addition to talking about the value of the maximum of $E$, we can gain some information about its position.
If our initial condition $E_0$ is spherically symmetric, our equation gives
\begin{eqnarray}
\frac{\partial E}{\partial t} &=& -r^{1-d} \frac{\partial}{\partial r} \left ( r^{d-1} \frac{\partial \beta(E)}{\partial r} \right ) \nonumber \\
\frac{\partial}{\partial t} \frac{\partial E}{\partial r} &=& -\frac{\partial}{\partial r} \left [ r^{1-d} \frac{\partial}{\partial r} \left ( r^{d-1} \frac{\partial \beta(E)}{\partial r} \right ) \right ] \nonumber \\
&=& -\frac{\partial}{\partial r} \left [ \frac{d-1}{r} \beta^{\prime}(E) \frac{\partial E}{\partial r} + \beta^{\prime\prime}(E) \left ( \frac{\partial E}{\partial r} \right )^2 + \beta^{\prime}(E) \frac{\partial^2 E}{\partial r^2} \right ] \nonumber \\
&=& \frac{d-1}{r^2} \beta^{\prime}(E) \frac{\partial E}{\partial r} - \frac{d-1}{r} \beta^{\prime\prime}(E) \left ( \frac{\partial E}{\partial r} \right )^2 - \frac{d-1}{r} \beta^{\prime}(E) \frac{\partial^2 E}{\partial r^2} - \beta^{\prime\prime\prime}(E) \left ( \frac{\partial E}{\partial r} \right )^3 \nonumber \\
&& - 3\beta^{\prime\prime}(E) \frac{\partial E}{\partial r} \frac{\partial^2 E}{\partial r^2} - \beta^{\prime}(E) \frac{\partial^3 E}{\partial r^3} \; . \nonumber
\end{eqnarray}
Now suppose that $E_0$ achieves its maximum at the origin and has no other local minima or maxima. In other words
$\frac{\partial E_0}{\partial r} < 0$ for $r \neq 0$. If $\frac{\partial E}{\partial r}$ were to become positive away
from the origin at some later time, it would have to first vanish. Also, since an extreme point has not formed yet,
$(r_0, t_0)$ is an inflection point satisfying $\frac{\partial^2 E}{\partial r^2}(r_0, t_0) = 0$ and
$\frac{\partial^3 E}{\partial r^3}(r_0, t_0) < 0$. Plugging these into our expression above, we see that
$\frac{\partial}{\partial t} \frac{\partial E}{\partial r} (r_0, t_0) = -\beta^{\prime}(E(r_0, t_0)) \frac{\partial^3 E}{\partial r^3}(r_0, t_0) < 0$.
If the radial derivative of $E$ ever gets to zero, it starts decreasing again and never passes zero to become positive.
Therefore the origin is the only local extremum at all times.

\subsubsection{Existence of a steady-state}
We now have everything we need to discuss the behaviour of solutions in the limit of infinite time. By setting
$\frac{\partial E}{\partial t}$ equal to zero, we can see that the limiting energy distribution satisfies
\begin{equation}
\begin{dcases}
\Delta \beta(E(x)) = 0 & x \in \Omega \\
\nabla E(x) \cdot n = 0 & x \in \partial \Omega
\end{dcases} \; . \nonumber
\end{equation}
The only solutions to the Neumann problem for Laplace's equation are constants. Therefore $\beta(E)$ must be
constant in space if $\frac{\partial E}{\partial t} = 0$. There are two ways for this to happen. One is for $E$ to
solely occupy the Hagedorn regime. If not, $\beta$ is invertible in a neighbourhood of at least one energy in the
steady-state and $E$ must be constant. By energy conservation, the value of this constant is of course
\begin{equation}
\bar{E_0} = \frac{1}{|\Omega|} \int_{\Omega} E_0(x) \textup{d}x \; . \nonumber
\end{equation}
In other words, we know what $\lim_{t \rightarrow \infty} E(\cdot, t)$ must be \textit{if it exists} but it is not
yet obvious that it exists. Even for an energy distribution whose extrema stay in the same place and smooth out over
time, it is possible for the intermediate regions to constantly oscillate without ever converging to any function. Our
proof that this does not happen for (\ref{filtration-equation}) will compare it to the heat equation.
\begin{figure}[h]
\includegraphics[scale=0.45]{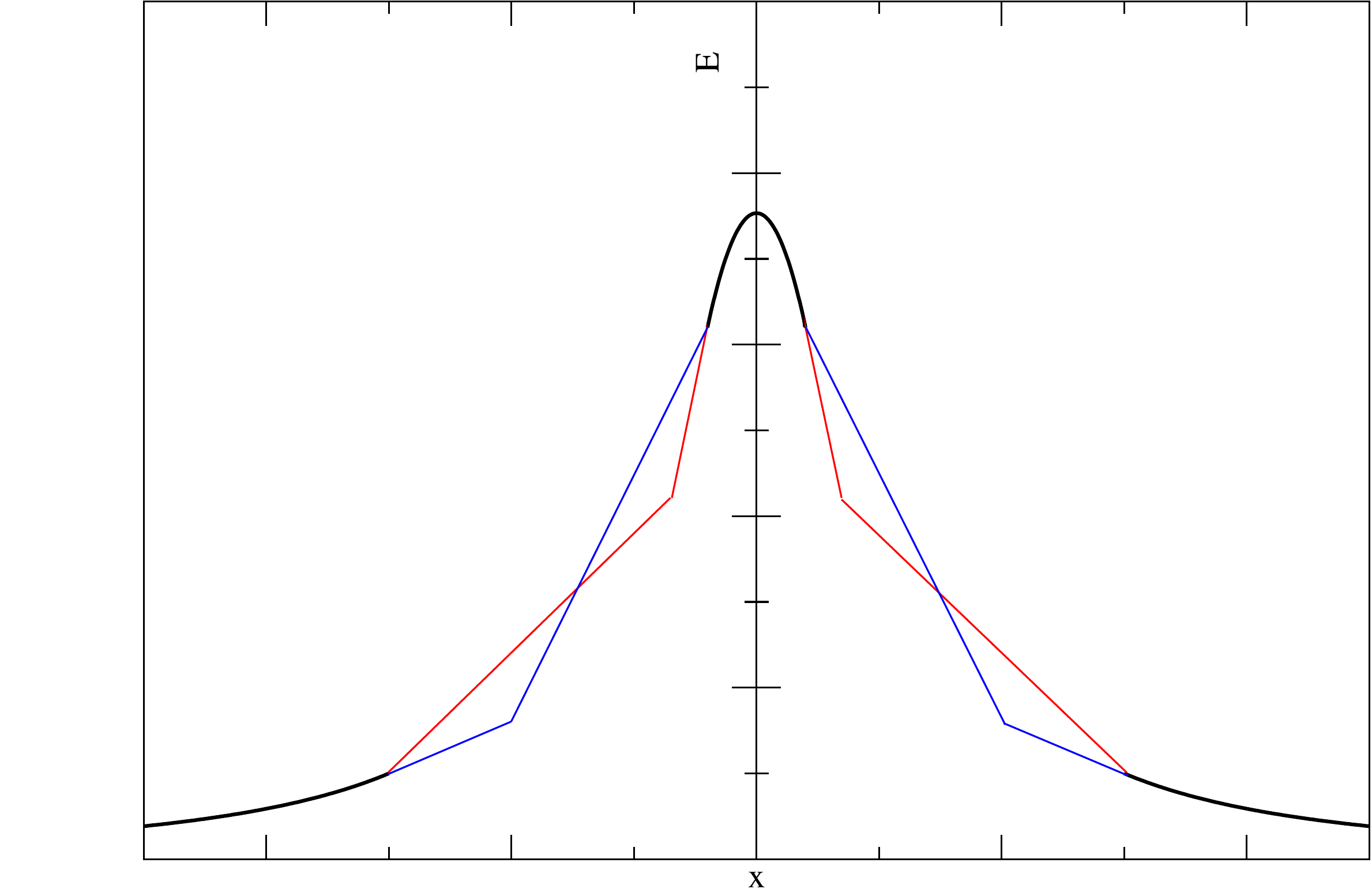}
\caption{The two energy distributions shown have equal minima, equal maxima and equal masses. They also have the
central peak as their only extreme value. A solution to (\ref{filtration-equation}) that alternated between these
functions indefinitely would satisfy all of the properties that we have proven so far but not have a steady-state.}
\label{esoteric-case}
\end{figure}
\begin{theorem}
Suppose that $E$ solves
\begin{equation}
\begin{dcases}
\frac{\partial E}{\partial t} (x, t) = - \Delta \beta(E(x, t)) & (x, t) \in \Omega \times (0, \infty) \\
E(x, 0) = E_0(x) & x \in \Omega \\
\nabla E(x, t) \cdot n = 0 & (x, t) \in \partial \Omega \times (0, \infty)
\end{dcases} \nonumber
\end{equation}
for a spherically symmetric $E_0$ with a central maximum as its only local extreme point. If $\beta^{\prime}(E) \leq 0$
becomes strict whenever $E_{\mathrm{min}} < E < \bar{E_0}$, then $E$ converges to its average.
\end{theorem}
\begin{proof}
Letting $B^d_R$ denote a ball of radius $R$, the mass satisfies
\begin{eqnarray}
M_R &=& \int_{B^d_R} E \textup{d}x \nonumber \\
\frac{\textup{d} M_R}{\textup{d}t} &=& \int_{B^d_R} \frac{\partial E}{\partial t} \textup{d}x \nonumber \\
&=& -\int_{B^d_R} \Delta \beta(E) \textup{d}x \nonumber \\
&=& -\int_{\mathbb{S}^{d-1}_R} \beta^{\prime}(E) \nabla E \cdot n \textup{d}S_x \nonumber \\
&\leq& 0 \; . \nonumber
\end{eqnarray}
The inequality came from the fact that the integrand is non-negative. The flux is non-positive because the value of $E$
must always decrease as we move away from the origin and $\beta^{\prime}(E)$ is non-positive by assumption. Let
$R(t, \epsilon)$ be the largest radius at which $E(R, t) = \bar{E_0} - \epsilon$. If $R(t)$ ever reaches the radius of
$\Omega$, the solution will have converged to $\bar{E_0}$ by energy conservation and the maximum principle. Therefore
assume the opposite and let $R = \max \{ R(t, \epsilon) : t > 0 \}$. Since $E(R, t) \leq \bar{E_0}$, we can make the
integral larger by replacing $\beta^{\prime}(E)$ with $\beta^{\prime}(\bar{E_0})$.
\begin{equation}
\frac{\textup{d} M_R}{\textup{d}t} \leq -\beta^{\prime}(\bar{E_0}) \int_{\mathbb{S}^{d-1}_R} \nabla E \cdot n \textup{d}S_x < 0 \nonumber
\end{equation}
We now see that $M_R$ for our equation (\ref{filtration-equation}) shrinks more quickly than $M_R$ for a heat equation
with diffusion constant $-\beta^{\prime}(\bar{E_0})$. The standard Fourier series method shows that the solution to
said heat equation with initial data $E_0$ will converge to $\bar{E_0}$. Thus $\mathbb{S}^{d-1}_R$, the sphere where
$\bar{E_0}$ is first achieved, eventually encloses a constant function. By energy conservation and the maximum
principle, $E$ must be equal to the same constant outside this sphere as well. Therefore for any $\epsilon > 0$, the
energy at the edge always comes within $\epsilon$ of $\bar{E_0}$ which means that ``the outside of this sphere'' must
have not existed all along.
\end{proof}

\subsection{Time scales}
So far we have shown that the nonlinear diffusion we are interested in shares a number of intuitive properties with
linear diffusion: energy is conserved, peaks smooth out over time and reasonable configurations of energy will decay to
a distribution that is completely uniform. However, we still have relatively little information on \textit{how quickly}
these energy distributions decay. Defining and bounding time scales for the diffusion will allow us to compare our
model to what is already known about gravitational bound states.

\subsubsection{The concentration comparison theorem}
\begin{figure}[h]
\includegraphics[scale=0.45]{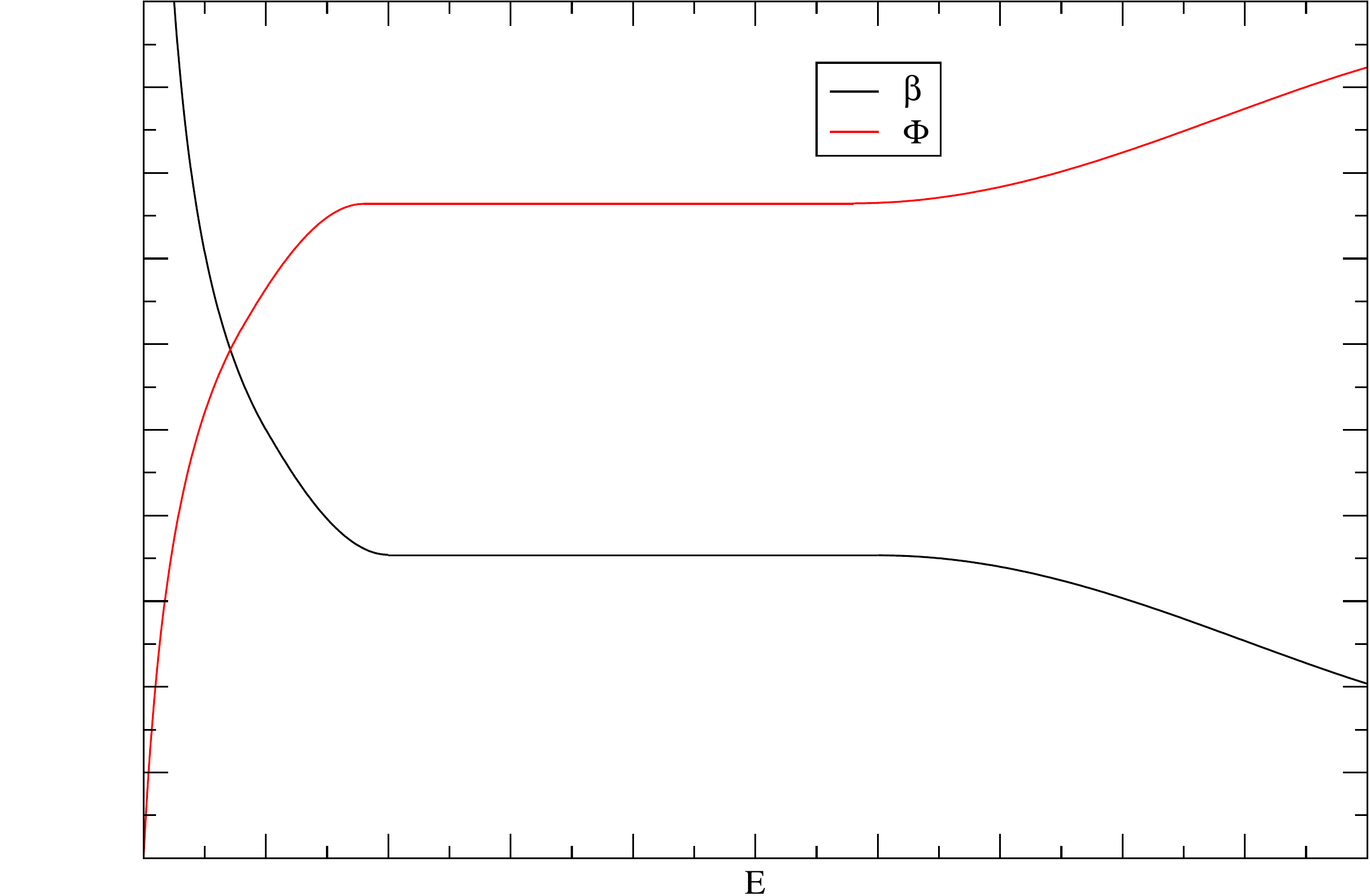}
\caption{This shows how we must change the function that appears inside the Laplacian in order to satisfy the
hypotheses of the theorem. Bounding the time scales arising from $\beta$ is no more difficult than bounding the
time scales arising from $\Phi$.}
\label{beta-to-phi}
\end{figure}
If we continued applying the basic properties, we would be able to derive some powerful results that are found in the
literature. One of these is the concentration comparison theorem \cite{vazquez1} which applies to equations of the form
(\ref{filtration-equation}) called \textit{filtration-equations}. The only difference is that filtration equations
are typically written $\frac{\partial E}{\partial t} = \Delta \Phi(E)$ where $\Phi$ is weakly increasing. A
fundamental property of diffusion is that the mass contained within a fixed ball at the origin $M_R(t)$ should decrease.
The total mass $M(t)$ which we found to be constant in time is $M_{\infty}(t)$. The concentration comparison theorem,
which we state below, is used to convert an inequality involving $\Phi$ to an inequality involving $M_R$.
\begin{theorem}
Let $\Phi_1, \Phi_2 \in C^1(\mathbb{R})$ be increasing functions sending $0$ to $0$ such that
$\Phi_1^{\prime} \geq \Phi_2^{\prime}$. Suppose that $u_1$ having mass $M_{1, R}(t)$ and $u_2$ having mass
$M_{2, R}(t)$ are solutions to $\frac{\partial u}{\partial t} = \Delta \Phi_1(u)$
and $\frac{\partial u}{\partial t} = \Delta \Phi_2(u)$ respectively with spherically symmetric data.
If $M_{1, R}(0) \leq M_{2, R}(0)$ for all $R$ then $M_{1, R}(t) \leq M_{2, R}(t)$ for all
$R$ and $t$.
\end{theorem}
To use this theorem, we must convert our function $\beta(E)$ into something satisfying $\Phi(0) = 0$. This is easily
done with $\Phi(E) = \beta(E_{\mathrm{min}}) - \beta(E + E_{\mathrm{min}})$. Figure \ref{beta-to-phi} shows what
kind of function $\Phi$ is if $\beta$ is the function in Figure \ref{3-betas}.
\begin{figure}[h]
\includegraphics[scale=0.45]{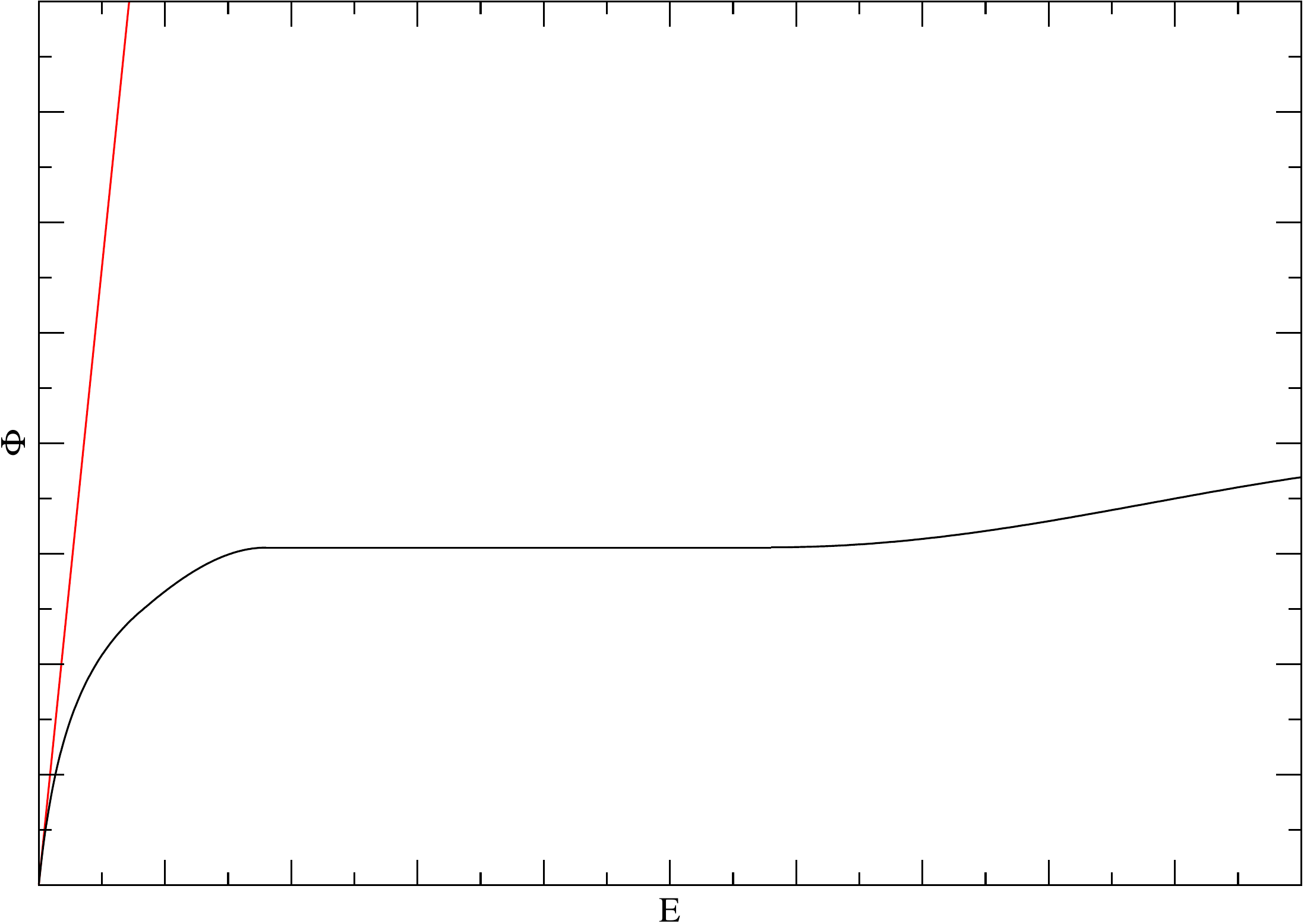}
\caption{We were able to show that our filtration equation had a well defined steady-state by comparing it to a heat
equation. With this linear function as our bounding function, an estimate for the decay time can also be obtained.
However, this would lead to a crude bound because the linear function above is not very similar to the true $\Phi$.}
\label{crude-bound}
\end{figure}
\begin{equation}
\frac{\partial}{\partial t} \left ( E + E_{\mathrm{min}} \right ) = -\Delta \beta \left ( E + E_{\mathrm{min}} \right ) \nonumber
\end{equation}
is now the same as saying
\begin{equation}
\frac{\partial E}{\partial t} = \Delta \Phi(E) \nonumber
\end{equation}
where $E_0$ is now some energy distribution equal to zero on $\partial \Omega$. If we can find an exact solution to
another filtration equation whose mass shrinks more quickly than that of $E$, we will have found a ``decay time'' that
is shorter than the one we are looking for. Similarly, a mass that shrinks more slowly would be associated with a longer
``decay time''. Time scales can therefore be determined if we choose bounding functions that are ``steeper''
than $\Phi$ or ``flatter'' than $\Phi$. The crudest thing we could do is find a lower bound on the decay time by
drawing a linear function above $\Phi$. The next logical step is to find proper estimates for the
decay time by comparing our filtration equation to something more non-trivial.

\subsubsection{Estimates in one dimension}
As stated before, we are considering a density of states that has two microcanonical phase transitions. We will call
the energy of the first one $E_{\mathrm{H}}$ for Hagedorn and the energy of the second one $E_{\mathrm{F}}$ for field
theory. A useful definition of decay time for us will be the time required for the maximum of an energy distribution
to descend from $E_{\mathrm{F}}$ to $E_{\mathrm{H}}$. Specifically, we will look at the initial condition
\begin{equation}
E_0(x) =
\begin{cases}
E_{\mathrm{F}} & |x_1|, \dots, |x_d| < a \\
0 & \mathrm{otherwise}
\end{cases}
\label{toy-condition}
\end{equation}
and see how long it takes until $E(x, t) < E_{\mathrm{H}}$ for all values of $x$. The steep vertical jump through the
Hagedorn regime is a feature that (\ref{toy-condition}) has in common with sharply peaked initial energy distributions.
The sharply peaked functions do not need to be constant at $E_{\mathrm{F}}$ but this will happen anyway once they are
allowed to evolve. The distributions in Figure \ref{non-square} will flatten out in a relatively short time because
diffusion dominates above $E_{\mathrm{F}}$ but not below. This suggests that the decay time for a general peak is
dominated by the decay time of (\ref{toy-condition}). Further evidence that (\ref{toy-condition}) is representative of
more general initial conditions will appear in the next chapter on numerics.
\begin{figure}[h]
\includegraphics[scale=0.45]{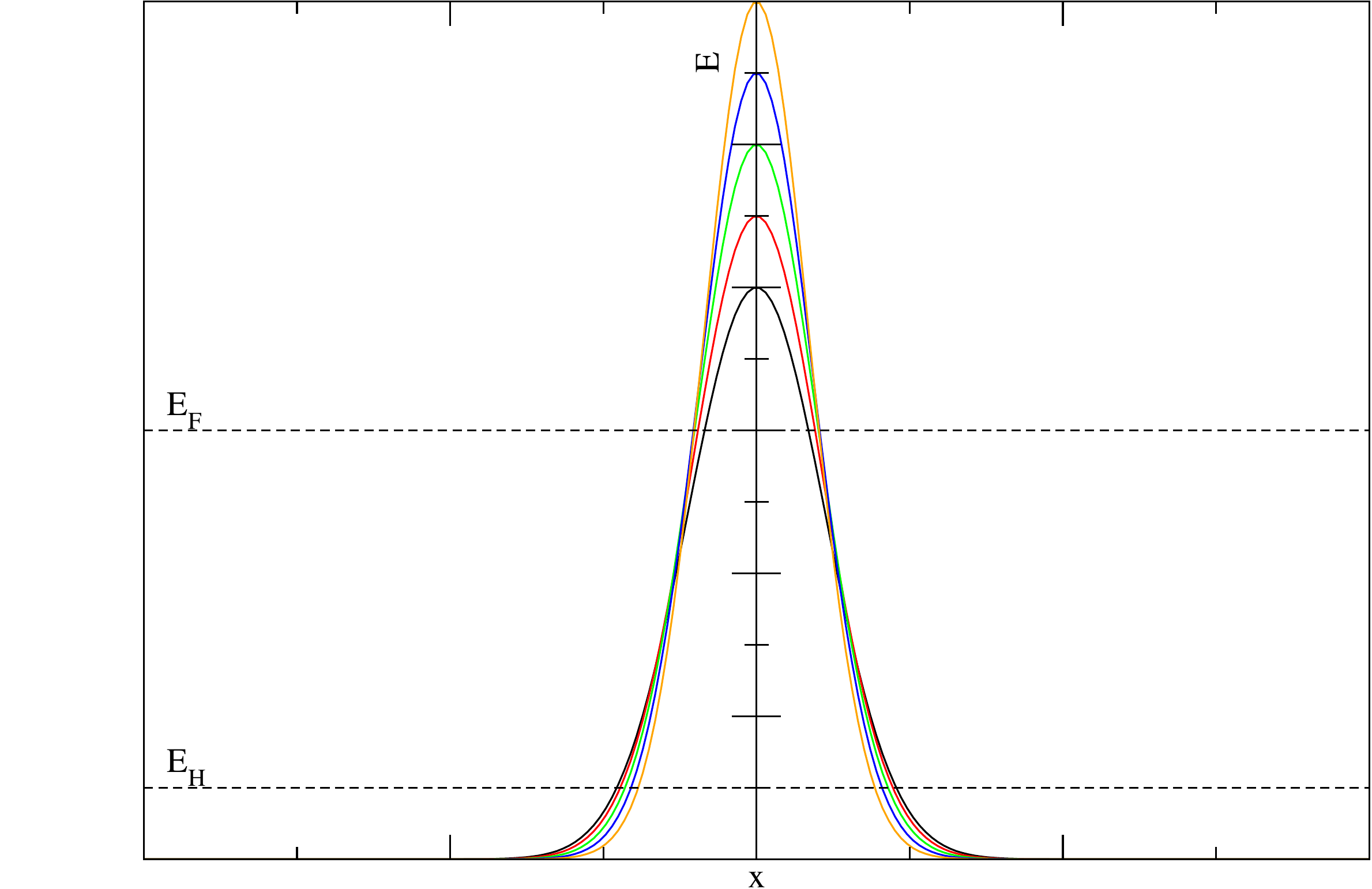}
\caption{Except for the piece on top that will quickly diffuse, these energy distributions are qualitatively similar
to the step function we are considering.}
\label{non-square}
\end{figure}
Finding the decay time of (\ref{toy-condition}) with respect to a general filtration function is still too hard, so we
will pick a simple filtration function with the goal of using the concentration comparison theorem. The one to pick is
\begin{equation}
\Phi(E) =
\begin{cases}
E & E < E_{\mathrm{H}} \\
E_{\mathrm{H}} & E \geq E_{\mathrm{H}}
\end{cases}
\; . \label{simplified-phi}
\end{equation}
If we take an energy profile solving $\frac{\partial E}{\partial t} = \Delta \Phi(E)$ and look at it at some instant
of time, the parts with energy above $E_{\mathrm{H}}$ should be static and the parts with energy below
$E_{\mathrm{H}}$ should be satisfying the heat equation. Energy crosses $E_{\mathrm{H}}$ at a particular distance
$x^*(t)$. Using the method of \cite{mvr}, we will construct a one-dimensional solution where $x^*$ starts off at $a$
and eventually shrinks to zero indicating that the decay time $T$ has been reached.
\begin{equation}
E(x, t) =
\begin{cases}
E_{\mathrm{F}} & |x| < x^*(t) \\
F(x, t) & |x| > x^*(t)
\end{cases}
\label{1d-solution-ansatz}
\end{equation}
where $F$ solves the regular heat equation. Clearly $F$ cannot be just any solution to the heat equation. To obey the
initial condition (\ref{toy-condition}), we must have $F(x, 0) = 0$ for $|x| > a$. Also, we must impose conservation
of energy. The mass contained between $0$ and $x^*(t)$ is $x^*(t) E_{\mathrm{F}}$. The mass everywhere else is
$\int_{x^*(t)}^{\infty} F(x, t) \textup{d}x$. We want their rates of change to be equal and opposite so
\begin{eqnarray}
E_{\mathrm{F}} \frac{\textup{d}x^*(t)}{\textup{d}t} &=& -\frac{\textup{d}}{\textup{d}t} \int_{x^*(t)}^{\infty} F(x, t) \textup{d}x \nonumber \\
&=& F(x^*(t), t) \frac{\textup{d}x^*(t)}{\textup{d}t} - \int_{x^*(t)}^{\infty} \frac{\partial F}{\partial t}(x, t) \textup{d}x \nonumber \\
&=& E_{\mathrm{H}} \frac{\textup{d}x^*(t)}{\textup{d}t} - \int_{x^*(t)}^{\infty} \Delta F(x, t) \textup{d}x \nonumber \\
&=& E_{\mathrm{H}} \frac{\textup{d}x^*(t)}{\textup{d}t} + \frac{\partial F}{\partial x} (x^*(t), t) \nonumber
\end{eqnarray}
where we have differentiated under the integral sign. Therefore
\begin{equation}
\left ( E_{\mathrm{F}} - E_{\mathrm{H}} \right ) \frac{\textup{d}x^*(t)}{\textup{d}t} = \frac{\partial F}{\partial x} (x^*(t), t) \label{1d-mass-matching}
\end{equation}
is a necessary condition for the ansatz (\ref{1d-solution-ansatz}) to work. It is also sufficient as we show in the
appendix. The initial condition for $F$ that will give us these necessary and sufficient conditions is
$F(x, 0) = A \Theta(x - a)$ where $A$ is a yet undetermined constant that will turn out to be between $E_{\mathrm{H}}$
and $2E_{\mathrm{H}}$. Figure \ref{gluing-solutions} shows the basic setup.
\begin{figure}[h]
\centering
\subfloat[][$t = 0$]{\includegraphics[scale=0.3]{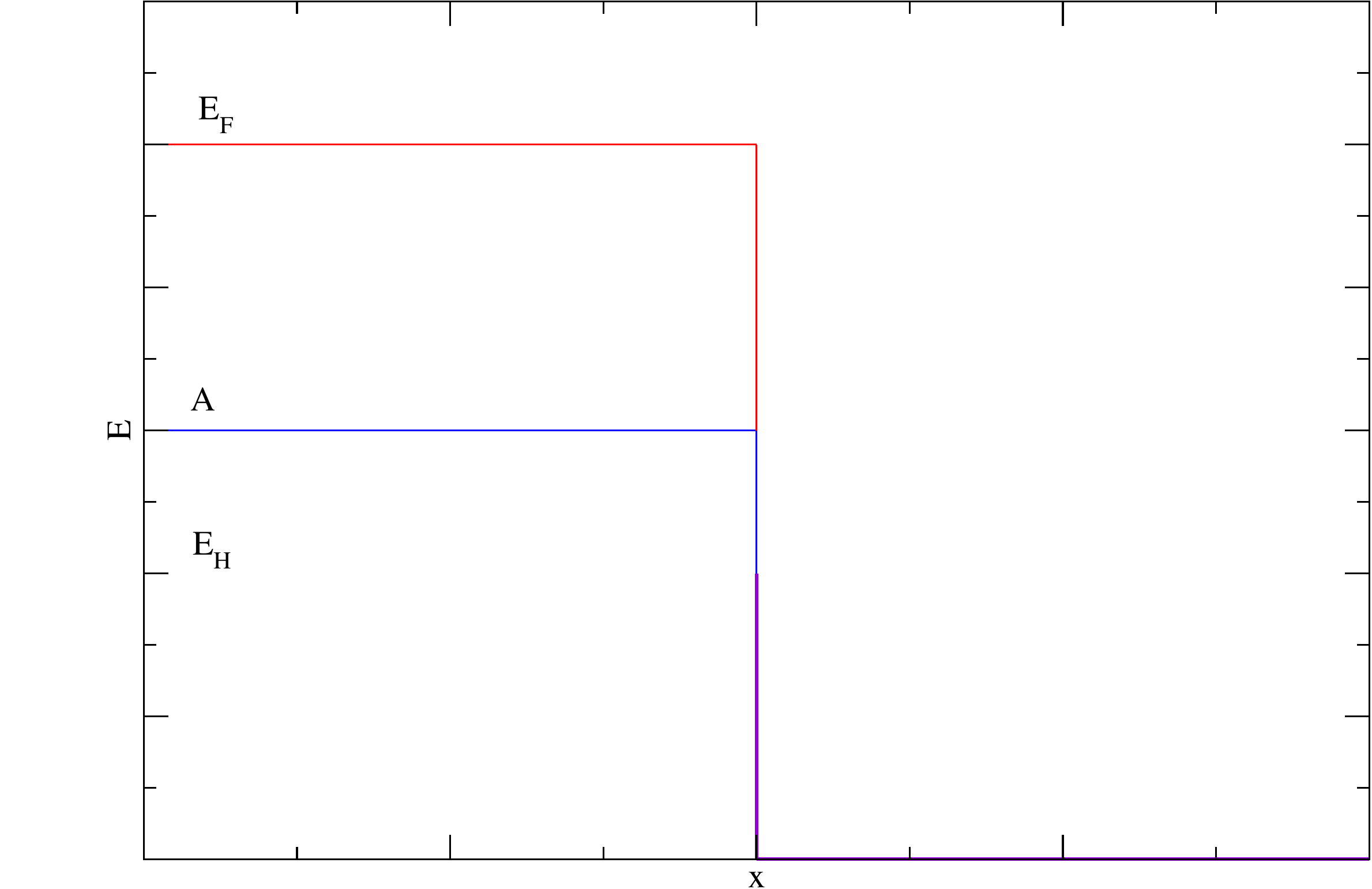}}
\subfloat[][Some later time]{\includegraphics[scale=0.3]{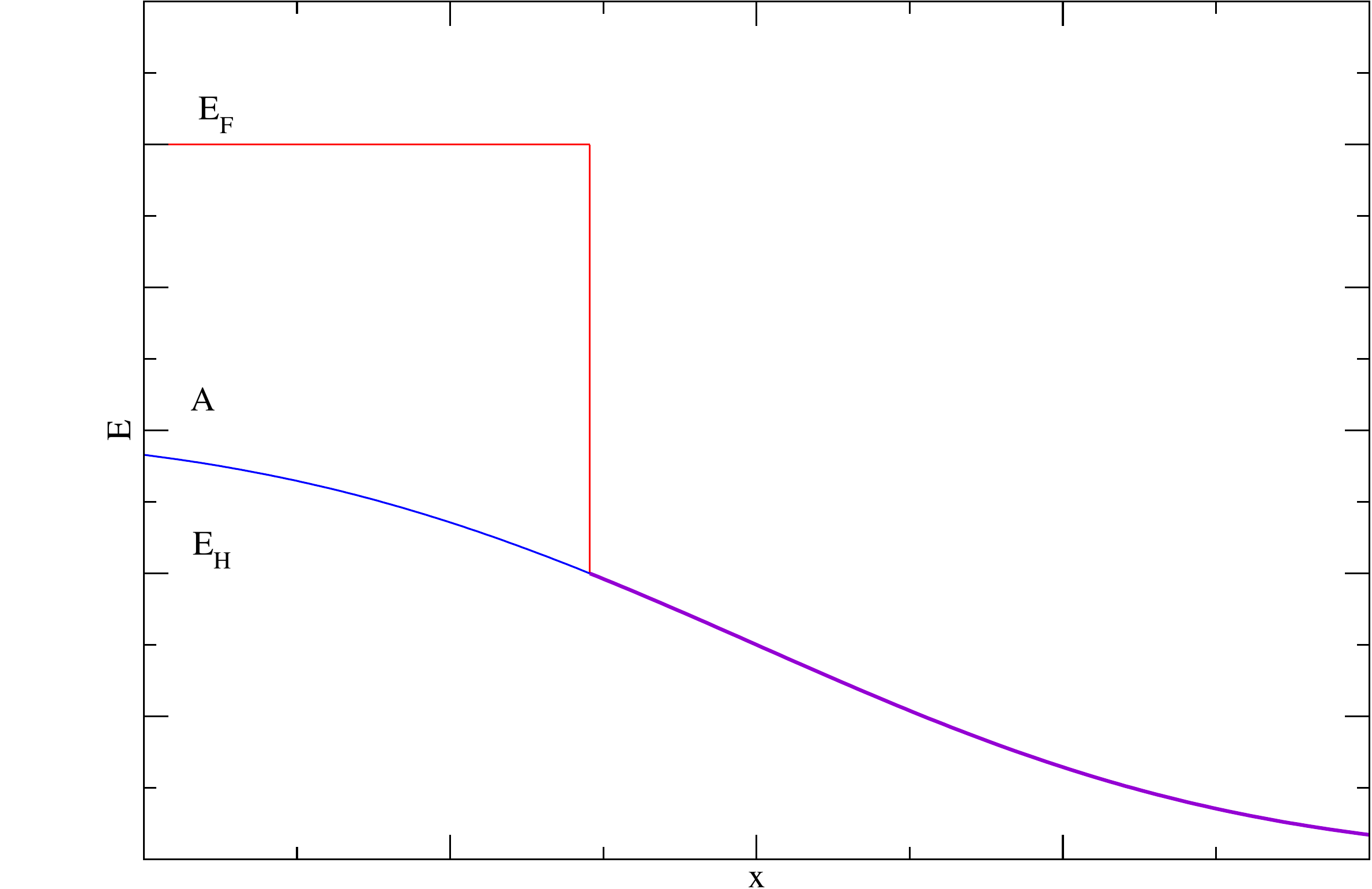}}
\caption{The union of the red and purple curves is what we are looking for. It is (\ref{1d-solution-ansatz}), the
solution to the (\ref{simplified-phi}) filtration equation for the (\ref{toy-condition}) Cauchy data. The solution to
the auxillary problem that we use to find it is the union of the blue and purple curves. This is $F$, the solution to
the heat equation for the Heaviside Cauchy data.}
\label{gluing-solutions}
\end{figure}
To solve for $F$, we first note that the heat kernel is
\begin{equation}
\Phi_{\mathrm{H}}(x,t) = \frac{1}{(4\pi t)^{\frac{d}{2}}} e^{-\frac{x^2}{4t}} \; . \nonumber
\end{equation}
The heat equation's Cauchy problem is solved by taking the convolution of the initial condition with the heat kernel.
Therefore
\begin{eqnarray}
F(x, t) &=& \frac{A}{\sqrt{4\pi t}} \int_{-\infty}^{\infty} \theta(y - a) e^{-\frac{(x - y)^2}{4t}} \textup{d}y \nonumber \\
&=& \frac{A}{\sqrt{4\pi t}} \int_{-\infty}^a e^{-\frac{(x - y)^2}{4t}} \textup{d}y \nonumber \\
&=& \frac{A}{2} \left ( 1 + \erf \left ( \frac{a - x}{2\sqrt{t}}\right ) \right ) \; . \label{1d-auxillary-solution}
\end{eqnarray}
Differentiating $F$ is now straightforward and the position of the interface can be found by setting
$F(x, t) = E_{\mathrm{H}}$. From this we obtain
\begin{eqnarray}
\frac{\partial F}{\partial x}(x, t) &=& -\frac{A}{\sqrt{4\pi t}} e^{-\frac{(x - a)^2}{4t}} \label{1d-f-derivative} \\
x^*(t) &=& a - 2\sqrt{t} \erf^{-1} \left ( \frac{2E_{\mathrm{H}}}{A} - 1 \right ) \label{1d-f-interface}
\end{eqnarray}
and the ratio between (\ref{1d-f-derivative}) at the interface and the derivative of (\ref{1d-f-interface}) is constant.
\begin{equation}
\frac{\frac{\partial F}{\partial x}(x^*(t), t)}{\frac{\textup{d}x^*(t)}{\textup{d}t}} = \frac{A}{\sqrt{4\pi}} \frac{e^{-\erf^{-1} \left ( \frac{2E_{\mathrm{H}}}{A} - 1 \right )^2}}{\erf^{-1} \left ( \frac{2E_{\mathrm{H}}}{A} - 1 \right )} \nonumber
\end{equation}
We must set this equal to $E_{\mathrm{F}} - E_{\mathrm{H}}$. The $A$ we obtain from doing so is given by
\begin{eqnarray}
A &=& \frac{2 E_{\mathrm{H}}}{1 + \erf(I)} \in \left [ E_{\mathrm{H}}, 2E_{\mathrm{H}} \right ] \label{a-value} \\
\sqrt{\pi}Ie^{I^2}(1 + \erf(I)) &=& \frac{E_{\mathrm{H}}}{E_{\mathrm{F}} - E_{\mathrm{H}}} \; . \label{transcendental-equation}
\end{eqnarray}
The transcendental equation defining $I$ can be solved because the range of the left hand side includes all positive
real numbers. The last thing we need to do is find a more explicit value of $A$ for the case when
$E_{\mathrm{F}} \gg E_{\mathrm{H}}$. If we were to consider $\sqrt{\pi}Ie^{I^2}(1 + \erf(I)) = 0$, the solution would
simply be $I = 0$. Since $E_{\mathrm{H}}$ is close to zero but not quite, it is appropriate to linearize the left hand
side of (\ref{transcendental-equation}) giving us $\sqrt{\pi}I = \frac{E_{\mathrm{H}}}{E_{\mathrm{F}} - E_{\mathrm{H}}}$.
Plugging this into (\ref{a-value}) gives
\begin{equation}
A \approx \frac{2 E_{\mathrm{H}}}{1 + \erf \left ( \frac{E_{\mathrm{H}} / \sqrt{\pi}}{E_{\mathrm{F}} - E_{\mathrm{H}}} \right )} \; . \nonumber
\end{equation}
Finally, the time scale can be found by setting $x^*(T) = 0$ or $F(0, T) = E_{\mathrm{H}}$.
\begin{eqnarray}
T &=& \frac{1}{4} \left ( \frac{a}{\erf^{-1} \left ( \frac{2E_{\mathrm{H}}}{A} - 1 \right )} \right )^2 \nonumber \\
&\approx& \frac{1}{4} \left ( \frac{a \left ( E_{\mathrm{F}} - E_{\mathrm{H}} \right )}{E_{\mathrm{H}} / \sqrt{\pi}} \right )^2 \nonumber \\
&\approx& \frac{\pi}{4} \left [ \frac{a E_{\mathrm{F}}}{E_{\mathrm{H}}} \right ]^2 \label{1d-timescale}
\end{eqnarray}

\subsubsection{Higher dimensional generalization}
The procedure above works in an arbitrary number of dimensions if the initial condition has energy $E_{\mathrm{F}}$
inside a hypercube of side length $2a$. We will take $F(x, 0) = A \theta(x_1 - a) \dots \theta(x_d - a)$, solve the
heat equation for $F$ with this Cauchy data and write
\begin{equation}
E(x, t) =
\begin{cases}
E_{\mathrm{F}} & x_1 < x_1^*(t), \dots, x_d < x_d^*(t) \\
F(x, t) & \mathrm{otherwise}
\end{cases}
\; . \label{nd-solution-ansatz}
\end{equation}
Performing the convolution with the heat kernel is straightforward because the solution has a product form.
\begin{eqnarray}
F(x, t) &=& \frac{A}{(4\pi t)^{\frac{d}{2}}} \int_{-\infty}^{\infty} \dots \int_{-\infty}^{\infty} \theta(y_1 - a) \dots \theta(y_d - a) e^{-\frac{(x - y)^2}{4t}} \textup{d}y_1 \dots \textup{d}y_d \nonumber \\
&=& \frac{A}{(4\pi t)^{\frac{d}{2}}} \int_{-\infty}^a \dots \int_{-\infty}^a e^{-\frac{(x_1 - y_1)^2 + \dots + (x_d - y_d)^2}{4t}} \textup{d}y_1 \dots \textup{d}y_d \nonumber \\
&=& \frac{A}{2^d} \left ( 1 + \erf \left ( \frac{a - x_1}{2\sqrt{t}}\right ) \right ) \dots \left ( 1 + \erf \left ( \frac{a - x_d}{2\sqrt{t}}\right ) \right ) \label{nd-auxillary-solution}
\end{eqnarray}
From this we see that the surface $x^*(t)$ where $F$ takes on the value $E_{\mathrm{H}}$ is the locus of points
\begin{equation}
x^*(t) = \left [
\begin{tabular}{c}
$a - 2 \sqrt{t} \erf^{-1} \left ( z_1 - 1 \right )$ \\
\vdots \\
$a - 2 \sqrt{t} \erf^{-1} \left ( z_d - 1 \right )$
\end{tabular}
\right ] \label{nd-f-interface}
\end{equation}
where $\prod_i z_i = 2^d \frac{E_{\mathrm{H}}}{A}$. Differentiating $F$ and evaluating it on this surface, we get
\begin{equation}
\frac{\partial F}{\partial x_i}(x^*(t), t) = \frac{2E_{\mathrm{H}}}{z_i} \frac{1}{\sqrt{4\pi}} e^{-\erf^{-1} (z_i - 1)^2} \; . \label{nd-f-derivative}
\end{equation}
The analogue of (\ref{1d-mass-matching}) to the higher dimensional case is
\begin{equation}
\left ( E_{\mathrm{F}} - E_{\mathrm{H}} \right ) \frac{\textup{d}V^*(t)}{\textup{d}t} = \int_{x^*(t)} \nabla F(x, t) \cdot n \textup{d}S_x \label{nd-mass-matching}
\end{equation}
where $V^*(t)$ is $a^d$ minus the volume enclosed by the surface (\ref{nd-f-interface}) in the positive orthant. This
can be visualized in Figure \ref{positive-orthant}. Clearly $V^*(0) = 0$. If we want to calculate the volume
$V^*(t)$, we are looking for $x_i$ values that go from the $x_i^*(t)$ curve to $a$ instead of from
$0$ to the $x_i^*(t)$ curve. This is the same as saying that
\begin{eqnarray}
z_1 &\in& \left [ 1, 2^d \frac{E_{\mathrm{H}}}{A} \right ] \nonumber \\
z_2 &\in& \left [ 1, \frac{2^d}{z_1} \frac{E_{\mathrm{H}}}{A} \right ] \nonumber \\
&\vdots& \nonumber \\
z_d &\in& \left [ 1, \frac{2^d}{z_1 \dots z_{d-1}} \frac{E_{\mathrm{H}}}{A} \right ] \; . \nonumber
\end{eqnarray}
\begin{figure}[h]
\centering
\subfloat[][Early time]{\includegraphics[scale=0.4]{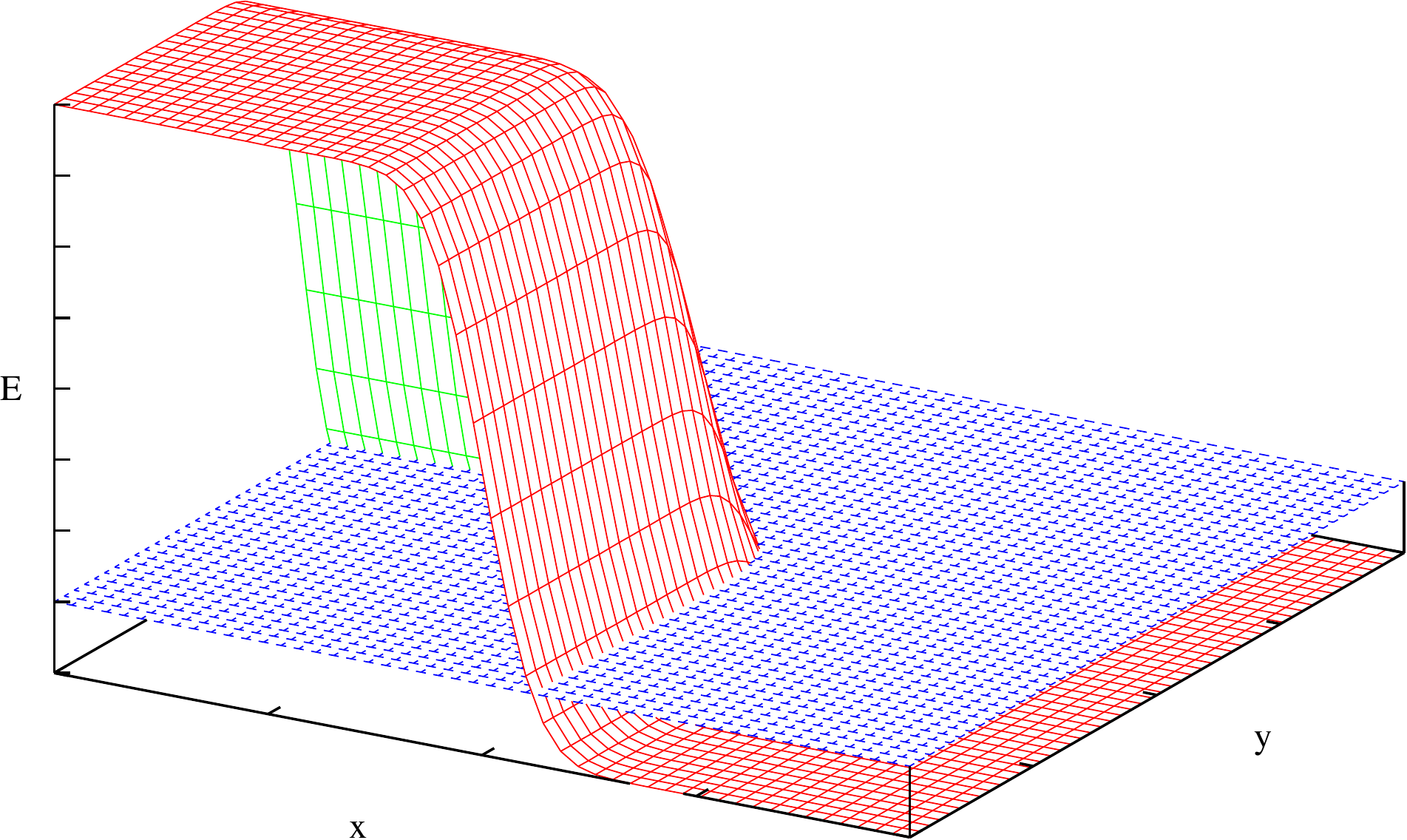}}
\subfloat[][Later time]{\includegraphics[scale=0.4]{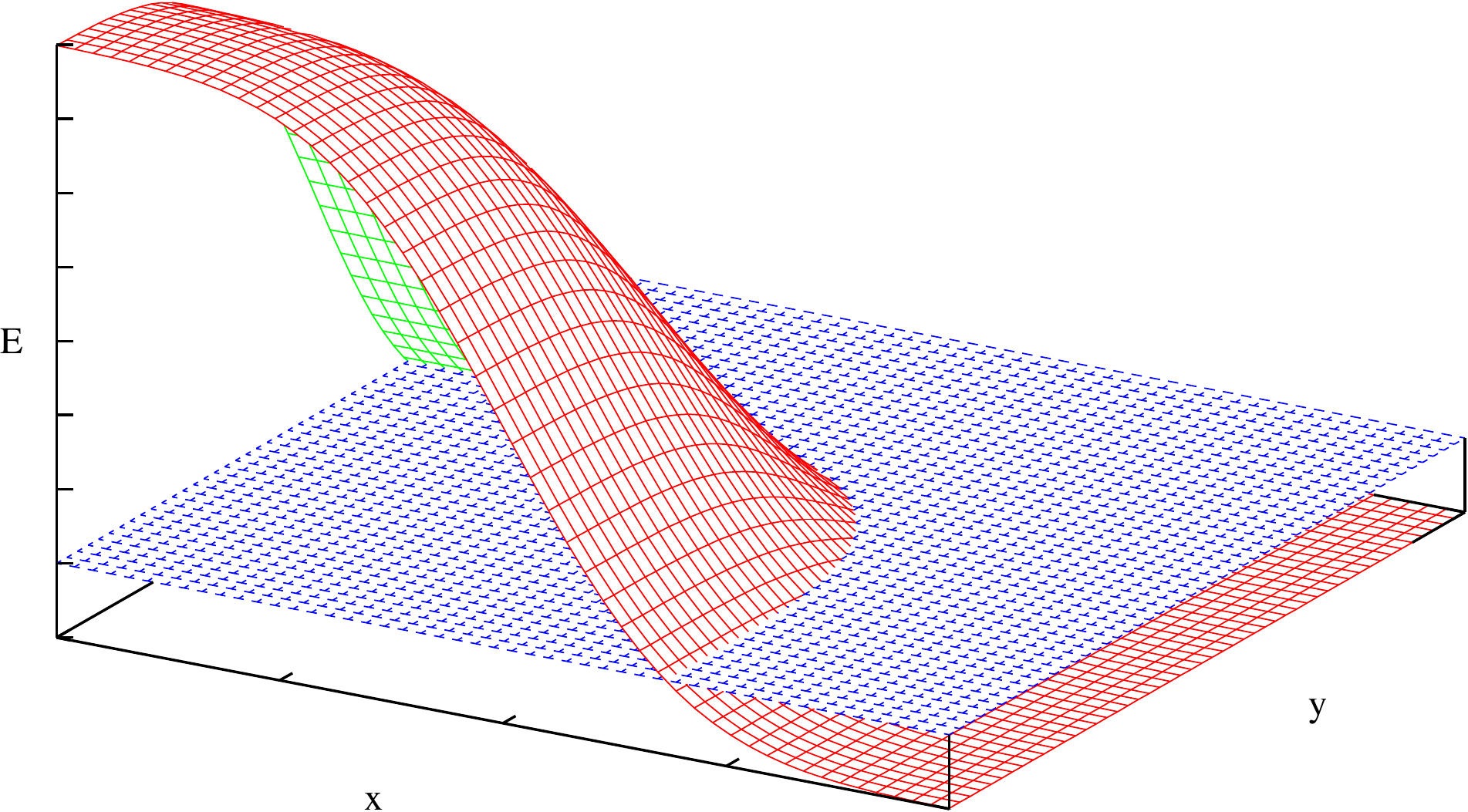}}
\caption{The curve where the blue sheet intersects the evolving surface is $x^*(t)$. The area of the blue sheet between
the surface and the co-ordinate planes is $a^2 - V^*(t)$. We are instead calculating the area $V^*(t)$ which has an
equal and opposite rate of change.}
\label{positive-orthant}
\end{figure}
Each $x_i$ only depends on one $z_i$ so it is trivial to find the Jacobian determinant.
\begin{eqnarray}
\frac{\textup{d}x_i}{\textup{d}z_i} &=& -\sqrt{\pi t} e^{\erf^{-1}(z_i - 1)^2} \nonumber \\
\left | \frac{\textup{d}x_1}{\textup{d}z_1} \dots \frac{\textup{d}x_d}{\textup{d}z_d} \right | &=& (\pi t)^{\frac{d}{2}} \prod_{i=1}^d e^{\erf^{-1}(z_i - 1)^2} \nonumber
\end{eqnarray}
Also, by multiplying all but one of these eigenvalues together, we can find the area element required to calculate the
surface integral in (\ref{nd-mass-matching}). An expression for the volume is
\begin{eqnarray}
V^*(t) &=& \int_1^{2^d \frac{E_{\mathrm{H}}}{A}} \dots \int_1^{\frac{2^d}{z_1 \dots z_{d - 1}} \frac{E_{\mathrm{H}}}{A}} \left | \frac{\textup{d}x_1}{\textup{d}z_1} \dots \frac{\textup{d}x_d}{\textup{d}z_d} \right | \textup{d}z_d \dots \textup{d}z_1 \nonumber \\
&=& (\pi t)^{\frac{d}{2}} \int_1^{2^d \frac{E_{\mathrm{H}}}{A}} \dots \int_1^{\frac{2^d}{z_1 \dots z_{d - 1}} \frac{E_{\mathrm{H}}}{A}} \prod_{i=1}^d e^{\erf^{-1}(z_i - 1)^2} \textup{d}z_d \dots \textup{d}z_1 \; . \label{nd-denominator}
\end{eqnarray}
We only need to know the length of $\nabla F$ to compute the surface integral. Working this out,
\begin{eqnarray}
\int_{x^*(t)} \nabla F(x, t) \cdot n \textup{d}S_x &=& \int_1^{2^d \frac{E_{\mathrm{H}}}{A}} \dots \int_1^{\frac{2^d}{z_1 \dots z_{d - 2}} \frac{E_{\mathrm{H}}}{A}} \left | \nabla F(x^*(t), t) \right | \left | \frac{\textup{d}x_1}{\textup{d}z_1} \dots \frac{\textup{d}x_{d-1}}{\textup{d}z_{d-1}} \right | \textup{d}z_{d-1} \dots \textup{d}z_1 \nonumber \\
&=& \frac{1}{\sqrt{4\pi t}} (\pi t)^{\frac{d-1}{2}} \int_1^{2^d \frac{E_{\mathrm{H}}}{A}} \dots \int_1^{\frac{2^d}{z_1 \dots z_{d - 1}} \frac{E_{\mathrm{H}}}{A}} 2E_{\mathrm{H}} \sqrt{\sum_{i=1}^d \frac{1}{z_i^2} e^{-2 \erf^{-1}(z_i - 1)^2}} \nonumber \\
&& \prod_{i=1}^{d-1} e^{\erf^{-1}(z_i - 1)^2} \textup{d}z_{d-1} \dots \textup{d}z_1 \label{nd-numerator}
\end{eqnarray}
where we have arbitrarily chosen $z_d$ as the variable to be determined from $z_1, \dots, z_{d-1}$.
The fact that (\ref{nd-denominator}) is proportional to one more factor of $t$ than (\ref{nd-numerator}) is what allows
(\ref{nd-mass-matching}) to be obeyed. There is a problem with these expressions however that only rears its head when
$d > 1$; not all points on the surface $x^*(t)$ reach the co-ordinate planes at the same time. In two dimensions
for instance, the endpoint
\begin{equation}
\left [
\begin{tabular}{c}
$a - 2 \sqrt{t} \erf^{-1} \left ( \frac{4E_{\mathrm{H}}}{A} - 1 \right )$ \\
$a$
\end{tabular}
\right ] \nonumber
\end{equation}
will reach $(0, a)$ before the midpoint
\begin{equation}
\left [
\begin{tabular}{c}
$a - 2 \sqrt{t} \erf^{-1} \left ( \sqrt{\frac{4E_{\mathrm{H}}}{A}} - 1 \right )$ \\
$a - 2 \sqrt{t} \erf^{-1} \left ( \sqrt{\frac{4E_{\mathrm{H}}}{A}} - 1 \right )$
\end{tabular}
\right ] \nonumber
\end{equation}
reaches $(0, 0)$. In other words, for late times, the volume (\ref{nd-denominator}) includes some regions outside the
positive orthant. It is not correct to integrate $x_i$ values from $a$ to the $x^*(t)$ curve. We must integrate from
$a$ to $\max \left ( x^*(t), 0 \right )$. The decay time obtained by setting (\ref{nd-auxillary-solution}) to
$E_{\mathrm{H}}$ is
\begin{equation}
T = \frac{1}{4} \left [ \frac{a}{\erf^{-1} \left ( 2 \left ( \frac{E_{\mathrm{H}}}{A} \right )^{\frac{1}{d}} - 1 \right )} \right ]^2 \; . \nonumber
\end{equation}
The expression (\ref{nd-denominator}) is most accurate when $T$ is large, \textit{i.e.} when
$A \approx 2^d E_{\mathrm{H}}$. If this is true, we can obtain a reasonable approximation to this complicated integral
by setting $V^*(T) \approx a^d$. Doing this, we see that
\begin{equation}
\int_1^{2^d \frac{E_{\mathrm{H}}}{A}} \dots \int_1^{\frac{2^d}{z_1 \dots z_{d - 1}} \frac{E_{\mathrm{H}}}{A}} \prod_{i=1}^d e^{\erf^{-1}(z_i - 1)^2} \textup{d}z_d \dots \textup{d}z_1 \approx \left ( \frac{2}{\sqrt{\pi}} \right )^d \erf^{-1} \left ( 2 \left ( \frac{E_{\mathrm{H}}}{A} \right )^{\frac{1}{d}} - 1 \right )^d \; . \nonumber
\end{equation}
With this, $E_{\mathrm{F}} - E_{\mathrm{H}}$ becomes much easier to express.
\begin{eqnarray}
E_{\mathrm{F}} - E_{\mathrm{H}} &=& \frac{\int_1^{2^d \frac{E_{\mathrm{H}}}{A}} \dots \int_1^{\frac{2^d}{z_1 \dots z_{d - 2}} \frac{E_{\mathrm{H}}}{A}} 2E_{\mathrm{H}} \sqrt{\sum_{i=1}^d \frac{1}{z_i^2} e^{-2 \erf^{-1}(z_i - 1)^2}} \prod_{i=1}^{d-1} e^{\erf^{-1}(z_i - 1)^2} \textup{d}z_{d-1} \dots \textup{d}z_1}{\pi d \int_1^{2^d \frac{E_{\mathrm{H}}}{A}} \dots \int_1^{\frac{2^d}{z_1 \dots z_{d - 1}} \frac{E_{\mathrm{H}}}{A}} \prod_{i=1}^d e^{\erf^{-1}(z_i - 1)^2} \textup{d}z_d \dots \textup{d}z_1} \nonumber \\
&\approx& \frac{2E_{\mathrm{H}} \sqrt{d} \left ( \frac{2}{\sqrt{\pi}} \right )^{d-1} \erf^{-1} \left ( 2 \left ( \frac{E_{\mathrm{H}}}{A} \right )^{\frac{1}{d-1}} - 1 \right )^{d-1}}{\pi d \left ( \frac{2}{\sqrt{\pi}} \right )^2 \erf^{-1} \left ( 2 \left ( \frac{E_{\mathrm{H}}}{A} \right )^{\frac{1}{d}} - 1 \right )^d} \nonumber \\
&\approx& \frac{E_{\mathrm{H}}}{\sqrt{\pi d} \erf^{-1} \left ( 2 \left ( \frac{E_{\mathrm{H}}}{A} \right )^{\frac{1}{d}} - 1 \right )} \left [ \lim_{\frac{2^d E_{\mathrm{H}}}{A} \rightarrow 1} \frac{\erf^{-1} \left ( 2 \left ( \frac{E_{\mathrm{H}}}{A} \right )^{\frac{1}{d-1}} - 1 \right )}{\erf^{-1} \left ( 2 \left ( \frac{E_{\mathrm{H}}}{A} \right )^{\frac{1}{d}} - 1 \right )} \right ]^{d - 1} \nonumber \\
&=& \frac{E_{\mathrm{H}}}{\sqrt{\pi d} \erf^{-1} \left ( 2 \left ( \frac{E_{\mathrm{H}}}{A} \right )^{\frac{1}{d}} - 1 \right )} \left ( \frac{d}{d - 1} \right )^{d - 1} \nonumber
\end{eqnarray}
In the first step, we have used the fact that
$\sqrt{\sum_{i=1}^d \frac{1}{z_i^2} e^{-2 \erf^{-1}(z_i - 1)^2}} \approx \sqrt{d}$
because $z_1, \dots, z_d$ are very close to $1$ while in the last step, we have used l'H\^{o}pital's rule. Linearizing
the inverse error function in the denominator and solving for $A$, we see that
\begin{equation}
\frac{2^d E_{\mathrm{H}}}{A} = \left [ 1 + \frac{E_{\mathrm{H}}}{E_{\mathrm{F}} - E_{\mathrm{H}}} \frac{1}{\sqrt{\pi d}} \left ( \frac{d}{d - 1} \right )^{d - 1} \right ] \approx 1 \; , \nonumber
\end{equation}
which is precisely the condition we needed in the first place to be able to say that $T$ was large and use this
approximation. Therefore when $E_{\mathrm{F}} \gg E_{\mathrm{H}}$, the time scale $T$ in $d$ dimensions becomes
\begin{eqnarray}
T &=& \frac{\pi d}{4} \left [ a \frac{E_{\mathrm{F}} - E_{\mathrm{H}}}{E_{\mathrm{H}}} \left ( \frac{d - 1}{d} \right )^{d - 1} \right ]^2 \nonumber \\
&\approx& \frac{\pi d}{4} \left [ \frac{a E_{\mathrm{F}}}{E_{\mathrm{H}}} \left ( \frac{d - 1}{d} \right )^{d - 1} \right ]^2 \; . \label{nd-timescale}
\end{eqnarray}

\subsubsection{Upper and lower bounds}
It is time to leverage these results to  improve the decay time estimates for the filtration function in Figure 
\ref{crude-bound}. To achieve the tightest possible bounds, we will make sure that our more-steep and less-steep
functions flatten out at $E_{\mathrm{H}}$. Figure \ref{bounding-functions1} shows what this
looks like. The function on the left $\Phi_1$ has a slope $\Phi_1^{\prime}(0)$ equal to $\Phi^{\prime}(0)$.
The function on the left $\Phi_2$ has a slope $\Phi_2^{\prime}(0)$ equal to
$\lim_{E \rightarrow E_{\mathrm{H}}^{-}} \Phi^{\prime}(E)$. When these slopes are
determined, the concentration comparison theorem tells us that our time scale will satisfy
\begin{equation}
\frac{\pi d}{4 \Phi_1^{\prime}(0)} \left [ \frac{a E_{\mathrm{F}}}{E_{\mathrm{H}}} \left ( \frac{d - 1}{d} \right )^{d - 1} \right ]^2 \leq T \leq \frac{\pi d}{4 \Phi_2^{\prime}(0)} \left [ \frac{a E_{\mathrm{F}}}{E_{\mathrm{H}}} \left ( \frac{d - 1}{d} \right )^{d - 1} \right ]^2 \; . \label{prototype-timescale}
\end{equation}
\begin{figure}[h]
\centering
\subfloat[][Comparing $\Phi$ to $\Phi_1$]{\includegraphics[scale=0.35]{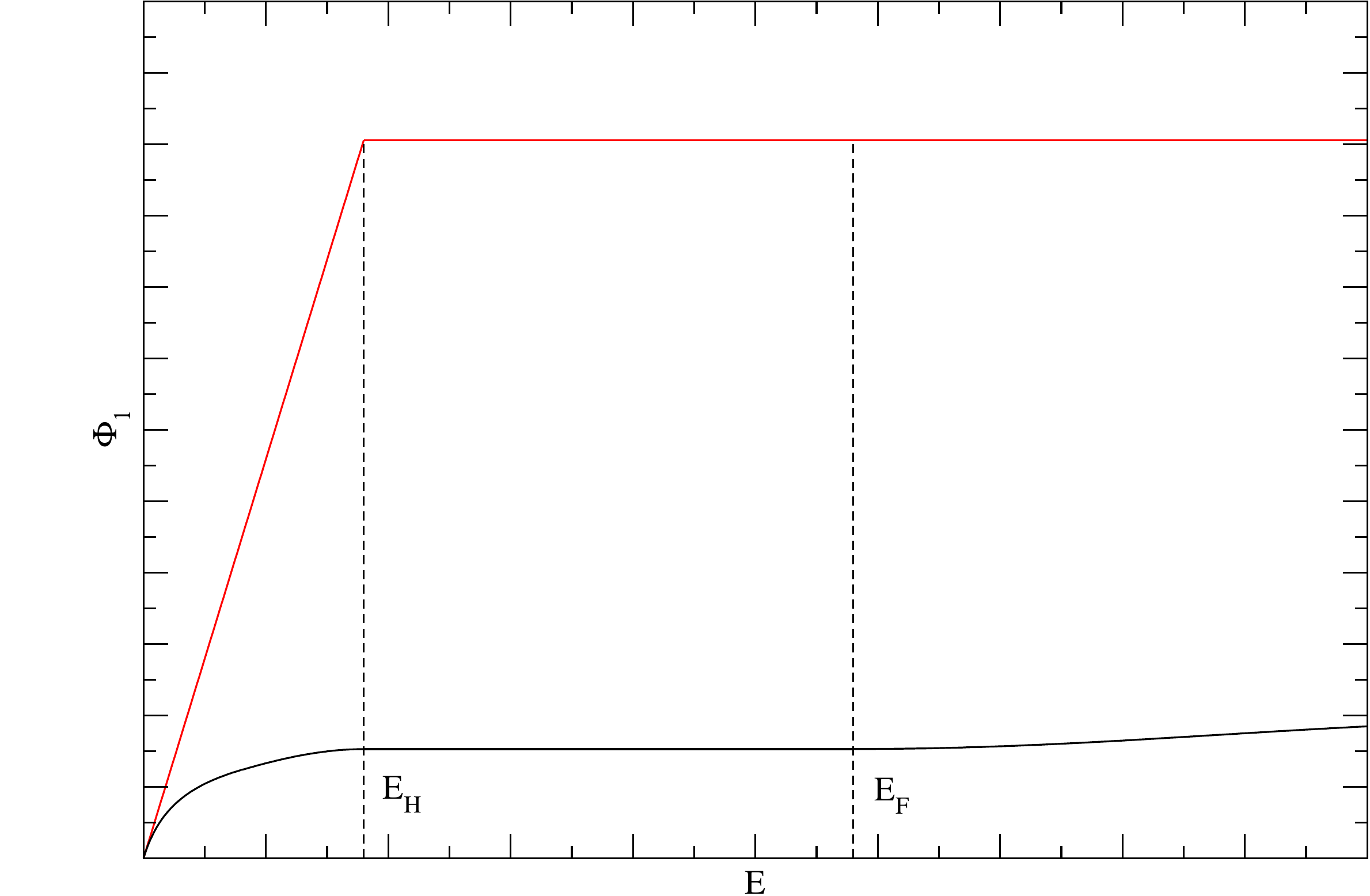}}
\subfloat[][Comparing $\Phi$ to $\Phi_2$]{\includegraphics[scale=0.35]{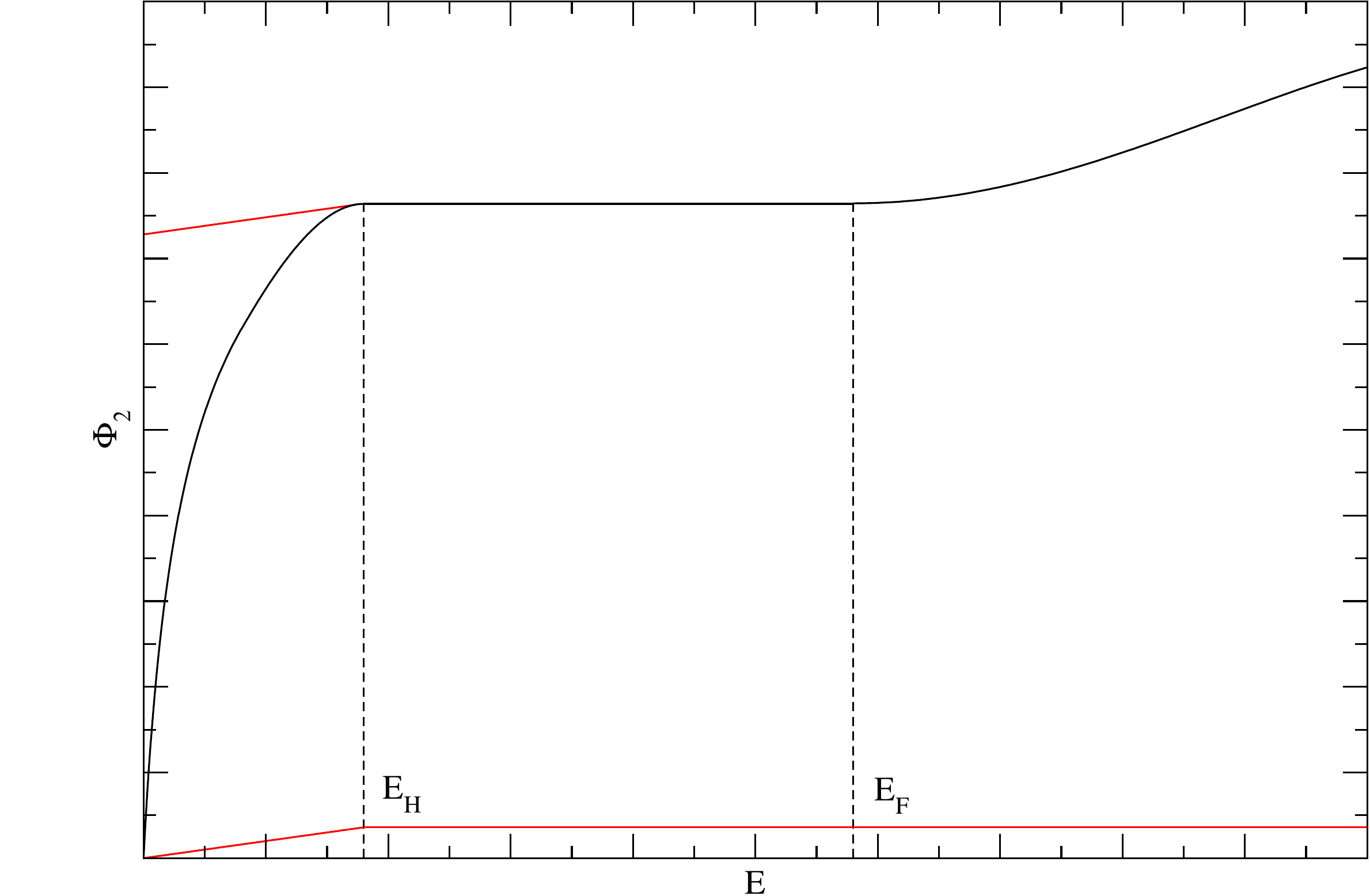}}
\caption{The piecewise-linear bounding functions must relate to our filtration function in this way. On the left,
$\Phi_1$ could cross over to Hagedorn behaviour at some energy $E_1 \geq E_{\mathrm{H}}$ and on the right, $\Phi_2$
could cross over to Hagedorn behaviour at some energy $E_2 \leq E_{\mathrm{H}}$. The tightest possible time scale
bounds are achieved when $E_1 = E_{\mathrm{H}} = E_2$.}
\label{bounding-functions1}
\end{figure}
This is the main part of the text that will use a specific formula for $\beta(E)$, so first consider the one that is
valid when $C(E) = \rho^{-2}(E)$. Recall that this is
$\beta(E) = cE^{\alpha - 1} = \beta \left ( E_{\mathrm{min}} \right )E_{\mathrm{min}}^{1 - \alpha} E^{\alpha - 1}$
in the lowest energy regime. The slopes that we must insert are simple:
\begin{eqnarray}
\Phi_1^{\prime}(0) &=& -\beta^{\prime} \left ( E_{\mathrm{min}} \right ) = (1 - \alpha)\frac{\beta \left ( E_{\mathrm{min}} \right )}{E_{\mathrm{min}}} \label{slope-1} \\
\Phi_2^{\prime}(0) &=& -\beta^{\prime} \left ( E_{\mathrm{H}} \right ) = (1 - \alpha) \beta \left ( E_{\mathrm{min}} \right ) \frac{E_{\mathrm{H}}^{\alpha - 2}}{E_{\mathrm{min}}^{\alpha - 1}} \; . \label{slope-2}
\end{eqnarray}
Plugging (\ref{slope-1}) and (\ref{slope-2}) into (\ref{prototype-timescale}), this time scale is now solved:
\begin{equation}
\frac{\pi d}{4(1 - \alpha) \beta \left ( E_{\mathrm{min}} \right )} \frac{E_{\mathrm{min}}}{E_{\mathrm{H}}^2} \left [ a E_{\mathrm{F}} \left ( \frac{d - 1}{d} \right )^{d - 1} \right ]^2 \leq T \leq \frac{\pi d}{4(1 - \alpha) \beta \left ( E_{\mathrm{min}} \right )} \frac{E_{\mathrm{min}}^{\alpha - 1}}{E_{\mathrm{H}}^{\alpha}} \left [ a E_{\mathrm{F}} \left ( \frac{d - 1}{d} \right )^{d - 1} \right ]^2 \; . \label{timescale-1}
\end{equation}
This dependence on $\left ( a E_{\mathrm{F}} \right )^2$ is a general feature of long Hagedorn regions as we will see
in our next time scale.
Recall that when we had $C(E) = 1$, the effective inverse temperature function to consider $\tilde{\beta}$ was
defined by
\begin{eqnarray}
\tilde{\beta}^{\prime}(E) &=& \rho^2(E) \beta^{\prime}(E) \; . \nonumber
\end{eqnarray}
We will construct a filtration function from this in the same way as before.
$\tilde{\Phi}(E) = \tilde{\beta} \left ( E_{\mathrm{min}} \right ) - \tilde{\beta} \left ( E + E_{\mathrm{min}} \right )$.
In the low energy regime, we have
\begin{eqnarray}
\tilde{\Phi}^{\prime}(E) &=& -\rho^2 \left ( E + E_{\mathrm{min}} \right ) \beta^{\prime} \left ( E + E_{\mathrm{min}} \right ) = \rho^2 \left ( E + E_{\mathrm{min}} \right ) \Phi^{\prime}(E) \nonumber \\
&=& (1 - \alpha) \beta \left ( E_{\mathrm{min}} \right ) e^{\frac{2}{\alpha} \frac{\beta \left ( E_{\mathrm{min}} \right ) \left ( E + E_{\mathrm{min}} \right )^{\alpha}}{E_{\mathrm{min}}^{\alpha - 1}}} \frac{\left ( E + E_{\mathrm{min}} \right )^{\alpha - 2}}{E_{\mathrm{min}}^{\alpha - 1}} \; . \nonumber
\end{eqnarray}
The first thing to notice is that $E_{\mathrm{H}}$ is the same for $\Phi$ and $\tilde{\Phi}$. One derivative vanishes
if and only if the other does. Also, at an energy of $E_{\mathrm{min}}$, the number of states is very close to $1$
so $\tilde{\Phi}^{\prime}(0) = \Phi^{\prime}(0)$. Finding the slopes of appropriate bounding functions
$\tilde{\Phi}_1$ and $\tilde{\Phi}_2$ comes down to extremizing $\tilde{\Phi}^{\prime}$.
\begin{figure}[h]
\centering
\subfloat[][$E_{\mathrm{H}} < E^*$]{\includegraphics[scale=0.35]{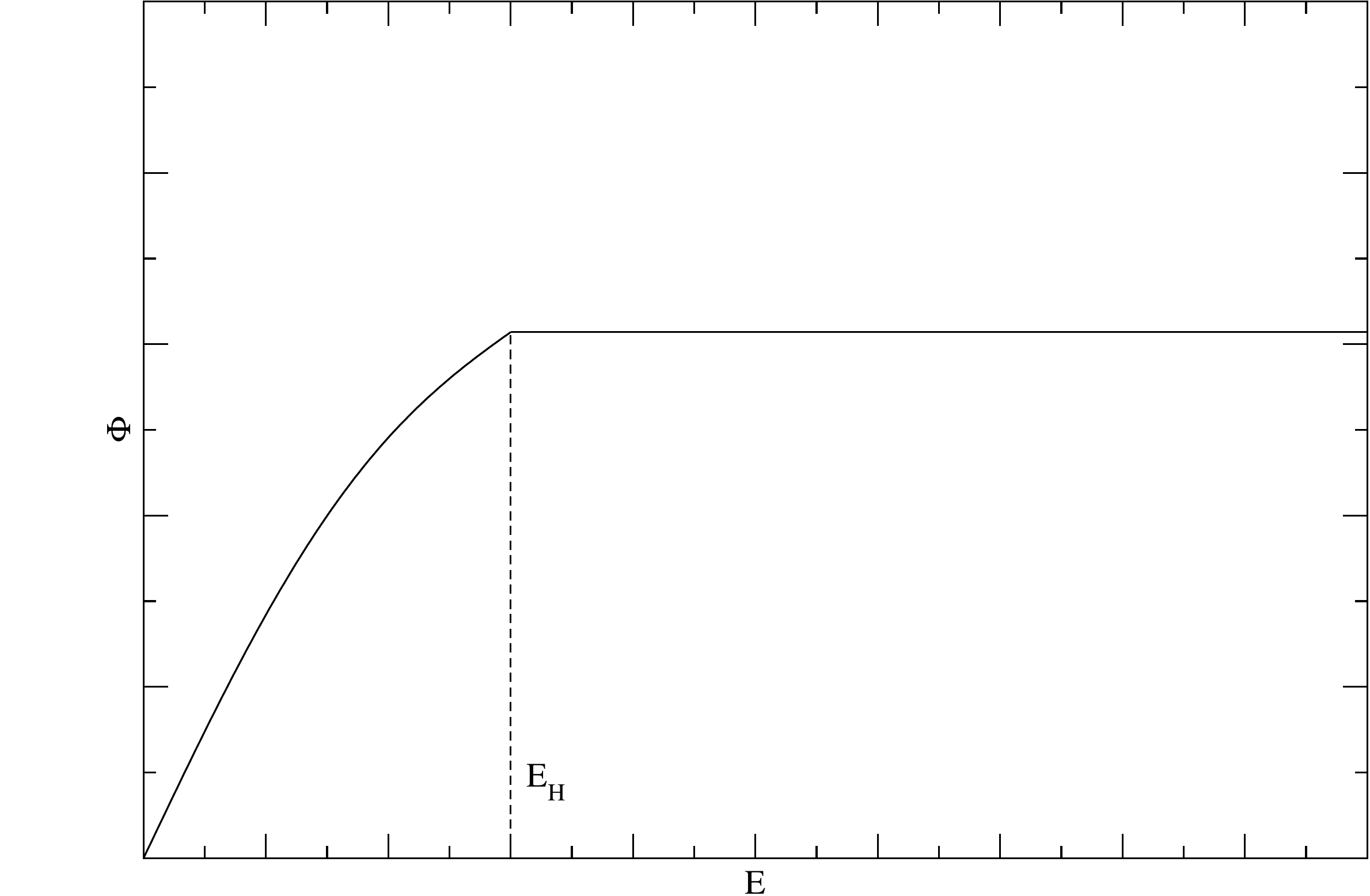}}
\subfloat[][$E_{\mathrm{H}} > E^*$]{\includegraphics[scale=0.35]{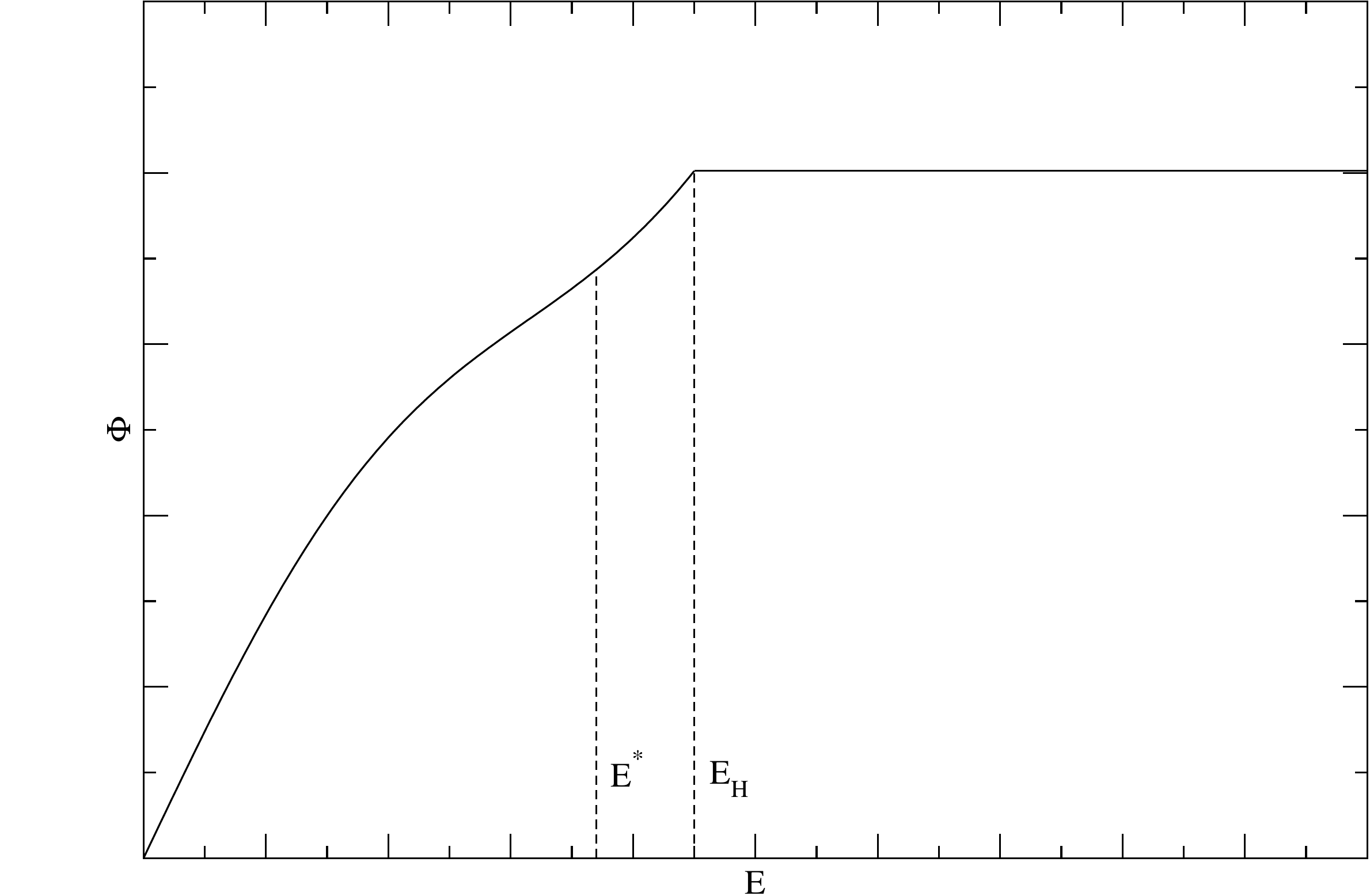}}
\caption{The upper bound for the decay time, given by $\tilde{\Phi}_2$ depends on which of the above cases is
realized. The figure on the left would lead to a situation much like the one for $\Phi_2$ that we saw before.
However, on the right, the Hagedorn energy is large enough to make $\tilde{\Phi}$ change concavity before reaching
the Hagedorn phase.}
\label{inflection-point}
\end{figure}
\begin{eqnarray}
\tilde{\Phi}^{\prime\prime}(E^*) &=& 0 \nonumber \\
2 \frac{\beta \left ( E_{\mathrm{min}} \right )}{E_{\mathrm{min}}^{\alpha - 1}} (1 - \alpha) E_{\mathrm{min}}^{\alpha - 1} \beta \left ( E_{\mathrm{min}} \right ) \left ( E^* + E_{\mathrm{min}} \right )^{2\alpha - 3} &=& (2 - \alpha) (1 - \alpha) E_{\mathrm{min}}^{\alpha - 1} \beta \left ( E_{\mathrm{min}} \right ) \left ( E^* + E_{\mathrm{min}} \right )^{\alpha - 3} \nonumber \\
E^* &=& \left [ \frac{2 - \alpha}{2 \beta \left ( E_{\mathrm{min}} \right ) E_{\mathrm{min}}^{1 - \alpha}} \right ]^{\frac{1}{\alpha}} - E_{\mathrm{min}} \nonumber
\end{eqnarray}
Therefore $\tilde{\Phi}$ starts off with a steep derivative at $0$ and becomes less steep until reaching the
inflection point $E^*$. This means that there are two important cases to consider: $E_{\mathrm{H}} < E^*$ and
$E_{\mathrm{H}} > E^*$. The flattest point occurs at $E_{\mathrm{H}}$ in the former case and at $E^*$ in the latter.
\begin{eqnarray}
\tilde{\Phi}_1^{\prime}(0) &=& (1 - \alpha)\frac{\beta \left ( E_{\mathrm{min}} \right )}{E_{\mathrm{min}}} \label{slope-3} \\
\tilde{\Phi}_2^{\prime}(0) &=&
\begin{cases}
(1 - \alpha) \beta \left ( E_{\mathrm{min}} \right ) \frac{E_{\mathrm{H}}^{\alpha - 2}}{E_{\mathrm{min}}^{\alpha - 1}} e^{\frac{2}{\alpha} \frac{\beta \left ( E_{\mathrm{min}} \right ) E_{\mathrm{H}}^{\alpha}}{E_{\mathrm{min}}^{\alpha - 1}}} & E_{\mathrm{H}} < E^* \\
(1 - \alpha) \left ( \frac{2 - \alpha}{2} \right )^{\alpha - \frac{2}{\alpha}} \beta \left ( E_{\mathrm{min}} \right )^{\frac{2}{\alpha}} E_{\mathrm{min}}^{\frac{2}{\alpha} - 2} e^{\frac{2}{\alpha} - 1} & E_{\mathrm{H}} > E^*
\end{cases} \label{slope-4}
\end{eqnarray}
Plugging (\ref{slope-3}) and (\ref{slope-4}) into (\ref{prototype-timescale}), the next time scale is solved as well:
\begin{equation}
\frac{\pi d}{4(1 - \alpha) \beta \left ( E_{\mathrm{min}} \right )} \frac{E_{\mathrm{min}}}{E_{\mathrm{H}}^2} \left [ a E_{\mathrm{F}} \left ( \frac{d - 1}{d} \right )^{d - 1} \right ]^2 \leq T \leq C_T \left ( E_{\mathrm{min}}, E_{\mathrm{H}} \right ) \left [ a E_{\mathrm{F}} \left ( \frac{d - 1}{d} \right )^{d - 1} \right ]^2 \label{timescale-2}
\end{equation}
where
\begin{equation}
C_T \left ( E_{\mathrm{min}}, E_{\mathrm{H}} \right ) = \begin{cases}
\frac{\pi d}{4(1 - \alpha) \beta \left ( E_{\mathrm{min}} \right )} \frac{E_{\mathrm{min}}^{\alpha - 1}}{E_{\mathrm{H}}^{\alpha}} e^{-\frac{2}{\alpha} \frac{\beta \left ( E_{\mathrm{min}} \right ) E_{\mathrm{H}}^{\alpha}}{E_{\mathrm{min}}^{\alpha - 1}}} & E_{\mathrm{H}} < \left [ \frac{2 - \alpha}{2 \beta \left ( E_{\mathrm{min}} \right ) E_{\mathrm{min}}^{1 - \alpha}} \right ]^{\frac{1}{\alpha}} \\
\frac{\pi d}{4(1 - \alpha)} \left ( \frac{2 - \alpha}{2} \right )^{\frac{2}{\alpha} - \alpha} \beta \left ( E_{\mathrm{min}} \right )^{-\frac{2}{\alpha}} \frac{E_{\mathrm{min}}^{2 - \frac{2}{\alpha}}}{E_{\mathrm{H}}^2} e^{1 - \frac{2}{\alpha}} & E_{\mathrm{H}} > \left [ \frac{2 - \alpha}{2 \beta \left ( E_{\mathrm{min}} \right ) E_{\mathrm{min}}^{1 - \alpha}} \right ]^{\frac{1}{\alpha}}
\end{cases} \; . \nonumber
\end{equation}
There is a third possibility; $E_{\mathrm{H}}$ being so large that $\Phi^{\prime} \left ( E_{\mathrm{H}} \right )$
surpasses $\Phi^{\prime}(0)$. However, $E_{\mathrm{min}}$ can always be made small enough to stop this from happening.

\subsubsection{Remaining problems}
A useful question to ask is how our decay time differs from a decay time that would arise from linear diffusion.
Solivng the heat equation with the (\ref{toy-condition}) initial data,
\begin{eqnarray}
E(x, t) &=& \frac{E_{\mathrm{F}}}{(4\pi t)^{\frac{d}{2}}} \int_{-a}^a \dots \int_{-a}^a e^{-\frac{(x-y)^2}{4t}} \textup{d}y_1 \dots \textup{d}y_d \nonumber \\
&=& \frac{E_{\mathrm{F}}}{2^d} \prod_{i = 1}^d \left ( \erf \left ( \frac{x_i + a}{2 \sqrt{t}} \right ) - \erf \left ( \frac{x_i - a}{2 \sqrt{t}} \right ) \right ) \; . \nonumber
\end{eqnarray}
The time needed for the central peak to come down to $E_{\mathrm{H}}$ is
\begin{equation}
T = \frac{a^2}{4} \erf^{-1} \left ( \frac{E_{\mathrm{H}}}{E_{\mathrm{F}}} \right )^{-\frac{2}{d}} \approx \frac{a^2}{4} \left ( \frac{2E_{\mathrm{F}}}{\sqrt{\pi}E_{\mathrm{H}}} \right )^{\frac{2}{d}} \; . \nonumber
\end{equation}
For $d > 1$, this decay time is parametrically less than the $O \left ( \left ( a E_{\mathrm{F}} \right )^2 \right )$
result that we found for diffusion with a Hagedorn regime. However, linear diffusion in $d = 1$ is much slower.
To get a sense of why the $d = 1$ times are similiar, we will look at the amount of energy outside $B^1_a$ after a
short amount of time \cite{mvr}. For diffusion with a Hagedorn regime, we find
\begin{eqnarray}
a E_{\mathrm{F}} - M_a(t) &=& A \int_a^{\infty} 1 + \erf \left ( \frac{a - x}{2 \sqrt{t}} \right ) \textup{d}x \nonumber \\
&=& 2A \sqrt{\frac{t}{\pi}} \nonumber \\
&\approx& 4 E_{\mathrm{H}} \sqrt{\frac{t}{\pi}} \; . \nonumber
\end{eqnarray}
For the heat equation with the same slope, the result is
\begin{eqnarray}
a E_{\mathrm{F}} - M_a(t) &=& E_{\mathrm{F}} \int_a^{\infty} \erf \left ( \frac{x + a}{2 \sqrt{t}} \right ) - \erf \left ( \frac{x - a}{2 \sqrt{t}} \right ) \textup{d}x \nonumber \\
&=& 2 E_{\mathrm{F}} \sqrt{\frac{t}{\pi}} \left ( 1 - e^{-\frac{a^2}{t}} \right ) + 2 E_{\mathrm{F}} a \left ( 1 - \erf \left ( \frac{a}{\sqrt{t}} \right ) \right ) \nonumber \\
&\approx& 2 E_{\mathrm{F}} \sqrt{\frac{t}{\pi}} \nonumber
\end{eqnarray}
where in the last step we have assumed $t \ll a^2$. Therefore, diffusion without a Hagedorn regime is faster at the
beginning, but not over longer time scales. We can explain this effect by noticing that a lump of energy diffusing
linearly becomes remarkably flat towards the ``end'' of its diffusion. Small spatial derivatives lead to small time
derivatives and a long time constant. Diffusion with a Hagedorn regime does not have this problem. The discontinuous
solutions we saw had steep profiles during all stages of what we called the diffusion. A different effect, the large
static phase, is what lengthened the time constant. Even though diffusion times are similar with and without a Hagedorn
phase, the numerics will show that if the energy distribution begins with a peak well above $E_{\mathrm{F}}$, its
diffusion will slow down significantly upon reaching $E_{\mathrm{F}}$. This does not contradict our analysis because
the diffusion in the high energy phase comes with a much larger constant than the diffusion that crosses over to
Hagedorn behaviour in the other two phases.

Regardless of how many dimensions there are, we should expect that our model overestimates the decay time of a plasma
ball. The energy in our model can only diffuse away to infinity so quickly once it leaves a ball of radius $a$.
In a realstic plasma ball, energy escaping the central region is quickly ejected because it is travelling through a
vacuum.

\subsection{Comments on unbounded domains}
For unbounded domains such as the whole space, the only sensible energy distributions are the ones that decay to zero
and filtration equations like (\ref{filtration-equation}) take on many new properties. Even though the
diffusion of energy governed by a master equation should certainly be well defined for a space that extends infinitely,
we found in the last chapter that our PDE is only a useful model for such behaviour when the energy is high. Therefore,
it is not necessarily true that a result derived for $\Omega = \mathbb{R}^d$ applies to the entropic limit of a field
theory.

\subsubsection{Barriers to uniqueness}
The first thing we note is that
\begin{equation}
\begin{dcases}
\frac{\partial E}{\partial t} (x, t) = - \Delta \beta(E(x, t)) & (x, t) \in \mathbb{R}^d \times (0, \infty) \\
E(x, 0) = E_0(x) & x \in \mathbb{R}^d
\end{dcases} \nonumber
\end{equation}
would have many solutions if we simply looked for them in the set of all differentiable functions. This property is
well known for the heat equation. In fact, for the heat equation \cite{chung} and similar equations \cite{dhungana}
in the whole space, nonzero solutions have been found that vanish at $t = 0$. These authors have found a restriction
that must be made in order to recover uniqueness; there must exist a paraboloid in $\mathbb{R}^d$ such that
$\log E$ stays below it at all times. Because our equation is different from the heat equation, it is not yet clear
that the same condition, or a stronger one, will yield unique solutions to (\ref{filtration-equation}). Nevertheless,
we will impose the stronger condition that $E$ be uniformly bounded in spacetime.

Requiring uniform boundedness is natural because it forces solutions to still obey the maximum principle. The maximum
principle phrased in terms of subsolutions and supersolutions made no assumption about the domain being bounded,
but in applying it, we set one solution equal to the maximum of the other. Our assumption of uniform boundedness
is the only reason that we can still write $u = \max \{ v(x, 0) : x \in \mathbb{R}^d \}$ because the maximum of a
continuous function is only guaranteed to exist on a compact domain. Having to make assumptions about how solutions
grow is a common theme when extending maximum principles to the unbounded case \cite{protter}. We now turn to the
question of uniqueness which is intimately related to energy conservation.

Consider the function $\beta(E) = -\frac{1}{m} E^m$ which is monotonically decreasing for
all $m \in \mathbb{R}$. The equation (\ref{filtration-equation}) for this choice
\begin{equation}
\frac{\partial E}{\partial t} = \nabla \cdot \left ( E^{m-1} \nabla E \right ) \label{very-fast}
\end{equation}
has been studied extensively \cite{esteban, rodriguez, vazquez2, carrillo}. It is called the porous medium equation for
$m \in (1, \infty)$, the fast diffusion equation for $m \in (0, 1]$ and the very fast diffusion equation for
$m \in (-1, 0]$. We are mainly interested in very fast diffusion because we had $0 < \alpha \leq 1$ and $m$ is
essentially $\alpha - 1$. We are only writing $m$ instead of $\alpha - 1$ here for consistency with the literature.

Trying to naively derive conservation of energy for (\ref{very-fast}) does not work. This is
related to the fact that $\frac{\partial E}{\partial t}$ does not need to be bounded at a given time, even if $E$ is
uniformly bounded. To illustrate this, let us work with the very fast diffusion equation in one dimension.
For the initial condition $E_0$, it is perfectly valid to choose a function that decays like $|x|^{-k}$ where $k > 1$.
For large $x$, this leads to a time derivative given by:
\begin{eqnarray}
\frac{\partial E}{\partial t} (x, 0) &=& \frac{\partial}{\partial x} \left [ E^{m - 1} (x, 0) \frac{\partial E}{\partial x} (x, 0) \right ] \nonumber \\
&=& \frac{\partial}{\partial x} \left [ E_0^{m - 1} (x) \frac{\partial E_0}{\partial x} (x) \right ] \nonumber \\
&\sim& \frac{\partial}{\partial x} \left ( x^{(1 - m)k} x^{-k - 1} \right ) \nonumber \\
&\sim& x^{-mk - 2} \; . \nonumber
\end{eqnarray}
The exponent tells us that it is possible for $E_0$ to decay ``too quickly''. If $k > \frac{2}{-m}$ then
$\frac{\partial E}{\partial t}$ is not bounded at $t = 0$. This affects our ability to compute
\begin{equation}
\frac{\textup{d}M}{\textup{d}t} = \frac{\textup{d}}{\textup{d}t} \int_{-\infty}^{\infty} E \textup{d}x \; . \nonumber
\end{equation}
In order to take the derivative inside the integral, there must exist an $\epsilon > 0$ such that
$\frac{\partial E}{\partial t}$ is bounded on $(-\infty, \infty) \times [0, \epsilon]$. We cannot do this with all
$k$ values, but even if we could, we would still not be out of the woods.

Pick a value of $k$ satisfying $\frac{1}{-m} < k < \frac{2}{-m}$. Since $k < \frac{2}{-m}$
$\frac{\partial E}{\partial t}$ as derived above must be finite at $t = 0$. Computing the time derivative of the
total mass,
\begin{eqnarray}
\frac{\textup{d}M}{\textup{d}t}(0) &=& \int_{-\infty}^{\infty} \frac{\partial E}{\partial t} (x, 0) \textup{d}x \nonumber \\
&=& \int_{-\infty}^{\infty} \frac{\partial}{\partial x} \left [ E^{m - 1} (x, 0) \frac{\partial E}{\partial x} (x, 0) \right ] \textup{d}x \nonumber \\
&=& \left. E_0^{m - 1} (x) \frac{\partial E_0}{\partial x} (x) \right|_{-\infty}^{\infty} \nonumber \\
&=& 2 \lim_{|x| \rightarrow \infty} E_0^{m - 1} (|x|) \frac{\partial E_0}{\partial x} (|x|) \nonumber \\
&\sim& \lim_{|x| \rightarrow \infty} |x|^{-mk - 1} \; . \nonumber
\end{eqnarray}
Our choice $k > \frac{1}{-m}$ tells us that the time derivative of $M$ is initially infinite. Therefore $M$ cannot
possibly be a constant function on $[0, \infty)$. We have demonstrated that there are infinitely many functions
$E_0 \in L^1(\mathbb{R})$ such that there is no energy conserving soluton to the following:
\begin{equation}
\begin{dcases}
\frac{\partial E}{\partial t}(x, t) = \frac{\partial}{\partial x} \left [ E^{m - 1}(x, t) \frac{\partial}{\partial x} E(x, t) \right ] & (x, t) \in \mathbb{R} \times (0, \infty) \\
E(x, 0) = E_0(x) & x \in \mathbb{R}
\end{dcases} \; . \label{naive-problem}
\end{equation}
As it turns out, however, for any $E_0 \in L^1(\mathbb{R})$, there \textit{is} \cite{esteban, vazquez2} an energy 
conserving solution to the similar problem:
\begin{equation}
\begin{dcases}
\frac{\partial E}{\partial t}(x, t) = \frac{\partial}{\partial x} \left [ E^{m - 1}(x, t) \frac{\partial}{\partial x} E(x, t) \right ] & (x, t) \in \mathbb{R} \times (0, \infty) \\
\lim_{t \rightarrow 0} E(x, t) = E_0(x) & x \in \mathbb{R}
\end{dcases} \; . \label{modified-problem}
\end{equation}
Differential equations can only hold in open sets so the notation $E(x, 0) = E_0(x)$ really means that any derivative
of $E(\cdot, t)$ will converge to the corresponding derivative of $E_0$ as $t \rightarrow 0$. The subtle distinction
between (\ref{naive-problem}) and (\ref{modified-problem}) is that the only thing that needs to converge is the
function itself. The $E(\cdot, t)$ will converge in $L^1(\mathbb{R})$ to $E_0$ but the
$\frac{\partial}{\partial t} E(\cdot, t)$ \textit{need not} converge to
$\frac{\partial}{\partial x} \left ( E_0^{m - 1} \frac{\partial E_0}{\partial x} \right )$. In other words,
$\frac{\textup{d}M(0)}{\textup{d}t}$ appeared to diverge because we were still calculating it incorrectly.

Esteban, Rodriguez and V\'azquez \cite{esteban} showed that an energy conserving solution to the very fast diffusion
Cauchy problem exists in one dimension. However, there are also solutions that vanish in finite time such as
\begin{equation}
E(x, t) = (T - t) \left [ 1 - \left ( \frac{1 - e^{-\sqrt{2}x}}{1 + e^{-\sqrt{2}x}} \right )^2 \right ] \; . \nonumber
\end{equation}
This solves (\ref{modified-problem}) for $m = 0$ and does not conserve energy. The theorem of \cite{esteban} still
applies to initial conditions like
\begin{equation}
E_0(x) = T \left [ 1 - \left ( \frac{1 - e^{-\sqrt{2}x}}{1 + e^{-\sqrt{2}x}} \right )^2 \right ] \; , \nonumber
\end{equation}
it just yields a different solution. In fact, for all $-1 < m \leq 0$ and all initial conditions, there are short lived
solutions to (\ref{modified-problem}) in addition to the one that conserves energy.
One of the following equivalent conditions is needed to make sure we are choosing the right solution.
\begin{enumerate}
\item
$\int_{-\infty}^{\infty} E(x, t) \textup{d}x = \int_{-\infty}^{\infty} E_0(x) \textup{d}x$ for $t \in (0, \infty)$.
\item
$\lim_{x \rightarrow \pm \infty} E^{m - 1}(x, t) \frac{\partial E}{\partial x}(x, t) = 0$ for $t \in (0, \infty)$.
\item
$E(x, t) > 0$ for $(x, t) \in \mathbb{R} \times (0, \infty)$.
\end{enumerate}
In fact, \cite{rodriguez} generalized this to arbitrary flux functions.
\begin{theorem}
Suppose $E_0 \in L^1(\mathbb{R})$ and $f ,g \in L^{\infty} ((0, \infty))$ are non-negative.
Then there exists a unique solution to
\begin{equation}
\begin{dcases}
\frac{\partial E}{\partial t}(x, t) = \frac{\partial}{\partial x} \left ( E^{m - 1}(x, t) \frac{\partial}{\partial x} E(x, t) \right ) & (x, t) \in \mathbb{R} \times (0, \infty) \\
\lim_{t \rightarrow 0} E(x, t) = E_0(x) & x \in \mathbb{R} \\
\lim_{x \rightarrow \infty} E^{m - 1}(x, t) \frac{\partial}{\partial x} E(x, t) = -f(t) & t \in (0, \infty) \\
\lim_{x \rightarrow -\infty} E^{m - 1}(x, t) \frac{\partial}{\partial x} E(x, t) = g(t) & t \in (0, \infty)
\end{dcases} \nonumber
\end{equation}
whose total mass satisfies
\begin{equation}
\int_{-\infty}^{\infty} E(x, t) \textup{d}x = \int_{-\infty}^{\infty} E_0(x) \textup{d}x - \int_0^t f(s) + g(s) \textup{d}s \; . \nonumber
\end{equation}
\end{theorem}
The energy conserving solution that interests us is precisely the case $f = g = 0$.
It is important to remember that even though the vanishing flux of $E$ is posed as a constraint, we are still free
to consider an initial profile $E_0$ that does not have a vanishing flux. This is an example of a problem with
\textit{inconsistent initial and boundary conditions}. Some numerical methods have been written specifically to
address this \cite{verma}. Infinite propogation speed ensures that the conditions are only inconsistent on a set of
measure zero. Depending on what we specify for the flux, the distribution might ``jump'' to a function that decays
quickly or to a function that decays slowly. When $m > 0$, energy conservation in (\ref{modified-problem}) is known to
hold without the complication of non-uniqueness \cite{vazquez2}.

We may summarize this discussion by saying that the maximum principle only holds for a general
(\ref{filtration-equation}) if we specify uniform boundedness and that uniqueness only holds for a general
(\ref{filtration-equation}) if we specify energy conservation. The very fast diffusion equation in one dimension 
provides an example of uniqueness not holding for (\ref{filtration-equation}) if we only specify uniform boundedness.
Even when we have a unique energy conserving solution, the time derivative of the total mass will not make sense at
$t = 0$ for the perfectly good initial conditions that are inconsistent with the boundary conditions. For those, the
energy flux that vanishes for the solution at all positive times will not vanish at $t = 0$.

\subsubsection{Barenblatt profiles}
Solutions to (\ref{filtration-equation}) on a bounded domain converge to their average values which means that
solutions to the same equation on the whole space should converge to zero. This will be an $L^{\infty}$ convergence
and not an $L^1$ convergence because of energy conservation. One should check whether the associated decay times
on $\mathbb{R}^d$ have anything to do with our main time scales (\ref{timescale-1}) and (\ref{timescale-2}). Some
known results about the very fast diffusion equation will help us do this. There is an important family of exact
solutions to the porous medium, fast diffusion and very fast diffusion equations given by:
\begin{equation}
U(x, t) = \left [ \frac{\left ( \frac{4}{1 - m} - 2d \right ) t}{|x|^2 + Bt^{\frac{2}{2 - d(1 - m)}}} \right ]^{\frac{1}{1 - m}} \label{barenblatt}
\end{equation}
where $B$ is a positive constant that determines the mass \cite{carrillo}. These self-similar functions called
Barenblatt solutions have a Dirac delta as their initial conditions. Solving for the mass,
\begin{eqnarray}
M(t) &=& d \omega_d \int_0^{\infty} U(r, t) r^{d - 1} \textup{d}r \nonumber \\
&=& d \omega_d \left [ \left ( \frac{4}{1 - m} - 2d \right ) t \right ]^{\frac{1}{1 - m}} \int_0^{\infty} \frac{d^{d - 1}}{\left ( r^2 + B t^{\frac{2}{2 - d (1 - m)}} \right )^{\frac{1}{1 - m}}} \textup{d}r \nonumber \\
&=& d \omega_d \left [ \frac{\left ( \frac{4}{1 - m} - 2d \right ) t}{Bt^{\frac{2}{2 - d (1 - m)}}} \right ]^{\frac{1}{1 - m}} \int_0^{\infty} \frac{B^{\frac{d}{2}} t^{\frac{d}{2 - d (1 - m)}} s^{d - 1}}{\left ( s^2 + 1 \right )^{\frac{1}{1 - m}}} \textup{d}s \nonumber \\
&=& d \omega_d \left ( \frac{4}{1 - m} - 2d \right )^{\frac{1}{1 - m}} B^{\frac{d}{2} - \frac{1}{1 - m}} \int_0^{\infty} \frac{s^{d - 1}}{\left ( s^2 + 1 \right )^{\frac{1}{1 - m}}} \textup{d}s \nonumber \\
&=& \frac{d \omega_d}{2} \left ( \frac{4}{1 - m} - 2d \right )^{\frac{1}{1 - m}} B^{\frac{d}{2} - \frac{1}{1 - m}} \frac{\Gamma \left ( \frac{d}{2} \right ) \Gamma \left ( \frac{1}{1 - m} - \frac{d}{2} \right )}{\Gamma \left ( \frac{1}{1 - m} \right )} \label{barenblatt-mass}
\end{eqnarray}
where we have recognized a beta function. The time that it takes for a Barenblatt profile's peak to reach
$E_{\mathrm{H}}$ is given by:
\begin{equation}
T = \left [ E_{\mathrm{H}} \left ( \frac{B}{\frac{4}{1 - m} - 2d} \right )^{\frac{1}{1 - m}} \right ]^{(1 - m) - \frac{2}{d}} \; . \label{time-with-b}
\end{equation}
If we wanted to call this a Barenblatt profile of mass $a E_{\mathrm{F}}$, we could use the mass relation
(\ref{barenblatt-mass}) to replace $B$:
\begin{equation}
T = \left ( \frac{4}{1 - m} - 2d \right )^{\frac{2m - d(1 - m)}{d(1 - m)}} E_{\mathrm{H}}^{\frac{d(1 - m) - 2}{d}} \left [ \frac{2 a E_{\mathrm{F}}}{d \omega_d \frac{\Gamma \left ( \frac{d}{2} \right ) \Gamma \left ( \frac{1}{1 - m} - \frac{d}{2} \right )}{\Gamma \left ( \frac{1}{1 - m} \right )}} \right ]^{\frac{\frac{4}{d} - 2(1 - m)}{2 - d(1 - m)}} \; . \label{time-without-b}
\end{equation}
Even though this decay time is concerned with a Dirac delta of mass $a E_{\mathrm{F}}$, it upper bounds the decay time
for a box of mass $a E_{\mathrm{F}}$ by the concentration comparison theorem. Things are even nicer than this; the
Barenblatt profiles attract all solutions to the Cauchy problems for these nonlinear diffusion equations. As time goes
on, general solutions $E(\cdot, t)$ will converge to $U(\cdot, t)$ and the rate of convergence has been found
\cite{carrillo}:
\begin{equation}
\lim_{t \rightarrow \infty} t^{\frac{d}{2 - d(1 - m)}} \left | \left | E(\cdot, t) - U(\cdot, t) \right | \right |_{L^{\infty}} = 0 \; . \nonumber
\end{equation}
This means that (\ref{time-without-b}) still depends on the mass in the same way as the true decay time. Since
$O \left ( \left ( aE_{\mathrm{F}} \right )^{\frac{\frac{4}{d} - 2 (1 - m)}{2 - d (1 - m)}} \right )$ and
$O \left ( \left ( aE_{\mathrm{F}} \right )^2 \right )$ are very different, we should not have used results from an
unbounded domain to derive (\ref{timescale-1}) and (\ref{timescale-2}). Implicit in all of this is a certain relation
between $d$ and $m$. Consider this list of conditions:
\begin{enumerate}
\item
The constant $\frac{4}{1 - m} - 2d$ must be positive for the Barenblatt solutions to be well defined.
\item
The overall power of $t$ in (\ref{barenblatt}) must be negative in order for the peak value to decay with time.
\item
In order for the mass (\ref{barenblatt-mass}) to be finite, $\frac{s^{d - 1}}{\left ( s^2 + 1 \right )^{\frac{1}{1 - m}}}$ must approach zero faster than $\frac{1}{s}$.
\item
An $O \left ( \left ( aE_{\mathrm{F}} \right )^{\frac{\frac{4}{d} - 2 (1 - m)}{2 - d (1 - m)}} \right )$ decay time should decrease with the mass.
\end{enumerate}
Any one of these four requirements will tell us that $m > \frac{d - 2}{d}$. Since diffusion in our model happens
when $0 < \alpha < 1$, our primary interest is $-1 < m < 0$. This is only above the critical exponent in one dimension
providing another reason why it was necessary to demand a bounded domain. We know from \cite{esteban, rodriguez} that
when $d = 1$, the very fast diffusion equation yields infinitely many solutions and one energy conserving solution.
The latter is indeed asymptotic to the Barenblatt profile having the same mass. However, in $d = 2$ we already see
that there is no energy conserving solution to the very fast diffusion equation. It would be interesting to find
some analogue of (\ref{barenblatt}) for bounded domains that acts as a fundamental solution and stays well defined for
all $m$ and $d$.

\section{Numerical analysis}
The expressions for our time scales rely on the fact that $E_{\mathrm{F}} \gg E_{\mathrm{H}}$ in a domain that has a
large but finite size. To see how sensitively the inequalities depend on these factors, we have undertaken a numerical
test of our time scales and (\ref{timescale-1}) in particular. The code used for most simulations can be found in the
appendix. One thing that is immediately visible in the code is our decision to use initial conditions that decay
according to a power law:
\begin{equation}
E_0(x) = E_{\mathrm{max}} \left ( \frac{1}{1 + x^2} \right )^{\frac{k}{2}} \; . \nonumber
\end{equation}
For the very fast diffusion equation on an unbounded domain, we saw that it was possible for the initial condition to
decay too quickly ($k > \frac{1}{-m}$). Since the nicest functions to simulate are the ones with a finite
$\frac{\textup{d}M}{\textup{d}t}(0)$, this rules out functions $E_0$ that decay exponentially or have compact support.
For low energies, the relevant exponent is $m + 1 = \alpha = \frac{9}{10}$. This confines us to using $k$ values smaller
than $10$ and not so much smaller that the decay takes forever.
The numerics are done on a bounded domain so $\frac{\textup{d}M}{\textup{d}t}(0)$ is technically finite for
any $k$. However, if we were to pick $k > 10$, the flux at the edges would experience unbounded growth as
we made the domain larger casting doubt on our ability to trust those results.

\subsection{Implementation details}
In (\ref{filtration-equation}), we have been considering a function $\beta = \frac{\textup{d}S}{\textup{d}E}$ which
has diffusive behaviour for low and high energies with a Hagedorn phase and possibly a clustering phase in between. The
chosen entropy should be of the Super Yang-Mills type (\ref{sym-entropy}) where
\begin{equation}
S(E) \sim
\begin{cases}
E^{\frac{9}{10}} & E < E_{\mathrm{H}} \\
E & E_{\mathrm{H}} < E < E^{\prime}_{\mathrm{F}} \\
E^{\frac{8}{7}} & E^{\prime}_{\mathrm{F}} < E < E_{\mathrm{F}} \\
E^{\frac{3}{4}} & E > E_{\mathrm{F}}
\end{cases}
\; . \nonumber
\end{equation}
A smooth approximation to this can be accomplished with the function:
\begin{equation}
S(E) = \left ( \frac{E}{E_{\mathrm{H}}} \right )^{\frac{9}{10}} \left ( 1 + \frac{E}{E_{\mathrm{H}}} \right )^{\frac{1}{10}} \left ( 1 + \frac{E}{E^{\prime}_{\mathrm{F}}} \right )^{\frac{1}{7}} \left ( 1 + \frac{E}{E_{\mathrm{F}}} \right )^{-\frac{11}{28}} \; . \label{numerical-s1}
\end{equation}
However, the clustering phase where $S$ is convex is difficult to simulate numerically and our results including
(\ref{timescale-1}) have neglected it. Assuming that $E^{\prime}_{\mathrm{F}} = E_{\mathrm{F}}$ so that this phase does
not exist, we arrive at:
\begin{equation}
S(E) = \left ( \frac{E}{E_{\mathrm{H}}} \right )^{\frac{9}{10}} \left ( 1 + \frac{E}{E_{\mathrm{H}}} \right )^{\frac{1}{10}} \left ( 1 + \frac{E}{E_{\mathrm{F}}} \right )^{-\frac{1}{4}} \; . \label{numerical-s2}
\end{equation}
\begin{figure}[h]
\centering
\subfloat[][$E^{\prime}_{\mathrm{F}} \neq E_{\mathrm{F}}$]{\includegraphics[scale=0.35]{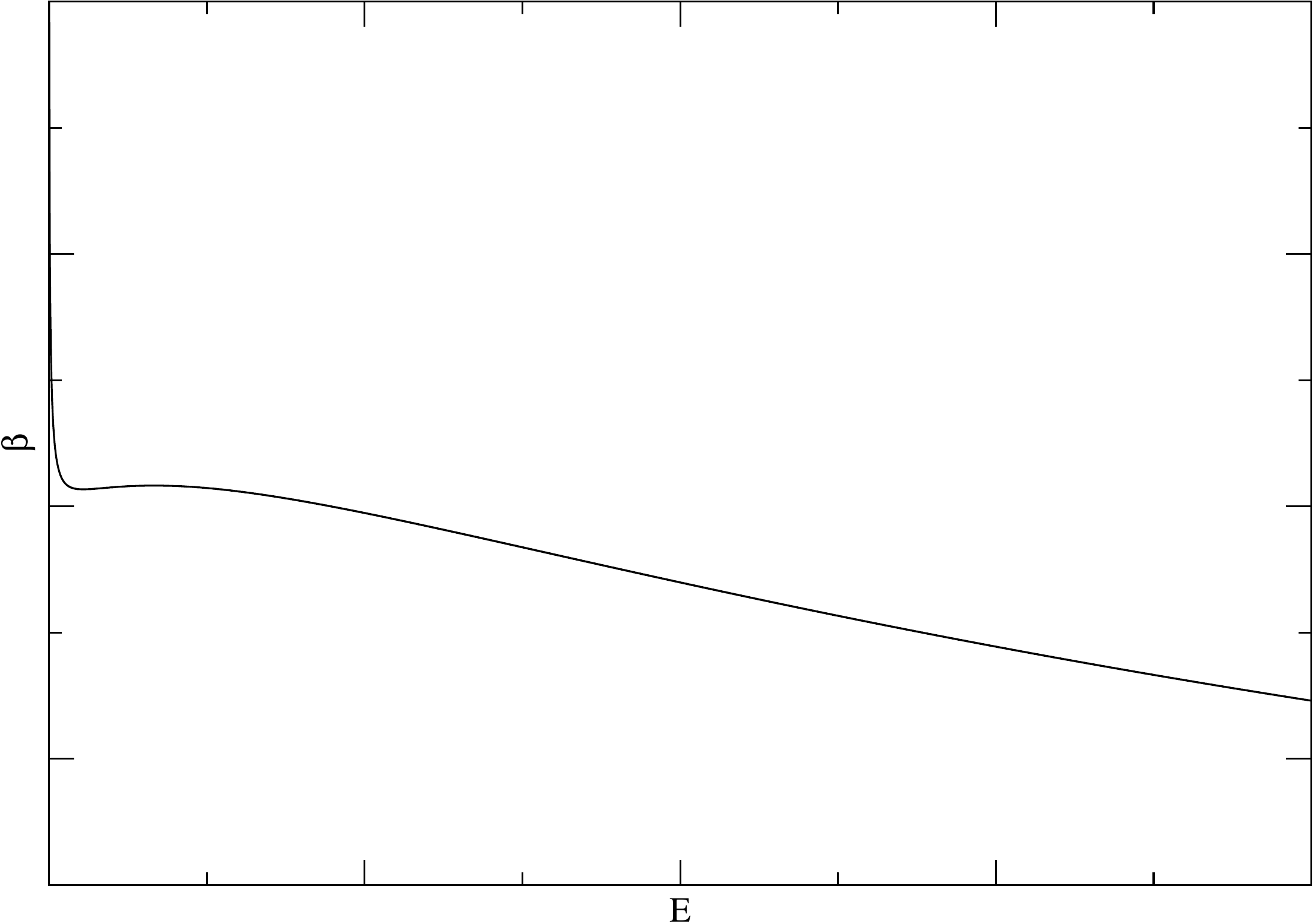}}
\subfloat[][$E^{\prime}_{\mathrm{F}} = E_{\mathrm{F}}$]{\includegraphics[scale=0.35]{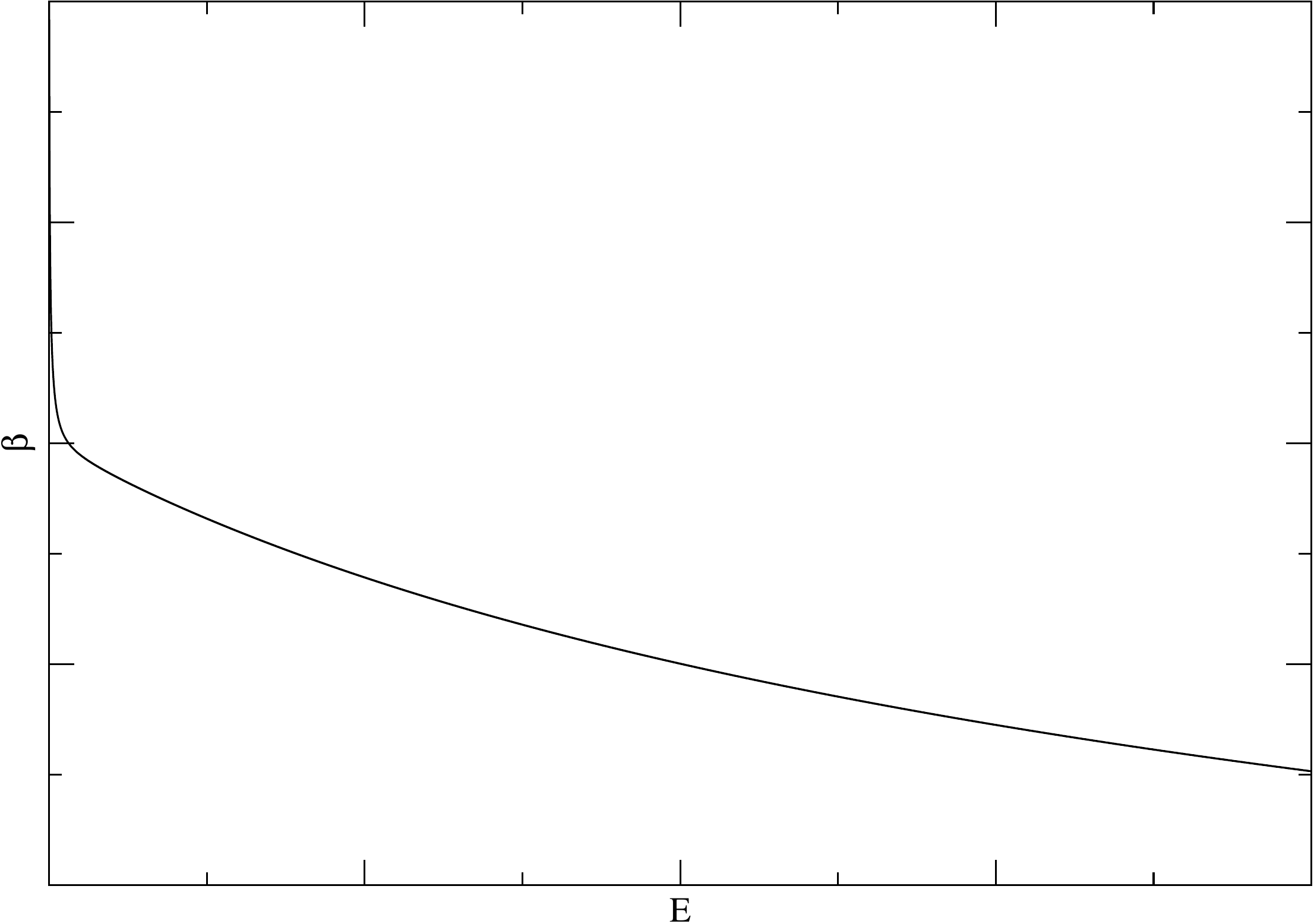}}
\caption{One the left is the derivative of (\ref{numerical-s1}) showing the three phases. The simplest choice for
removing the clustering phase leads us to differentiate (\ref{numerical-s2}), the plot on the right.}
\label{naive-cluster-removal}
\end{figure}
Plotting the $\beta(E)$ derived from these two functions, we see a slight problem. After the flat region of
(\ref{numerical-s1}), $\beta(E)$ increases reaching a local maximum. After this it begins to decrease more and more
quickly until it reaches a point of inflection and asymptotically approaches 0. The plot for (\ref{numerical-s2}), on
the other hand, does not look very flat in any region and seems to always be asymptotically approaching 0. There is no
point of inflection where some other behaviour ``crosses over'' to $E^{-\frac{1}{4}}$ behaviour. We will now show by
brute force that there is no inflection point, ruling out the possibility that this only happens for certain
$E_{\mathrm{H}}$ and $E_{\mathrm{F}}$.
\begin{claim}
For any $N > 0$, $\beta: \mathbb{R}_{> 0} \rightarrow \mathbb{R}$ defined by
\begin{equation}
\beta(x) = \frac{9}{10} x^{-\frac{1}{10}} (1 + x)^{\frac{1}{10}} \left ( 1 + \frac{x}{N} \right )^{-\frac{1}{4}} + \frac{1}{10} x^{\frac{9}{10}} (1 + x)^{-\frac{9}{10}} \left ( 1 + \frac{x}{N} \right )^{-\frac{1}{4}} - \frac{1}{4N} x^{\frac{9}{10}} (1 + x)^{\frac{1}{10}} \left ( 1 + \frac{x}{N} \right )^{-\frac{5}{4}} \nonumber
\end{equation}
has no point of inflection.
\end{claim}
\begin{proof}
First we will write
\begin{eqnarray}
\beta^{\prime}(x) &=& -\frac{9}{100} \left ( 1 + \frac{x}{N} \right )^{-\frac{1}{4}} \left [ x^{-\frac{11}{10}} (1 + x)^{\frac{1}{10}} - 2x^{-\frac{1}{10}} (1 + x)^{-\frac{9}{10}} + x^{\frac{9}{10}} (1 + x)^{-\frac{19}{10}} \right ] \nonumber \\
&& -\frac{1}{20N} \left ( 1 + \frac{x}{N} \right )^{-\frac{5}{4}} \left [ 9 x^{-\frac{1}{10}} (1 + x)^{\frac{1}{10}} + x^{\frac{9}{10}} (1 + x)^{-\frac{9}{10}} \right ] + \frac{5}{16N^2} x^{\frac{9}{10}} (1 + x)^{\frac{1}{10}} \left ( 1 + \frac{x}{N} \right )^{-\frac{9}{4}} \; . \nonumber
\end{eqnarray}
Now the equation $\beta^{\prime\prime}(x) = 0$ becomes:
\begin{eqnarray}
&& \frac{9}{1000} \left ( 1 + \frac{x}{N} \right )^{-\frac{1}{4}} \left [ 11 x^{-\frac{21}{10}} (1 + x)^{\frac{1}{10}} - 3 x^{-\frac{11}{10}} (1 + x)^{-\frac{9}{10}} - 27 x^{-\frac{1}{10}} (1 + x)^{-\frac{19}{10}} + 19 x^{\frac{9}{10}} (1 + x)^{-\frac{29}{10}} \right ] \nonumber \\
&& +\frac{27}{400N} \left ( 1 + \frac{x}{N} \right )^{-\frac{5}{4}} \left [ x^{-\frac{11}{10}} (1 + x)^{\frac{1}{10}} - 2 x^{-\frac{1}{10}} (1 + x)^{-\frac{9}{10}} + x^{\frac{9}{10}} (1 + x)^{-\frac{19}{10}} \right ] \nonumber \\
&& +\frac{3}{32N^2} \left ( 1 + \frac{x}{N} \right )^{-\frac{9}{4}} \left [ 9 x^{-\frac{1}{10}} (1 + x)^{\frac{1}{10}} + x^{\frac{9}{10}} (1 + x)^{-\frac{9}{10}} \right ] - \frac{45}{64N^3} x^{\frac{9}{10}} (1 + x)^{\frac{1}{10}} \left ( 1 + \frac{x}{N} \right )^{-\frac{13}{4}} = 0 \nonumber \\
&& \frac{3}{125} \left ( 1 + \frac{x}{N} \right )^3 \left [ 11 (1 + x)^3 - 3 x (1 + x)^2 -27 x^2 (1 + x) + 19 x^3 \right ] \nonumber \\
&& +\frac{9}{50N} \left ( 1 + \frac{x}{N} \right )^2 \left [ x (1 + x)^3 - 2 x^2 (1 + x)^2 + x^3 (1 + x) \right ] \nonumber \\
&& +\frac{1}{4N^2} \left ( 1 + \frac{x}{N} \right ) \left [ 9 x^2 (1 + x)^3 + x^3 (1 + x)^2 \right ] - \frac{15}{8N^3} x^3 (1 + x)^3 = 0 \nonumber \\
&& 625 x^6 + (2500N + 1625) x^5 + (7250 N + 2275) x^4 + (9250N + 819) x^3 \nonumber \\
&& + (2340N^2 + 3402N) x^2 + (720N^3 + 972N^2) x + 264N^3 = 0 \; . \nonumber
\end{eqnarray}
Even though we cannot factor this sixth degree polynomial, we may conclude that it has no positive real roots via
Descartes' rule of signs.
\end{proof}
What we should take away from this is that removing the clustering phase by hand is too crude to give us a good
filtration function to use in the numerics. Instead we will use a function like the one plotted in Figure
\ref{numerical-beta} where quadratic interpolation has been used to guarantee that it is differentiable and decreasing
with an almost-flat section and two inflection points.
\begin{figure}[h]
\includegraphics[scale=0.6]{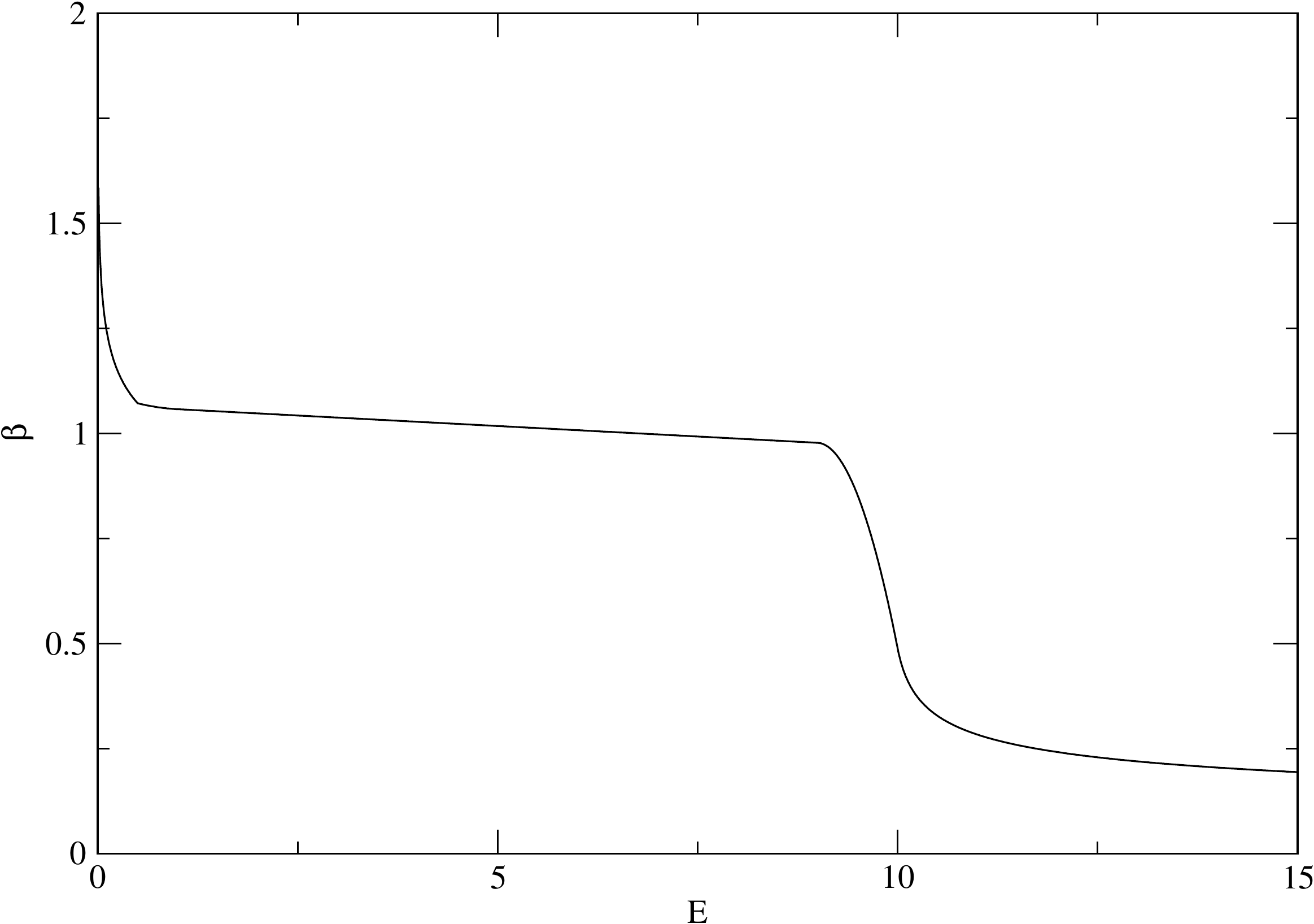}
\caption{For this function, the flat section begins at $E = 1.0$. The region $0.5 < E < 1.0$ is used to interpolate
between the flat section and the $E^{-\frac{1}{10}}$ power law. The flat section ends at $E = 0.9 E_{\mathrm{F}}$ and
the $E^{-\frac{1}{4}}$ power law begins at $E = E_{\mathrm{F}}$. In the region
$0.9 E_{\mathrm{F}} < E < E_{\mathrm{F}}$, a parabola again interpolates. The plot is for $E_{\mathrm{F}} = 10.0$ but
this general rule has been followed for all $E_{\mathrm{F}}$.}
\label{numerical-beta}
\end{figure}

\subsubsection{The Crank-Nicolson method}
Now that the form of (\ref{filtration-equation}) has been decided, we need a method capable of finding approximate
solutions to it. Our PDE involves time and space, so it is natural to discretize space and reduce it to the following
set of coupled ODEs:
\begin{eqnarray}
\frac{\textup{d}}{\textup{d}t} E_i(t) &=& f_i(E(t), t) \nonumber \\
f_i(E(t), t) &=& -\frac{\beta(E_{i+1}(t)) - 2\beta(E_i(t)) + \beta(E_{i-1}(t))}{(dx)^2} \; . \nonumber
\end{eqnarray}
The simplest method for integrating this type of ODE system is the forward Euler method:
\begin{equation}
E_i(t + dt) = E_i(t) + f_i(E(t), t) dt \; . \label{forward-euler}
\end{equation}
One might expect a method like (\ref{forward-euler}) to be accurate as long as $dx$ and $dt$ are small. However, one
can use the heat equation to show that things are not so simple.
\begin{figure}[h]
\centering
\subfloat[][$dt = 0.005$, $dx = 0.1$]{\includegraphics[scale=0.35]{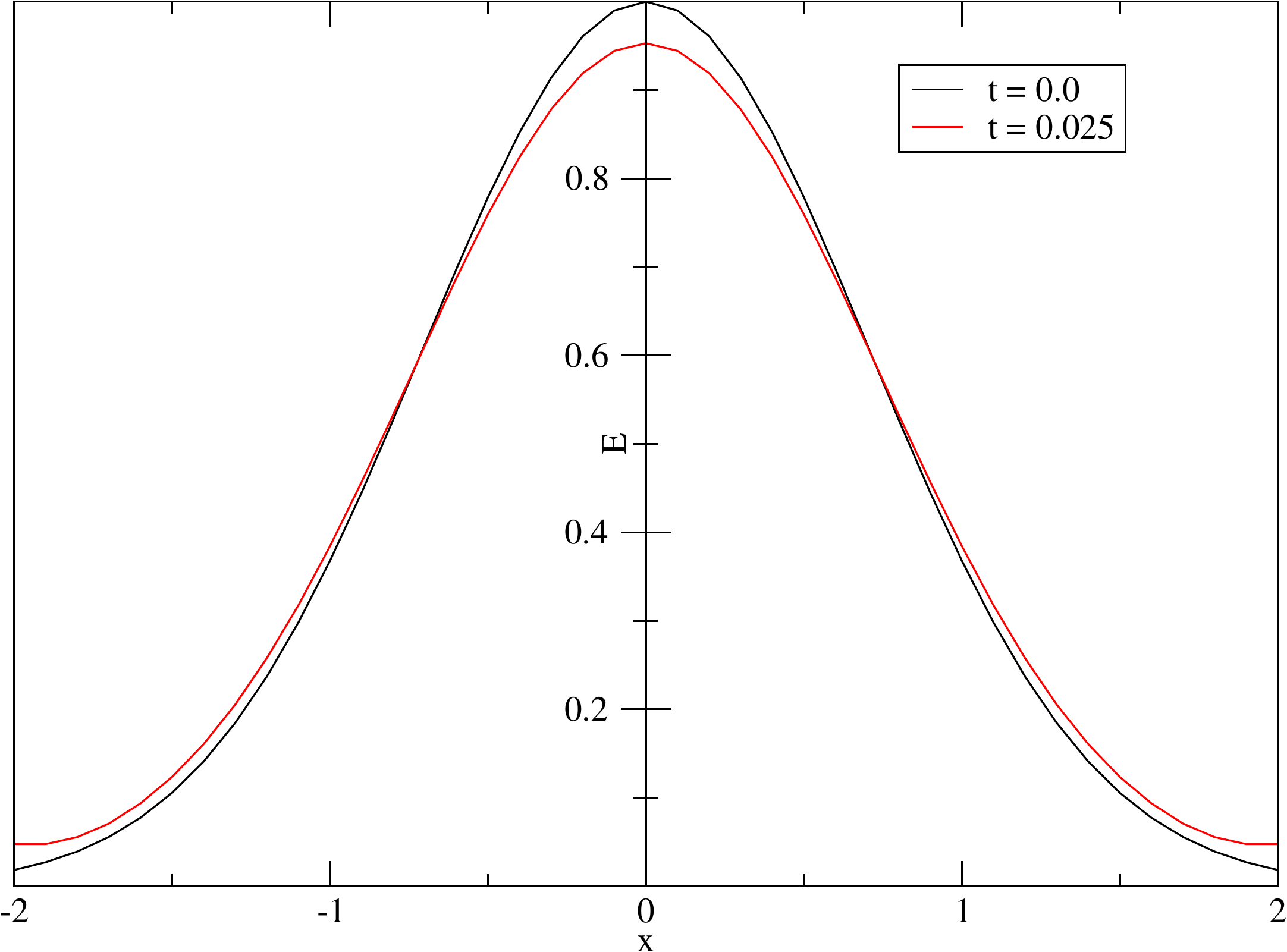}}
\subfloat[][$dt = 0.005$, $dx = 0.05$]{\includegraphics[scale=0.35]{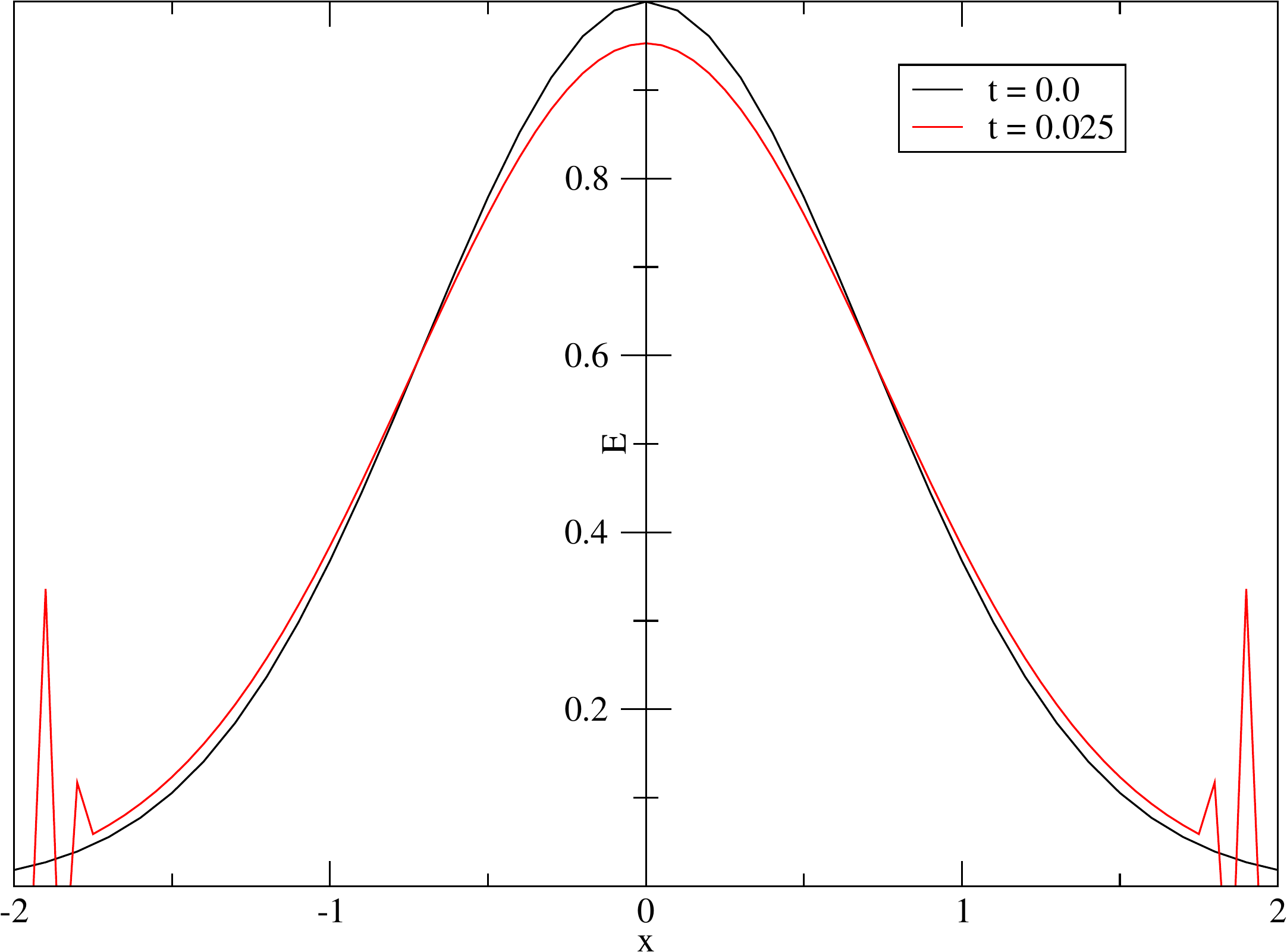}}
\caption{A standard demonstration that the forward Euler method is numerically unstable for the heat equation.}
\label{numerical-stability}
\end{figure}
Figure \ref{numerical-stability} shows that making the discretization smaller is not always an improvement. To avoid
numerical instabilities, the restriction that must be obeyed is $dt < \frac{1}{2} (dx)^2$. This is the precise
restriction for the heat equation but a nonlinear equation would suffer from a similar restriction. This is especially
harsh for investigating the long decay times associated with a Hagedorn phase. Although they are more difficult to
apply, there are other methods known to be stable for large timesteps \cite{press}. One of these is the backward Euler
method:
\begin{equation}
E_i(t + dt) = E_i(t) + f_i(E(t + dt), t + dt) dt \; . \label{backward-euler}
\end{equation}
In (\ref{forward-euler}), we evolve $E$ forward in time by adding the derivative that it \textit{has now}. In
(\ref{backward-euler}) which looks similar, we add the derivative that it \textit{will have} after we add it. Since
some inversion is clearly necessary, we refer to this type of method as an \textit{implicit method}. In this work, we
will use the average of the forward and backward Euler methods, known as the Crank-Nicolson method. In order to
isolate $E_i(t + dt)$ in (\ref{backward-euler}) when treating the heat equation, we must solve a system of linear
equations - one for each site in the lattice. For a nonlinear equation like (\ref{filtration-equation}), these become
nonlinear algebraic equations, and hence often require numerical methods themselves. Something that might look familiar
is:
\begin{figure}[h]
\centering
\subfloat[][Linear scale]{\includegraphics[scale=0.35]{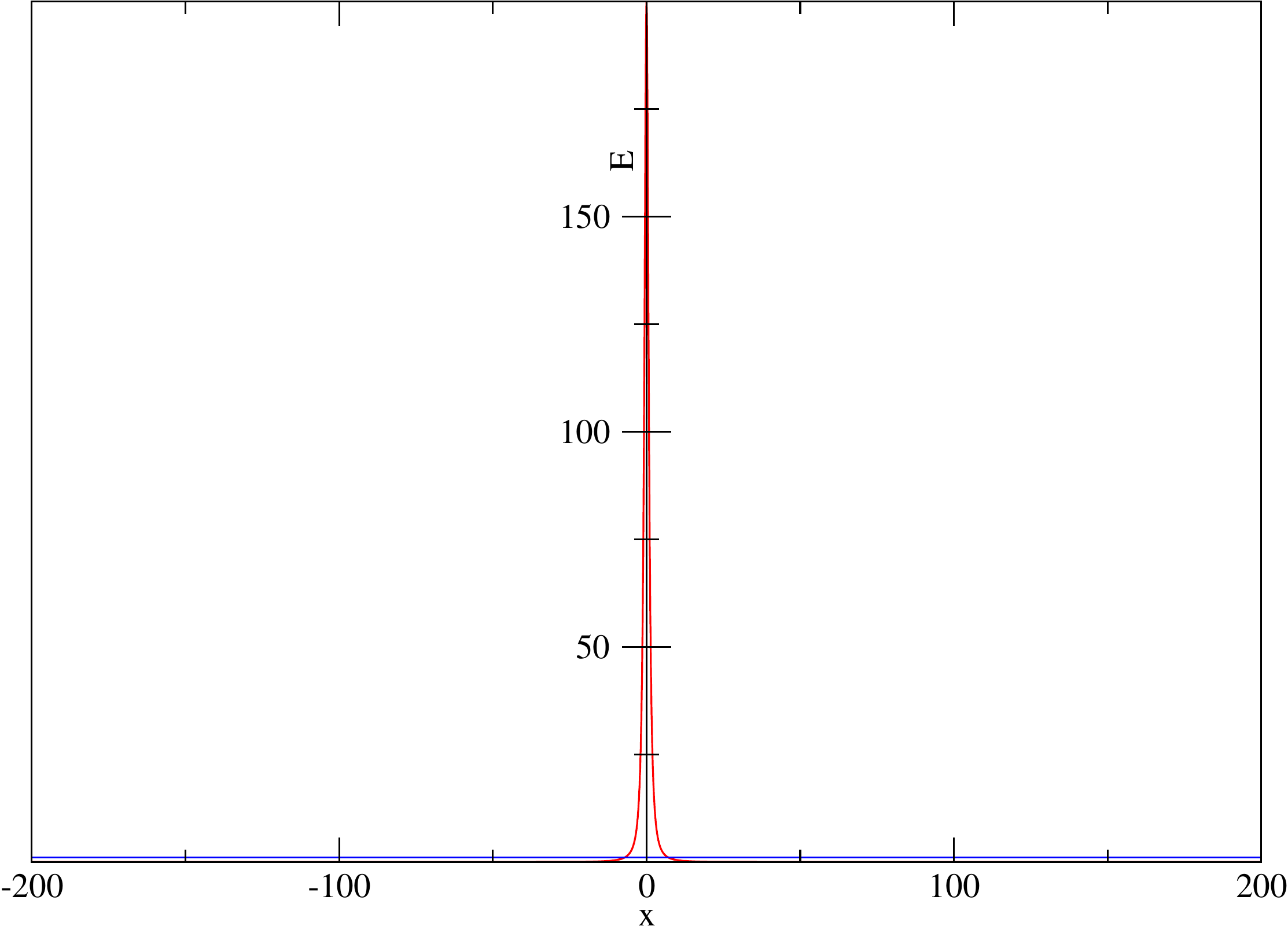}}
\subfloat[][Modified logarithmic scale]{\includegraphics[scale=0.35]{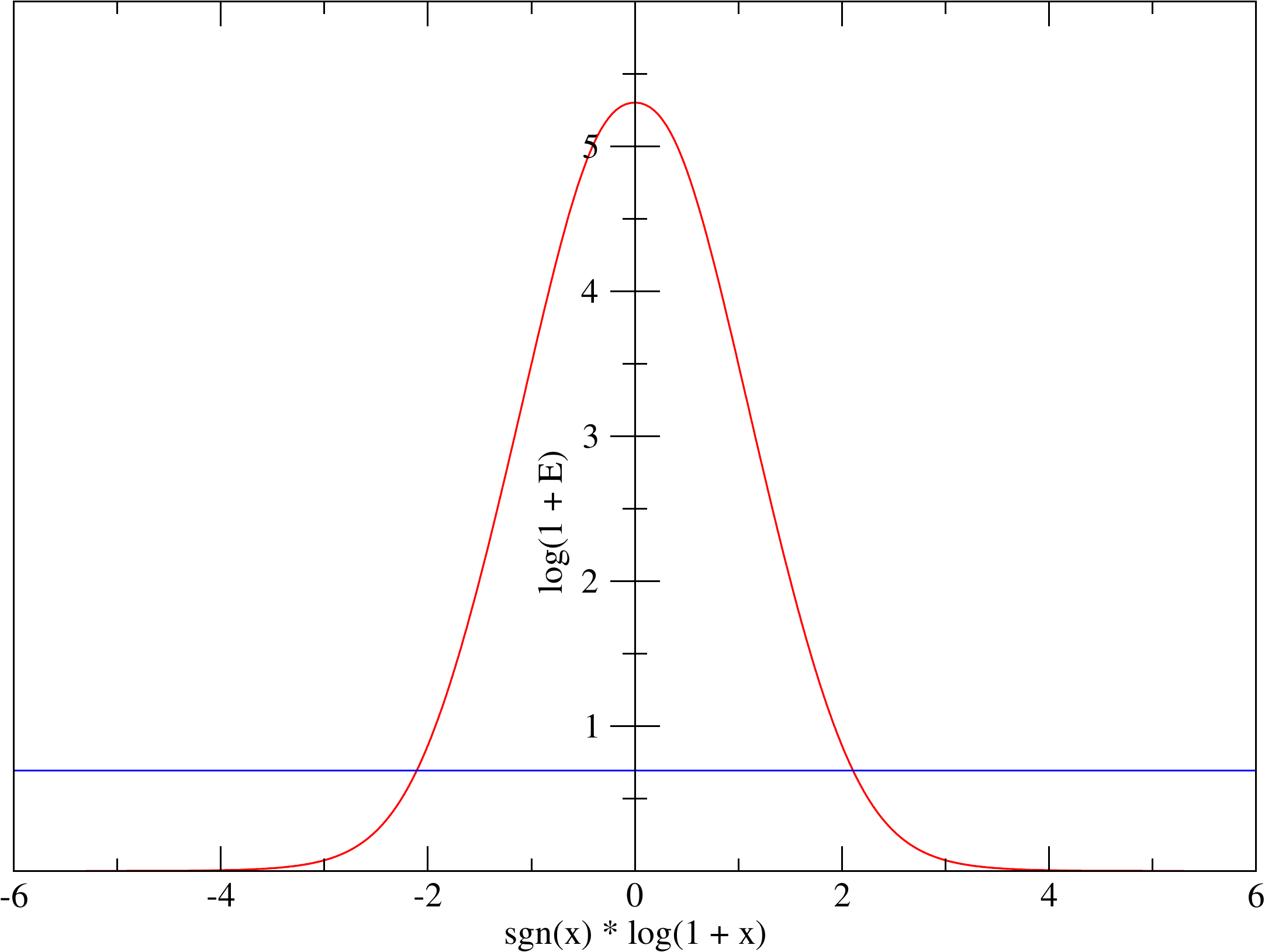}}
\caption{Typically we will work with an initial condition like that shown in red. The width of the domain is chosen
so that all of the function's mass can fit below the Hagedorn energy in blue. The plot on the left shows that $E_0(x)$
is very flat for most of the $x$ values. An adaptive $dx$ makes sense and will be chosen so that the minimum $dx$ is
much smaller than the distance between the inflection points. The plot on the right shows the graph of the same function
but more conveniently.}
\label{no-exaggeration}
\end{figure}
\begin{equation}
x_{n + 1} = x_{n} - \frac{f(x_n)}{f^{\prime}(x_n)} \; . \nonumber
\end{equation}
This is Newton's method where the sequence converges to a zero of $f$ unless the initial guess $x_0$ is sufficiently
far away. The derivative becomes the Jacobian when we have more than one variable.
\begin{eqnarray}
g^{\prime}(E^n(t + dt))_{ij} (E_j^{n + 1}(t + dt) - E_j^n(t + dt)) = -g_i(E^n(t + dt)) \nonumber \\
g_i(E(t + dt)) = E_i(t) + f_i(E(t + dt)) dt - E_i(t + dt) \label{newton-method}
\end{eqnarray}
We should solve this (tridiagonal) linear system for an initial guess of $E_i^0(t + dt) = E_i(t)$ since $dt$ is small
enough to not change the energy much with each step. The Crank-Nicolson method for us is then
\begin{eqnarray}
E_i(t + dt) &=& \frac{1}{2} \left [ E_i(t) + f_i(E(t)) dt + E^{\infty}_i(t + dt) \right ] \nonumber \\
f_i(E(t)) &=& -\frac{\beta(E_{i+1}(t)) - 2\beta(E_i(t)) + \beta(E_{i-1}(t))}{(dx)^2} \; , \label{crank-nicolson}
\end{eqnarray}
with $E^{\infty}_i(t + dt)$ generated in the (\ref{newton-method}) way. One aspect of this which is still undesirable
is having a constant $dx$. Given the energy profiles we wish to evolve, it will be more convenient to have $dx$ depend
on $i$ so that the lattice can be finely grained in regions where the energy profile changes the most. Figure
\ref{no-exaggeration} shows the need for a variable spatial step and also establishes the conventions we will use to
have logarithmic axes.

\subsubsection{Convergence tests}
We wish to test this algorithm, but do not have any exact solutions at our disposal. To verify that our finite
difference scheme converges to a solution as $O(h^2) \equiv O ((dx)^2)$, there is a useful test suggested by
\cite{choptuik} based on the Richardson expansion:
\begin{equation}
E^h(x, t) = E(x, t) + e_2(x, t) h^2 + e_4(x, t) h^4 + \dots \; . \label{richardson}
\end{equation}
Here $E^h$ is the approximate solution computed by the Crank-Nicolson method, which differs from the exact
solution according to functions $e_2$, $e_4$, \textit{etc}. If we choose different discretizations,
$dx \in \{h, \; 2h, \; 4h\}$ for example, (\ref{richardson}) states that for a given $(x, t)$,
\begin{equation}
Q(t) = \lim_{h \rightarrow 0} \frac{E^{4h}(x, t) - E^{2h}(x, t)}{E^{2h}(x, t) - E^h(x, t)} = \lim_{h \rightarrow 0} \frac{16h^2 - 4h^2}{4h^2 - h^2} = 4 \; . \nonumber
\end{equation}
As a better check, we will compare the 2-norms of the numerical solutions instead of choosing a single point $x$.
\begin{equation}
Q(t) = \lim_{h \rightarrow 0} \frac{\left | \left | E^{4h}(\cdot, t) - E^{2h}(\cdot, t) \right | \right |_{\ell^2}}{\left | \left | E^{2h}(\cdot, t) - E^h(\cdot, t) \right | \right |_{\ell^2}} = 4 \label{q-factor}
\end{equation}
Plotting this for five different values of $h$, it looks like $Q(t) = 2$. This would be expected for a non-centred
finite difference method whose Richardson expansion looks like:
\begin{equation}
E^h(x, t) = E(x, t) + e_1(x, t) h + e_2(x, t) h^2 + \dots \; . \nonumber
\end{equation}
\begin{figure}[h]
\includegraphics[scale=0.6]{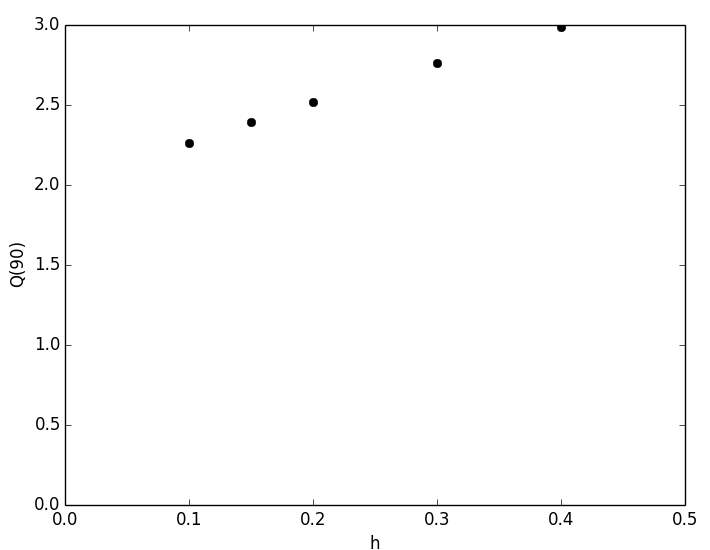}
\caption{The factor $Q(t)$ plotted for simulations that lasted $t = 90.0$. The grid was uniform and five step sizes
were used.}
\label{convergence-test}
\end{figure}
Our departure from the expected error $O(h^2)$ is likely a result of the nonlinear algebraic equations. As soon as we
use Newton's method, Crank-Nicolson is not being followed exactly and an $O(h)$ error can be introduced. Linear
combinations such as $2E^h - E^{2h}$ can be used to cancel this error term in a technique called Richardson
extrapolation. This has been successfully used for other diffusion equations with the Crank-Nicolson method
\cite{zlatev}.

\subsection{Simulation results}
Going ahead with the Crank-Nicolson method, one can print out slices of the energy every so often to see how it is
diffusing. The numerics make it clear that there are two time scales of interest. What we have predicted
in (\ref{timescale-1}) is the long decay time needed for a distribution of energy to descend below $E_{\mathrm{H}}$.
Before this, the distribution will reach $E_{\mathrm{F}}$ or some number close to it, at what we call the termalization
time. Figure \ref{commonality} shows that for a variety of initial conditions, the energy profiles are qualitatively
similar to step functions when this happens.
\begin{figure}[h]
\centering
\subfloat[][Wider function]{\includegraphics[scale=0.35]{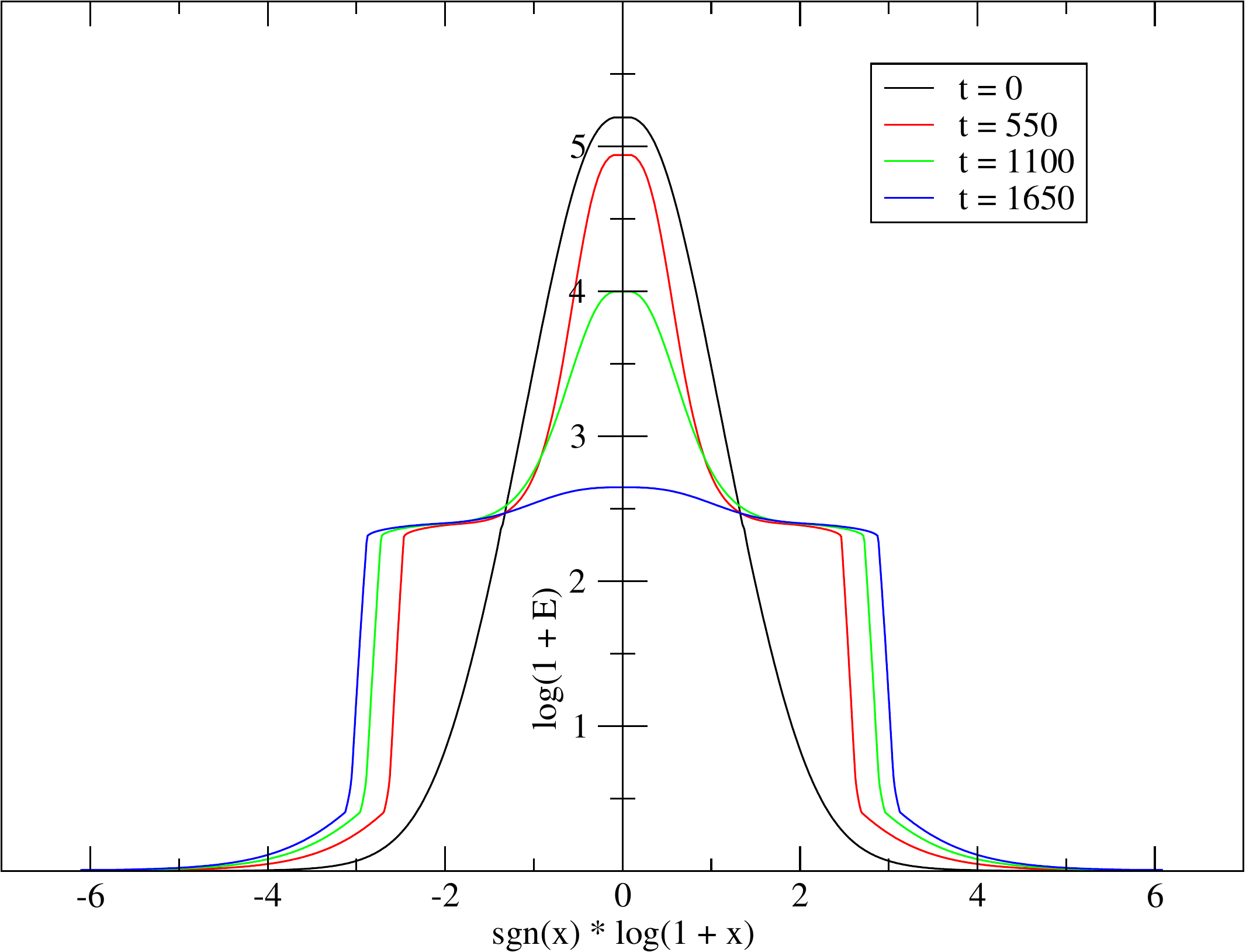}}
\subfloat[][Narrower function]{\includegraphics[scale=0.35]{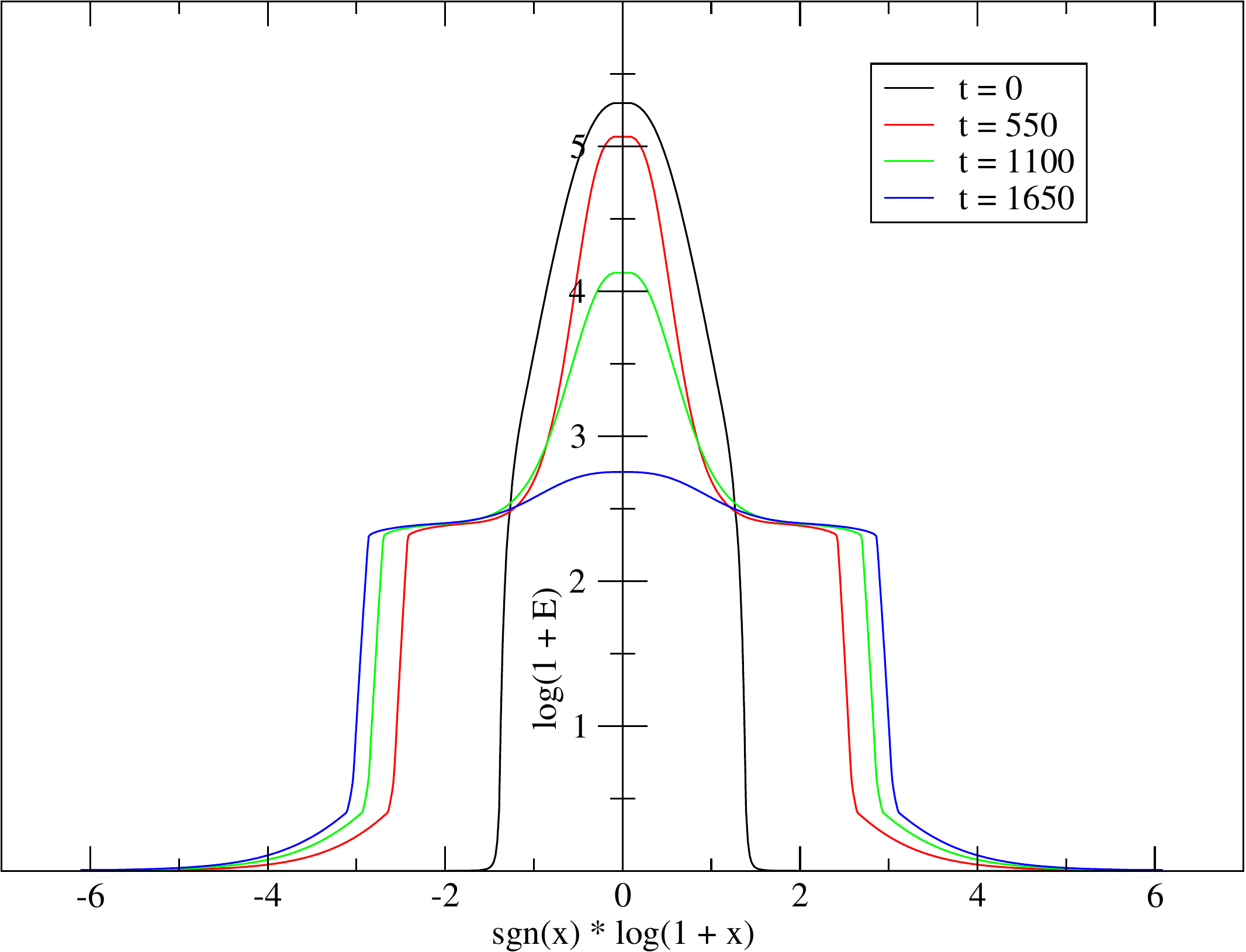}}
\caption{These initial conditions with the same total mass both evolve toward a state that has a flat line
in the high energy phase near $E_{\mathrm{F}} = 10.0$.}
\label{commonality}
\end{figure}

\subsubsection{Short time dynamics}
\begin{figure}[!t]
\subfloat[][Same peak]{\includegraphics[scale=0.4]{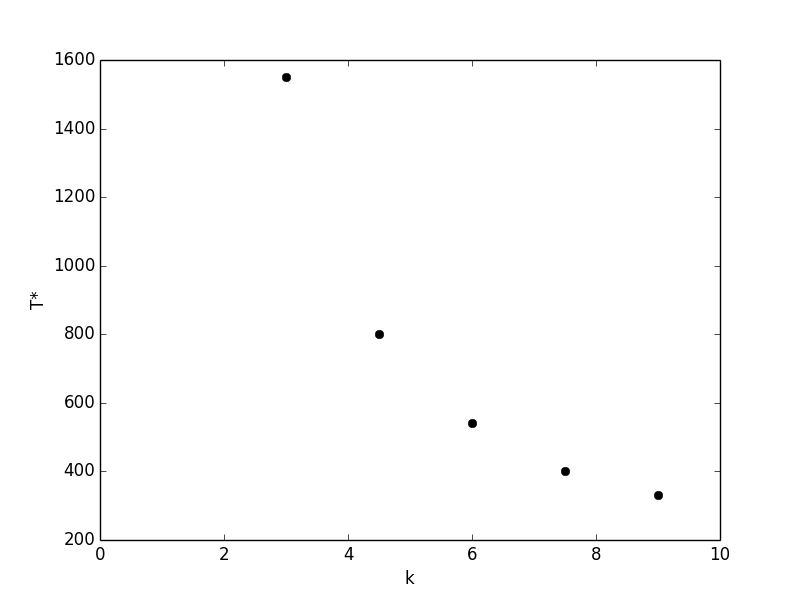}}
\subfloat[][Same mass]{\includegraphics[scale=0.4]{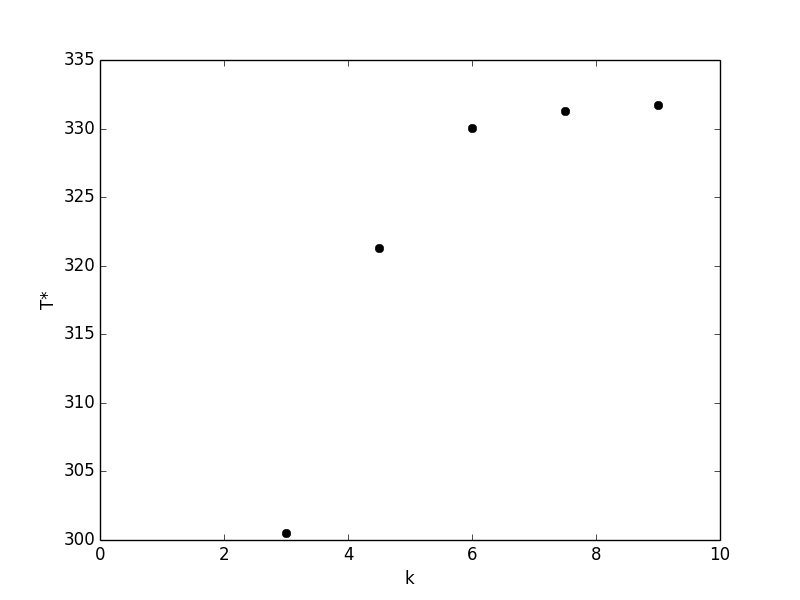}}
\caption{Two plots of the thermalization time for field theory energy $E_{\mathrm{F}} = 10.0$ and different values of
$k$ in the initial condition.}
\label{therm-times}
\end{figure}
\begin{figure}[!t]
\includegraphics[scale=0.4]{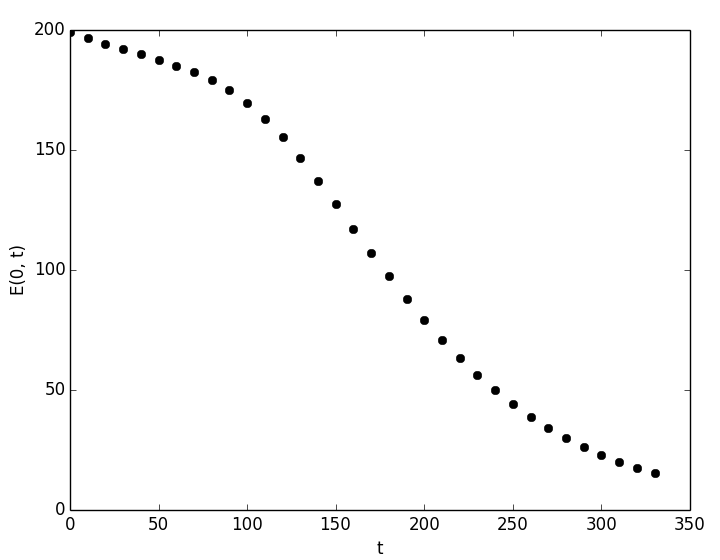}
\caption{A plot of how the peak energy moves down with time. This corresponds to the $k = 9$ decay in Figure
\ref{therm-times}.}
\label{nopower}
\end{figure}
The left plot in Figure \ref{commonality} uses the
$E_0(x) = 200 \left ( \frac{1}{1 + x^2} \right )^{\frac{k}{2}}$ form. It is fairly clear from the figure
that the profile is flattening out somewhere above $\log \left (1 + E_{\mathrm{F}} \right ) = \log 11$. After
interpolation, the effective Hagedorn and field theory energies become slightly shifted with respect to the parameters
$E_{\mathrm{H}}$ and $E_{\mathrm{F}}$. Therefore to measure thermalization times as a function of $k$, we have
used $\frac{3}{2} E_{\mathrm{F}}$ as a rule of thumb for where the peak energy should be. After making a plot for
five values of $k$, we see that the thermalization time is much more sensitive to variations in $L E_{\mathrm{F}} / M$
than $E_{\mathrm{F}} / E_{\mathrm{max}}$.

Another trend we may investigate is how $E(0, t)$ behaves as a function of $t$. Plotting this for $k = 9$ will give us
Figure \ref{nopower}. After a slow start at early times, the plot becomes steeper before levelling off again. This
basic shape holds for other $k$ values as well. It is different from the decay of a Barenblatt profile which would
always have $E(0, t)$ as a power law. The time $t^*$ when $E(0, t)$ is changing most quickly appears somewhere in
Figure \ref{nopower}. If we work with the functions having similar thermalization times (the ones normalized to have
the same mass), we find that $t^* \approx 0.6 T^*$ in all five cases. The ratio between
$\frac{\textup{d}}{\textup{d}t} E(0, t^*)$ and the average rate of change in the interval $(0, T^*)$ is about $1.8$. We
have not found a way to predict these numbers analytically.

\subsubsection{Long time dynamics}
\begin{figure}[!t]
\includegraphics[scale=0.4]{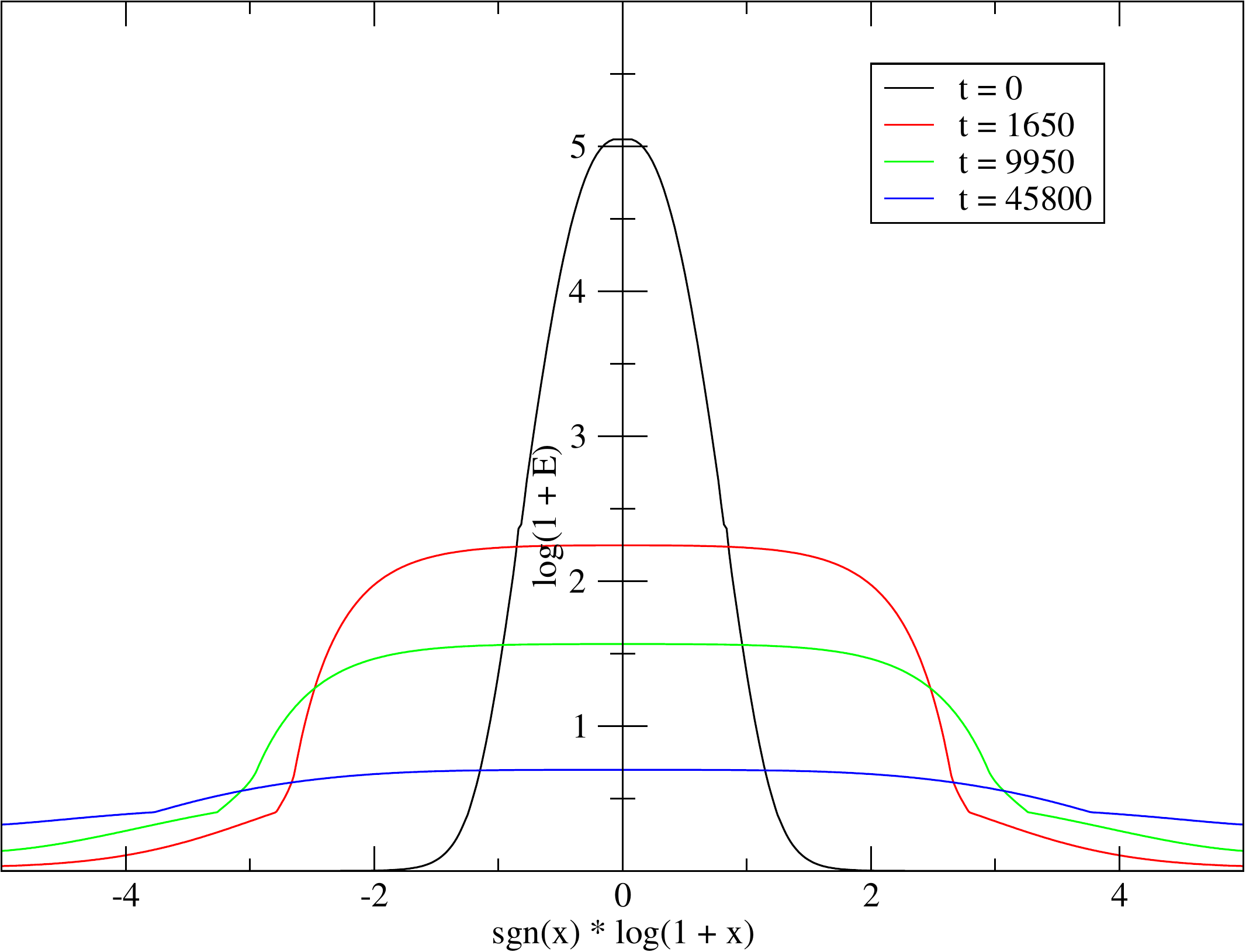}
\caption{This shows one of our initial conditions diffusing all the way through the Hagedorn regime. It takes about
a hundred times longer to do this than it does to reach the thermalization time.}
\label{decay-profiles}
\end{figure}
\begin{figure}[!b]
\includegraphics[scale=0.5]{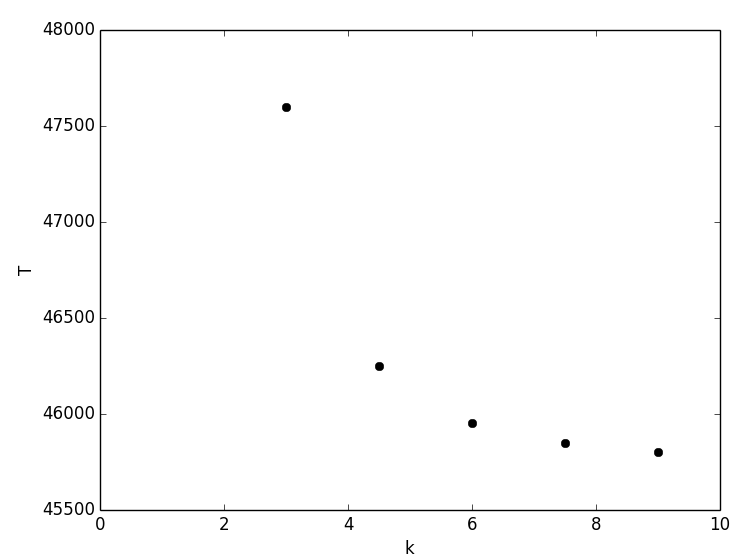}
\caption{A plot of the decay time for $E_{\mathrm{F}} = 10.0$, $E_{\mathrm{H}} = 1.0$ and five values of $k$. Even
though they specify different shapes for the initial condition, the times are all within $4\%$ of eachother.}
\label{decay-times}
\end{figure}
\begin{figure}[!h]
\includegraphics[scale=0.5]{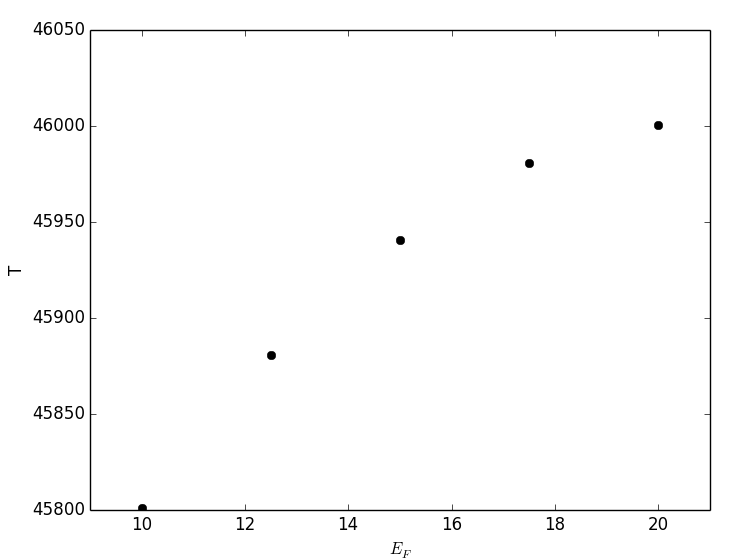}
\caption{A plot of the decay time for $k = 9$, $E_{\mathrm{H}} = 1.0$ and five values of $E_{\mathrm{F}}$. Even
though the profiles begin to flatten out at different heights, their final decay times are all within $0.4\%$ of
eachother.}
\label{decay-times-closer}
\end{figure}
The next step is to wait until $E(0, t) < E_{\mathrm{H}}$. An example of what this diffusion looks
like is in Figure \ref{decay-profiles}. This time, $T \gg T^*$ is predicted in (\ref{timescale-1}) to depend on the
squared mass. After all, a cylinder of radius $a$ and height $E_{\mathrm{F}}$ has a volume of
$2aE_{\mathrm{F}}$ if the base is one-dimensional. When investigating the short time dynamics, initial conditions of
five different shapes were normalized to have the same mass. Since $T^*$ was found to be almost the same for them, it
is no surprise that $T$ shares this property. Figure \ref{decay-times} shows a plot of these times.
As long as the mass is kept constant, early and late time dynamics are largely insensitive to $k$ so we will use the
initial condition
\begin{equation}
E_0(x) = 200 \left ( \frac{1}{1 + x^2} \right )^{\frac{9}{2}} \label{main-initial-condition}
\end{equation}
from now on. The major quantity we have not changed yet is the $E_{\mathrm{F}} / E_{\mathrm{H}}$ ratio. We have always
had $E_{\mathrm{F}} = 10.0$ and $E_{\mathrm{H}} = 1.0$ so far. As we know in one dimension, $T$ is proportional to
$M^2$ and not a power of $E_{\mathrm{F}}$. Whether or not we can change $E_{\mathrm{F}}$ and leave $T$ invariant will
be the real test of (\ref{timescale-1}). The condition
$1.0 = E_{\mathrm{H}} \ll E_{\mathrm{F}} \ll E_{\mathrm{max}} = 200.0$ needs to be satisfied so we will choose a few
$10.0 \leq E_{\mathrm{F}} \leq 20.0$ values. The decay times for these choices, plotted in Figure
\ref{decay-times-closer}, are remarkably close. The exact prefactors in (\ref{timescale-1}) still need to be checked
but it turns out that they are not very constraining. We know that $\alpha = \frac{9}{10}$ and taking the length $L$
from the code, we can plug it into (\ref{main-initial-condition}) to find $E_{\mathrm{min}} = E_0(L)$. Because
$E_{\mathrm{min}}$ is so small, the prefactor multiplying $M^2$ in the upper bound from (\ref{timescale-1}) is very
different from the prefactor multiplying $M^2$ in the lower bound from (\ref{timescale-1}). In this simulation, they
differ by a factor of
\begin{equation}
E_{\mathrm{min}}^{\alpha - 2} = \left ( 8.69 \cdot 10^{-19} \right )^{\frac{9}{10} - 2} = 7.36 \cdot 10^{19} \; , \nonumber
\end{equation}
which is much larger than any of the numbers put into the simulation by hand.

\subsubsection{Higher dimensions}
\begin{figure}[!h]
\subfloat[][Zoomed in]{\includegraphics[scale=0.3]{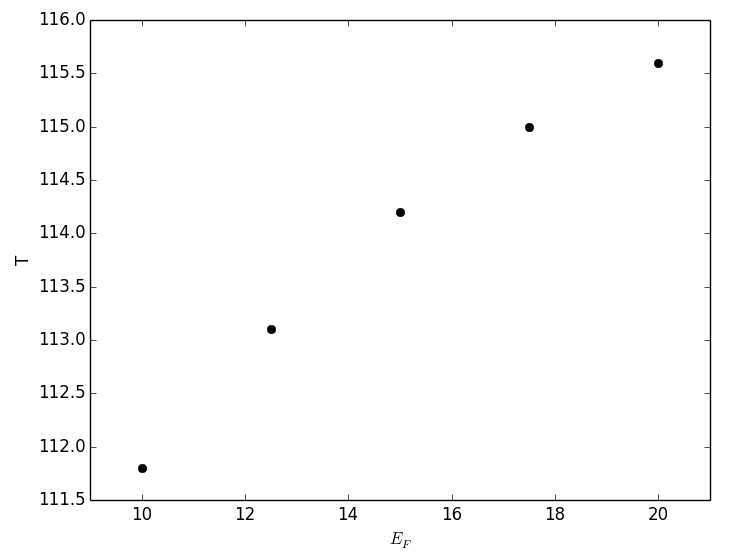}}
\subfloat[][Zoomed out]{\includegraphics[scale=0.3]{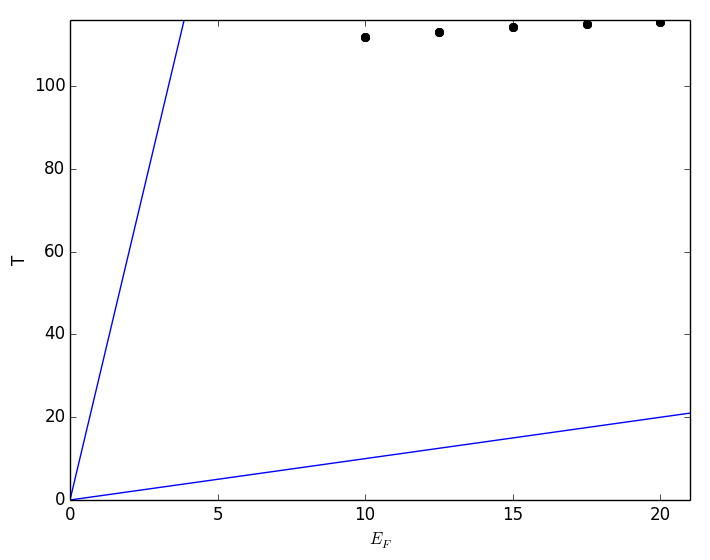}}
\caption{On the left is a plot of the decay time in two dimensions for $k = 9$, $E_{\mathrm{H}} = 1.0$ and five values
of $E_{\mathrm{F}}$. These are much smaller than the one-dimensional decay times in Figure \ref{decay-times-closer}.
The zoomed out version on the right shows that even though these dots do not vary linearly, they are sandwiched between
two bounds which do vary linearly. For the sake of the plot, the blue lines have been understated. In reality the 
lower bound is much closer to being horizontal and the upper bound is much closer to being vertical.}
\label{decay-times-2d}
\end{figure}
Good agreement between our prediction and the numerical results has been demonstrated in one dimension. In our case,
extending the numerics to higher dimensions is simple because we are focusing on spherically symmetric data. Our
equation becomes
\begin{eqnarray}
\frac{\partial E}{\partial t} &=& -r^{1 - d} \frac{\partial}{\partial r} \left ( r^{d - 1} \frac{\partial \beta(E)}{\partial r} \right ) \nonumber \\
&=& -\frac{\partial^2 \beta(E)}{\partial r^2} - \frac{d - 1}{r} \frac{\partial \beta(E)}{\partial r} \; . \nonumber
\end{eqnarray}
This adds a term with a single spatial derivative to our one-dimensional equation from before. A plot like Figure
\ref{decay-times-closer} should no longer be expected because $T$ in (\ref{timescale-1}) is no longer proportional to
$M^2$. In two dimensions,
$\left ( a E_{\mathrm{F}} \right )^2 = \left ( a^2 E_{\mathrm{F}} \right ) E_{\mathrm{F}} \propto M E_{\mathrm{F}}$.
This means we are looking for a decay time that varies linearly with the field theory energy. After modifying the
Crank-Nicolson code in the appendix, two-dimensional simulations testing this have been performed. The results, plotted
in Figure \ref{decay-times-2d}, do not appear to be proportional to $E_{\mathrm{F}}$. The $E_{\mathrm{F}} = 20.0$ decay
time is much less than twice the $E_{\mathrm{F}} = 10.0$ decay time. Even though the times scale differently from the
bounds in (\ref{timescale-1}), they do not actually violate the bounds because of the extreme prefactors that differ
by a factor of $10^{19}$.

The typical values for $T$ differ greatly between dimensions one and two. More generally, the decay times when
$d \geq 2$ are much smaller than when $d = 1$. One dimension is special because it is only in this case that the
$\alpha < 1$ Barenblatt profile exists. Figure \ref{decay-times-2d} shows decay times that are only nonzero because of
the bounded domain - a phenomenon known as instantaneous extinction. This states that $T \rightarrow 0$ as
$E_{\mathrm{min}} \rightarrow 0$, or in other words, the process we are simulating is not well defined in infinite
volume.

\section{Room for improvement}
The dynamics of our model have shown many similarities to plasma balls but there are three main problems. To start,
our solutions show that energy diffuses to infinity rather slowly once it escapes the central region. The phenomenon
of decay via hadron ejection is not realized. A more serious problem is the instantaneous extinction we have seen for
a lattice with two or more dimensions. Finally, we have seen that some punishing prefactors allow the admissible decay
times to span several orders of magnitude.

Though one can imagine several \textit{ad hoc} changes to our model that might eliminate these problems, an extension
that refers to conserved quantities other than energy is physically well motivated. The remainder of this thesis
will focus on developing an explicit PDE description for the momentum in a Poincar\'e invariant theory.

\subsection{A common approximation}
There is already a widely applicable PDE framework called hydrodynamics that puts energy and momentum on equal footing.
Hydrodynamics is concerned with systems that are close to equilibrium. Because of this, a system may deviate from a
translation invariant state only slightly leading to slowly varrying charges $E(\textbf{x}, t)$,
$\textbf{P}(\textbf{x}, t)$ and $Q(\textbf{x}, t)$. These are associated with long distance variations in their sources
$T(\textbf{x}, t)$, $\textbf{v}(\textbf{x}, t)$ and $\mu(\textbf{x}, t)$ respectively.

A theory with a symmetric stress-energy tensor and a $U(1)$ current has $d + 1 + \frac{(d + 1)(d + 2)}{2}$ independent
components. It is therefore not possible to describe every system using the $d + 2$ hydrodynamic variables above.
Nevertheless, hydrodynamics provides a good description of field theories with high occupation numbers and has
demonstrated a particular aptitude for problems relating to plasma balls and black holes.

\subsubsection{Basic hydrodynamics}
The hydrodynamics equations are nothing but the local conservation laws
\begin{eqnarray}
\partial_{\mu} T^{\mu \nu} &=& 0 \nonumber \\
\partial_{\mu} J^{\mu} &=& 0 \; . \label{hydro-equations}
\end{eqnarray}
The conserved currents are built from $T$, $\mu$ and $u^{\mu}$ making this a system of $d + 2$ equations for $d + 2$
unknown functions (the velocity satisfies $u^2 = -1$). The rules for writing (\ref{hydro-equations}) in terms of the
hydrodynamic variables are called the \textit{constitutive relations}. There is no limit to how complicated the
constitutive relations might be so it is helpful to use the assumption that the functions are slowly varrying. This
allows one to consider different versions of hydrodynamics based on how many derivatives are kept. In zeroth order
hydrodynamics, the temperature, velocity and chemical potential are not differentiated in the expressions for
$T^{\mu\nu}$ and $J^{\mu}$. In first order hydrodynamics, they are differentiated at most once. These orders are often
referred to as ``ideal'' and ``dissipative'' hydrodynamics. It is important to note that this is different from the
linearization of hydrodynamics. Nonlinear equations (\textit{e.g.} the Navier-Stokes equations) can easily arise from
zeroth or first order hydrodynamics, so it is common to introduce a further order-by-order expansion that drops terms
with too many variables multiplied together.

The first step in developing constitutive relations is decomposing the currents into components that are transverse
and longitudinal with respect to $u^{\mu}$. This leads to
\begin{eqnarray}
T^{\mu\nu} &=& \mathcal{E}u^{\mu}u^{\nu} + \mathcal{P}\Delta^{\mu\nu} + \left ( q^{\mu}u^{\nu} + q^{\nu}u^{\mu} \right ) + t^{\mu\nu} \nonumber \\
J^{\mu} &=& \mathcal{N}u^{\mu} + j^{\mu} \label{hydro-decomposition}
\end{eqnarray}
where we have defined the projector $\Delta^{\mu\nu} = \eta^{\mu\nu} + u^{\mu}u^{\nu}$. Transverse quantities like
$j^{\mu}$, $q^{\mu}$ and $t^{\mu\nu}$ cannot be built out of $T$, $\mu$ and $u^{\mu}$ without derivatives. Therefore
ideal hydrodynamics takes the form:
\begin{eqnarray}
T^{\mu\nu} &=& \varepsilon u^{\mu}u^{\nu} + P \Delta^{\mu\nu} \nonumber \\
J^{\mu} &=& n u^{\mu} \; . \label{ideal-constituents}
\end{eqnarray}
The pressure $P$, while not a charge or a source, is usually given in terms of sources by an equation of state
$P(T, \mu)$. Going to the fluid's rest frame makes it clear that $\varepsilon$ is the energy density and $n$ is the
charge density.

Things become more complicated when we include one order of dissipation \cite{kovtun}. Looking at
(\ref{hydro-decomposition}), many quantities such as $u^{\mu}$, $q^{\mu}$ and $j^{\mu}$ depend on position for a
general non-equilibrium configuration. There is nothing that prevents us from changing $u^{\mu}$ to a different
function of spacetime as long as $q^{\mu}$ and $j^{\mu}$ change as well to keep $T^{\mu\nu}$ and $J^{\mu}$ invariant.
This redundancy, similar to a gauge freedom, is called \textit{frame invariance} in hydrodynamics \cite{ambrosetti}.
Consider a shift $u^{\prime \mu}(\textbf{x}, t) = u^{\mu}(\textbf{x}, t) + \delta u^{\mu}(\textbf{x}, t)$. To preserve
the normalization, $\delta u^{\mu}$ is transverse to $u^{\mu}$. Using the inverse of (\ref{hydro-decomposition}),
corresponding changes in the coefficients can be calculated to first order.
\begin{eqnarray}
\mathcal{E} &=& u_{\mu} u_{\nu} T^{\mu\nu} \nonumber \\
\delta \mathcal{E} &=& 2 u_{\mu} \delta u_{\nu} T^{\mu\nu} = 2 u_{\mu} \delta u_{\nu} t^{\mu\nu} \approx 0 \nonumber \\
\mathcal{P} &=& \frac{1}{d} \Delta_{\mu\nu} T^{\mu\nu} \nonumber \\
\delta \mathcal{P} &=& \frac{2}{d} u_{\mu} \delta u_{\nu} T^{\mu\nu} \approx 0 \nonumber \\
\mathcal{N} &=& -u_{\mu} J^{\mu} \nonumber \\
\delta \mathcal{N} &=& -\delta u_{\mu} J^{\mu} = -\delta u_{\mu} j^{\mu} \approx 0 \; . \nonumber
\end{eqnarray}
Note that purely dissipative quantities like $j^{\mu}$ and $t^{\mu\nu}$ become second order when multiplied by
$\delta u_{\mu}$. These are all zero which explains why we could write down (\ref{ideal-constituents}) without worrying
about frame invariance. The same analysis would give $\delta t^{\mu\nu} = 0$ as well. Conversely,
\begin{eqnarray}
q_{\mu} &=& -\Delta_{\mu \alpha} u_{\beta} T^{\alpha\beta} \nonumber \\
\delta q_{\mu} &=& -\delta u^{\beta} T_{\mu\beta} - 2 u_{\mu} \delta u_{\alpha} u_{\beta} T^{\alpha\beta} - \delta u_{\mu} u_{\alpha} u_{\beta} T^{\alpha\beta} \nonumber \\
&=& -\mathcal{P} \delta u_{\mu} - q^{\beta} u^{\mu} \delta u_{\beta} - \delta u_{\beta} t^{\mu\beta} + 2 u_{\mu} \delta u_{\alpha} q^{\alpha} - \mathcal{E} \delta u_{\mu} \nonumber \\
&\approx& -(\mathcal{E} + \mathcal{P}) \delta u_{\mu} \nonumber \\
j_{\mu} &=& \Delta_{\mu\nu} J^{\nu} \nonumber \\
\delta j_{\mu} &=& u_{\mu} \delta u_{\nu} J^{\nu} + u_{\nu} \delta u_{\mu} J^{\nu} \nonumber \\
&=& u_{\mu} \delta u_{\nu} j^{\nu} - \mathcal{N} \delta u_{\mu} \nonumber \\
&\approx& -\mathcal{N} \delta u_{\mu} \; . \nonumber
\end{eqnarray}
We see that a suitable definition of local velocity reduces the number of terms in (\ref{hydro-decomposition}). The one
that makes $j^{\mu} = 0$ is called the Eckart frame while the one that makes $q^{\mu} = 0$ is called the Landau frame.
Naturally, $T^{\prime}(\textbf{x}, t) = T(\textbf{x}, t) + \delta T(\textbf{x}, t)$ and $\mu^{\prime}(\textbf{x}, t) = \mu(\textbf{x}, t) + \delta \mu(\textbf{x}, t)$ are
allowed shifts of the other hydro variables. Under these redefinitions, the coefficients from
(\ref{ideal-constituents}) become
\begin{eqnarray}
\varepsilon(T^{\prime}, \mu^{\prime}) &=& \varepsilon(T, \mu) + \frac{\partial \varepsilon}{\partial T} \delta T + \frac{\partial \varepsilon}{\partial \mu} \delta \mu \nonumber \\
P(T^{\prime}, \mu^{\prime}) &=& P(T, \mu) + \frac{\partial P}{\partial T} \delta T + \frac{\partial P}{\partial \mu} \delta \mu \nonumber \\
n(T^{\prime}, \mu^{\prime}) &=& n(T, \mu) + \frac{\partial n}{\partial T} \delta T + \frac{\partial n}{\partial \mu} \delta \mu \; . \nonumber
\end{eqnarray}
This means we can define temperature and chemical potential such that $\mathcal{E} = \varepsilon$ and
$\mathcal{N} = n$ \cite{kovtun}.

Choosing the Landau frame, constitutive relations are expressions for the scalar $\mathcal{P}$, the transverse vector
$j^{\mu}$ and the transverse traceless symmetric tensor $t^{\mu\nu}$. These may contain any combination of $T$,
$\mu$ and $u^{\mu}$ with one derivative. However, the equations of zeroth order hydrodynamics (\ref{ideal-constituents})
give relations between many of these terms up to higher order corrections. The Landau frame expressions that follow
from this are:
\begin{eqnarray}
T^{\mu\nu} &=& \varepsilon u^{\mu} u^{\nu} + \left ( P - \zeta \partial_{\lambda}u^{\lambda} \right ) \Delta^{\mu\nu} - \eta \Delta^{\mu\alpha} \Delta^{\nu\beta} \left ( \partial_{\alpha} u_{\beta} + \partial_{\beta} u_{\alpha} - \frac{2}{d} \eta_{\alpha\beta} \partial_{\lambda} u^{\lambda} \right ) \nonumber \\
J^{\mu} &=& nu^{\mu} - \sigma T \Delta^{\mu\nu} \partial_{\nu} \left ( \frac{\mu}{T} \right ) + \chi_T \Delta^{\mu\nu} \partial_{\nu} T \; . \label{dissipative-constituents}
\end{eqnarray}
The frame invariant parameters $\zeta$, $\eta$ and $\sigma$ are functions of $T$ and $\mu$ that must be determined from
experiment or the microscopic theory. They are called the bulk viscosity, shear viscosity and conductivity respectively.
It turns out that $\chi_T$ must be zero for time-reversal invariance to be satisfied. The linearization of
(\ref{dissipative-constituents}) will be important in what follows.

We linearize around the equilibrium solution which has constant hydro variables and zero velocity. This allows us to
write $u^{\mu} = (1, v^i)$. Explicitly, what we seek are evolution equations for
\begin{eqnarray}
J^0(\textbf{x}, t) &=& n(T(\textbf{x}, t), \mu(\textbf{x}, t)) \nonumber \\
&\approx& n(T_0, \mu_0) + \left ( \frac{\partial n}{\partial T} \right )_0 \delta T(\textbf{x}, t) + \left ( \frac{\partial n}{\partial \mu} \right )_0 \delta \mu(\textbf{x}, t) \nonumber \\
&\equiv& n_0 + \tilde{n}(\textbf{x}, t) \nonumber \\
T^{00}(\textbf{x}, t) &=& \varepsilon(T(\textbf{x}, t), \mu(\textbf{x}, t)) \nonumber \\
&\approx& \varepsilon(T_0, \mu_0) + \left ( \frac{\partial \varepsilon}{\partial T} \right )_0 \delta T(\textbf{x}, t) + \left ( \frac{\partial \varepsilon}{\partial \mu} \right )_0 \delta \mu(\textbf{x}, t) \nonumber \\
&\equiv& \varepsilon_0 + \tilde{\varepsilon}(\textbf{x}, t) \nonumber \\
T^{0i}(\textbf{x}, t) &=& (\varepsilon(T(\textbf{x}, t), \mu(\textbf{x}, t)) + P(T(\textbf{x}, t), \mu(\textbf{x}, t))) v^i(\textbf{x}, t) \nonumber \\
&\approx& (\varepsilon_0 + P_0) v^i(\textbf{x}, t) \nonumber \\
&\equiv& \tilde{P}^i(\textbf{x}, t) \; . \nonumber
\end{eqnarray}
It is important to remember that most coefficients in (\ref{ideal-constituents}) and (\ref{dissipative-constituents})
are functions of $T$ and $\mu$. Their derivatives evaluated at $T_0$ and $\mu_0$ should therefore appear in the
linearization. These derivatives, known as susceptibilities, do not have a functional dependence on $T$ and $\mu$
anymore; they are simply numbers. For this reason, we should not be worried if the linearized hydrodynamics equations
appear to have more coefficients than the six we have seen so far. One equation of hydrodynamics is always the
continuity equation $\frac{\partial \varepsilon}{\partial t} = \partial_i T^{0i}$. It is easy to see that this becomes
\begin{equation}
\frac{\partial \tilde{\varepsilon}}{\partial t} = \partial_i \tilde{P}^i \label{continuity-equation}
\end{equation}
because it is already linear. The equation involving stresses linearizes to
\begin{eqnarray}
\frac{\partial \tilde{P}^j}{\partial t} &=& \partial_i \left [ \left ( P - \zeta \partial_k v^k \right ) \delta^{ij} - \eta \left ( \partial^i v^j + \partial^j v^i  - \frac{2}{d} \delta^{ij} \partial_k v^k \right ) \right ] \nonumber \\
&=& \partial_i \left [ \left ( P - \frac{\zeta}{\varepsilon_0 + P_0} \partial_k \tilde{P}^k \right ) \delta^{ij} - \frac{\eta}{\varepsilon_0 + P_0} \left ( \partial^i \tilde{P}^j + \partial^j \tilde{P}^i  - \frac{2}{d} \delta^{ij} \partial_k \tilde{P}^k \right ) \right ] \nonumber \\
&=& \partial_i \left [ \delta^{ij} \left ( \frac{\partial P}{\partial \varepsilon} \right )_0  \tilde{\varepsilon} + \delta^{ij} \left ( \frac{\partial P}{\partial n} \right )_0  \tilde{n} + \delta^{ij} \frac{\frac{2}{d} \eta_0 - \zeta_0}{\varepsilon_0 + P_0} \partial_k \tilde{P}^k - \frac{\eta_0}{\varepsilon_0 + P_0} \left ( \partial^i \tilde{P}^j + \partial^j \tilde{P}^i \right ) \right ] \nonumber \\
&=& \left ( \frac{\partial P}{\partial \varepsilon} \right )_0 \partial^j \tilde{\varepsilon} + \left ( \frac{\partial P}{\partial n} \right )_0 \partial^j \tilde{n} + \frac{\frac{2 - d}{d} \eta_0 - \zeta_0}{\varepsilon_0 + P_0} \partial^j \partial_k \tilde{P}^k - \frac{\eta_0}{\varepsilon_0 + P_0} \partial_k \partial^k \tilde{P}^j \; . \label{stress-equation}
\end{eqnarray}
Finally, to handle the $U(1)$ current
\begin{eqnarray}
\frac{\partial \tilde{n}}{\partial t} &=& \partial_i \left [ nv^i - \sigma T \delta^{ij} \partial_j \left ( \frac{\mu}{T} \right ) \right ] \nonumber \\
&=& \partial_i \left [ \frac{n_0}{\varepsilon_0 + P_0} \tilde{P}^i - \sigma T \delta^{ij} \partial_j \left ( \frac{\mu}{T} \right ) \right ] \nonumber \\
&=& \partial_i \left [ \frac{n_0}{\varepsilon_0 + P_0} \tilde{P}^i - \sigma_0 T \delta^{ij} \left ( \frac{\partial \mu / T}{\partial \varepsilon} \right )_0 \partial_j \tilde{\varepsilon} - \sigma_0 T \delta^{ij} \left ( \frac{\partial \mu / T}{\partial n} \right )_0 \partial_j \tilde{n} \right ] \nonumber \\
&=& \frac{n_0}{\varepsilon_0 + P_0} \partial_i \tilde{P}^i - \sigma_0 T \left ( \frac{\partial \mu / T}{\partial \varepsilon} \right )_0 \partial_i \partial^i \tilde{\varepsilon} - \sigma_0 T \left ( \frac{\partial \mu / T}{\partial n} \right )_0 \partial_i \partial^i \tilde{n} \; . \label{current-equation}
\end{eqnarray}
It is interesting to note that the continuity equation (\ref{continuity-equation}) effectively splits into two
equations in non-relativistic hydrodynamics. In this limit, one must insert the speed of light back into
(\ref{dissipative-constituents}) and collect powers of $c$. Taking $c \rightarrow \infty$, inverse powers
vanish leaving only $O(c)$ and $O(1)$ terms in the hydro equations. Since these orders can be considered separate,
they yield conservation of mass and conservation of kinetic energy \cite{kaminski}.

Matching the predictions of hydrodynamics with those of the interaction picture leads to the Kubo formulas. A simple
Kubo formula can be derived for the diffusion constant in the heat equation
$\frac{\partial n}{\partial t} - D \frac{\partial^2 n}{\partial x^2} = 0$. Performing a two-sided Fourier transform
in space and a one-sided Fourier transform in time, the equation can be solved as
\begin{equation}
n(k, z) = \frac{n_0(k)}{Dk^2 - iz} = \frac{\chi \mu_0(k)}{Dk^2 - iz} \; . \label{kubo-step1}
\end{equation}
We assume that the initial buildup of charge is due to a local chemical potential that existed before $t = 0$. If this
chemical potential increased adiabatically starting at $t = -\infty$, the corresponding perturbation to the Hamiltonian
is given by
\begin{eqnarray}
H &\mapsto& H - \int_{-\infty}^{\infty} \mu(x, t) n(x, t) \textup{d}x \nonumber \\
&=& H - \int_{-\infty}^{\infty} e^{\epsilon t} \mu_0(x) \theta(-t) n(x, t) \textup{d}x \nonumber \\
&\equiv& H - \delta H \; . \nonumber
\end{eqnarray}
Heisenberg's equation now gives
\begin{eqnarray}
\frac{\textup{d}}{\textup{d}t} \left < n(x, t) \right > &=& -i \left < [ n(x, t), \delta H(x, t) ] \right > \nonumber \\
\left < n(x, t) \right > &=& -i \int_{-\infty}^t \int_{-\infty}^{\infty} e^{\epsilon t^{\prime}} \mu_0(x) \theta(-t^{\prime}) \left < [n(x, t), n(x^{\prime}, t^{\prime})] \right > dx^{\prime} dt^{\prime} \nonumber \\
&=& -i \int_{-\infty}^0 \int_{-\infty}^{\infty} e^{\epsilon t^{\prime}} \mu_0(x) \theta(t-t^{\prime}) \left < [n(x, t), n(x^{\prime}, t^{\prime})] \right > dx^{\prime} dt^{\prime} \nonumber \\
&=& -\int_{-\infty}^o \int_{-\infty}^{\infty} e^{\epsilon t^{\prime}} \mu_0(x) G_{nn}^{\mathrm{R}} (t - t^{\prime}, x - x^{\prime}) \textup{d}x^{\prime} \textup{d}t^{\prime} \; . \nonumber
\end{eqnarray}
Here $G_{nn}^{\mathrm{R}} (t - t^{\prime}, x - x^{\prime})$ is the retarded Green's function of $n$ with itself.
The one fact we need is that the zero mode of $G_{nn}^{\mathrm{R}}$ is the susceptibility $-\chi$. It is a matter of
complex analysis to show that
\begin{eqnarray}
\left < n(k, z) \right > &=& -\mu_0(k) \int_{-\infty}^{\infty} G_{nn}^{\mathrm{R}} (\omega, k) \frac{1}{(\epsilon + i\omega)(\epsilon + i (\omega - z))} \frac{\textup{d}\omega}{2\pi} \nonumber \\
&=& \frac{\mu_0(k)}{iz} \left [ G_{nn}^{\mathrm{R}} (0, k) - G_{nn}^{\mathrm{R}} (z, k) \right ] \nonumber \\
&=& \frac{i \mu_0(k)}{z} \left [ \chi + G_{nn}^{\mathrm{R}} (z, k) \right ] \; . \label{kubo-step2}
\end{eqnarray}
Comparing (\ref{kubo-step1}) to (\ref{kubo-step2}), we see that a quantum theory well described by linear diffusion
should be one whose momentum space retarded Green's function is $\frac{\chi D k^2}{iz - Dk^2}$. A calculation that is
longer but equally straightforward has the linear hydrodynamics equations in place of the heat equation.
This is done in \cite{kovtun}. Additionally a field theory is developed that allows one to compute corrections to the
Green's functions arising from small nonlinearities in the hydrodynamics equations.

\subsubsection{The fluid / gravity correspondence}
Correlation functions in certain CFTs are related to supergravity amplitudes in certain AdS backgrounds.
The link between Green's functions and linearized hydrodynamics opens up the possibility of describing fluid
phenomena using gravity. This was demonstrated in \cite{policastro} which found that the shear viscosity in Super
Yang-Mills theory is given by $\eta = \frac{\pi}{8} N^2 T^3$. By now it is known that in arbitrary dimension, the
holographic shear viscosity differs from the entropy density by a factor of $\frac{1}{4\pi}$ \cite{rangamani}. A
major result of 2007 is that all transport coefficients in the fluid stress-energy tensor can be algorithmically found
from the gravity side as well \cite{hubeny}. This expansion, which is valid to all orders, is called the
fluid / gravity correspondence.

In the examples discussed previously, we related the energy of a state on the boundary of AdS to the mass of a black
hole in the bulk. More generally, there is a systematic way to find the boundary stress-energy tensor corresponding to
a given bulk metric \cite{balasubramanian, fg}. Because every metric yields a conserved stress-energy tensor, it is not
surprising that the fluid / gravity correspondence exists. However, the limit in which temperature and velocity suffice
to describe $T^{\mu\nu}$ could have corresponded to an intractable $g_{\mu\nu}$. Showing otherwise, \cite{hubeny}
started with the metric
\begin{equation}
ds^2 = L^2 \left [ -\rho^2 \left ( 1 - \frac{1}{z_0^d \rho^d} \right ) dv^2 + 2 dv d\rho + \rho^2 dx_i dx^i \right ] \; . \label{ed-fink-poincare}
\end{equation}
This is nothing but the black brane (\ref{schwarzschild-poincare}) written in ingoing Eddington-Finkelstein
co-ordinates with $\rho = \frac{1}{z}$. Boosting the brane to a particular velocity, this becomes
\begin{equation}
ds^2 = L^2 \left [ -\rho^2 \left ( 1 - \frac{1}{z_0^d \rho^d} \right ) u_{\mu} u_{\nu} dx^{\mu} dx^{\nu} - 2 u_{\mu} dx^{\mu} d\rho + \rho^2 \Delta_{\mu\nu} dx^{\mu} dx^{\nu} \right ] \; . \label{boosted-ed-fink-poincare}
\end{equation}
Greek indices have been used for all co-ordinates except $\rho$ which is not in the field theory. The key step is
to promote $z_0$ (which determines the temperature) and $u^{\mu}$ to slowly varrying functions of spacetime. One
could repeat the calculation for a charged black hole if she wanted chemical potential to vary as well. Clearly, these
functions are heavily constrained for (\ref{boosted-ed-fink-poincare}) to still solve Einstein's equations. The
constraints turn out to be those of hydrodynamics with infinitely many orders of dissipation. Put another way, we may
say that temperature and velocity functions satisfying SYM hydrodynamics at a given order, cause
(\ref{boosted-ed-fink-poincare}) to only violate Einstein's equations at a higher order.

Instead of making $z_0$ and $u_{\mu}$ functions of $x$, \cite{hubeny} makes them functions of $\epsilon x$ to keep
track of derivatives before setting $\epsilon = 1$. Using $g^{(0)}$ to denote the metric of
(\ref{boosted-ed-fink-poincare}), $g^{(0)}$ with unconstrained $z_0$ and $u_{\mu}$ has a stress-energy tensor like
(\ref{ideal-constituents}) and violates the Einstein equations
at order $\epsilon$. To correct this and make the violation order $\epsilon^2$, one must do two things. The first
is to write $z_0 = z_0^{(0)} + \epsilon z_0^{(1)}$, $u_{\mu} = u_{\mu}^{(0)} + \epsilon u_{\mu}^{(1)}$ so that $g^{(0)}$
picks up terms of order $\epsilon$. When these corrections are explicitly calculated, $z_0^{(1)}$ \textit{e.g.} will
be some multiple of a derivative of $z_0^{(0)}$. The second is to add a new piece of the metric
$g = g^{(0)} + \epsilon g^{(1)}$. If we plug this new metric into Einstein's equation, a vanishing $\epsilon$ term
will dictate that $z_0^{(0)}$ and $u_{\mu}^{(0)}$ satisfy ideal hydrodynamics (\ref{ideal-constituents}). It will also
allow us to compute the corrected metric and show that it corresponds to the stress-energy tensor
(\ref{dissipative-constituents}). The result of this is \cite{wiseman}
\begin{eqnarray}
ds^2 &=& L^2 \left [ -\rho^2 \left ( 1 - \frac{1}{z_0^{(0) d} \rho^d} \right ) u^{(0)}_{\mu} u^{(0)}_{\nu} dx^{\mu} dx^{\nu} - 2 u^{(0)}_{\mu} dx^{\mu} d\rho + \rho^2 \Delta^{(0)}_{\mu\nu} dx^{\mu} dx^{\nu} \right. \nonumber \\
&& + 2 \rho^2 z_0^{(0)} F \left ( z_0^{(0)} \rho \right ) \left ( \Delta^{(0)}_{\mu \alpha} \Delta^{(0)}_{\mu \beta} \left ( \partial^{\alpha} u^{(0) \beta} + \partial^{\beta} u^{(0) \alpha} \right ) - \frac{1}{d - 1} \Delta^{(0)}_{\mu \nu} \partial_{\lambda} u^{(0) \lambda} \right ) dx^{\mu} dx^{\nu} \nonumber \\
&& \left. + \frac{2}{d - 1} \rho u^{(0)}_{\mu} u^{(0)}_{\nu} \partial_{\lambda} u^{(0) \lambda} dx^{\mu} dx^{\nu} - \rho u^{(0) \lambda} \partial_{\lambda} \left ( u^{(0)}_{\mu} u^{(0)}_{\nu} \right ) dx^{\mu} dx^{\nu} \right ] \label{fgc-metric}
\end{eqnarray}
where
\begin{equation}
F(x) = \int_x^{\infty} \frac{y^{d - 1} - 1}{y(y^d - 1)} \textup{d}y \; . \nonumber
\end{equation}
Of course (\ref{fgc-metric}) is still in the form of (\ref{boosted-ed-fink-poincare}). The choice to not collect
all the $dx^{\mu} dx^{\nu}$ terms makes it easy to see the non-derivative part whose components satisfy ideal
hydrodynamics. Continuing the procedure with $z_0^{(2)}$ and $u_{\mu}^{(2)}$ and $g^{(2)}$,
these same components are constrained to satisfy dissipative hydrodynamics. Rather
than (\ref{fgc-metric}) associated with the dissipative stress-energy tensor (\ref{dissipative-constituents}),
Einstein's equations then yield a longer metric associated with the stress-energy tensor for a conformal fluid at two
orders of dissipation. This logic continues inductively showing many non-trivial constitutive relations.
Beyond just conformal fluids, similar constitutive relations may be derived for fluids in spaces with compactified
directions. These relations have been used to study plasma balls and other objects that appear in the Witten
model \cite{cardoso, rocha, lahiri}. These studies assume that a black hole is already present because the process of
black hole formation is outside the regime of hydrodynamics \cite{wiseman}.

Simply replacing the stochastic model of this thesis with hydrodynamics is not what we plan to do. For one thing, it is
hoped that the initial conditions of Figure \ref{no-exaggeration} correspond to black holes that have not formed yet.
For hydrodynamics to apply, we must wait for the inhomogeneities to be smoothed out. Therefore, we should only expect
agreement between hydrodynamics and a suitable extension of our model if we linearize both of them. Another reason to
extend our model comes from the importance of the density of states. When we only had energy being transferred, this
function allowed us to read off diffusing or clustering behaviour. It is plausible that something similar can be done
when we include momentum.

\subsection{Restricting the density of states}
If momentum is viewed as a quantity exchanged between nearest neighbour sites, the rate for each transition naturally
depends on a \textit{momentum restricted density of states} $\rho(E, \textbf{P})$. While expressions for
$\rho(E, \textbf{P})$ appear in some models of electron structure \cite{weick}, the restricted density of states for
a field theory is a barely studied quantity. Here, we attempt to rectify this by deriving some properties of the
restricted density of states for simple field theories. A strongly coupled SYM expression analogous to
(\ref{sym-entropy}) is most likely beyond our reach.

\subsubsection{Some conformal field theory}
Some of the most interesting techniques for counting states are applicable to two-dimensional CFTs. Consider a
Minkowskian theory on $\mathbb{S}^1 \times \mathbb{R}$ where the circle has radius 1 by convention. At thermal
equilibrium, this can be thought of as a Euclidean theory on $\mathbb{S}^1 \times \mathbb{S}^1$ since temperature and
time are related by a Wick rotation. The radius of the second circle must be $\beta$. If we rescale each $\mathbb{S}^1$
by the same factor, one such rescaling gives us another Euclidean theory on a circle of radius 1 as shown in
Figure \ref{tori}.
\begin{figure}[h]
\includegraphics[scale=0.6]{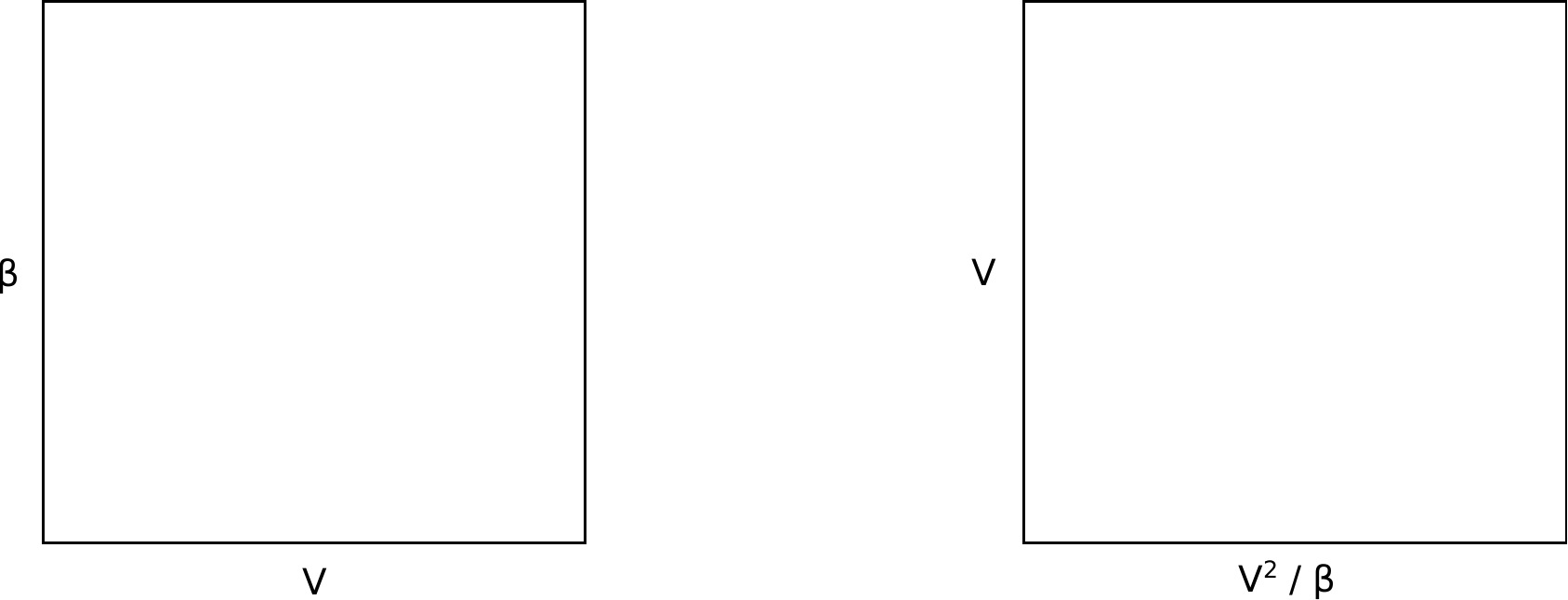}
\caption{These two theories live on a circle of volume $V$ but their temperatures are different. They are related
through multiplication by $\frac{V}{\beta}$ which is a conformal transformation.}
\label{tori}
\end{figure}
Conformal invariance then tells us that inverse temperature $\beta$ and inverse temperature $\frac{V^2}{\beta}$
are completely equivalent \cite{tong2}. In terms of the modular parameter we defined for (\ref{gen-partfunc}),
$\tau = i\frac{\beta}{2\pi}$ becomes $i\frac{4\pi^2 / \beta}{2\pi} = -\frac{1}{\tau}$. For a $\tau$ that is not
imaginary, the factor we trace over to compute the generalized partition function is
$e^{-2\pi \tau_2 H + 2\pi i \tau_1 P}$. Clearly, this does not change if we increase the real part of $\tau$ by 1.
Symmetries of the generalized partition function can therefore be written as
\begin{eqnarray}
\tau \mapsto \tau + 1 \nonumber \\
\tau \mapsto -\frac{1}{\tau} \; . \label{modular-transformations}
\end{eqnarray}
These two transformations generate the \textit{modular group} which is the discrete but infinite group
$SL(2, \mathbb{Z})$. The utility of modular invariance for studying thermodynamics was first noticed by Cardy in 1986
\cite{cardy}. Our previous expressions for the partition function of a free theory do not posess modular
invariance. This is because they are only asymptotic expressions for the high temperature limit of $Z$. Low temperature
information is lost when we convert sums to integrals. It is perhaps for this reason that most sources prefer to write
infinite products or special functions when simple expressions like (\ref{partfunc}) and (\ref{gen-partfunc}) are all
that are needed.

For an infinite collection of oscillators with frequencies
$1, 2, 3, \dots$, the energy of an arbitrary state $\left | N_1, N_2, N_3, \dots \right >$ is
$N_1 + 2N_2 + 3N_3 + \dots$. The degeneracy of energy $n$ is simply the number of ways to make $n$ by picking a certain
number of ones, a certain number of twos, a certain number of threes and so on. This well-studied object $p(n)$ is
usually called ``the partition function'' in number theory. We will call this microcanonical quantity
``the partition sum'' to distinguish it from the canonical partition function we have been using so far. The first
asymptotic $p(n) \sim \frac{1}{4 \sqrt{3} n} e^{\pi \sqrt{\frac{2}{3} n}}$ due to Hardy and Ramanujan was based on
transformation properties of the generating function
\begin{eqnarray}
F(q) &=& \sum_{n = 1}^{\infty} p(n) q^n \nonumber \\
&=& \left [ 1 + q + q^2 + \dots \right ] \left [ 1 + q^2 + q^4 + \dots \right ] \dots \nonumber \\
&=& \prod_{n = 1}^{\infty} \left ( 1 - q^n \right )^{-1} \; . \label{genfunc}
\end{eqnarray}
If we take $q = e^{2\pi i \tau}$ (and $\tau = i\frac{\beta}{2\pi}$ as before), (\ref{genfunc}) becomes a sum of
Boltzmann factors. It is not quite the partition function of our oscillator family because we have not yet included
the Casimir effect. Accounting for zero point energies, each term in the product (\ref{genfunc}) picks up a factor of
$q^{\frac{n}{2}}$ giving us
\begin{eqnarray}
Z_1(\tau) &=& \prod_{n = 1}^{\infty} q^{\frac{n}{2}} \left ( 1 - q^n \right )^{-1} \nonumber \\
&=& q^{-\frac{1}{24}} \prod_{n = 1}^{\infty} \left ( 1 - q^n \right )^{-1} \nonumber \\
&=& \eta^{-1}(\tau) \; . \label{partfunc-half}
\end{eqnarray}
We have analytically continued and written this partition function in terms of the \textit{Dedekind eta function}
\begin{equation}
\eta(\tau) = e^{\frac{\pi i \tau}{12}} \prod_{n = 1}^{\infty} \left ( 1 - e^{2\pi i n \tau} \right ) \; . \label{eta-function}
\end{equation}
Referring to the eta function, Rademacher gave a convergent series for $p(n)$
\begin{eqnarray}
p(n) &=& \frac{1}{\sqrt{2} \pi} \sum_{k = 1}^{\infty} A_k(n) \sqrt{k} \frac{\textup{d}}{\textup{d}n} \left ( \frac{\sinh \left ( \frac{\pi}{k} \sqrt{\frac{2}{3} \left ( n - \frac{1}{24} \right )} \right )}{\sqrt{n - \frac{1}{24}}} \right ) \label{rademacher} \\
A_k(n) &=& \sum_{m \in \left ( \mathbb{Z} / k\mathbb{Z} \right )^*} e^{\pi i \left ( s(m, k) - \frac{2nm}{k} \right )} \nonumber
\end{eqnarray}
where the coefficients are written in terms of the Dedekind sum
\begin{equation}
s(m, k) = \sum_{l = 1}^{k - 1} \frac{l}{k} \left ( \frac{ml}{k} - \bigg \lfloor \frac{ml}{k} \bigg \rfloor - \frac{1}{2} \right ) \; . \nonumber
\end{equation}
Just as we treated $\rho(E)$ as the inverse Laplace transform of $Z(\beta)$, the partition sum is an integral transform
of its generating function.
\begin{equation}
p(n) = \frac{1}{2\pi i} \int_{\gamma} \frac{F(q)}{q^{n + 1}} \textup{d}q \nonumber
\end{equation}
Instead of a vertical line, the contour $\gamma$ that leads to (\ref{rademacher}) is a beautiful shape built up
from fractals in the complex plane. A very readable thesis on this derivation is \cite{hsu}.

The main fact about $F(q)$ used in the proof is that it is an imaginary exponential times a function which plays
nicely with the modular group --- as it would be even if our theory were not free. Under the modular
generators (\ref{modular-transformations}), the Dedekind eta function transforms in the following way:
\begin{eqnarray}
\eta(\tau + 1) &=& e^{\frac{\pi i}{12}} \eta(\tau) \nonumber \\
\eta \left ( -\frac{1}{\tau} \right ) &=& \sqrt{-i \tau} \eta(\tau) \; . \label{modular-eta}
\end{eqnarray}
The first of these equalities is trivial. The second is not but there is a clever proof containing no more
than a page of algebra \cite{elstrodt}. A function with transformation properties similar to (\ref{modular-eta}) is
called a \textit{modular form}. More precisely, a holomorphic function $f$ in the upper half plane is a modular form of
weight $w$ if
\begin{equation}
\left | f \left ( -\frac{1}{\tau} \right ) \right | = \tau^w \left | f(\tau) \right | \; . \nonumber
\end{equation}
According to this definition, $\eta$ is a modular form of weight $\frac{1}{2}$. The Rademacher formula can be
rederived for (multiples of) other modular forms. The result is \cite{dijkgraaf, birmingham}
\begin{eqnarray}
p_w(n) &=& 2\pi \sum_{m - \frac{c}{24} < 0} \left ( \frac{n - \frac{c}{24}}{\left | m - \frac{c}{24} \right |} \right )^{\frac{w - 1}{2}} p_w(m) \nonumber \\
&& \sum_{k = 1}^{\infty} \frac{1}{k} Kl \left ( n - \frac{c}{24}, m - \frac{c}{24}; k \right ) I_{1 - w} \left ( \frac{4\pi}{k} \sqrt{\left | m - \frac{c}{24} \right | \left ( n - \frac{c}{24} \right )} \right ) \label{rademacher-w}
\end{eqnarray}
where the Kloosterman sum is defined by
\begin{equation}
Kl(n, m; k) = \sum_{d \in \left ( \mathbb{Z} / k\mathbb{Z} \right )^*} e^{\frac{2\pi i}{k} \left ( dn + d^{-1} m \right )} \; . \nonumber
\end{equation}
For our purposes, we will choose a weight $w$ and use $k = m = 0$. Uniqueness of the vacuum gives $p_w(0) = 1$ and the
central charge is 2 when we have left movers and right movers. This formula can be understood as the origin of the
Bessel function in (\ref{free-dos-bessel}).

In order to use the partition sum and its asymptotics to study $\rho(E, P)$, there are some changes that must be made
to our theory of positively indexed oscillators.
\begin{figure}[h]
\includegraphics[scale=0.6]{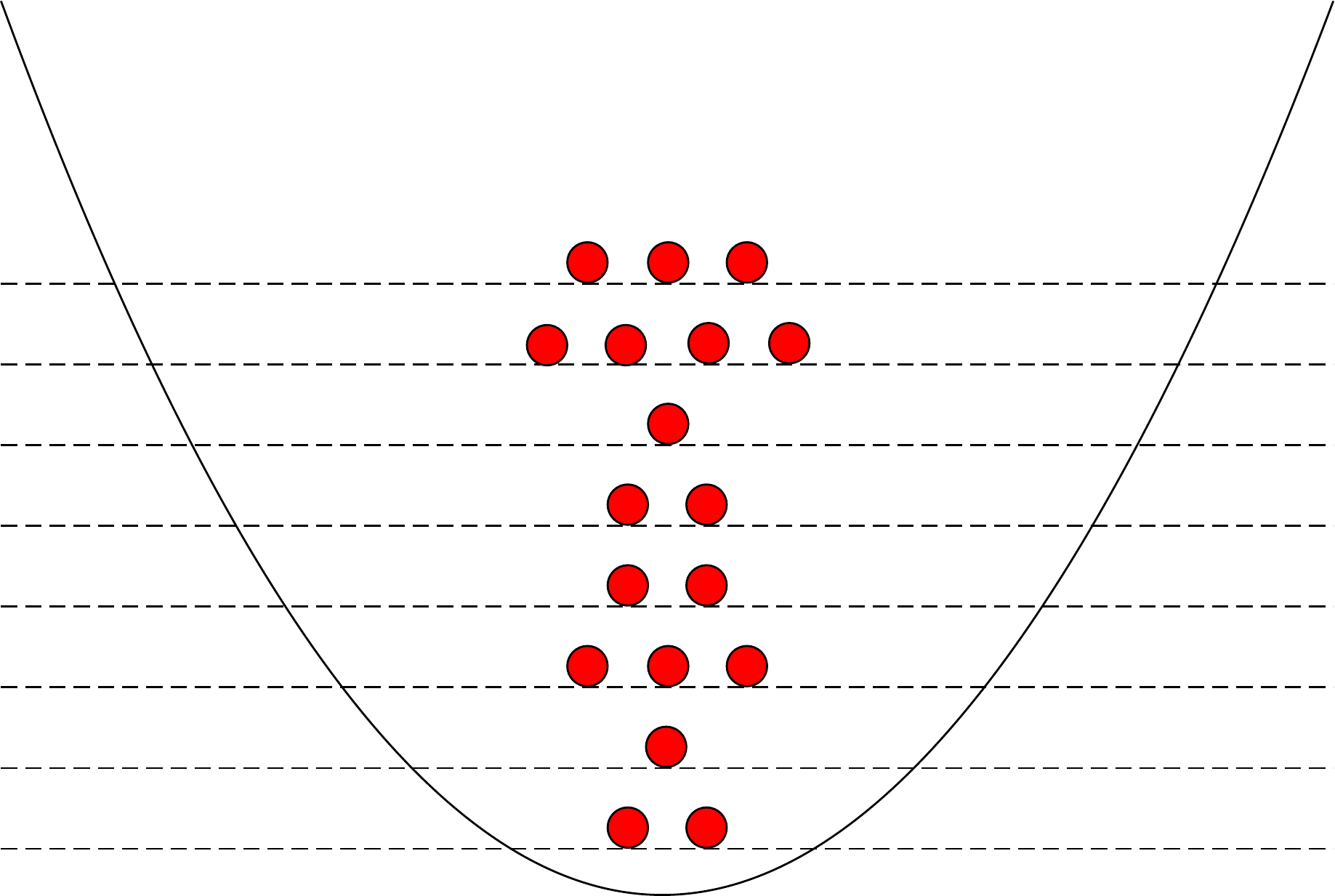}
\caption{The theory we have discussed so far resembles a harmonic potential where arbitrarily many non-interacting
bosons can be piled into each level.}
\label{well}
\end{figure}
To start, we should introduce negatively indexed oscillators to make it more like a field theory in momentum space.
An excitation of any oscillator increases the energy. Whether it increases or decreases momentum depends on whether
the frequency is positive or negative. This means that in the contribution of the newly introduced oscillators,
$\tau_2$ has the same sign as in $Z_1(\tau)$ and $\tau_1$ has the opposite. Including a sum of
$e^{-2\pi i n \tau^*}$ Boltzmann factors along with our sum of $e^{2\pi i n \tau}$ Boltzmann factors suggests that
the partition function for the doubly infinite family of oscillators is $\left | \eta(\tau) \right |^{-2}$.
This is almost correct, but we need to remember that the Casimir contribution doubles instead of vanishes when we add
more oscillators. Fixing this, the appropriate partition function is given by
\begin{equation}
Z_2(\tau) = \eta^{-1}(\tau) \eta^{-1}(-\tau^*)^* \; . \label{partfunc-full}
\end{equation}
This is still not modular invariant because we have not yet included the oscillator with zero frequency. The zero mode
has a continuous spectrum which is quadratic in the free variable $\xi$ \cite{tong2}. This can be seen \textit{e.g.} in
the scalar field mode expansion that was discussed in the context of the string action (\ref{string-action}).
The most exact expression for the partition function of scalar field theory is therefore
\begin{eqnarray}
Z_3(\tau) &=& Z_2(\tau) \int_{-\infty}^{\infty} e^{-\pi \Im \tau \xi^2} \frac{\textup{d}\xi}{2\pi} \nonumber \\
&=& \frac{1}{2\pi \sqrt{\Im \tau}} \eta^{-1}(\tau) \eta^{-1}(-\tau^*)^* \; . \label{partfunc-final}
\end{eqnarray}
To summarize these results, the partition function $Z_1$ is a modular form of weight $-\frac{1}{2}$. By
(\ref{rademacher-w}), its density of states is asymptotic to the partition sum;
$p_{-\frac{1}{2}}(n) \sim \frac{1}{4 \sqrt{3} n} e^{\pi \sqrt{\frac{2}{3} n}}$. Doubling the number of oscillators,
we made the partition function $Z_2$ which is a modular form of weight $-1$. Its density of states has the
behaviour $p_{-1}(n) \sim \frac{3^{\frac{1}{4}}}{12 n^{\frac{5}{4}}} e^{2\pi \sqrt{\frac{1}{3} n}}$. The partition
function $Z_3$, since it describes a CFT, is a modular form of weight 0. If we take the density of states
$p_0(n) \sim \frac{3^{\frac{1}{4}}}{\sqrt{12} n^{\frac{3}{4}}} e^{2\pi \sqrt{\frac{1}{3} n}}$ and substitute
$n = \frac{VE}{2\pi}$ for the conformal weight, our expression agrees with (\ref{free-dos}) which gave $\rho(E)$ for a
free CFT from an entirely different setup.

\subsubsection{Convoluted functions}
Evidently, the Casimir term and the zero mode do not affect $\log p_w (n)$ to leading order. It is therefore valid to
use $F(q)^2$ instead of $Z_3(\tau)$ as the generating function. With this approximation, the convolution formula gives
a very clear picture of how the density of states is comprised.
\begin{eqnarray}
p_0(n) &\approx& \frac{1}{n!} \frac{\textup{d}^n}{\textup{d}q^n} F(q)^2 \nonumber \\
&=& \sum_{k = 0}^n p(k) p(n - k) \nonumber
\end{eqnarray}
The first term $p(0) p(n)$ has $n$ right movers and $0$ left movers and corresponds to energy $n$ and momentum $n$.
The last term $p(n) p(0)$ has $0$ right movers and $n$ left movers and corresponds to energy $n$ and momentum $-n$.
The middle term $p \left ( \frac{n}{2} \right )^2$ has a left mover for every right mover and corresponds to energy $n$
and momentum $0$. This tells us that in a $1 + 1$-dimensional possibly interacting CFT,
\begin{equation}
\rho(E, P) = \frac{1}{2} p \left ( \frac{V}{4\pi} (E - P) \right ) p \left ( \frac{V}{4\pi} (E + P) \right ) \label{rdos-2d}
\end{equation}
is a good approximation to the momentum restricted density of states. This only becomes difficult to evaluate when
$|P| \neq E$ are of the same order.

Essentially we derived (\ref{rdos-2d}) by starting with a function whose asymptotic behaviour is
$\log p(n) \sim \pi \sqrt{\frac{2}{3} n}$. The convolution of two of them yields a function that follows
$\log p_0(n) \sim \pi \sqrt{\frac{4}{3} n}$. In other words, convolution doubles the central charge. At least for free
theories, this phenomenon can be readily seen in higher dimensions. The crudest form of (\ref{free-dos}) is
\begin{equation}
\log \rho(E) \sim \left ( \frac{(d + 1)^{d + 1}}{d^d} AVE^d \right )^{\frac{1}{d + 1}} \; . \nonumber
\end{equation}
It will be helpful to write this density of states explicitly as $\rho(E; A)$ from now on. For a bosonic theory, we had
$A$ proportional to $s$ and for a fermionic theory, we had $A$ proportional to $s^*$. Writing the partition
function $\prod_{\textbf{p}} Z(\textbf{p})^s Z^*(\textbf{p})^{s^*}$, it becomes clear that $\rho(\cdot; A)$ is
$\rho \left ( \cdot; \frac{A}{2} \right )$ convolved with itself. Performing substitutions in the
resulting integral,
\begin{eqnarray}
\rho(E; A) &=& \int_0^E \rho \left (\xi; \frac{A}{2} \right ) \rho \left (E - \xi; \frac{A}{2} \right ) \textup{d}\xi \nonumber \\
&=& \frac{1}{2} \int_{-E}^E \rho \left ( \frac{E - P}{2}; \frac{A}{2} \right ) \rho \left ( \frac{E + P}{2}; \frac{A}{2} \right ) \textup{d}P \nonumber \\
&=& \int_0^E \rho \left ( \frac{E - P}{2}; \frac{A}{2} \right ) \rho \left ( \frac{E + P}{2}; \frac{A}{2} \right ) \textup{d}P \nonumber \\
&=& \int_0^\infty \rho \left ( \frac{E - P}{2}; \frac{A}{2} \right ) \rho \left ( \frac{E + P}{2}; \frac{A}{2} \right ) \textup{d}P \nonumber \\
&=& \frac{1}{d \omega_d} \int_{\mathbb{R}^d} \rho \left ( \frac{E - | \textbf{P} |}{2}; \frac{A}{2} \right ) \rho \left ( \frac{E + | \textbf{P} |}{2}; \frac{A}{2} \right ) \textup{d}\textbf{P} \; . \label{posterior-integral}
\end{eqnarray}
In the second last step we have used the fact that densities of states are zero for negative arguments. Part of the
definition of the restricted density of states is
\begin{equation}
\rho(E; A) = \int_{\mathbb{R}^d} \rho(E, \textbf{P}; A) \textup{d}\textbf{P} \; . \label{prior-integral}
\end{equation}
Comparing (\ref{posterior-integral}) to (\ref{prior-integral}), a plausible formula for the restricted density of
states in a free CFT is
\begin{equation}
\log \rho(E, \textbf{P}) \sim \left ( \frac{(d + 1)^{d + 1}}{2(2d)^d} AV(E - | \textbf{P} |)^d \right )^{\frac{1}{d + 1}} + \left ( \frac{(d + 1)^{d + 1}}{2(2d)^d} AV(E + | \textbf{P} |)^d \right )^{\frac{1}{d + 1}} \; . \label{rdos-free}
\end{equation}

In order to use the convolution formula, we factored the partition function in a way that is not always valid. This
means we again have results for a free CFT in arbitrary dimension and an arbitrary CFT in two dimensions. It is
possible that further results could be obtained using an entropy proposed by Erik Verlinde in 2000
\cite{verlinde2, kutasov}. This Cardy-Verlinde formula suggests that higher dimensional CFTs are thermodynamically more
similar to their two-dimensional cousins than previously thought.

\section{Entropic dynamics of momentum}
Retracing our steps, there are nonlinear PDEs associated with changes in the momentum distribution on a lattice.
Each lattice direction, which we will call a large direction, needs to have one evolution equation associated with it.
Additionally, each site is allowed to have compact or small directions which do not show up in the PDEs.
Therefore the $\rho(E, \textbf{P})$ we should use comes from a more restricted function
$\rho(E, \textbf{P}, \textbf{P}^*)$ for the whole field theory and has $\textbf{P}^*$ for the small directions
integrated out. Most of the simulations for our previous model had one large direction. Because we used the density of
states (\ref{sym-entropy}) for SYM on $\mathbb{S}^3$, it is tempting to say that these simulations had three small
directions. This is misleading as Figure \ref{orthogonality} shows.
\begin{figure}[h]
\centering
\subfloat[][Sharing no directions]{\includegraphics[scale=0.45]{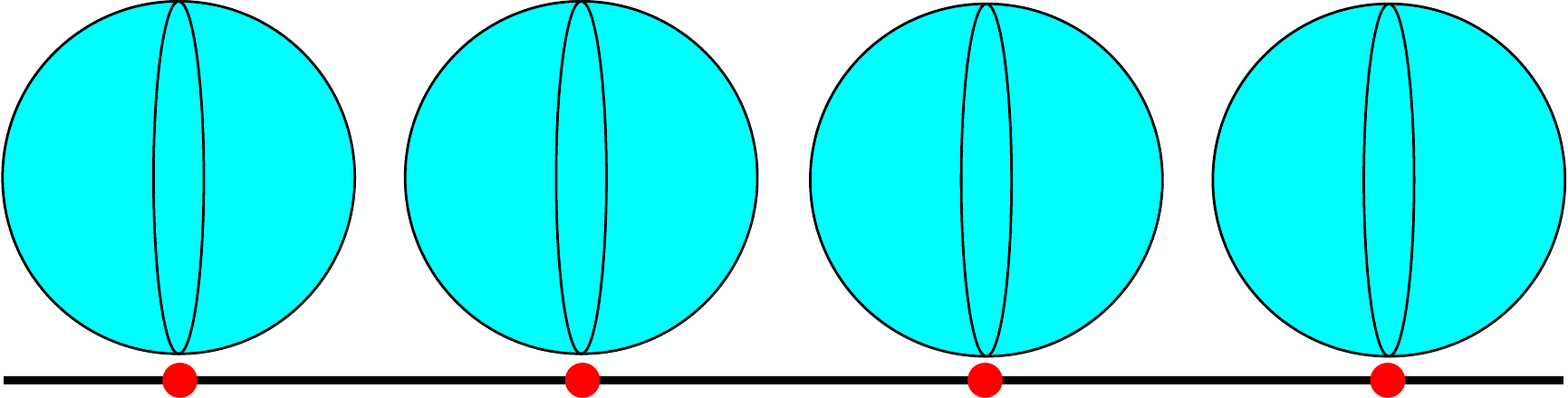}} \\
\subfloat[][Sharing one direction]{\includegraphics[scale=0.45]{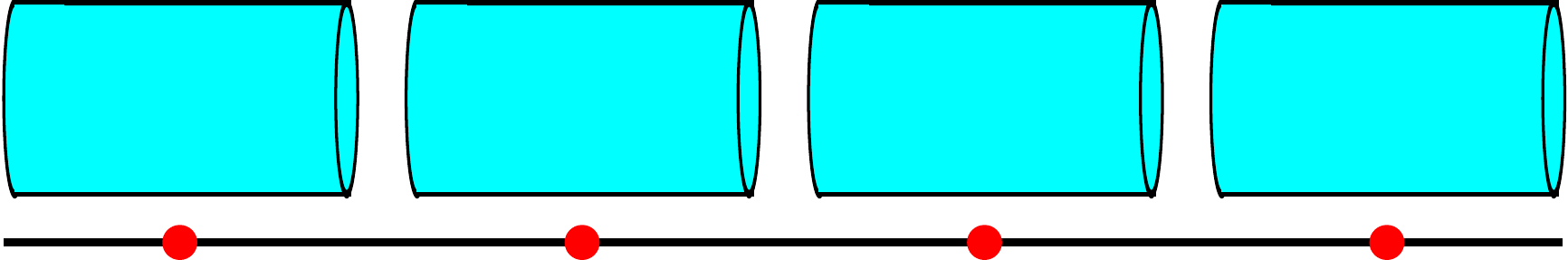}}
\caption{If one uses a density of states for two spatial dimensions, both of these situations are compatible with the
lattice being a line. One has an overall dimensionality of $2 + 1$, the other $3 + 1$. The task of introducing
momentum to our model forces us to address this ambiguity.}
\label{orthogonality}
\end{figure}
If each site is viewed as a miniature field theory whose spatial directions are all orthogonal to the lattice, there is
nothing left after integrating out $\textbf{P}^*$. The result is simply an unrestricted density of states
$\rho(E)$ precluding our attempts to extend the model. What we really need are sites that include the directions of
the lattice in addition to their compact directions. Therefore what we considered previously was not $\mathbb{S}^3$
with radius $r$ but $\mathbb{S}^2 \times [0, r]$. This is a background with only two small directions because the
non-periodic $[0, r]$ matches up with the lattice. By arguments in \cite{aharony}, the density of states for such a
theory should be similar to (\ref{sym-entropy}). Including time, this theory with one large direction and two small
ones lives on $\mathbb{R}^2 \times \mathbb{S}^2$. The Witten model, which is known to host plasma balls, has the
background $\mathbb{R}^4 \times \mathbb{S}^1$ suggesting that its lattice sites look like
$\mathbb{S}^1 \times [0, r]^3$. The equations that follow will have $d$ referring to the number of large spatial
dimensions only.

\subsection{Proceeding by analogy}
In order to derive our PDE for the energy
\begin{equation}
\frac{\partial E}{\partial t} = \partial_i \left ( C(E) \rho^2(E) \partial_i \frac{\textup{d} \log \rho(E)}{\textup{d}E} \right ) \; , \label{original-equation}
\end{equation}
we took the continuum limit of the equation
\begin{equation}
\frac{\partial n_c}{\partial t} = \sum_{\left < b, c \right >} \sum_{k \neq 0} k W_{(n_c, n_b) \rightarrow (n_c + k, n_b - k)} \; . \label{first-non-continuum}
\end{equation}
The main steps involved were assuming that $k$ was some small amount of energy $\pm \epsilon$ and writing $n_c$ as
$E(\textbf{x})$ where site $b$ was ``one lattice vector'' away from site $c$. Then (\ref{first-non-continuum}) became
\begin{equation}
\frac{\partial E(\textbf{x})}{\partial t} = \epsilon \sum_{\textbf{e} \in \{ \pm \textbf{e}_1, \dots, \pm \textbf{e}_d \}} \left [ W_{(E(\textbf{x}), E(\textbf{x} + a\textbf{e})) \rightarrow (E(\textbf{x}) + \epsilon, E(\textbf{x} + a\textbf{e}) - \epsilon)} - W_{(E(\textbf{x}), E(\textbf{x} + a\textbf{e})) \rightarrow (E(\textbf{x}) - \epsilon, E(\textbf{x} + a\textbf{e}) + \epsilon)}\right ] \label{first-continuum}
\end{equation}
and we took $\epsilon$ and $a$ to zero. The sum runs over positive and negative versions of the standard basis
vectors so that in a cubic lattice, all $2d$ nearest neighbours are accessed.

\subsubsection{Allowed transitions and rates}
In (\ref{first-continuum}) above, $\frac{\partial E(\textbf{x})}{\partial t}$ was proportional to $\epsilon$, the
smallest amount of energy allowed to move in one timestep. Therefore, we should presumably add an equation for
momentum where $\frac{\partial \textbf{P}(\textbf{x})}{\partial t}$ is proportional to some $\textbf{q}$, the smallest
momentum vector that moves in one timestep. This can be written as $q \textbf{e}$, the length times some unit vector,
where $q \rightarrow 0$ just like $\epsilon$ and $a$.

However, $q$ and $\epsilon$ are not independent. A CFT in $d = 1$ for instance has energy and momentum operators
that can be written in terms of left movers and right movers:
\begin{eqnarray}
E &=& N_{\textup{R}} + N_{\textup{L}} - \frac{c}{12} \nonumber \\
P &=& N_{\textup{R}} - N_{\textup{L}} \; . \nonumber
\end{eqnarray}
The only way to add to $E$ is to excite either $N_{\textup{R}}$ or $N_{\textup{L}}$ with $\epsilon$ units of energy
and this necessarily adds $\epsilon$ to or subtracts $\epsilon$ from $P$. In this case we must have $q = \epsilon$. We
will make the assumption that the other theories we wish to model also have the energy transferred equal to the
magnitude of the momentum transferred.

Now that we know $E$ increases or decreases by $\epsilon$ whenever $\textbf{P}$ increases or decreases by
$\epsilon \textbf{e}$, an allowed transition takes the form
\begin{equation}
(E(\textbf{x}), \textbf{P}(\textbf{x}), E(\textbf{x} + a\textbf{e}), \textbf{P}(\textbf{x} + a\textbf{e})) \rightarrow (E(\textbf{x}) \pm \epsilon, \textbf{P}(\textbf{x}) - \epsilon\textbf{e}^{\prime}, E(\textbf{x} + a\textbf{e}) \mp \epsilon, \textbf{P}(\textbf{x} + a\textbf{e}) + \epsilon\textbf{e}^{\prime}) \nonumber
\end{equation}
where $\textbf{e}$ is a unit vector indicating the two sites involved and $\textbf{e}^{\prime}$ is \textit{a priori}
some other unit vector. However, it does not make sense for momentum in the vertical direction to move between sites
that are horizontally displaced and it does not make sense for momentum pointing in the up direction to move down.
Therefore we require that $\textbf{e}^{\prime} = \textbf{e}$. This is where we see that the vector of large momentum
components $\textbf{P}$ lives in the same number of dimensions as the lattice vector $\textbf{x}$.

When focusing on energy distributions, the form of the transition rates in (\ref{first-continuum}) followed from
thermodynamic arguments as
\begin{equation}
W_{(E(\textbf{x}), E(\textbf{x} + a\textbf{e})) \rightarrow (E(\textbf{x}) + \epsilon, E(\textbf{x} + a\textbf{e}) - \epsilon)} = C\left ( \frac{E(\textbf{x}) + E(\textbf{x} + a\textbf{e})}{2} \right ) \rho(E(\textbf{x}) + \epsilon) \rho(E(\textbf{x} + a\textbf{e}) - \epsilon) \label{first-transition-rates}
\end{equation}
with the number of states having energy $E$ given by $\rho(E)$. The analogous transition rate involving momentum
should depend on $\rho(E, \textbf{P})$ in the same way. We are led to:
\begin{eqnarray}
&& W_{(E(\textbf{x}), \textbf{P}(\textbf{x}), E(\textbf{x} + a\textbf{e}), \textbf{P}(\textbf{x} + a\textbf{e})) \rightarrow (E(\textbf{x}) + \epsilon, \textbf{P}(\textbf{x}) - \epsilon\textbf{e}^{\prime}, E(\textbf{x} + a\textbf{e}) - \epsilon, \textbf{P}(\textbf{x} + a\textbf{e}) + \epsilon\textbf{e}^{\prime})} = \nonumber \\
&& \delta_{\textbf{e}, \textbf{e}^{\prime}} C\left ( \frac{E(\textbf{x}) + E(\textbf{x} + a\textbf{e})}{2}, \frac{\textbf{P}(\textbf{x}) + \textbf{P}(\textbf{x} + a\textbf{e})}{2} \right ) \nonumber \\
&& \rho(E(\textbf{x}) + \epsilon, \textbf{P}(\textbf{x}) - \epsilon\textbf{e}^{\prime}) \rho(E(\textbf{x} + a\textbf{e}) - \epsilon, \textbf{P}(\textbf{x} + a\textbf{e}) + \epsilon\textbf{e}^{\prime}) \; . \label{second-transition-rates}
\end{eqnarray}

\subsubsection{A tensor identity}
Consider what happens when we differentiate $E(\textbf{x} + a\textbf{e})$ twice with respect to $a$.
\begin{eqnarray}
\frac{\partial}{\partial a} E(\textbf{x} + a\textbf{e}) &=& e_i \partial_i E(\textbf{x} + a\textbf{e}) \nonumber \\
\frac{\partial^2}{\partial a^2} E(\textbf{x} + a\textbf{e}) &=& e_i e_j \partial_i \partial_j E(\textbf{x} + a\textbf{e}) \label{example-derivative}
\end{eqnarray}
Note that $e_i$ and $e_j$ without the boldface are not the $i^{\mathrm{th}}$ and $j^{\mathrm{th}}$ standard basis
vectors, they are the $i^{\mathrm{th}}$ and $j^{\mathrm{th}}$ components of the general unit vector $\textbf{e}$.
All of the terms that we must differentiate with respect to $a$ occur inside a sum over $\textbf{e}$ so we must know
how to deal with sums of components of unit vectors. When we sum $e_i e_j$, there is only one standard basis vector
for which $e_i$ is nonzero and only one standard basis vector for which $e_j$ is nonzero. They have to be the same one
for the resulting sum to be nonzero, so it is clear that:
\begin{equation}
\sum_{\textbf{e} \in \{ \pm \textbf{e}_1, \dots, \pm \textbf{e}_d \}} e_i e_j = 2\delta_{ij} \; . \nonumber
\end{equation}
We used this fact when deriving the PDE for our first model.
For reasons that are not initially clear, we will show that one can get the same answer up to a constant by replacing
the sum over lattice vectors with a surface integral over a sphere.
\begin{claim}
\begin{equation}
\int_{\mathbb{S}^{d - 1}} e_i e_j \textup{d}S_{\textbf{e}} = \omega_d \delta{ij} \nonumber
\end{equation}
where $\omega_d$ is the volume of the unit ball in $\mathbb{R}^d$.
\end{claim}
\begin{proof}
If $i = j$, we simply apply the divergence theorem to the identity vector field. Let
$\textbf{v}(\textbf{e}) = \textbf{e}$, defined on the unit ball $B^d$. Then,
\begin{equation}
\int_{B^d} \nabla \cdot \textbf{v} \textup{d}\textbf{e} = \int_{\partial B^d} \textbf{v} \cdot \textbf{n} \textup{d}S_{\textbf{e}} \nonumber
\end{equation}
for this vector field reads
\begin{equation}
d \int_{B^d} \textup{d}\textbf{e} = \int_{\mathbb{S}^{d - 1}} (e_1)^2 + \dots + (e_d)^2 \textup{d}S_{\textbf{e}} \nonumber
\end{equation}
because if the unit vector $\textbf{e}$ extends from the origin to a point on $\mathbb{S}^{d - 1}$, the unit normal
to the surface at this point is just $\textbf{e}$. By symmetry, the right hand side is $d$ copies of what we are
trying to compute, so this cancels the $d$ on the left hand side yielding:
\begin{equation}
\omega_d = \int_{\mathbb{S}^{d - 1}} (e_i)^2 \textup{d}S_{\textbf{e}} \; . \nonumber
\end{equation}
To see that this integral vanishes when $i$ and $j$ are different, parametrize the unit sphere as
\begin{eqnarray}
x_1 &=& \cos \theta_1 \nonumber \\
x_2 &=& \sin \theta_1 \cos \theta_2 \nonumber \\
&\dots& \nonumber \\
x_{d - 1} &=& \sin \theta_1 \dots \sin \theta_{d - 2} \cos \theta_{d - 1} \nonumber \\
x_d &=& \sin \theta_1 \dots \sin \theta_{d - 2} \sin \theta_{d - 1} \nonumber
\end{eqnarray}
and without loss of generality, choose $e_i = x_d$ and $e_j = x_{d - 1}$. The desired integral is then
\begin{equation}
\int_0^{2\pi} \int_0^{\pi} \dots \int_0^{\pi} [ \sin^2 \theta_1 \dots \sin^2 \theta_{d - 2} \sin \theta_{d - 1} \cos \theta_{d - 1} ] \sin^{d - 2} \theta_1 \sin^{d - 3} \theta_2 \dots \sin \theta_{d - 2} \textup{d}\theta_1 \dots \textup{d}\theta_{d - 2} \textup{d}\theta_{d - 1} \nonumber
\end{equation}
where the part in brackets is $e_i e_j$. Since the area element only goes up to $\theta_{d - 2}$, the part involving
$\theta_{d - 1}$ is just a sine and a cosine which we know integrates to zero.
\end{proof}
Unlike in our first model, equations involving momenta will have several places in which there are \textit{four}
components of $\textbf{e}$. This will make it important that we integrate rather than sum. A sum of four components
will give zero unless the indices are all the same.
\begin{equation}
\sum_{\textbf{e} \in \{ \pm \textbf{e}_1, \dots, \pm \textbf{e}_d \}} e_i e_j e_k e_l =
\begin{cases}
2 & i = j = k = l \\
0 & \mathrm{otherwise}
\end{cases} \nonumber
\end{equation}
The continuous version, on the other hand, is a much nicer object.
\begin{claim}
\begin{equation}
\int_{\mathbb{S}^{d - 1}} e_i e_j e_k e_l \textup{d}S_{\textbf{e}} = \frac{\omega_d}{d + 2} \left ( \delta_{ij}\delta_{kl} + \delta_{ik}\delta_{jl} + \delta_{il}\delta_{jk} \right ) \nonumber
\end{equation}
where $\omega_d$ is the volume of the unit ball in $\mathbb{R}^d$.
\end{claim}
\begin{proof}
We will first prove the special case when $i = j = k = l$. Define a vector field
$\textbf{v} : B^d \rightarrow \mathbb{R}^d$ which cubes every component of its argument:
\begin{equation}
v_i(\textbf{e}) = (e_i)^3 \; . \nonumber
\end{equation}
Using the divergence theorem on this, we get:
\begin{eqnarray}
\int_{B^d} \nabla \cdot \textbf{v} \textup{d}\textbf{e} &=& \int_{\partial B^d} \textbf{v} \cdot \textbf{n} \textup{d}S_{\textbf{e}} \nonumber \\
3 \int_{B^d} (e_1)^2 + \dots (e_d)^2 \textup{d}\textbf{e} &=& \int_{\mathbb{S}^{d - 1}} (e_1)^4 + \dots + (e_d)^4 \textup{d}S_{\textbf{e}} \nonumber \\
3d \omega_d \int_0^1 r^{d + 1} \textup{d}r &=& d \int_{\mathbb{S}^{d - 1}} (e_i)^4 \textup{d}S_{\textbf{e}} \nonumber \\
\frac{3 \omega_d}{d + 2} &=& \int_{\mathbb{S}^{d - 1}} (e_i)^4 \textup{d}S_{\textbf{e}} \; . \nonumber
\end{eqnarray}
We will now apply this to the case when $i$ and $k$ are different.
\begin{eqnarray}
\sum_{k \neq i} \int_{\mathbb{S}^{d - 1}} (e_i)^2 (e_k)^2 \textup{d}S_{\textbf{e}} &=& \sum_{k = 1}^d \int_{\mathbb{S}^{d - 1}} (e_i)^2 (e_k)^2 \textup{d}S_{\textbf{e}} - \int_{\mathbb{S}^{d - 1}} (e_i)^4 \textup{d}S_{\textbf{e}} \nonumber \\
&=& \int_{\mathbb{S}^{d - 1}} (e_i)^2 \textup{d}S_{\textbf{e}} - \int_{\mathbb{S}^{d - 1}} (e_i)^4 \textup{d}S_{\textbf{e}} \nonumber \\
&=& \omega_d - \frac{3 \omega_d}{d + 2} \nonumber \\
&=& \omega_d \frac{d - 1}{d + 2}
\end{eqnarray}
We summed $k$ over all but one of the $d$ choices, so by symmetry, this should be $d - 1$ times the actual value.
This shows that:
\begin{equation}
\int_{\mathbb{S}^{d - 1}} e_i e_j (e_k)^2 \textup{d}S_{\textbf{e}} = \frac{\omega_d}{d + 2} \delta_{ij} \; . \nonumber
\end{equation}
We have computed the integral when $k$ and $l$ are the same. Similarly, if $k$ and $i$ were the same, the result would
be proportional to $\delta_{jl}$ and if $k$ and $j$ were the same, the result would be proportional to $\delta_{il}$.
This covers the three ways to have pairs of equal indices. It remains to be seen that the expression vanishes when
$i$, $j$, $k$ and $l$ are all different. As before, we can show this using explicit co-ordinates.
\begin{eqnarray}
e_i &=& \sin \theta_1 \dots \sin \theta_{d - 2} \sin \theta_{d - 1} \nonumber \\
e_j &=& \sin \theta_1 \dots \sin \theta_{d - 2} \cos \theta_{d - 1} \nonumber \\
e_k &=& \sin \theta_1 \dots \sin \theta_{d - 3} \cos \theta_{d - 2} \nonumber \\
e_l &=& \sin \theta_1 \dots \sin \theta_{d - 4} \cos \theta_{d - 3} \nonumber
\end{eqnarray}
Again, $e_i$ and $e_j$ contribute an odd function of $\theta_{d - 1}$ while the surface measure does not depend on
$\theta_{d - 1}$. The other components cannot change this because $e_k$ and $e_l$ are both distinct from $e_i$ and
$e_j$.
\end{proof}

\subsubsection{Setting up the equations}
In (\ref{first-continuum}), $\frac{\partial E(\textbf{x})}{\partial t}$ was schematically given by
``$W$ for a transition that adds $\epsilon$ to site $\textbf{x}$'' minus
``$W$ for a transition that subtracts $\epsilon$ from site $\textbf{x}$'', summed over $\textbf{e}$ and multiplied
by $\epsilon$. We will write out this type of expression except we will integrate instead of sum.
\begin{eqnarray}
\frac{\partial E(\textbf{x})}{\partial t} &=& \epsilon \int_{\mathbb{S}^{d - 1}} W_{(E(\textbf{x}), \textbf{P}(\textbf{x}), E(\textbf{x} + a\textbf{e}), \textbf{P}(\textbf{x} + a\textbf{e})) \rightarrow (E(\textbf{x}) + \epsilon, \textbf{P}(\textbf{x}) - \epsilon\textbf{e}, E(\textbf{x} + a\textbf{e}) - \epsilon, \textbf{P}(\textbf{x} + a\textbf{e}) + \epsilon\textbf{e})} \nonumber \\
&-& W_{(E(\textbf{x}), \textbf{P}(\textbf{x}), E(\textbf{x} + a\textbf{e}), \textbf{P}(\textbf{x} + a\textbf{e})) \rightarrow (E(\textbf{x}) - \epsilon, \textbf{P}(\textbf{x}) - \epsilon\textbf{e}, E(\textbf{x} + a\textbf{e}) + \epsilon, \textbf{P}(\textbf{x} + a\textbf{e}) + \epsilon\textbf{e})} \textup{d}S_{\textbf{e}} \label{second-continuum-e}
\end{eqnarray}
Note that $\textbf{P}(\textbf{x})$ always receives a $-\epsilon \textbf{e}$ contribution because we are adopting a
convention where $\textbf{x} + a\textbf{e}$ is the ``other site''. If our convention had $\textbf{x} - a\textbf{e}$
as the other site, $\textbf{P}(\textbf{x})$ would receive a $+\epsilon \textbf{e}$ contribution. The fact that there
are two ways to raise the energy $E(\textbf{x}) \mapsto E(\textbf{x}) + \epsilon$ (adding a left mover from the right
site and adding a right mover from the left site) is accounted for by the integral that causes $\textbf{e}$ to change
direction. Writing down the momentum equation and being careful about the same type of thing, we get:
\begin{eqnarray}
\frac{\partial \textbf{P}(\textbf{x})}{\partial t} &=& -\epsilon \int_{\mathbb{S}^{d - 1}} \textbf{e} \left [ W_{(E(\textbf{x}), \textbf{P}(\textbf{x}), E(\textbf{x} + a\textbf{e}), \textbf{P}(\textbf{x} + a\textbf{e})) \rightarrow (E(\textbf{x}) + \epsilon, \textbf{P}(\textbf{x}) - \epsilon\textbf{e}, E(\textbf{x} + a\textbf{e}) - \epsilon, \textbf{P}(\textbf{x} + a\textbf{e}) + \epsilon\textbf{e})} \right. \nonumber \\
&+& \left. W_{(E(\textbf{x}), \textbf{P}(\textbf{x}), E(\textbf{x} + a\textbf{e}), \textbf{P}(\textbf{x} + a\textbf{e})) \rightarrow (E(\textbf{x}) - \epsilon, \textbf{P}(\textbf{x}) - \epsilon\textbf{e}, E(\textbf{x} + a\textbf{e}) + \epsilon, \textbf{P}(\textbf{x} + a\textbf{e}) + \epsilon\textbf{e})} \right ] \textup{d}S_{\textbf{e}} \; . \label{second-continuum-p}
\end{eqnarray}
Substituting our transition rates (\ref{second-transition-rates}), these become:
\begin{eqnarray}
\frac{\partial E(\textbf{x})}{\partial t} &=& \epsilon \int_{\mathbb{S}^{d - 1}} C\left ( \frac{E(\textbf{x}) + E(\textbf{x} + a\textbf{e})}{2}, \frac{\textbf{P}(\textbf{x}) + \textbf{P}(\textbf{x} + a\textbf{e})}{2} \right ) \nonumber \\
&& \left [ \rho(E(\textbf{x}) + \epsilon, \textbf{P}(\textbf{x}) - \epsilon \textbf{e}) \rho(E(\textbf{x} + a\textbf{e}) - \epsilon, \textbf{P}(\textbf{x} + a\textbf{e}) + \epsilon \textbf{e}) \right. \nonumber \\
&& \left. - \rho(E(\textbf{x}) - \epsilon, \textbf{P}(\textbf{x}) - \epsilon \textbf{e}) \rho(E(\textbf{x} + a\textbf{e}) + \epsilon, \textbf{P}(\textbf{x} + a\textbf{e}) + \epsilon \textbf{e}) \right ] \textup{d}S_{\textup{e}} \nonumber \\
\frac{\partial \textbf{P}(\textbf{x})}{\partial t} &=& -\epsilon \int_{\mathbb{S}^{d - 1}} \textbf{e} C\left ( \frac{E(\textbf{x}) + E(\textbf{x} + a\textbf{e})}{2}, \frac{\textbf{P}(\textbf{x}) + \textbf{P}(\textbf{x} + a\textbf{e})}{2} \right ) \nonumber \\
&& \left [ \rho(E(\textbf{x}) + \epsilon, \textbf{P}(\textbf{x}) - \epsilon \textbf{e}) \rho(E(\textbf{x} + a\textbf{e}) - \epsilon, \textbf{P}(\textbf{x} + a\textbf{e}) + \epsilon \textbf{e}) \right. \nonumber \\
&& \left. + \rho(E(\textbf{x}) - \epsilon, \textbf{P}(\textbf{x}) - \epsilon \textbf{e}) \rho(E(\textbf{x} + a\textbf{e}) + \epsilon, \textbf{P}(\textbf{x} + a\textbf{e}) + \epsilon \textbf{e}) \right ] \textup{d}S_{\textup{e}} \; . \label{ep-equations}
\end{eqnarray}
Our task in the following section will be to differentiate the right hand sides with respect to $\epsilon$ and $a$
and plug in $\epsilon = 0$ and $a = 0$. Recall that in our model for the energy, the only non-vanishing derivative of
order $\leq 4$ was $\frac{\partial^4}{\partial \epsilon^2 \partial a^2}$. The momentum equation in (\ref{ep-equations})
has other low order derivatives that do not vanish. For example, set $C = 1$ and compute
$\frac{\partial^4}{\partial \epsilon \partial a^3}$. Differentiating with respect to $\epsilon$ once and setting
$\epsilon = 0$ simply removes the factor of $\epsilon$ in front. Therefore this contribution to
$\frac{\partial P_i}{\partial t}$ is:
\begin{eqnarray}
&-2& \int_{\mathbb{S}^{d - 1}} e_i \rho(E(\textbf{x}), \textbf{P}(\textbf{x})) \frac{\partial^3}{\partial a^3} \left. \rho(E(\textbf{x} + a\textbf{e}), \textbf{P}(\textbf{x} + a\textbf{e})) \right |_{a=0} \nonumber \textup{d}S_{\textbf{e}} \\
&=& -2 \rho(E(\textbf{x}), \textbf{P}(\textbf{x})) \partial_j \partial_k \partial_l \rho(E(\textbf{x}), \textbf{P}(\textbf{x})) \int_{\mathbb{S}^{d - 1}} e_i e_j e_k e_l \textup{d}S_{\textbf{e}} \nonumber \\
&=& \frac{-6 \omega_d}{d} \rho(E(\textbf{x}), \textbf{P}(\textbf{x})) \partial_i \partial_j \partial_j \rho(E(\textbf{x}), \textbf{P}(\textbf{x})) \; . \label{simple-momentum}
\end{eqnarray}
This produces a rotationally covariant equation precisely because of the integral. From the discrete sum of
$e_i e_j e_k e_l$ that we computed previously, we can see that this would have given us
$-4 \rho \partial_i \partial_i \partial_i \rho$ which does not transform as a vector.

\subsection{The continuum limit}
Rewriting (\ref{ep-equations}) for brevity,
\begin{eqnarray}
\frac{\partial E}{\partial t} &=& X(\epsilon, a) = \epsilon \tilde{X}(\epsilon, a) \nonumber \\
\frac{\partial P_i}{\partial t} &=& Y_i(\epsilon, a) = \epsilon \tilde{Y}_i(\epsilon, a) \; . \nonumber
\end{eqnarray}
We must Taylor expand $X$ and $Y_i$ around small arguments.

\subsubsection{To first non-vanishing order}
It is clear that if we do not differentiate at all with respect to $\epsilon$, $X$ and $Y_i$ will vanish at
$\epsilon = 0$ no matter how many times were differentiate with respect to $a$. It is also true that
$\tilde{X}$ and $\tilde{Y}_i$ will vanish identically if we do not differentiate them with respect to $a$. The integrand
in $\tilde{X}$ becomes the zero function of $\epsilon$ once we plug in $a = 0$.
The integrand in $\tilde{Y}_i$ does not but it becomes proportional to $e_i$ which is odd.

$X(0, a) = 0 = X(\epsilon, 0)$ and $Y_i(0, a) = 0 = Y_i(\epsilon, 0)$ so the lowest order derivative that could
possibly survive is $\frac{\partial^2}{\partial \epsilon \partial a}$. Looking at $X(\epsilon, a)$ once more,
$\left. \frac{\partial X}{\partial \epsilon} \right |_{\epsilon = 0} = \tilde{X}(a, 0)$ which is the zero function of
$a$. Therefore,
$\left. \frac{\partial^2 X}{\partial \epsilon \partial a} \right |_{\substack{a = 0 \\ \epsilon = 0}} = 0$. We will
now compute this same derivative for $Y_i$.
\begin{eqnarray}
\left. \frac{\partial^2 Y_i}{\partial \epsilon \partial a} \right|_{\substack{a = 0 \\ \epsilon = 0}} &=& \left. \frac{\partial \tilde{Y}_i}{\partial a} \right|_{\substack{a = 0 \\ \epsilon = 0}} \nonumber \\
&=& -2\int_{\mathbb{S}^{d - 1}} e_i \rho(E(\textbf{x}), \textbf{P}(\textbf{x})) \frac{\partial}{\partial a} C\left ( \frac{E(\textbf{x}) + E(\textbf{x} + a\textbf{e})}{2}, \frac{\textbf{P}(\textbf{x}) + \textbf{P}(\textbf{x} + a\textbf{e})}{2} \right ) \nonumber \\
&& \left. \rho(E(\textbf{x} + a\textbf{e}), \textbf{P}(\textbf{x} + a\textbf{e})) \right|_{a = 0} \textup{d}S_{\textbf{e}} \nonumber \\
&=& -2\rho \left ( \frac{1}{2} \rho \partial_j C + C \partial_j \rho \right ) \int_{\mathbb{S}^{d - 1}} e_i e_j \textup{d}S_{\textbf{e}} \nonumber \\
&=& -\omega_d \partial_i (C \rho^2) \nonumber
\end{eqnarray}
What this shows is that to second order,
\begin{eqnarray}
\frac{\partial E}{\partial t} &=& 0 \nonumber \\
\frac{\partial P_i}{\partial t} &=& -\epsilon a \omega_d \partial_i (C \rho^2) \; . \label{second-order-ep}
\end{eqnarray}

This is consistent with the intuition about which way a distribution of momentum should move. In
(\ref{second-order-ep}), $C \rho^2$ is a function of the energy (which is static) and the momentum magnitude
$| \textbf{P} |$. Consider the simplest case where we let $d = 1$ and Taylor expand $C \rho^2$ as some constant
plus $|P|$:
\begin{equation}
\frac{\partial P(x, t)}{\partial t} \propto -\frac{\partial |P(x, t)|}{\partial x} \; . \nonumber
\end{equation}
A solution to the above equation is $P(x, t) = f(x - t)$ where $f$ is a non-negative function. Because of the
absolute value, another solution is $P(x, t) = -f(x + t)$. This confirms that the direction of motion for a
disturbance is equal to the sign of the disturbance; $f(x - t)$ is a lump of positive (right) momentum that moves
to the right, while $-f(x + t)$ is a lump of negative (left) momentum that moves to the left. This behaviour appears
to rectify a shortcoming of our last model; an ejected piece of a plasma ball propagating through empty space.

For our purposes, it is not enough to stop at the second order expansion where the energy is constant in time.
As we did with our first model, we will compute as many derivatives of $X$ as we need to see the energy dynamics.
We argued above that we must differentiate $X$ at least once with respect to $a$ and at least twice with respect to
$\epsilon$. In fact, we must differentiate even more than this; $\frac{\partial^3 X}{\partial \epsilon^2 \partial a}$
still vanishes. What we find by looking at the fourth order derivatives is:
\begin{equation}
\frac{\partial E}{\partial t} = \frac{1}{4} \epsilon^2 a^2 \left. \frac{\partial^4 X}{\partial \epsilon^2 \partial a^2} \right|_{\substack{a = 0 \\ \epsilon = 0}} \; . \label{e-same-derivative}
\end{equation}
When evaluating this and similar expressions, we will surpress the steps in calculating the $\epsilon$ derivatives,
saving them for the appendix. Additionally, we will adopt the notation
$\rho_{+} = \rho(E(\textbf{x} + a\textbf{e}), \textbf{P}(\textbf{x} + a\textbf{e}))$ and
$C_{+} = C\left ( \frac{E(\textbf{x}) + E(\textbf{x} + a\textbf{e})}{2}, \frac{\textbf{P}(\textbf{x}) + \textbf{P}(\textbf{x} + a\textbf{e})}{2} \right )$
and recall that:
\begin{eqnarray}
\frac{\partial \rho_{+}}{\partial a} &=& e_j \partial_j \rho_{+} \nonumber \\
\frac{\partial C_{+}}{\partial a} &=& \frac{1}{2} e_j \partial_j C_{+} \; . \label{convenient-relations}
\end{eqnarray}

If we apply this to (\ref{e-same-derivative}), our steps are:
\begin{eqnarray}
\left. \frac{\partial^4 X}{\partial \epsilon^2 \partial a^2} \right|_{\substack{a = 0 \\ \epsilon = 0}} &=& 2 \left. \frac{\partial^3 \tilde{X}}{\partial \epsilon \partial a^2} \right|_{\substack{a = 0 \\ \epsilon = 0}} \nonumber \\
&=& 4 \int_{\mathbb{S}^{d - 1}} \frac{\partial^2}{\partial a^2} \left. \left [ C_{+} \left ( \rho_{+} \frac{\partial \rho}{\partial E} - \rho \frac{\partial \rho_{+}}{\partial E} \right ) \right ] \right |_{a = 0} \textup{d}S_{\textbf{e}} \nonumber \\
&=& 4 \int_{\mathbb{S}^{d - 1}} e_i \frac{\partial}{\partial a} \left. \left [ \frac{1}{2} \partial_i C_{+} \left ( \rho_{+} \frac{\partial \rho}{\partial E} - \rho \frac{\partial \rho_{+}}{\partial E} \right ) + C_{+} \left ( \partial_i \rho_{+} \frac{\partial \rho}{\partial E} - \rho \partial_i \frac{\partial \rho_{+}}{\partial E} \right ) \right ] \right |_{a = 0} \textup{d}S_{\textbf{e}} \nonumber \\
&=& 4 \int_{\mathbb{S}^{d - 1}} e_i e_j \left [ C \left ( \partial_i \partial_j \rho \frac{\partial \rho}{\partial E} - \rho \partial_i \partial_j \frac{\partial \rho}{\partial E} \right ) -\partial_i C \left ( \rho \partial_j \frac{\partial \rho}{\partial E} - \partial_j \rho \frac{\partial \rho}{\partial E} \right ) \right ] \textup{d}S_{\textbf{e}} \nonumber \\
&=& 4 \omega_d \left [ C \left ( \partial_i \partial_i \rho \frac{\partial \rho}{\partial E} - \rho \partial_i \partial_i \frac{\partial \rho}{\partial E} \right ) - \partial_i C \left ( \rho \partial_i \frac{\partial \rho}{\partial E} - \partial_i \rho \frac{\partial \rho}{\partial E} \right ) \right ] \nonumber \\
&=& -4 \omega_d \partial_i \left ( C \rho^2 \partial_i \frac{\textup{d} \log \rho}{\textup{d} E} \right ) \; . \nonumber
\end{eqnarray}
This shows that to first non-vanishing order,
\begin{equation}
\frac{\partial E}{\partial t} = -\epsilon^2 a^2 \omega_d \partial_i \left ( C \rho^2 \partial_i \frac{\partial \log \rho}{\partial E} \right ) \; . \label{fourth-order-e}
\end{equation}
Notice that when we only considered the energy, we dropped factors of $\epsilon$ and $a$ from (\ref{original-equation}).
We must keep factors of $\epsilon$ and $a$ in the PDEs for energy and momentum because they appear with different
exponents. The dominant contribution to $\frac{\partial \textbf{P}}{\partial t}$ includes $\epsilon a$, while in the
dominant contribution to $\frac{\partial E}{\partial t}$, it is $(\epsilon a)^2$. The ratio between these factors is
a physical quantity because it tells us how long the diffusion time scale in (\ref{fourth-order-e}) is relative to the
transport time scale in (\ref{second-order-ep}).

\subsubsection{To first consistent order}
We now have non-trivial differential equations for both energy and momentum. However, the former is correct to fourth
order while the latter is only correct to second order. To be consistent, we will find further terms in the momentum
PDE. The terms are
\begin{equation}
\frac{\partial P_i}{\partial t} = \epsilon a \left. \frac{\partial^2 Y_i}{\partial \epsilon \partial a} \right |_{\substack{a = 0 \\ \epsilon = 0}} + \frac{1}{6} \epsilon^3 a \left. \frac{\partial^4 Y_i}{\partial \epsilon^3 \partial a} \right |_{\substack{a = 0 \\ \epsilon = 0}} + \frac{1}{4} \epsilon^2 a^2 \left. \frac{\partial^4 Y_i}{\partial \epsilon^2 \partial a^2} \right |_{\substack{a = 0 \\ \epsilon = 0}} + \frac{1}{6} \epsilon a^3 \left. \frac{\partial^4 Y_i}{\partial \epsilon \partial a^3} \right |_{\substack{a = 0 \\ \epsilon = 0}} \; , \label{p-4-derivatives}
\end{equation}
which we will calculate one by one.
\begin{eqnarray}
\left. \frac{\partial^4 Y_i}{\partial \epsilon^3 \partial a} \right |_{\substack{a = 0 \\ \epsilon = 0}} &=& 3 \left. \frac{\partial^3 \tilde{Y_i}}{\partial \epsilon^2 \partial a} \right |_{\substack{a = 0 \\ \epsilon = 0}} \nonumber \\
&=& 6 \int_{\mathbb{S}^{d - 1}} e_i \frac{\partial}{\partial a} \left [ C_{+} \left ( \frac{\partial \rho}{\partial E} - e_j \frac{\partial \rho}{\partial P_j} \right ) \left ( \frac{\partial \rho_{+}}{\partial E} - e_k \frac{\partial \rho_{+}}{\partial P_k} \right ) \right. \nonumber \\
&& + C_{+} \left ( \frac{\partial \rho}{\partial E} + e_j \frac{\partial \rho}{\partial P_j} \right ) \left ( \frac{\partial \rho_{+}}{\partial E} + e_k \frac{\partial \rho_{+}}{\partial P_k} \right ) - C_{+} \rho_{+} \left ( \frac{\partial^2 \rho}{\partial E^2} + e_j e_k \frac{\partial^2 \rho}{\partial P_j \partial P_k} \right ) \nonumber \\
&& \left. \left. - C_{+} \rho \left ( \frac{\partial^2 \rho_{+}}{\partial E^2} + e_j e_k \frac{\partial^2 \rho_{+}}{\partial P_j \partial P_k} \right ) \right ] \right |_{a = 0} \textup{d}S_{\textbf{e}} \nonumber \\
&=& 6 \int_{\mathbb{S}^{d - 1}} e_i e_l \left [ 2 \frac{\partial \rho}{\partial E} \left ( \frac{1}{2} \partial_l C \frac{\partial \rho}{\partial E} + C \partial_l \frac{\partial \rho}{\partial E} \right ) + 2 \frac{\partial \rho}{\partial P_j} e_j e_k \left ( \frac{1}{2} \partial_l C \frac{\partial \rho}{\partial P_k} + C \partial_l \frac{\partial \rho}{\partial P_k} \right ) \right. \nonumber \\
&& - \left ( \frac{1}{2} \partial_l C \rho + C \partial_l \rho \right ) \left ( \frac{\partial^2 \rho}{\partial E^2} + e_j e_k \frac{\partial^2 \rho}{\partial P_j \partial P_k} \right ) \nonumber \\
&& \left. - \rho \left ( \frac{1}{2} \partial_l C \frac{\partial^2 \rho}{\partial E^2} + \frac{1}{2} \partial C e_j e_k \frac{\partial^2 \rho}{\partial P_j \partial P_k} + C \partial_l \frac{\partial^2 \rho}{\partial E^2} + C e_j e_k \partial_l \frac{\partial^2 \rho}{\partial P_j \partial P_k} \right ) \right ] \textup{d}S_{\textbf{e}} \nonumber \\
&=& \frac{6 \omega_d}{d + 2} \left [ \partial_l C \left ( \frac{\partial \rho}{\partial E} \right )^2 \delta_{jk} + \partial_l C \frac{\partial \rho}{\partial P_j} \frac{\partial \rho}{\partial P_k} + C \partial_l \left ( \frac{\partial \rho}{\partial E} \right )^2 \delta{jk} + C \partial_l \left ( \frac{\partial \rho}{\partial P_j} \frac{\partial \rho}{\partial P_k} \right ) \right. \nonumber \\
&& - \rho \partial_l C \frac{\partial^2 \rho}{\partial E^2} \delta_{jk} - \rho \partial_l C \frac{\partial^2 \rho}{\partial P_j \partial P_k} - C \partial_l \rho \frac{\partial^2 \rho}{\partial E^2} \delta_{jk} - C \partial_l \rho \frac{\partial^2 \rho}{\partial P_j \partial P_k} \nonumber \\
&& \left. - C \rho \partial_l \frac{\partial^2 \rho}{\partial E^2} \delta_{jk} - C \rho \partial_l \frac{\partial^2 \rho}{\partial P_j \partial P_k} \right ] \left ( \delta_{ij} \delta_{kl} + \delta_{ik} \delta_{jl} + \delta_{il} \delta_{jk} \right ) \nonumber \\
&=& -\frac{6 \omega_d}{d + 2} \partial_l \left [ C \rho^2 \left ( \frac{\partial^2 \log \rho}{\partial E^2} \delta_{jk} + \frac{\partial^2 \log \rho}{\partial P_j \partial P_k} \right ) \right ] \left ( \delta_{ij} \delta_{kl} + \delta_{ik} \delta_{jl} + \delta_{il} \delta_{jk} \right ) \nonumber
\end{eqnarray}
In the next one, the steps are very similar.
\begin{eqnarray}
\left. \frac{\partial^4 Y_i}{\partial \epsilon^2 \partial a^2} \right |_{\substack{a = 0 \\ \epsilon = 0}} &=& 2 \left. \frac{\partial^3 \tilde{Y_i}}{\partial \epsilon \partial a^2} \right |_{\substack{a = 0 \\ \epsilon = 0}} \nonumber \\
&=& 4 \int_{\mathbb{S}^{d - 1}} e_i e_j \left. \frac{\partial^2}{\partial a^2} \left [ C_{+} \rho_{+} \frac{\partial \rho}{\partial P_j} - C_{+} \rho \frac{\partial \rho_{+}}{\partial P_j} \right ] \right |_{a = 0} \textup{d}S_{\textbf{e}} \nonumber \\
&=& 4 \int_{\mathbb{S}^{d - 1}} e_i e_j e_k \left. \frac{\partial}{\partial a} \left [ \left ( \frac{1}{2} \partial_k C_{+} \rho_{+} + C_{+} \partial_k \rho_{+} \right ) \frac{\partial \rho}{\partial P_j} - \left ( \frac{1}{2} \partial_k C_{+} \frac{\partial \rho_{+}}{\partial P_j} + C_{+} \partial_k \frac{\partial \rho_{+}}{\partial P_j} \right ) \rho \right ] \right |_{a = 0} \textup{d}S_{\textbf{e}} \nonumber \\
&=& 4 \int_{\mathbb{S}^{d - 1}} e_i e_j e_k e_l \left [ \left ( \partial_k C \partial_l \rho  + C \partial_k \partial_l \rho \right ) \frac{\partial \rho}{\partial P_j} - \left ( \partial_k C \partial_l \frac{\partial \rho}{\partial P_j} + C \partial_k \partial_l \frac{\partial \rho}{\partial P_j} \right ) \rho \right ] \textup{d}S_{\textbf{e}} \nonumber \\
&=& \frac{4 \omega_d}{d + 2} \left ( \partial_k C \partial_l \rho \frac{\partial \rho}{\partial P_j} + C \partial_k \partial_l \rho \frac{\partial \rho}{\partial P_j} - \partial_k C \partial_l \frac{\partial \rho}{\partial P_j} \rho - C \partial_k \partial_l \frac{\partial \rho}{\partial P_j} \rho \right ) \left ( \delta_{ij} \delta_{kl} + \delta_{ik} \delta_{jl} + \delta_{il} \delta_{jk} \right ) \nonumber \\
&=& -\frac{4 \omega_d}{d + 2} \partial_k \left ( C \rho^2 \partial_l \frac{\partial \log \rho}{\partial P_j} \right ) \left ( \delta_{ij} \delta_{kl} + \delta_{ik} \delta_{jl} + \delta_{il} \delta_{jk} \right ) \nonumber
\end{eqnarray}
The final expression for
$\left. \frac{\partial^4 Y_i}{\partial \epsilon^3 \partial a} \right |_{\substack{a = 0 \\ \epsilon = 0}}$ involves
no more than one derivative of $C$ becaues there is only one derivative with respect to $a$. Since
$\left. \frac{\partial^4 Y_i}{\partial \epsilon^2 \partial a^2} \right |_{\substack{a = 0 \\ \epsilon = 0}}$ has two
$a$ derivatives, there is a chance that the expression will involve two derivatives of $C$. However, we know from above
that the double derivatives of $C$ cancel out. This will not happen when we compute
$\left. \frac{\partial^4 Y_i}{\partial \epsilon \partial a^3} \right |_{\substack{a = 0 \\ \epsilon = 0}}$. In this
case all three derivatives of $C$ will survive and we will have to extend (\ref{convenient-relations}) so that we know
how to deal with up to three derivatives of $C_{+}$. A simple $C(E(x))$ can be used to illustrate the problem.
\begin{eqnarray}
\frac{\partial C}{\partial x} &=& \frac{\partial C}{\partial E} \frac{\partial E}{\partial x} \nonumber \\
\frac{\partial^2 C}{\partial x^2} &=& \frac{\partial^2 C}{\partial E^2} \left ( \frac{\partial E}{\partial x} \right )^2 + \frac{\partial C}{\partial E} \frac{\partial^2 E}{\partial x^2} \nonumber \\
\frac{\partial^3 C}{\partial x^3} &=& \frac{\partial^3 C}{\partial E^3} \left ( \frac{\partial E}{\partial x} \right )^3 + 3 \frac{\partial^2 C}{\partial E^2} \frac{\partial E}{\partial x} \frac{\partial^2 E}{\partial x^2} + \frac{\partial C}{\partial E} \frac{\partial^3 E}{\partial x^3} \nonumber
\end{eqnarray}
These expansions are straightforward. However, when we consider $C_{+}$, we do not get the benefit of being able to
write multiple derivatives in a compact form.
\begin{eqnarray}
\left. \frac{\partial C_{+}}{\partial a} \right |_{a = 0} &=& \frac{1}{2} \frac{\partial C}{\partial E} \frac{\partial E}{\partial x} \nonumber \\
&=& \frac{1}{2} \frac{\partial C}{\partial x} \nonumber \\
\left. \frac{\partial^2 C_{+}}{\partial a^2} \right |_{a = 0} &=& \left. \frac{\partial}{\partial a} \left ( \frac{1}{2} \frac{\partial C_{+}}{\partial E} \frac{\partial E_{+}}{\partial x} \right ) \right |_{a = 0} \nonumber \\
&=& \frac{1}{4} \frac{\partial^2 C}{\partial E^2} \left ( \frac{\partial E}{\partial x} \right )^2 + \frac{1}{2} \frac{\partial C}{\partial E} \frac{\partial^2 E}{\partial x^2} \nonumber \\
&\neq& \frac{1}{4} \frac{\partial^2 C}{\partial x^2} \nonumber \\
\left. \frac{\partial^3 C_{+}}{\partial a^3} \right |_{a = 0} &=& \left. \frac{\partial}{\partial a} \left [ \frac{1}{4} \frac{\partial^2 C_{+}}{\partial E^2} \left ( \frac{\partial E_{+}}{\partial x} \right )^2 + \frac{1}{2} \frac{\partial C_{+}}{\partial E} \frac{\partial^2 E_{+}}{\partial x^2} \right ] \right |_{a = 0} \nonumber \\
&=& \frac{1}{8} \frac{\partial^3 C}{\partial E^3} \left ( \frac{\partial E}{\partial x} \right )^3 + \frac{3}{4} \frac{\partial^2 C}{\partial E^2} \frac{\partial E}{\partial x} \frac{\partial^2 E}{\partial x^2} + \frac{1}{2} \frac{\partial C}{\partial E} \frac{\partial^3 E}{\partial x^3} \nonumber \\
&\neq& \frac{1}{8} \frac{\partial^3 C}{\partial x^3} \label{inconvenient-relations}
\end{eqnarray}
We must be careful to use these relations in the last part of the calculation we are carrying out.
\begin{eqnarray}
\left. \frac{\partial^4 Y_i}{\partial \epsilon \partial a^3} \right |_{\substack{a = 0 \\ \epsilon = 0}} &=& \left. \frac{\partial^3 \tilde{Y}_i}{\partial a^3} \right |_{a = 0} \nonumber \\
&=& -2 \int_{\mathbb{S}^{d - 1}} e_i \rho \left. \frac{\partial}{\partial a^3} \left ( C_{+} \rho_{+} \right ) \right |_{a = 0} \textup{d}S_{\textbf{e}} \nonumber \\
&=& -\frac{2 \omega_d}{d + 2} \left [ C \partial_j \partial_k \partial_l \rho + \frac{3}{2} \partial_j C \partial_k \partial_l \rho  + 3 \partial_j \rho \left ( \frac{1}{4} \frac{\partial^2 C}{\partial E^2} \partial_k E \partial_l E + \frac{1}{2} \frac{\partial^2 C}{\partial E \partial P_m} \partial_k P_m \partial_l E \right. \right. \nonumber \\
&& \left. + \frac{1}{4} \frac{\partial^2 C}{\partial P_m \partial P_n} \partial_k P_m \partial_l P_n + \frac{1}{2} \frac{\partial C}{\partial E} \partial_k \partial_l E + \frac{1}{2} \frac{\partial C}{\partial P_m} \partial_k \partial_l P_m \right ) + \rho \left ( \frac{1}{8} \frac{\partial^3 C}{\partial E^3} \partial_j E \partial_k E \partial_l E \right. \nonumber \\
&& + \frac{3}{8} \frac{\partial^3 C}{\partial E^2 \partial P_m} \partial_j P_m \partial_k E \partial_l E + \frac{3}{8} \frac{\partial^3 C}{\partial E \partial P_m \partial P_n} \partial_j P_m \partial_k P_n \partial_l E + \frac{1}{8} \frac{\partial^3 C}{\partial P_m \partial P_n \partial P_o} \partial_j P_m \partial_k P_n \partial_l P_o \nonumber \\
&& + \frac{3}{4} \frac{\partial^2 C}{\partial E^2} \partial_j \partial_k E \partial_l E + \frac{3}{4} \frac{\partial^2 C}{\partial P_m \partial E} \partial_j P_m \partial_k \partial_l E + \frac{3}{4} \frac{\partial^2 C}{\partial E \partial P_m} \partial_j \partial_k P_m \partial_l E \nonumber \\
&& \left. \left. + \frac{3}{4} \frac{\partial^2 C}{\partial P_m \partial P_n} \partial_j \partial_k P_m \partial_l P_n + \frac{1}{2} \frac{\partial C}{\partial E} \partial_j \partial_k \partial_l E + \frac{1}{2} \frac{\partial C}{\partial P_m} \partial_j \partial_k \partial_l P_m \right ) \right ] \left ( \delta_{ij} \delta_{kl} + \delta_{ik} \delta_{jl} + \delta_{il} \delta_{jk} \right ) \nonumber
\end{eqnarray}
We saw in (\ref{simple-momentum}) that the fourth order term in $\frac{\partial P_i}{\partial t}$ involving
$\epsilon a^3$ is quite simple when $C = 1$. The above shows that it is considerably more messy for a general $C$.
Putting together all of the derivatives that have been computed, the end result is:
\begin{eqnarray}
\frac{\partial E}{\partial t} &=& -\epsilon^2 a^2 \omega_d \partial_i \left ( C \rho^2 \partial_i \frac{\partial \log \rho}{\partial E} \right ) \nonumber \\
\frac{\partial P_i}{\partial t} &=& -\epsilon a \omega_d \partial_i (C \rho^2) - \epsilon^3 a \frac{\omega_d}{d + 2} \partial_l \left [ C \rho^2 \left ( \frac{\partial^2 \log \rho}{\partial E^2} \delta_{jk} + \frac{\partial^2 \log \rho}{\partial P_j \partial P_k} \right ) \right ] \left ( \delta_{ij} \delta_{kl} + \delta_{ik} \delta_{jl} + \delta_{il} \delta_{jk} \right ) \nonumber \\
&& - \epsilon^2 a^2 \frac{\omega_d}{d + 2} \partial_k \left ( C \rho^2 \partial_l \frac{\partial \log \rho}{\partial P_j} \right ) \left ( \delta_{ij} \delta_{kl} + \delta_{ik} \delta_{jl} + \delta_{il} \delta_{jk} \right ) \nonumber \\
&& - \epsilon a^3 \frac{\omega_d}{3(d + 2)} \left [ C \partial_j \partial_k \partial_l \rho + \frac{3}{2} \partial_j C \partial_k \partial_l \rho  + 3 \partial_j \rho \left ( \frac{1}{4} \frac{\partial^2 C}{\partial E^2} \partial_k E \partial_l E + \frac{1}{2} \frac{\partial^2 C}{\partial E \partial P_m} \partial_k P_m \partial_l E \right. \right. \nonumber \\
&& \left. + \frac{1}{4} \frac{\partial^2 C}{\partial P_m \partial P_n} \partial_k P_m \partial_l P_n + \frac{1}{2} \frac{\partial C}{\partial E} \partial_k \partial_l E + \frac{1}{2} \frac{\partial C}{\partial P_m} \partial_k \partial_l P_m \right ) + \rho \left ( \frac{1}{8} \frac{\partial^3 C}{\partial E^3} \partial_j E \partial_k E \partial_l E \right. \nonumber \\
&& + \frac{3}{8} \frac{\partial^3 C}{\partial E^2 \partial P_m} \partial_j P_m \partial_k E \partial_l E + \frac{3}{8} \frac{\partial^3 C}{\partial E \partial P_m \partial P_n} \partial_j P_m \partial_k P_n \partial_l E + \frac{1}{8} \frac{\partial^3 C}{\partial P_m \partial P_n \partial P_o} \partial_j P_m \partial_k P_n \partial_l P_o \nonumber \\
&& + \frac{3}{4} \frac{\partial^2 C}{\partial E^2} \partial_j \partial_k E \partial_l E + \frac{3}{4} \frac{\partial^2 C}{\partial P_m \partial E} \partial_j P_m \partial_k \partial_l E + \frac{3}{4} \frac{\partial^2 C}{\partial E \partial P_m} \partial_j \partial_k P_m \partial_l E + \frac{3}{4} \frac{\partial^2 C}{\partial P_m \partial P_n} \partial_j \partial_k P_m \partial_l P_n \nonumber \\
&& \left. \left. + \frac{1}{2} \frac{\partial C}{\partial E} \partial_j \partial_k \partial_l E + \frac{1}{2} \frac{\partial C}{\partial P_m} \partial_j \partial_k \partial_l P_m \right ) \right ] \left ( \delta_{ij} \delta_{kl} + \delta_{ik} \delta_{jl} + \delta_{il} \delta_{jk} \right ) \; . \label{fourth-order-ep}
\end{eqnarray}
These equations take up a lot of space and that is all because of the last term.

\subsubsection{Static configurations}
For a general density of states $\rho(E, \textbf{P})$ and function $C(E, \textbf{P})$, (\ref{fourth-order-ep}) is only
static when energy and momentum are both uniform. If only one of them is, derivatives of $\rho$ and $C$ will not
vanish and the equations will introduce non-uniformities. If the momentum is initially zero everywhere, we can compare
the energy equation just derived with the one from our first model:
\begin{eqnarray}
\frac{\partial E}{\partial t} &=& -\epsilon^2 a^2 \omega_d \partial_i \left ( C(E, 0) \rho^2(E, 0) \partial_i \frac{\partial \log \rho(E, 0)}{\partial E} \right ) \nonumber \\
\frac{\partial E}{\partial t} &=& -\epsilon^2 a^2 \omega_d \partial_i \left ( C(E) \rho^2(E) \partial_i \frac{\textup{d} \log \rho(E)}{\textup{d} E} \right ) \; . \nonumber
\end{eqnarray}
The unrestricted density of states $\rho(E)$ should be qualitatively similar to the zero momentum density of states
$\rho(E, 0)$. Therefore these models predict similar diffusive behaviour at early times. However, this diffusion
generates momentum which can cause very different behaviour to occur at late times once the non-uniformities become
significant. For example, a Hagedorn density of states with $\textbf{P}(\textbf{x}, 0) = 0$ requires
$\frac{\partial E}{\partial t}$ to be zero (indicating slow diffusion) at $t = 0$ but not necessarily later times.
This supports the idea that our first model overestimates the amount of time taken for a cluster of energy to diffuse.

Before, we saw that the dynamics were frozen for any energy distribution when $\rho(E)$ was a Hagedorn density of
states. If such a density of states were to exist for this system of PDEs, $\log \rho(E, \textbf{P})$ would have to be
linear in both $E$ and $\textbf{P}$. Of the five terms in (\ref{fourth-order-ep}), this form causes the three involving
logarithms to vanish. What about the term in momentum proportional to $\epsilon a$? This will not vanish unless we
choose a special $C$ function as well. The one to choose is $C \propto \rho^{-2}$. It is not immediately obvious, but
this choice causes the one remaining term (the one proportional to $\epsilon a^3$) to vanish as well. This is because
\begin{eqnarray}
\left. \frac{\partial^4 Y_i}{\partial \epsilon \partial a^3} \right |_{\substack{a = 0 \\ \epsilon = 0}} &=& -2 \int_{\mathbb{S}^{d - 1}} e_i \left. \frac{\partial}{\partial a^3} \left ( C_{+} \rho_{+} \rho \right ) \right |_{a = 0} \textup{d}S_{\textbf{e}} \nonumber \\
&=& -2 \int_{\mathbb{S}^{d - 1}} e_i \left. \frac{\partial}{\partial a^3} \left ( e^{-2 \left ( \frac{aE + b_j P_j}{2} + \frac{\beta_{\mathrm{H}} E_{+} + \beta_{\mathrm{H}} v_j P_{j+}}{2} \right )} e^{\beta_{\mathrm{H}} E_{+} + \beta_{\mathrm{H}} v_j P_{j+}} e^{\beta_{\mathrm{H}} E + \beta_{\mathrm{H}} v_j P_j} \right ) \right |_{a = 0} \nonumber \\
&=& -2 \int_{\mathbb{S}^{d - 1}} e_i \left. \frac{\partial}{\partial a^3} 1 \right |_{a = 0} \textup{d}S_{\textbf{e}} \nonumber \\
&=& 0 \; . \nonumber
\end{eqnarray}
Therefore $C \propto \rho^{-2}$ and $C \propto 1$ appear to be natural choices once again. We will see that the
part of our equation that looks complicated for a more general $C$ simplifies considerably if we take the linearization.

\subsection{Consistency check}
Hydrodynamics is concerned with fluctuations of a system around some equilibrium state. This system is often a field
theory like the one being modelled by (\ref{fourth-order-ep}). An encouraging result we have seen already is the
solution at order $\epsilon a$ in terms of left and right moving waves. These are exact solutions to hydrodynamics for
a CFT in two spacetime dimensions. It will not be possible to compare solutions in the other cases of interest so we
will focus on matching coefficients in the linearized equations.

\subsubsection{Linearized equations}
We will use $(E_0, \textbf{P}_0)$ as our equilibrium state and insert
\begin{eqnarray}
E(\textbf{x}, t) &=& E_0 + \tilde{E}(\textbf{x}, t) \nonumber \\
\textbf{P}(\textbf{x}, t) &=& \textbf{P}_0 + \tilde{\textbf{P}}(\textbf{x}, t) \nonumber
\end{eqnarray}
into our equations of motion leaving only one power of the perturbations $(\tilde{E}, \tilde{\textbf{P}})$.
Working with the term inside the derivatives in our energy equation (\ref{fourth-order-ep}), we may write
\begin{eqnarray}
\frac{\partial \log \rho (E, \textbf{P})}{\partial E} &\approx& \frac{\partial \log \rho (E_0, \textbf{P}_0)}{\partial E} + \tilde{E} \frac{\partial^2 \log \rho (E_0, \textbf{P}_0)}{\partial E^2} + \tilde{P}_m \frac{\partial^2 \log \rho (E_0, \textbf{P}_0)}{\partial E \partial P_m} \nonumber \\
&=& \frac{\partial \log \rho_0}{\partial E} + \tilde{E} \frac{\partial^2 \log \rho_0}{\partial E^2} + \tilde{P}_m \frac{\partial^2 \log \rho_0}{\partial E \partial P_m} \nonumber
\end{eqnarray}
where we have used the subscript $0$ to denote evaluation at $(E_0, \textbf{P}_0)$. Taking one derivative kills the
constant so we have:
\begin{eqnarray}
C \rho^2 \partial_i \frac{\partial \log \rho}{\partial E} &\approx& C \rho^2 \partial_i \left [ \frac{\partial \log \rho_0}{\partial E} + \tilde{E} \frac{\partial^2 \log \rho_0}{\partial E^2} + \tilde{P}_m \frac{\partial^2 \log \rho_0}{\partial E \partial P_m} \right ] \nonumber \\
&=& C \rho^2 \left ( \frac{\partial^2 \log \rho_0}{\partial E^2} \partial_i \tilde{E} + \frac{\partial^2 \log \rho_0}{\partial E \partial P_m} \partial_i \tilde{P}_m \right ) \nonumber \\
&\approx& C_0 \rho^2_0 \left ( \frac{\partial^2 \log \rho_0}{\partial E^2} \partial_i \tilde{E} + \frac{\partial^2 \log \rho_0}{\partial E \partial P_m} \partial_i \tilde{P}_m \right ) \; . \nonumber
\end{eqnarray}
When we use this type of logic on the momentum equation in (\ref{fourth-order-ep}), a nice thing happens. Only a few
terms in the long $\epsilon a^3$ contribution do not have derivatives of $\tilde{E}$ and $\tilde{\textbf{P}}$
multiplied together: the one with $\partial_j \partial_k \partial_l \rho$, the one with
$\partial_j \partial_k \partial_l E$ and the one with $\partial_j \partial_k \partial_l P_m$.
Carrying out the straightforward linearization, we see that (\ref{fourth-order-ep}) becomes:
\begin{eqnarray}
\frac{\partial \tilde{E}}{\partial t} &=& -\epsilon^2 a^2 \omega_d C_0 \rho^2_0 \left ( \frac{\partial^2 \log \rho_0}{\partial E^2} \partial_i \partial_i \tilde{E} + \frac{\partial^2 \log \rho_0}{\partial E \partial P_m} \partial_i \partial_i \tilde{P}_m \right ) \nonumber \\
\frac{\partial \tilde{P}_i}{\partial t} &=& -\epsilon a \omega_d \rho^2_0 \left [ \left ( \frac{\partial C_0}{\partial E} + 2 C_0 \frac{\partial \log \rho_0}{\partial E} \right ) \partial_i \tilde{E} + \left ( \frac{\partial C_0}{\partial P_m} + 2 C_0 \frac{\partial \log \rho_0}{\partial P_m} \right ) \partial_i \tilde{P}_m \right ] \nonumber \\
&& - \epsilon^3 a \frac{\omega_d}{d + 2} \left [ \left [ C_0 \rho^2_0 \left ( \frac{\partial^3 \log \rho_0}{\partial E^3} \delta_{jk} + \frac{\partial^3 \log \rho_0}{\partial E \partial P_j \partial P_k} \right ) \right. \right. \nonumber \\
&& \left. + \rho^2_0 \left ( \frac{\partial^2 \log \rho_0}{\partial E^2} \delta_{jk} + \frac{\partial^2 \log \rho_0}{\partial P_j \partial P_k} \right ) \left ( \frac{\partial C_0}{\partial E} + 2 C_0 \frac{\partial \log \rho_0}{\partial E} \right ) \right ] \partial_l \tilde{E} \nonumber \\
&& + \left [ C_0 \rho^2_0 \left ( \frac{\partial^3 \log \rho_0}{\partial E^2 \partial P_m} \delta_{jk} + \frac{\partial^3 \log \rho_0}{\partial P_j \partial P_k \partial P_m} \right ) \right. \nonumber \\
&& \left. \left. + \rho^2_0 \left ( \frac{\partial^2 \log \rho_0}{\partial E^2} \delta_{jk} + \frac{\partial^2 \log \rho_0}{\partial P_j \partial P_k} \right ) \left ( \frac{\partial C_0}{\partial P_m} + 2 C_0 \frac{\partial \log \rho_0}{\partial P_m} \right ) \right ] \partial_l \tilde{P}_m \right ] \left ( \delta_{ij} \delta_{kl} + \delta_{ik} \delta_{jl} + \delta_{il} \delta_{jk} \right ) \nonumber \\
&& - \epsilon^2 a^2 \frac{\omega_d}{d + 2} C_0 \rho^2_0 \left [ \frac{\partial^2 \log \rho_0}{\partial P_j \partial E} \partial_k \partial_l \tilde{E} + \frac{\partial^2 \log \rho_0}{\partial P_j \partial P_m} \partial_k \partial_l \tilde{P}_m \right ] \left ( \delta_{ij} \delta_{kl} + \delta_{ik} \delta_{jl} + \delta_{il} \delta_{jk} \right ) \label{linearized-ep} \\
&& - \epsilon a^3 \frac{\omega_d}{2(d + 2)} \rho^2_0 \left [ \left ( \frac{\partial C_0}{\partial E} + 2 C_0 \frac{\partial \log \rho_0}{\partial E} \right ) \partial_i \partial_j \partial_j \tilde{E} + \left ( \frac{\partial C_0}{\partial P_m} + 2 C_0 \frac{\partial \log \rho_0}{\partial P_m} \right ) \partial_i \partial_j \partial_j \tilde{P}_m \right ] \; . \nonumber
\end{eqnarray}

These equations are still rather long. The $\epsilon a^3$ term is now only one line but the linearization has made the
$\epsilon^3 a$ term expand. The situation can be improved if we assume that the momentum $\textbf{P}_0$ around which we
linearize is not only a constant function but a constant close to zero. This is reminiscent of a common assumption in
hydrodynamics where the fluid velocity must be much less than the speed of sound to yield propagating hydrodynamic
modes \cite{kovtun}. We also used this assumption previously when we linearized hydrodynamics. This allows us to remove
all coefficients above that involve an odd number of
derivatives with respect to components of $\textbf{P}$. After all, functions like $C$ and $\log \rho$ are spherically
symmetric in $\textbf{P}$, so their odd order derivatives vanish at the origin. Moreover, second derivatives become
proportional to $\delta_{jk}$:
\begin{eqnarray}
\frac{\partial^2 \log \rho}{\partial P_j \partial P_k} &=& \frac{\partial}{\partial P_j} \left ( \frac{\partial \log \rho}{\partial |\textbf{P}|} \frac{P_k}{|\textbf{P}|} \right ) \nonumber \\
&=& \frac{\partial^2 \log \rho}{\partial |\textbf{P}|^2} \frac{P_j P_k}{|\textbf{P}|^2} + \frac{\partial \log \rho}{\partial |\textbf{P}|} \frac{\partial}{\partial P_j} \left ( \frac{P_k}{|\textbf{P}|} \right ) \nonumber \\
&=& \frac{\partial^2 \log \rho}{\partial |\textbf{P}|^2} \frac{P_j P_k}{|\textbf{P}|^2} + \frac{\partial \log \rho}{\partial |\textbf{P}|} \frac{|\textbf{P}| \delta_{jk} - \frac{P_j P_k}{|\textbf{P}|}}{|\textbf{P}|^2} \nonumber \\
&=& \frac{\partial \log \rho}{\partial |\textbf{P}|} \frac{\delta_{jk}}{|\textbf{P}|} + \frac{P_j P_k}{|\textbf{P}|^2} \left ( \frac{\partial^2 \log \rho}{\partial |\textbf{P}|^2} - \frac{1}{|\textbf{P}|} \frac{\partial \log \rho}{\partial |\textbf{P}|} \right ) \nonumber \\
&\approx& \frac{\partial \log \rho}{\partial |\textbf{P}|} \frac{\delta_{jk}}{|\textbf{P}|} \; . \nonumber
\end{eqnarray}
If we substitute this into (\ref{linearized-ep}), we arrive at:
\begin{eqnarray}
\frac{\partial \tilde{E}}{\partial t} &=& -\epsilon^2 a^2 \omega_d C_0 \rho^2_0 \frac{\partial^2 \log \rho_0}{\partial E^2} \partial_i \partial_i \tilde{E} \nonumber \\
\frac{\partial \tilde{P}_i}{\partial t} &=& -\epsilon a \omega_d \rho^2_0 \left ( \frac{\partial C_0}{\partial E} + 2 C_0 \frac{\partial \log \rho_0}{\partial E} \right ) \partial_i \tilde{E} \nonumber \\
&& - \epsilon^3 a \omega_d \left [ C_0 \rho^2_0 \left ( \frac{\partial^3 \log \rho_0}{\partial E^3} + \frac{1}{|\textbf{P}_0|} \frac{\partial^2 \log \rho}{\partial E \partial |\textbf{P}|} \right ) + \rho^2_0 \left ( \frac{\partial C_0}{\partial E} + 2 C_0 \frac{\partial \log \rho_0}{\partial E} \right ) \right ] \partial_i \tilde{E} \nonumber \\
&& - \epsilon^2 a^2 \frac{\omega_d}{d + 2} C_0 \rho^2_0 \frac{1}{|\textbf{P}_0|} \frac{\partial \log \rho_0}{\partial |\textbf{P}|} \left ( \partial_j \partial_j \tilde{P}_i + 2 \partial_i \partial_j \tilde{P}_j \right ) \nonumber \\
&& - \epsilon a^3 \frac{\omega_d}{2(d + 2)} \rho^2_0 \left ( \frac{\partial C_0}{\partial E} + 2 C_0 \frac{\partial \log \rho_0}{\partial E} \right ) \partial_i \partial_j \partial_j \tilde{E} \; . \label{small-linearized-ep}
\end{eqnarray}

\subsubsection{Comparison with hydrodynamics}
We are now in a position to compare (\ref{small-linearized-ep}) to hydrodynamics. Since our equations have second and
third derivatives, we should not attempt to relate them to ideal hydrodynamics. In fact, ideal hydrodynamics has a
conserved entropy and one of the fundamental assumptions in our model was that the entropy was driven to increase
\cite{bhattacharya}. We will have to use a stress-energy tensor that includes derivatives --- a situation we referred
to as dissipative hydrodynamics.

In the Landau frame, dissipative hydrodynamics came from (\ref{dissipative-constituents}). For our purposes, it is
enough that these equations have two derivatives even though (\ref{small-linearized-ep}) has three. For a proper
comparison, we need the linearization of the hydro equations. An alternative expression for the system
(\ref{continuity-equation}) and (\ref{stress-equation}) with no charge is
\begin{eqnarray}
&& \frac{\partial \tilde{\varepsilon}}{\partial t} + ik \tilde{P}_{\parallel} = 0 \nonumber \\
&& \frac{\partial \tilde{P}_{\parallel}}{\partial t} + ik \frac{\partial P_0}{\partial \varepsilon} \tilde{\varepsilon} + \gamma_{\textup{s}} k^2 \tilde{P}_{\parallel} = 0 \nonumber \\
&& \frac{\partial \tilde{P}_{\perp}}{\partial t} + \gamma_{\eta} k^2 \tilde{P}_{\perp} = 0 \nonumber \; .
\end{eqnarray}
In this form, we have implicitly Fourier transformed the fields and decomposed the momentum into a part
parallel to $\textbf{k}$ and a part perpendicular to $\textbf{k}$. The coefficients
\begin{eqnarray}
\gamma_{\textup{s}} &=& \frac{\frac{d - 2}{d} \eta_0 + \zeta_0}{\varepsilon_0 + P_0} \nonumber \\
\gamma_{\eta} &=& \frac{\eta_0}{\varepsilon_0 + P_0} \nonumber
\end{eqnarray}
have been defined as in \cite{kovtun}. The most obvious difference we see is that the energy decouples in
(\ref{small-linearized-ep}) instead of satisfying a continuity equation. If a model with only energy is governed
by the heat equation, one hopes that this is an effective description of a continuity equation with Fick's law:
$P_i \propto \partial_i E$. The linearization of (\ref{fourth-order-ep}) has shown that this is not the case; the
heat equation for energy is still an explicit part of our model.

One way around this is to consider an incompressible fluid: $\partial_i \tilde{P}_i = 0$. In hydrodynamics this is
equivalent to the approximation that the energy is constant in time. For this to be true in the model
(\ref{small-linearized-ep}), we need an energy that is already uniform so that only momentum is flowing. The simplest
comparison we can make is between the incompressible hydro equations and a pure momentum version of our model.
Since this is the $\epsilon^2 a^2$ term in (\ref{small-linearized-ep}), the relations
\begin{eqnarray}
\gamma_{\eta} &=& \epsilon^2 a^2 \frac{\omega_d}{d + 2} C_0 \rho^2_0 \frac{1}{|\textbf{P}_0|} \frac{\partial \log \rho_0}{\partial |\textbf{P}|} \nonumber \\
\gamma_{\mathrm{s}} &=& 2 \gamma_{\eta} \label{transport-coefficients}
\end{eqnarray}
are produced.  This implies $\zeta = \frac{4 - d}{2d} \eta$, which is much different from the results of
\cite{policastro}. For a conformal fluid in any number of dimensions, tracelessness of the stress-energy tensor
demands $\zeta = 0$. This also differs from the results of \cite{takeuchi} which considered the hydrodynamics of
a non-conformal theory of holographic QCD. One would have to contend with this problem even if she found a way of
going beyond incompressible hydrodynamics (\textit{e.g.} introducing auxiliary conserved currents to cancel the
problematic terms in (\ref{small-linearized-ep})).

These problems (no continuity equation and the coefficients of $\partial_j \partial_j P_i$ and
$\partial_i \partial_j P_j$ not being independent) suggest that our entropic model is fundamentally incompatible
with the long distance effective description that is hydrodynamics. Some evidence for this can be seen in our
expression for the mean-field variance (\ref{need-high-energy}). To have energy variances grow more slowly than squared
energies, our $E$ in the denominator had to be large. Similarly, requiring $|\textbf{P}|$ to be large is the most
obvious way of ensuring that the growth of the momentum covariance matrix is small. While not necessarily incompatible
with hydrodynamics, this is certainly incompatible with the linearization of it.

\subsubsection{One more simulation}
\begin{figure}[h]
\centering
\subfloat[][Energy]{\includegraphics[scale=0.4]{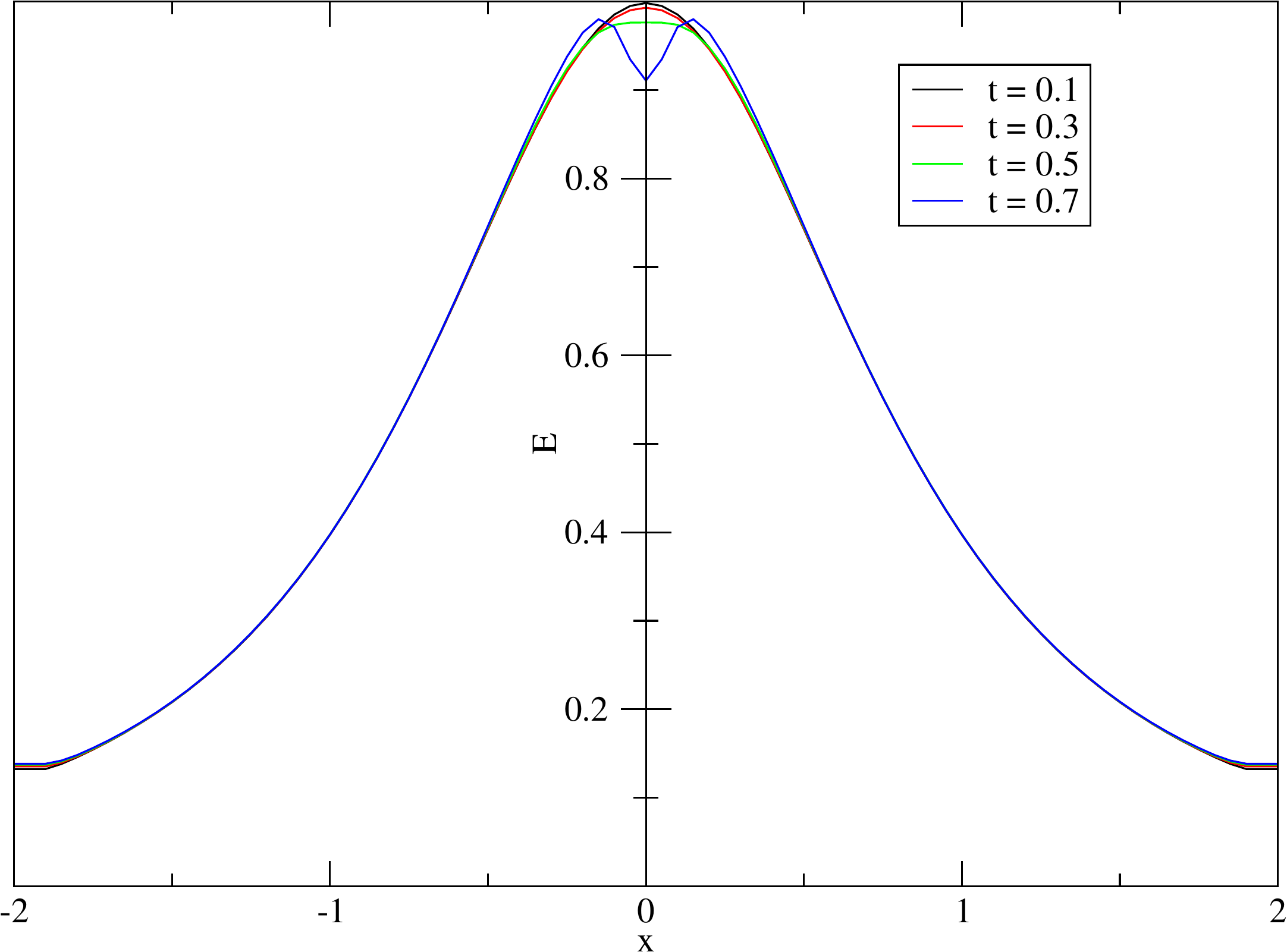}}
\subfloat[][Momentum]{\includegraphics[scale=0.4]{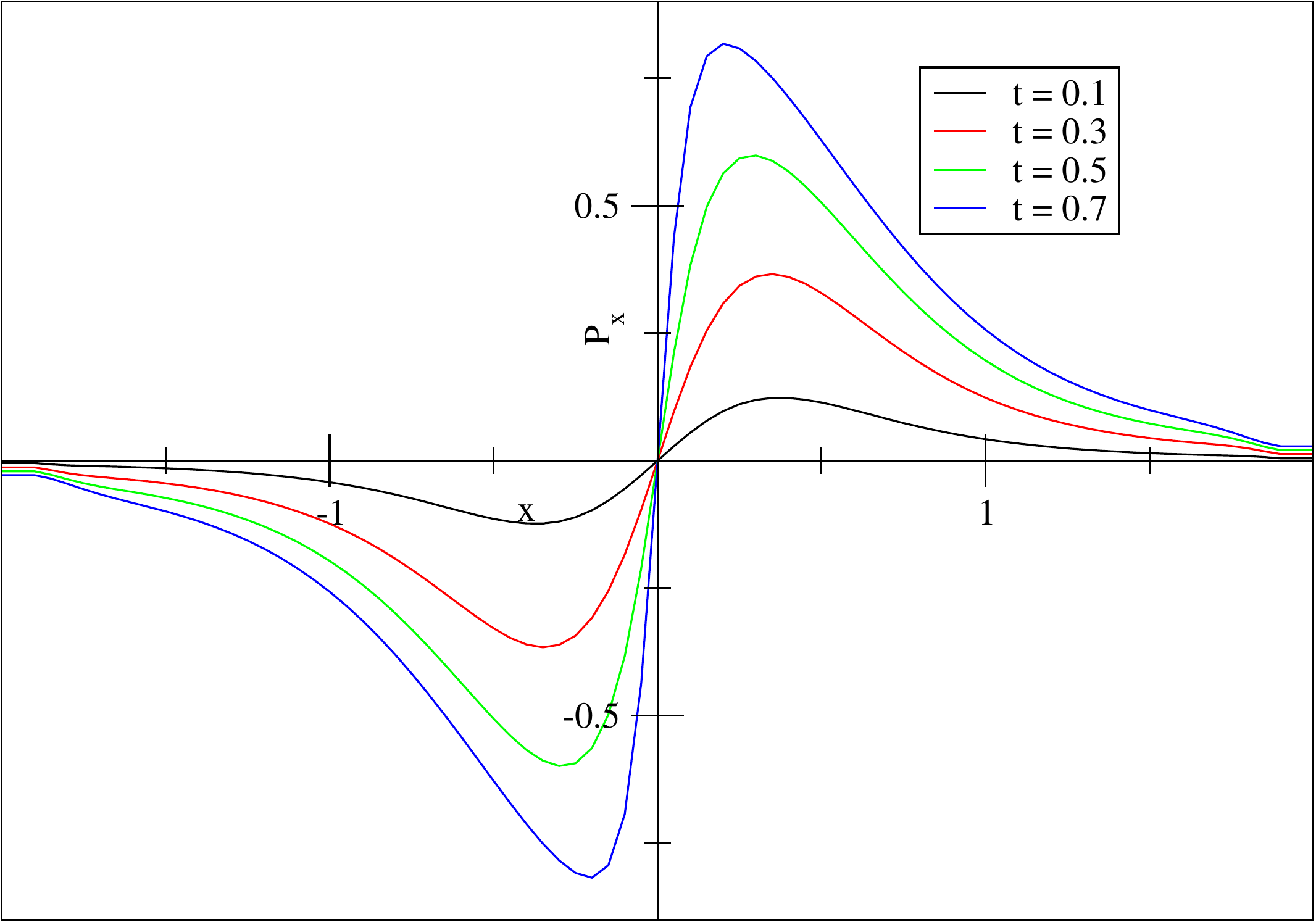}}
\caption{Energy and momentum profiles for the simulation that only has an $x$-axis. They are shown up until the time
when momentum values of $\pm 0.9$ form.}
\label{momentum-simulation-1d}
\end{figure}
The most fundamental problem with our first model is that it only had well defined decay times in one dimension. The
only way to make this model sensible is to remove the low energy phase from the density of states, allowing
(\ref{need-high-energy}) to be strictly obeyed. This has the effect of removing decay times altogether leading to
infinitely long lived black holes. A pressing question is whether the decay times of (\ref{fourth-order-ep}) remain
comparable when moving from one dimension to the next.

Since instantaneous extinction is a low energy effect, it suffices to use much easier initial conditions than the ones
in Figure \ref{no-exaggeration}. Its presence is also agnostic to whether the field theory is strongly coupled or
weakly coupled. Therefore, we may use the restricted density of states expressions (\ref{rdos-free}) that are valid for
free fields. Unfortunately, a numerical investigation of (\ref{fourth-order-ep}) shows that momenta grow very quickly
even if they start from zero. This leads to a time scale for satisfying
$|| \textbf{P} ||_{\infty} \approx || E ||_{\infty}$. After this time, momenta at various points are comparable to the
energies at those points and (\ref{rdos-free}) is no longer valid. This time scale is shorter than the black hole
decay time and probably also shorter than the black hole thermalization time that we defined earlier. Nevertheless, we
present results showing that this ``momentum generation time'' is similar in $d = 1$ and $d = 2$. There may be
arguments involving the other time scales that follow from this.

Figure \ref{momentum-simulation-1d} shows the results for $d = 1$. In this case, the sites only have directions that
are aligned with the lattice so the $d = 1$ restricted density of states is used. Choices made for
(\ref{fourth-order-ep}) are $\epsilon = a = 0.1$ and $C = 1$. Since we are only interested in short time dynamics, the
simplest forward difference Euler method has been used.
\begin{figure}[h]
\centering
\subfloat[][Energy]{\includegraphics[scale=0.4]{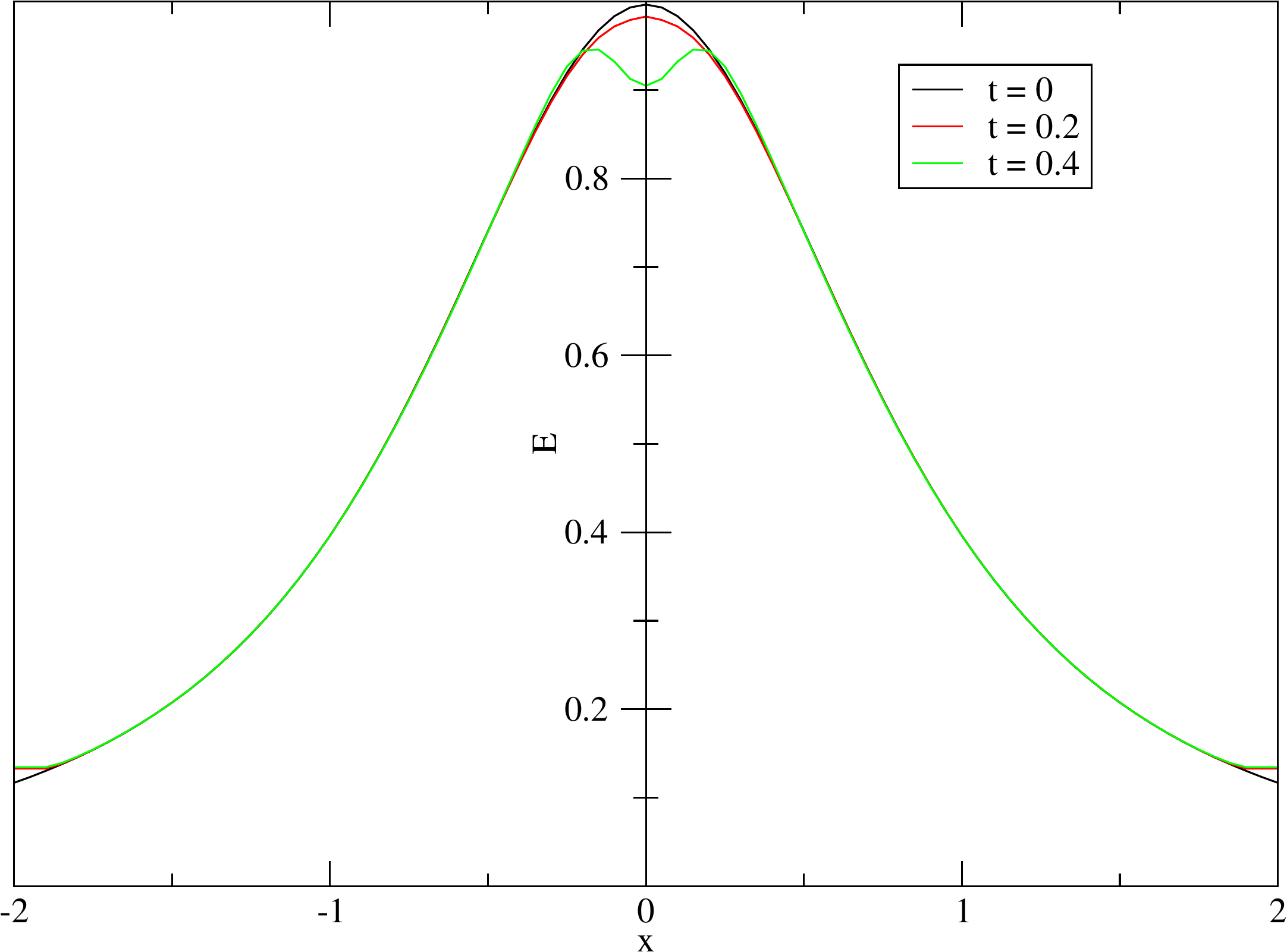}}
\subfloat[][Momentum]{\includegraphics[scale=0.4]{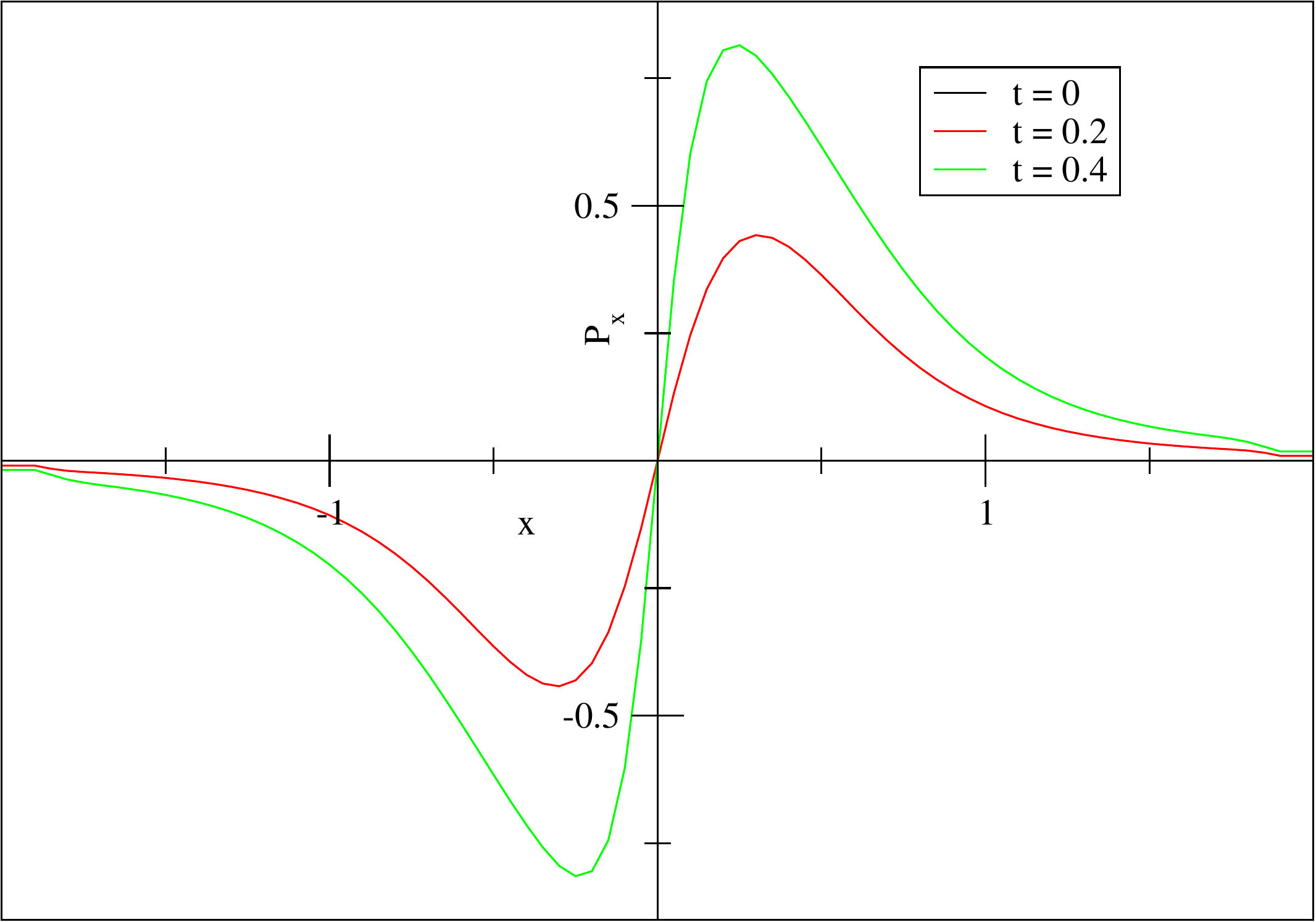}}
\caption{On the left are slices of the $E$ distributions along the $x$-axis. These are spherically symmetric. On
the right are $P_x$ distributions along the $x$-axis. These smoothly approach a $y$-axis value of zero as one rotates
the direction along which they are plotted. The $P_y$ distributions behave in an equal and opposite way.}
\label{momentum-simulation-2d}
\end{figure}
It is clear that these nonlinear equations do not satisfy the maximum principle. The centre of the distribution splits
into left and right moving waves instead. The $d = 2$ simulation has also been done with no extra small dimensions
leading to the $d = 2$ formula (\ref{rdos-free}). Even though the momentum generation time is shorter in this case,
Figure \ref{momentum-simulation-2d} shows that it still has the same order of magnitude.

If more accurate expressions for the restricted density of states were to be found, it would be interesting to return
to the study of decay times for (\ref{fourth-order-ep}) using a better numerical method. In addition to being more
trustworthy than the estimates (\ref{timescale-1}) and (\ref{timescale-2}), these decay times are also likely to be
shorter and therefore easier to find numerically.

\section{Conclusion}
Motivated by the universality of black hole physics in a variety of holographic models, we studied general properties
of field theories where excitations move throughout a lattice purely due to statistical noise. This led to two
interesting models: one in which only energy is dynamical and one in which momentum is as well. Many shortcomings of
the first model appeared to be resolved by the second when the most basic calculations were carried out. A more
quantitative comparison between them is likely to be difficult because both models require one to deal with nonlinear
partial differential equations.

To elabourate on our first model (\ref{main-pde}), a so-called filtration equation, we found that it made sense for
high energies and predicted frozen dynamics for systems that are governed by a Hagedorn density of states \cite{mvr}.
As input to the PDE, we chose a filtration function based on the thermodynamics of $\mathcal{N} = 4$ Super Yang-Mills,
which have been famously explored using the AdS / CFT correspondence \cite{maldacena, magoo}. Additional references
to the correspondence have been made throughout the thesis, since we likened the solutions of (\ref{main-pde}) to
the plasma balls in confining gauge theories that were discovered numerically in 2005 \cite{aharony}. To simplify the
analysis of our PDE, we removed the ``small black hole'' phase of SYM which is responsible for a first order phase
transition in the original background for holography. Because we still saw solutions that were similar to plasma balls,
our results suggest one of two things. Either the assumptions about the phase transitions in \cite{aharony} can be
relaxed or it is true that first order phase transitions in large field theories emerge whenever small field theories
with a Hagedorn phase are assembled on a lattice.

A detailed analysis of the model (\ref{main-pde}) revealed several problems. One is the aforementioned difficulty
of simulating a phase that has a convex microcanonical entropy. The resulting equation is unstable for the same
reason as the reverse heat equation. Even though the reverse heat equation cannot be simulated starting with only
Cauchy data, the situation can be improved if further constraints on the solutions are imposed. A very physical
one is positivity \cite{john}. Reverse heat equation methods developed recently \cite{fu, wang} have the potential to
fill in the missing phase in our numerics but this effort might not be justified in view of the other problems. To see
this, the decay times (\ref{timescale-1}) and (\ref{timescale-2}) derived for our model have upper and lower bounds
proportional to $E_{\mathrm{F}}^2$ where $E_{\mathrm{F}}$ is the critical energy of the plasma ball. In a large $N$
gauge theory, such an energy is proportional to $N^2$ leading to a predicted decay time of $O(N^4)$. This conflicts
with the $O(N^2)$ prediction of \cite{aharony}. Another problem with these time scales is that they contain extreme
prefactors: $E_{\mathrm{min}}$ which is very small and $E_{\mathrm{min}}^{\alpha - 1}$ which is very large. It is not
possible to fix these values due to a vicious cycle that plagues any realistic simulation. The requirement of a small
Hagedorn energy $E_{\mathrm{H}} \ll E_{\mathrm{F}}$ forces the domain kept in the simulation to be very large ---
large enough to accommodate the entire mass of the initial condition below $E_{\mathrm{H}}$. The large domain
forces any distribution with power law tails to have a very small $E_{\mathrm{min}}$. In one dimension, our inability
to find more restrictive prefactors can be viewed as a purely mathematical problem.

In dimension two and higher, this is not the case. Our numerics and the theory of Barenblatt profiles agreed on the
behaviour of the plasma ball decay time. It hugs the lower bound in (\ref{timescale-1}) and (\ref{timescale-2})
leading to a decay time of zero in the infinite volume limit. The source of this pathology (a diverging diffusion
constant) is clear but again not easy to fix. The high energy parts of energy distributions, for which our model is
valid, have their dynamics contaminated by the low energy parts which necessarily appear in the same distributions.
This serious problem was the main motivation for our second model (\ref{fourth-order-ep}).

The second model consisted of PDEs for energy and momentum that resulted in faster nonlinear diffusion, possibly
consistent with the $O(N^2)$ plasma ball prediction. This system also appeared not to suffer from instantaneous
extinction but the only check that was possible with out current understanding was limited to very short times.
Exploring this nonlinear model further requires one to derive a momentum restricted density of states that is valid for
small differences between $E$ and $| \textbf{P} |$. In the linear regime, we attempted to strengthen the model's
connection with holography by comparing the limit (\ref{small-linearized-ep}) to the equations for hydrodynamics. The
equations were found to have only superficial similarities and not offer meaningul predictions for transport
coefficients. As it makes no explicit reference to a Lagrangian, our model seems to apply equally well to strongly
coupled and weakly coupled field theories. Conversely, the procedure for calculating transport coefficients at weak
coupling \cite{amy1, amy2, york} is more difficult to carry out than the method revealed by the fluid / gravity
correspondence \cite{hubeny}. Any future method claiming to fix this aspect of our model would be worthwhile to
pursue, but it may be that the assumptions of (\ref{fourth-order-ep}) are fundamentally incompatible with those of
hydrodynamics.

Using two widely applicable stochastic models and some basic input about the Super Yang-Mills theory, were were able to
show that long lifetimes of black holes are associated with the Hagedorn density of states of a string worldsheet.
Despite our calculation being insufficient in practice, we showed that it is possible to study nonlinear diffusion,
both analytically and numerically as a means of predicting the time required for these black holes to evaporate through
Hawking radiation. Our second model in particular showed promise as the number of large spatial dimensions was not
restricted to one. Future studies can focus on at least three broad areas: improving the restricted density of states,
making the model more consistent with hydrodynamics and finding high energy phenomena other than black holes that are
common to the many known examples of holography. Even though our model is by no means safe from being replaced by other
effective theories that may describe holography more accurately, it is refreshing to see existence and uniqueness play
important roles in a situation where they are so often taken for granted.

\newpage
\bibliographystyle{unsrt}
\bibliography{references}

\begin{thebibliography}{100}

\bibitem{maldacena}
J.~Maldacena.
\newblock The large {N} limit of superconformal field theories and
  supergravity.
\newblock {\em Advances in Theoretical and Mathematical Physics}, 2:231--252,
  1998.
\newblock arXiv:hep-th/9711200.

\bibitem{mvr}
C.~Behan; K. Larjo; N. Lashkari; B. Swingle; M.~Van Raamsdonk.
\newblock Energy trapping from {H}agedorn densities of states.
\newblock {\em Journal of High Energy Physics}, 2013(10), 2013.
\newblock arXiv:1304.7275.

\bibitem{jacobson}
T.~Jacobson.
\newblock Thermodynamics of spacetime: {T}he {E}instein equation of state.
\newblock {\em Physical Review Letters}, 75(7):1260--1263, 1995.
\newblock arXiv:gr-qc/9504004.

\bibitem{verlinde1}
E.~Verlinde.
\newblock On the origin of gravity and the laws of {N}ewton.
\newblock {\em Journal of High Energy Physics}, 2011(4), 2011.
\newblock arXiv:1001.0785.

\bibitem{lashkari}
N.~Lashkari; M. B. McDermott; M.~Van Raamsdonk.
\newblock Gravitational dynamics from entanglement ``thermodynamics''.
\newblock {\em Journal of High Energy Physics}, 2014(4), 2014.
\newblock arXiv:1308.3716.

\bibitem{faulkner}
T.~Faulkner; M. Guica; T. Hartman; R. C. Myers; M.~Van Raamsdonk.
\newblock Gravitation from entanglement in holographic {CFT}s.
\newblock 2013.
\newblock arXiv:1312.7856.

\bibitem{aharony}
O.~Aharony; S. Minwalla;~T. Wiseman.
\newblock Plasma balls in large {N} gauge theories and localized black holes.
\newblock {\em Classical and Quantum Gravity}, 23(7), 2006.
\newblock arXiv:hep-th/0507219.

\bibitem{magoo}
O.~Aharony; S. S. Gubser; J. Maldacena; H. Ooguri;~Y. Oz.
\newblock Large {N} field theories, string theory and gravity.
\newblock {\em Physics Reports}, 323, 2000.
\newblock arXiv:hep-th/9905111.

\bibitem{lt1}
D.~A. Lowe;~L. Thorlacius.
\newblock Ad{S} / {CFT} and the information paradox.
\newblock {\em Physical Review D}, 60(10), 1999.
\newblock arXiv:hep-th/9903237.

\bibitem{lt2}
D.~A. Lowe;~L. Thorlacius.
\newblock Comments on the black hole information problem.
\newblock {\em Physical Review D}, 73(10), 2006.
\newblock arXiv:hep-th/0601059.

\bibitem{horowitz}
G.~T. Horowitz;~J. Maldacena.
\newblock The black hole final state.
\newblock {\em Journal of High Energy Physics}, 2004(2), 2004.
\newblock arXiv:hep-th/0310281.

\bibitem{marolf}
V.~E. Hubeny; D. Marolf;~M. Rangamani.
\newblock Hawking radiation from {A}d{S} black holes.
\newblock 2009.
\newblock arXiv:0911.4144.

\bibitem{polchinski}
J.~Polchinski.
\newblock {\em String theory}.
\newblock Cambridge Monographs on Mathematical Physics. Cambridge University
  Press, 1998.

\bibitem{abjm}
O.~Aharony; O. Bergman; D. L. Jafferis;~J. Maldacena.
\newblock N = 6 superconformal {C}hern-{S}imons-matter theories, {M}2-branes
  and their gravity duals.
\newblock {\em Journal of High Energy Physics}, 2008(10), 2008.
\newblock arXiv:0806.1218.

\bibitem{lunin}
O.~Lunin; J. Maldacena;~L. Maoz.
\newblock Gravity solutions for the {D}1-{D}5 system with angular momentum.
\newblock 2002.
\newblock arXiv:hep-th/0212210.

\bibitem{tong1}
D.~Tong.
\newblock The holographic dual of {A}d{S}3 x {S}3 x {S}3 x {S}1.
\newblock 2014.
\newblock arXiv:1402.5135.

\bibitem{behan}
C.~Behan.
\newblock Density of states in a free {CFT} and finite volume corrections.
\newblock {\em Physical Review D}, 88(2), 2013.
\newblock arXiv:1210.5655.

\bibitem{loran}
F.~Loran; M. M. Sheikh-Jabbari;~M. Vincon.
\newblock Beyond logarithmic corrections to {C}ardy formula.
\newblock {\em Journal of High Energy Physics}, 2011(1), 2011.
\newblock arXiv:1010.3561.

\bibitem{bbs}
M.~Becker; K. Becker; J.~H. Schwarz.
\newblock {\em String theory and {M}-theory: {A} modern introduction}.
\newblock Cambridge University Press, 2007.

\bibitem{wald}
R.~M. Wald.
\newblock The thermodynamics of black holes.
\newblock {\em Living Reviews in Relativity}, 4(6), 2001.

\bibitem{emparan}
R.~Emparan; H.~S. Reall.
\newblock Black holes in higher dimensions.
\newblock {\em Living Reviews in Relativity}, 11(6), 2008.

\bibitem{adm}
R.~Arnowitt; S. Deser; C.~W. Misner.
\newblock Dynamical structure and definition of energy in general relativity.
\newblock {\em Physical Review}, 116(3):1322--1330, 1959.

\bibitem{ashtekar}
A.~Ashtekar;~S. Das.
\newblock Asymptotically anti-de {S}itter space-times: {C}onserved quantities.
\newblock {\em Classical and Quantum Gravity}, 17(2), 2000.
\newblock arXiv:hep-th/9911230.

\bibitem{hawking}
S.~W. Hawking.
\newblock Particle creation by black holes.
\newblock {\em Communications in Mathematical Physics}, 43(3):199--220, 1975.

\bibitem{parker}
L.~Parker;~D. Toms.
\newblock {\em Quantum field theory in curved spacetime: {Q}uantized fields and
  gravity}.
\newblock Cambridge University Press, 2009.

\bibitem{thompson}
R.~T. Thompson.
\newblock {\em Characteristics of fluctuating space-time geometries:
  {M}inkowski and {S}chwarzschild examples}.
\newblock PhD thesis, Tufts University, 2008.

\bibitem{bekenstein1}
J.~D. Bekenstein.
\newblock Black holes and entropy.
\newblock {\em Physical Review D}, 7(8), 1973.

\bibitem{page}
S.~W. Hawking; D.~N. Page.
\newblock Thermodynamics of black holes in anti-de {S}itter space.
\newblock {\em Communications in Mathematical Physics}, 87(4):577--588, 1983.

\bibitem{witten1}
E.~Witten.
\newblock Anti-de {S}itter space, thermal phase transition, and confinement in
  gauge theories.
\newblock {\em Advances in Theoretical and Mathematical Physics}, 2:505--532,
  1998.
\newblock arXiv:hep-th/9803131.

\bibitem{bekenstein2}
J.~D. Bekenstein.
\newblock Universal upper bound on the entropy-to-energy ratio for bounded
  systems.
\newblock {\em Physical Review D}, 23(2), 1981.

\bibitem{strominger}
A.~Strominger;~C. Vafa.
\newblock Microscopic origin of the {B}ekenstein-{H}awking entropy.
\newblock {\em Physics Letters B}, 379:99--104, 1996.
\newblock arXiv:hep-th/9601029.

\bibitem{braunstein}
S.~L. Braunstein; S. Pirandola;~K. Zyczkowski.
\newblock Better late than never: {I}nformation retrieval from black holes.
\newblock {\em Physical Review Letters}, 110(10), 2013.
\newblock arXiv:0907.1190.

\bibitem{amps}
A.~Almheiri; D. Marolf; J. Polchinski;~J. Sully.
\newblock Black holes: {C}omplimentarity or firewalls?
\newblock {\em Journal of High Energy Physics}, 2013(2), 2013.
\newblock arXiv:1207.3123.

\bibitem{thooft}
G.~`tHooft.
\newblock A planar diagram theory for strong interactions.
\newblock {\em Nuclear Physics B}, 72(3):461--470, 1974.

\bibitem{hull}
C.~M. Hull; P.~K. Townsend.
\newblock Unity of superstring dualities.
\newblock {\em Nuclear Physics B}, 438:109--137, 1995.
\newblock arXiv:hep-th/9410167.

\bibitem{tseytlin}
A.~A. Tseytlin.
\newblock Self duality of {B}orn-{I}nfeld action and {D}irichlet 3-brane of
  type {IIB} superstring theory.
\newblock {\em Nuclear Physics B}, 469:51--67, 1996.
\newblock arXiv:hep-th/9602064.

\bibitem{green}
M.~B. Green;~M. Gutperle.
\newblock Comments on three-branes.
\newblock {\em Physics Letters B}, 377:28--35, 1996.
\newblock arXiv:hep-th/9602077.

\bibitem{csaki}
C.~Csaki; M. Reece;~J. Terning.
\newblock The {A}d{S} / {QCD} correspondence: {S}till undelivered.
\newblock {\em Journal of High Energy Physics}, 2009(5), 2009.
\newblock arXiv:0811.3001.

\bibitem{sachdev}
S.~Sachdev.
\newblock Condensed matter and {A}d{S} / {CFT}.
\newblock In {\em Lectures at the 5th Aegean summer school ``From gravity to
  thermal gauge theories: The {A}d{S} / {CFT} correspondence''}, 2010.
\newblock arXiv:1002.2947.

\bibitem{witten2}
E.~Witten.
\newblock Anti-de {S}itter space and holography.
\newblock {\em Advances in Theoretical and Mathematical Physics}, 2:253--291,
  1998.
\newblock arXiv:hep-th/9802150.

\bibitem{bf1}
P.~Breitenlohner; D.~Z. Freedman.
\newblock Positive energy in anti de-{S}itter backgrounds and gauged extended
  supergravity.
\newblock {\em Physics Letters B}, 115:197--201, 1982.

\bibitem{bf2}
P.~Breitenlohner; D.~Z. Freedman.
\newblock Stability in gauged extended supergravity.
\newblock {\em Annals of Physics}, 144:249--281, 1982.

\bibitem{burgess}
C.~P. Burgess; N. R. Constable; R.~C. Myers.
\newblock The free energy of {N}=4 super {Y}ang-mills and the {A}d{S} / {CFT}
  correspondence.
\newblock {\em Journal of High Energy Physics}, 1999(8), 1999.
\newblock arXiv:hep-th/9907188.

\bibitem{banks}
T.~Banks;~W. Fischler.
\newblock A model for high energy scattering in quantum gravity.
\newblock 1999.
\newblock arXiv:hep-th/9906038.

\bibitem{giddings}
S.~B. Giddings;~S. Thomas.
\newblock High energy colliders as black hole factories: {T}he end of short
  distance physics.
\newblock {\em Physical Review D}, 65(5), 2002.
\newblock arXiv:hep-ph/0106219.

\bibitem{dimopoulos}
S.~Dimopoulos;~G. Landsberg.
\newblock Black holes at the {LHC}.
\newblock {\em Physical Review Letters}, 87(16), 2001.
\newblock arXiv:hep-ph/0106295.

\bibitem{mcgreevy}
J.~McGreevy; L. Susskind;~N. Toumbas.
\newblock Invasion of the giant gravitons from anti-de {S}itter space.
\newblock {\em Journal of High Energy Physics}, 2000(6), 2000.
\newblock hep-th/0003075.

\bibitem{larjo}
K.~Larjo.
\newblock {\em On the emergence of classical gravity}.
\newblock PhD thesis, University of Pennsylvania, 2008.

\bibitem{gubser}
S.~S. Gubser; I. R. Klebanov; A.~A. Tseytlin.
\newblock Coupling constant dependence in the thermodynamics of {N}=4
  supersymmetric {Y}ang-{M}ills theory.
\newblock {\em Nuclear Physics B}, 534:202--222, 1998.
\newblock arXiv:hep-th/9805156.

\bibitem{fotopoulos}
A.~Fotopoulos; T.~R. Taylor.
\newblock Remarks on two-loop free energy in {N}=4 supersymmetric
  {Y}ang-{M}ills theory at finite temperature.
\newblock {\em Physical Review D}, 59(6), 1999.
\newblock arXiv:hep-th/9811224.

\bibitem{kim}
C.~Kim; S.~J. Rey.
\newblock Thermodynamics of large-{N} super {Y}ang-{M}ills theory and {A}d{S} /
  {CFT} correspondence.
\newblock {\em Nuclear Physics B}, 564:430--440, 2000.
\newblock arXiv:hep-th/9905205.

\bibitem{figueras}
P.~Figueras;~S. Tunyasuvunakool.
\newblock Localized plasma balls.
\newblock 2014.
\newblock arXiv:1404.0018.

\bibitem{bhardwaj}
S.~Bhardwaj;~J. Bhattacharya.
\newblock Thermodynamics of plasmaballs and plasmarings in 3 + 1 dimensions.
\newblock {\em Journal of High Energy Physics}, 2009(3), 2009.
\newblock arXiv:0806.1897.

\bibitem{milanesi}
R.~Emparan;~G. Milanesi.
\newblock Exact gravitational dual of a plasma ball.
\newblock {\em Journal of High Energy Physics}, 2009(8), 2009.
\newblock arXiv:0905.4590.

\bibitem{schaposnik}
F.~A. Schaposnik.
\newblock Gauge-gravity duality and phase transitions.
\newblock In {\em Ver{\~a}o Qu{\^a}ntico 2011}, 2011.
\newblock http://vq2011.ufes.br/fidel.pdf.

\bibitem{wiseman}
D.~Marolf; M. Rangamani;~T. Wiseman.
\newblock Holographic thermal field theory on curved spacetimes.
\newblock {\em Classical and Quantum Gravity}, 31(6), 2014.
\newblock arXiv:1312.0612.

\bibitem{freivogel}
B.~Freivogel; J. McGreevy; S.~J. Suh.
\newblock Exactly stable collective oscillations in conformal field theory.
\newblock {\em Physical Review D}, 85(10), 2012.
\newblock arXiv:1109.6013.

\bibitem{btz}
M.~Ba{\~n}ados; C. Teitelboim;~J. Zanelli.
\newblock The black hole in three dimensional space time.
\newblock {\em Physical Review Letters}, 69(13), 1992.
\newblock arXiv:hep-th/9204099v3.

\bibitem{ngvk}
N.~G.~Van Kampen.
\newblock {\em Stochastic processes in physics and chemistry}.
\newblock North-Holland Personal Library, 2007.

\bibitem{mccumber}
D.~E. McCumber.
\newblock Einstein relations connecting broadband emission and absorbtion
  spectra.
\newblock {\em Physical Review}, 136:A954--A957, 1964.

\bibitem{protter}
M.~H. Protter; H.~F. Weinberger.
\newblock {\em Maximum principles in differential equations}.
\newblock Prentice-Hall, 1967.

\bibitem{tao}
T.~Tao.
\newblock {\em Poincar{\'e}'s lagacies: {P}ages from year two of a mathematical
  blog}.
\newblock American Mathematical Society, 2009.

\bibitem{arena}
O.~Arena.
\newblock A strong maximum principle for quasilinear parabolic differential
  inequalities.
\newblock {\em Proceedings of the American Mathematical Society}, 32(2), 1972.

\bibitem{vazquez1}
J.~L. V{\'a}zquez.
\newblock Symmetrization and mass concentration for degenerate nonlinear
  parabolic and related elliptic equations.
\newblock {\em Advanced Nonlinear Studies}, 5:87--131, 2005.

\bibitem{chung}
S.~Y. Chung;~D. Kim.
\newblock An example of nonuniqueness of the {C}auchy problem for the heat
  equation.
\newblock {\em Communications in Partial Differential Equations},
  19:1257--1261, 1994.

\bibitem{dhungana}
B.~P. Dhungana.
\newblock An example of nonuniqueness of the {C}auchy problem for the {H}ermite
  heat equation.
\newblock {\em Proceedings of the Japan Academy, Series A, Mathematical
  Sciences}, 81(3):37--39, 2005.

\bibitem{esteban}
J.~R. Esteban; A. Rodriguez; J.~L. V{\'a}zquez.
\newblock A nonlinear heat equation with singular diffusivity.
\newblock {\em Communications in Partial Differential Equations}, 13:985--1039,
  1988.

\bibitem{rodriguez}
A.~Rodriguez; J.~L. V{\'a}zquez.
\newblock A well posed problem in singular {F}ickian diffusion.
\newblock {\em Archive for Rational Mechanics and Analysis}, 110(2):141--163,
  1990.

\bibitem{vazquez2}
J.~L. V{\'a}zquez.
\newblock Failure of the strong maximum principle in nonlinear diffusion:
  {E}xistence of needles.
\newblock {\em Communications in Partial Differential Equations},
  30(9):1263--1303, 2005.

\bibitem{carrillo}
J.~A. Carrillo; J.~L. V{\'a}zquez.
\newblock Fine asymptotics for fast diffusion equations.
\newblock {\em Communications in Partial Differential Equations},
  28(5--6):1023--1056, 2003.

\bibitem{verma}
L.~Verma.
\newblock L-stable derivative free error corrected trapezoidal rule for
  {B}urgers' equation with inconsistent initial and boundary conditions.
\newblock {\em International Journal of Mathematics and Mathematical Sciences},
  2012.

\bibitem{press}
W.~H. Press; S. A. Teukolsky; W. T. Vetterling; B.~P. Flannery.
\newblock {\em Numerical recipies in {C}}.
\newblock Cambridge University Press, 2002.

\bibitem{choptuik}
M.~Choptuik.
\newblock Numerical analysis for numerical relativists.
\newblock In {\em Lectures for {VII} {M}exican school on gravitational and
  mathematical physics}, 2006.
\newblock
  http://laplace.physics.ubc.ca/People/matt/Teaching/06Mexico/mexico06.pdf.

\bibitem{zlatev}
Z.~Zlatev; I. Dimov; I. Farag{\'o}; K. Georgiev; A. Havasi;~T. Ostromsky.
\newblock Richardson extrapolated numerical methods for treatment of
  one-dimensional advection equations.
\newblock In {\em 7th international conference on numerical methods and
  applications}, pages 198--206, 2011.

\bibitem{kovtun}
P.~Kovtun.
\newblock Lecutres on hydrodynamic fluctuations in relativistic theories.
\newblock {\em Journal of Physics A}, 45(47), 2012.
\newblock arXiv:1205.5040.

\bibitem{ambrosetti}
N.~Ambrosetti; J. Charbonneau;~S. Weinfurtner.
\newblock The fluid / gravity correspondence: {L}ecture notes from the 2008
  summer school on particles, fields and strings.
\newblock 2008.
\newblock arXiv:0810.2631.

\bibitem{kaminski}
M.~Kaminski;~S. Moroz.
\newblock Non-relativistic parity-violating hydrodynamics in two spatial
  dimensions.
\newblock {\em Physical Review B}, 89(11), 2014.
\newblock arXiv:1310.8305.

\bibitem{policastro}
G.~Policastro; D. T. Son; A.~O. Starinets.
\newblock Shear viscosity of strongly coupled {N} = 4, supersymmetric
  {Y}ang-{M}ills plasma.
\newblock {\em Physical Review Letters}, 87(8), 2001.
\newblock hep--th/0104066.

\bibitem{rangamani}
M.~Rangamani.
\newblock Gravity and hydrodynamics: {L}ectures on the fluid-gravity
  correspondence.
\newblock {\em Classical and Quantum Gravity}, 26(22), 2009.
\newblock arXiv:0905.4352.

\bibitem{hubeny}
S.~Bhattacharyya; V. E. Hubeny; S. Minwalla;~M. Rangamani.
\newblock Nonlinear fluid dynamics from gravity.
\newblock {\em Journal of High Energy Physics}, 2008(2), 2008.
\newblock arXiv:0712.2456.

\bibitem{balasubramanian}
V.~Balasubramanian;~P. Kraus.
\newblock A stress tensor for anti-de {S}itter gravity.
\newblock {\em Communications in Mathematical Physics}, 208(2):413--428, 1999.
\newblock arXiv:hep-th/9902121.

\bibitem{fg}
C.~Fefferman; C.~R. Graham.
\newblock Conformal invariants.
\newblock {\em Ast\'erisque}, (Numero Hors Serie):95--116, 1985.

\bibitem{cardoso}
V.~Cardoso; O. J.~C. Dias.
\newblock Bifurcation of plasma balls and black holes to lobed configurations.
\newblock {\em Journal of High Energy Physics}, 2009(4), 2009.
\newblock arXiv:0902.3560.

\bibitem{rocha}
V.~Cardoso; O. J. C. Dias; J.~V. Rocha.
\newblock Phase diagram for non-axisymmetric plasma balls.
\newblock {\em Journal of High Energy Physics}, 2010(1), 2010.
\newblock arXiv:0910.0020.

\bibitem{lahiri}
J.~Bhattacharya;~S. Lahiri.
\newblock Lumps of plasma in arbitrary dimensions.
\newblock {\em Journal of High Energy Physics}, 2010(8), 2010.
\newblock arXiv:0903.4734.

\bibitem{weick}
G.~Weick.
\newblock {\em Quantum dissipation and decoherence of collective excitations in
  metallic nanoparticles}.
\newblock PhD thesis, Universit{\'e} Louis Pasteur, 2006.

\bibitem{tong2}
D.~Tong.
\newblock Lectures on string theory.
\newblock 2009.
\newblock arXiv:0908.0333.

\bibitem{cardy}
J.~L. Cardy.
\newblock Operator content of two-dimensional conformally invariant theories.
\newblock {\em Nuclear Physics B}, 270:186--204, 1986.

\bibitem{hsu}
J.~Hsu.
\newblock The partition function and modular forms.
\newblock Master's thesis, Stanford University, 2011.

\bibitem{elstrodt}
J.~Elstrodt.
\newblock A very simple proof of the eta transformation formula.
\newblock {\em Manuscripta Mathematica}, 121:457--459, 2006.

\bibitem{dijkgraaf}
R.~Dijkgraaf; J. Maldacena; G. Moore;~E. Verlinde.
\newblock A black hole {F}arey tail.
\newblock 2000.
\newblock arXiv:hep-th/0005003.

\bibitem{birmingham}
D.~Birmingham;~S. Sen.
\newblock Exact black hole entropy bound in conformal field theory.
\newblock {\em Physical Review D}, 63(4), 2001.
\newblock arXiv:hep-th/0008051.

\bibitem{verlinde2}
E.~Verlinde.
\newblock On the holographic principle in a radiation dominated universe.
\newblock 2000.
\newblock arXiv:hep-th/0008140.

\bibitem{kutasov}
D.~Kutasov;~F. Larsen.
\newblock Partition sums and entropy bounds in weakly coupled {CFT}.
\newblock {\em Journal of High Energy Physics}, 2001(1), 2001.
\newblock arXiv:hep-th/0009244.

\bibitem{bhattacharya}
J.~Bhattacharya; S. Bhattacharyya; S. Minwalla;~A. Yarom.
\newblock A theory of first order dissipative superfluid dynamics.
\newblock 2011.
\newblock arXiv:1105.3733.

\bibitem{takeuchi}
S.~Takeuchi; Y.~L. Wu.
\newblock Hydrodynamics and transport coefficients in an infrared-deformed
  soft-wall {A}d{S} / {QCD} model at finite temperature.
\newblock {\em Physical Review D}, 88(2), 2013.
\newblock arXiv:1302.2428.

\bibitem{john}
F.~John.
\newblock Numerical solution of the equation of heat conduction for preceding
  times.
\newblock {\em Annali di Mathematica Pura ed Aplicata}, 40:129--142, 1955.

\bibitem{fu}
C.~L. Fu; X. T. Xiong;~Z. Qian.
\newblock Fourier regularization for a backward heat equation.
\newblock {\em Journal of Mathematical Analysis and Applications},
  331:472--480, 2007.

\bibitem{wang}
J.~R. Wang.
\newblock Shannon wavelet regularization methods for a backward heat equation.
\newblock {\em Journal of Computational and Applied Mathematics},
  235:3079--3085, 2011.

\bibitem{amy1}
P.~Arnold; G. D. Moore; L.~G. Yaffe.
\newblock Transport coefficients in high temperature gauge theories: {I}
  {L}eading-log results.
\newblock {\em Journal of High Energy Physics}, 2000(11), 2000.
\newblock arXiv:hep-ph/0010177.

\bibitem{amy2}
P.~Arnold; G. D. Moore; L.~G. Yaffe.
\newblock Transport coefficients in high temperature gauge theories: {II}
  beyond leading log.
\newblock {\em Journal of High Energy Physics}, 2003(5), 2003.
\newblock arXiv:hep-ph/0302165.

\bibitem{york}
M.~A. York; G.~D. Moore.
\newblock Second order hydrodynamic coefficients from kinetic theory.
\newblock {\em Physical Review D}, 79(5), 2009.
\newblock arXiv:0811.0729.

\end{thebibliography}
\newpage

\appendix
\section{Cumbersome derivatives}
In a few parts of the thesis we have dealt with discrete amounts of energy and momentum denoted by $\epsilon$ and
$\epsilon \textbf{e}$ respectively. In order to convert energy and momentum to continuous quantities we have had to
differentiate expressions where $\epsilon$ appears inside the argument of the density of states $\rho$. These
calculations are collected here for reference.

In the calculation of (\ref{e-same-derivative}), we had the quantity
\begin{eqnarray}
\tilde{X}(a, \epsilon) &=& \rho(E(\textbf{x}) + \epsilon, \textbf{P}(\textbf{x}) - \epsilon \textbf{e}) \rho(E(\textbf{x} + a\textbf{e}) - \epsilon, \textbf{P}(\textbf{x} + a\textbf{e}) + \epsilon \textbf{e}) \nonumber \\
&& - \rho(E(\textbf{x}) - \epsilon, \textbf{P}(\textbf{x}) - \epsilon \textbf{e}) \rho(E(\textbf{x} + a\textbf{e}) + \epsilon, \textbf{P}(\textbf{x} + a\textbf{e}) + \epsilon \textbf{e}) \label{first-appendix-derivative}
\end{eqnarray}
(actually there was a $C$ multiplying this expression and a surface integral around the whole thing)
which we had to differentiate once. The steps involved are:
\begin{eqnarray}
\left. \frac{\partial \tilde{X}}{\partial \epsilon} \right |_{\epsilon = 0} &=& \rho_{+} \left ( \frac{\partial \rho}{\partial E} - e_j \frac{\partial \rho}{\partial P_j} \right ) + \rho \left ( -\frac{\partial \rho_{+}}{\partial E} + e_j \frac{\partial \rho_{+}}{\partial P_j} \right ) \nonumber \\
&& - \rho_{+} \left ( -\frac{\partial \rho}{\partial E} - e_j \frac{\partial \rho}{\partial P_j} \right ) - \rho \left ( \frac{\partial \rho_{+}}{\partial E} + e_j \frac{\partial \rho_{+}}{\partial P_j} \right ) \nonumber \\
&=& 2 \rho_{+} \frac{\partial \rho}{\partial E} - 2 \rho \frac{\partial \rho_{+}}{\partial E} \; . \nonumber
\end{eqnarray}
The quantity that we had to differentiate in the calculation of (\ref{e-derivative}) was the same as
(\ref{first-appendix-derivative}) except without the $\textbf{P}$ dependence. Because the terms with an explicit
$\textbf{P}$ cancelled above, the answer for this case is the same.

A similar quantity to this one came up in the calculation of (\ref{p-4-derivatives}):
\begin{eqnarray}
\tilde{Y}(a, \epsilon) &=& -\rho(E(\textbf{x}) + \epsilon, \textbf{P}(\textbf{x}) - \epsilon \textbf{e}) \rho(E(\textbf{x} + a\textbf{e}) - \epsilon, \textbf{P}(\textbf{x} + a\textbf{e}) + \epsilon \textbf{e}) \nonumber \\
&& - \rho(E(\textbf{x}) - \epsilon, \textbf{P}(\textbf{x}) - \epsilon \textbf{e}) \rho(E(\textbf{x} + a\textbf{e}) + \epsilon, \textbf{P}(\textbf{x} + a\textbf{e}) + \epsilon \textbf{e}) \; . \label{second-appendix-derivative}
\end{eqnarray}
Since the middle sign is all that distinguishes (\ref{second-appendix-derivative}) from
(\ref{first-appendix-derivative}), it is easy to see that the energy pieces will cancel this time instead of the
momentum pieces:
\begin{eqnarray}
\left. \frac{\partial \tilde{Y}}{\partial \epsilon} \right |_{\epsilon = 0} &=& -\rho_{+} \left ( \frac{\partial \rho}{\partial E} - e_j \frac{\partial \rho}{\partial P_j} \right ) - \rho \left ( -\frac{\partial \rho_{+}}{\partial E} + e_j \frac{\partial \rho_{+}}{\partial P_j} \right ) \nonumber \\
&& - \rho_{+} \left ( -\frac{\partial \rho}{\partial E} - e_j \frac{\partial \rho}{\partial P_j} \right ) - \rho \left ( \frac{\partial \rho_{+}}{\partial E} + e_j \frac{\partial \rho_{+}}{\partial P_j} \right ) \nonumber \\
&=& 2 e_j \rho_{+} \frac{\partial \rho}{\partial P_j} - 2 e_j \rho \frac{\partial \rho_{+}}{\partial P_j} \; . \nonumber
\end{eqnarray}
However, there is one more calculation we have to do: differentiating (\ref{second-appendix-derivative}) twice
before plugging in $\epsilon = 0$.
\begin{eqnarray}
\left. \frac{\partial^2 \tilde{Y}}{\partial \epsilon^2} \right |_{\epsilon = 0} &=& -\left. \frac{\partial}{\partial \epsilon} \left [ \left ( \frac{\partial \rho}{\partial E}(E + \epsilon, \textbf{P} - \epsilon\textbf{e}) - e_j \frac{\partial \rho}{\partial P_j}(E + \epsilon, \textbf{P} - \epsilon\textbf{e}) \right ) \rho(E_{+} - \epsilon, \textbf{P}_{+} + \epsilon\textbf{e}) \right. \right. \nonumber \\
&& \left. \left. + \rho(E + \epsilon, \textbf{P} - \epsilon\textbf{e}) \left ( -\frac{\partial \rho}{\partial E}(E_{+} - \epsilon, \textbf{P}_{+} + \epsilon\textbf{e}) + e_j \frac{\partial \rho}{\partial P_j}(E_{+} - \epsilon, \textbf{P}_{+} + \epsilon\textbf{e}) \right ) \right. \right. \nonumber \\
&& \left. \left. + \left ( -\frac{\partial \rho}{\partial E}(E - \epsilon, \textbf{P} - \epsilon\textbf{e}) - e_j \frac{\partial \rho}{\partial P_j}(E - \epsilon, \textbf{P} - \epsilon\textbf{e}) \right ) \rho(E_{+} + \epsilon, \textbf{P}_{+} + \epsilon\textbf{e}) \right. \right. \nonumber \\
&& \left. \left. + \rho(E - \epsilon, \textbf{P} - \epsilon\textbf{e}) \left ( \frac{\partial \rho}{\partial E}(E_{+} + \epsilon, \textbf{P}_{+} + \epsilon\textbf{e}) + e_j \frac{\partial \rho}{\partial P_j}(E_{+} + \epsilon, \textbf{P}_{+} + \epsilon\textbf{e}) \right ) \right ] \right |_{\epsilon = 0} \nonumber \\
&=& -\rho_{+} \left ( \frac{\partial^2 \rho}{\partial E^2} - 2 e_j \frac{\partial^2 \rho}{\partial E \partial P_j} + e_j e_k \frac{\partial^2 \rho}{\partial P_j \partial P_k} \right ) - \rho \left ( \frac{\partial^2 \rho_{+}}{\partial E^2} - 2 e_j \frac{\partial^2 \rho_{+}}{\partial E \partial P_j} + e_j e_k \frac{\partial^2 \rho_{+}}{\partial P_j \partial P_k} \right ) \nonumber \\
&& - 2 \left ( \frac{\partial \rho}{\partial E} - e_j \frac{\partial \rho}{\partial P_j} \right ) \left ( -\frac{\partial \rho_{+}}{\partial E} + e_k \frac{\partial \rho_{+}}{\partial P_k} \right ) - \rho_{+} \left ( \frac{\partial^2 \rho}{\partial E^2} + 2 e_j \frac{\partial^2 \rho}{\partial E \partial P_j} + e_j e_k \frac{\partial^2 \rho}{\partial P_j \partial P_k} \right ) \nonumber \\
&& - \rho \left ( \frac{\partial^2 \rho_{+}}{\partial E^2} + 2 e_j \frac{\partial^2 \rho_{+}}{\partial E \partial P_j} + e_j e_k \frac{\partial^2 \rho_{+}}{\partial P_j \partial P_k} \right ) - 2 \left ( -\frac{\partial \rho}{\partial E} - e_j \frac{\partial \rho}{\partial P_j} \right ) \left ( \frac{\partial \rho_{+}}{\partial E} + e_k \frac{\partial \rho_{+}}{\partial P_k} \right ) \nonumber \\
&=& 2 \left ( \frac{\partial \rho}{\partial E} - e_j \frac{\partial \rho}{\partial P_j} \right ) \left ( \frac{\partial \rho_{+}}{\partial E} - e_k \frac{\partial \rho_{+}}{\partial P_k} \right ) + 2 \left ( \frac{\partial \rho}{\partial E} + e_j \frac{\partial \rho}{\partial P_j} \right ) \left ( \frac{\partial \rho_{+}}{\partial E} + e_k \frac{\partial \rho_{+}}{\partial P_k} \right ) \nonumber \\
&& - 2 \rho_{+} \left ( \frac{\partial^2 \rho}{\partial E^2} + e_j e_k \frac{\partial^2 \rho}{\partial P_j \partial P_k} \right ) - 2 \rho \left ( \frac{\partial^2 \rho_{+}}{\partial E^2} + e_j e_k \frac{\partial^2 \rho_{+}}{\partial P_j \partial P_k} \right ) \nonumber
\end{eqnarray}

\section{Discontinuous PDE solutions}
When deriving our main bounds on black hole decay times, a key step was finding a solution to
$\frac{\partial E}{\partial t} = \frac{\partial}{\partial x} \left ( \theta(E - E_{\mathrm{H}}) \frac{\partial E}{\partial x} \right )$.
We argued that such a solution is given by
\begin{equation}
E(x, t) =
\begin{cases}
E_{\mathrm{F}} & |x| < x^*(t) \\
F(x, t) & |x| > x^*(t)
\end{cases}
\nonumber
\end{equation}
where $F$ solves the heat equation and satisfies the mass conservation condition
\begin{equation}
\left ( E_{\mathrm{F}} - E_{\mathrm{H}} \right ) \frac{\textup{d}x^*(t)}{\textup{d}t} = \frac{\partial F}{\partial x} (x^*(t), t) \; . \nonumber
\end{equation}
Here, $x^*(t)$ is defined by
\begin{equation}
F(x^*(t), t) = E_{\mathrm{H}} \nonumber
\end{equation}
A function constructed this way retains the intuitive properties that we expect a solution to have but in some sense it
is not a solution; it is discontinuous and derivatives acting on a discontinuous function have no meaning. We need to
show that it is a solution in a precise generalized sense. Consider the solution to the heat equation $F$. Since
\begin{equation}
\frac{\partial F}{\partial t} - \frac{\partial^2 F}{\partial x^2} = 0 \; , \label{appendix-strong}
\end{equation}
it is clearly true that
\begin{equation}
\int_0^T \int_{-\infty}^{\infty} \left [ \frac{\partial F}{\partial t} - \frac{\partial^2 F}{\partial x^2} \right ] \varphi \textup{d}x\textup{d}t = 0 \label{appendix-weak}
\end{equation}
where $\varphi \in C_0^{\infty}(\mathbb{R} \times (0, T))$ is a test function with compact support. The important part
is that (\ref{appendix-strong}) and (\ref{appendix-weak}) are not equivalent. If $F$ were not differentiable,
(\ref{appendix-weak}) would still make sense because $\varphi$ is smooth and the derivatives can be shifted onto
$\varphi$ through integration by parts. A function solving a differential equation but only inside an integral with a
test function like this is called a \textit{weak solution}. The following two expressions
\begin{eqnarray}
\frac{\textup{d}}{\textup{d}t} \int_{-x^*}^{x^*} \varphi \textup{d}x &=& \frac{\textup{d}x^*}{\textup{d}t} \biggl. \varphi \biggl|_{-x^*}^{x^*} + \int_{-x^*}^{x^*} \frac{\partial \phi}{\partial t} \textup{d}x \nonumber \\
\frac{\textup{d}}{\textup{d}t} \int_{-x^*}^{x^*} F \varphi \textup{d}x &=& F \frac{\textup{d}x^*}{\textup{d}t} \biggl. \varphi \biggl|_{-x^*}^{x^*} + \int_{-x^*}^{x^*} \frac{\partial}{\partial t}(F \varphi) \textup{d}x \nonumber \\
&=& E_{\mathrm{H}} \frac{\textup{d}x^*}{\textup{d}t} \biggl. \varphi \biggl|_{-x^*}^{x^*} + \int_{-x^*}^{x^*} F \frac{\partial \varphi}{\partial t} + \varphi \frac{\partial F}{\partial t} \textup{d}x \nonumber
\end{eqnarray}
which use Feynman's trick of differentiating under the integral sign, will be useful in showing that $E$ is a weak
solution.
\begin{eqnarray}
\int_0^T \int_{-\infty}^{\infty} \left [ \frac{\partial E}{\partial t} - \frac{\partial^2 \Phi(E)}{\partial x^2} \right ] \varphi \textup{d}x\textup{d}t &=& -\int_0^T \int_{-\infty}^{\infty} E \frac{\partial \varphi}{\partial t} \textup{d}x\textup{d}t + \int_0^T \int_{-\infty}^{\infty} \frac{\partial \Phi(E)}{\partial x} \frac{\partial \varphi}{\partial x} \textup{d}x\textup{d}t \nonumber \\
&=& -\int_0^T \int_{\mathbb{R} \setminus (-x^*, x^*)} E \frac{\partial \varphi}{\partial t} \textup{d}x\textup{d}t - \int_0^T \int_{-x^*}^{x^*} E \frac{\partial \varphi}{\partial t} \textup{d}x\textup{d}t \nonumber \\
&& + \int_0^T \int_{\mathbb{R} \setminus (-x^*, x^*)} \frac{\partial \Phi(E)}{\partial x} \frac{\partial \varphi}{\partial x} \textup{d}x\textup{d}t + \int_0^T \int_{-x^*}^{x^*} \frac{\partial \Phi(E)}{\partial x} \frac{\partial \varphi}{\partial x} \textup{d}x\textup{d}t \nonumber \\
&=& -\int_0^T \int_{\mathbb{R} \setminus (-x^*, x^*)} F \frac{\partial \varphi}{\partial t} \textup{d}x\textup{d}t - E_{\mathrm{F}} \int_0^T \int_{-x^*}^{x^*} \frac{\partial \varphi}{\partial t} \textup{d}x\textup{d}t \nonumber \\
&& + \int_0^T \int_{\mathbb{R} \setminus (-x^*, x^*)} \frac{\partial F}{\partial x} \frac{\partial \varphi}{\partial x} \textup{d}x\textup{d}t \nonumber \\
&=& \int_0^T \int_{-x^*}^{x^*} F \frac{\partial \varphi}{\partial t} \textup{d}x\textup{d}t - \int_0^T \int_{-\infty}^{\infty} F \frac{\partial \varphi}{\partial t} \textup{d}x\textup{d}t - E_{\mathrm{F}} \int_0^T \int_{-x^*}^{x^*} \frac{\partial \varphi}{\partial t} \textup{d}x\textup{d}t \nonumber \\
&& + \int_0^T \int_{-\infty}^{\infty} \frac{\partial F}{\partial x} \frac{\partial \varphi}{\partial x} \textup{d}x\textup{d}t - \int_0^T \int_{-x^*}^{x^*} \frac{\partial F}{\partial x} \frac{\partial \varphi}{\partial x} \textup{d}x\textup{d}t \nonumber \\
&=& \int_0^T \int_{-x^*}^{x^*} \left ( F - E_{\mathrm{F}} \right ) \frac{\partial \varphi}{\partial t} - \frac{\partial F}{\partial x} \frac{\partial \varphi}{\partial x} \textup{d}x\textup{d}t \nonumber \\
&=& \int_0^T \int_{-x^*}^{x^*} \left ( F - E_{\mathrm{F}} \right ) \frac{\partial \varphi}{\partial t} \textup{d}x - \frac{\partial F}{\partial x} \biggl. \varphi \biggl|_{-x^*}^{x^*} + \int_{-x^*}^{x^*} \frac{\partial^2 F}{\partial x^2} \varphi \textup{d}x\textup{d}t \nonumber \\
&=& \int_0^T \int_{-x^*}^{x^*} \left ( F - E_{\mathrm{F}} \right ) \frac{\partial \varphi}{\partial t} \textup{d}x - \left ( E_{\mathrm{F}} - E_{\mathrm{H}} \right ) \frac{\textup{d}x^*}{\textup{d}t} \biggl. \varphi \biggl|_{-x^*}^{x^*} + \int_{-x^*}^{x^*} \frac{\partial^2 F}{\partial x^2} \varphi \textup{d}x\textup{d}t \nonumber \\
&=& \int_0^T \int_{-x^*}^{x^*} F \frac{\partial \varphi}{\partial t} + \frac{\partial^2 F}{\partial x^2} \varphi \textup{d}x\textup{d}t + E_{\mathrm{F}} \int_0^T \frac{\textup{d}x^*}{\textup{d}t} \biggl. \varphi \biggl|_{-x^*}^{x^*} \textup{d}t \nonumber \\
&& - E_{\mathrm{F}} \int_0^T \frac{\textup{d}x^*}{\textup{d}t} \biggl. \varphi \biggl|_{-x^*}^{x^*} + \int_{-x^*}^{x^*} \frac{\partial \varphi}{\partial t} \textup{d}x\textup{d}t \nonumber \\
&=& \int_0^T \int_{-x^*}^{x^*} F \frac{\partial \varphi}{\partial t} + \frac{\partial^2 F}{\partial x^2} \varphi \textup{d}x\textup{d}t + E_{\mathrm{F}} \int_0^T \frac{\textup{d}x^*}{\textup{d}t} \biggl. \varphi \biggl|_{-x^*}^{x^*} \textup{d}t \nonumber \\
&& - E_{\mathrm{F}} \int_0^T \frac{\textup{d}}{\textup{d}t} \int_{-x^*}^{x^*} \varphi \textup{d}x \textup{d}t \nonumber \\
&=& \int_0^T \int_{-x^*}^{x^*} \frac{\partial^2 F}{\partial x^2} \varphi \textup{d}x\textup{d}t + \int_0^T \int_{-x^*}^{x^*} F \frac{\partial \varphi}{\partial t} \textup{d}x + E_{\mathrm{F}} \frac{\textup{d}x^*}{\textup{d}t} \biggl. \varphi \biggl|_{-x^*}^{x^*} \textup{d}t \nonumber \\
&=& \int_0^T \int_{-x^*}^{x^*} \frac{\partial^2 F}{\partial x^2} \varphi \textup{d}x\textup{d}t + \int_0^T \frac{\textup{d}}{\textup{d}t} \int_{-x^*}^{x^*} F \varphi \textup{d}x - \int_{-x^*}^{x^*} \varphi \frac{\partial F}{\partial t} \textup{d}x\textup{d}t \nonumber \\
&=& \int_0^T \int_{-x^*}^{x^*} \frac{\partial^2 F}{\partial x^2} \varphi \textup{d}x\textup{d}t - \int_0^T \int_{-x^*}^{x^*} \varphi \frac{\partial F}{\partial t} \textup{d}x\textup{d}t \nonumber \\
&=& 0 \nonumber
\end{eqnarray}

\section{Crank-Nicolson code}

\begin{verbatim}
/* Changing float to double might help to avoid crashes. */
#include <stdio.h>
#include <stdlib.h>
#include <stdbool.h>
#include <math.h>

#define SMALL_ERROR 5e-7
#define TOO_MANY_LOOPS 30

/* Parameters for the initial condition. */
#define PMAX 200.0
#define EXP 1.3333

/* Ten magic constants handle the interpolation. */
#define EF 10.0
float c0, c1, c2, c3, c4, c5, c6, c7, c8, c9;

/* The maximum distance, the maximum time and the individual steps. */
float length, dx_min, dx_max, dt, max_time;

/* The number of sites in our discretized space. */
int n_s;

/* Arrays of this size that will be needed for the solution. */
float *energy;
float *new_energy;
float *jac_sub_diag;
float *jac_diag;
float *jac_sup_diag;
float *constant_part;
float *diff;
float *xs;

float beta(float x) {
  if (x < 0.5) {
    return pow(x, -0.1);
  } else if ((x > 0.5) && (x < 1.0)) {
    return c0 * x * x + c1 * x + c2;
  } else if ((x > 1.0) && (x < 0.9 * EF)) {
    return c3 * x + c4;
  } else if ((x > 0.9 * EF) && (x < EF)) {
    return c5 * x * x + c6 * x + c7;
  } else {
    return c9 * pow(x - c8, -0.25);
  }
}

float beta_pr(float x) {
  if (x < 0.5) {
    return -0.1 * pow(x, -1.1);
  } else if ((x > 0.5) && (x < 1.0)) {
    return 2.0 * c0 * x + c1;
  } else if ((x > 1.0) && (x < 0.9 * EF)) {
    return c3;
  } else if ((x > 0.9 * EF) && (x < EF)) {
    return 2.0 * c5 * x + c6;
  } else {
    return -0.25 * c9 * pow(x - c8, -1.25);
  }
}

/* The tridiagonal matrix algorithm which fills result with the Ax = b solution
 * from Wikipedia.
 */
void tridiag(int num, float *sub_diag, float *diag, float *sup_diag, float *res) {
  int i;
  
  sup_diag[0] /=  diag[0];
  res[0] /= diag[0];
  
  for (i = 1; i < num; i++) {
    float factor = 1.0 / (diag[i] - sub_diag[i - 1] * sup_diag[i - 1]);
    sup_diag[i] *= factor;
    res[i] = (res[i] - sub_diag[i - 1] * res[i - 1]) * factor;
  }
  
  for (i = num - 2; i >= 0; i--) {
    res[i] -= sup_diag[i] * res[i + 1];
  }
}

void step_crank_nicolson() {
  float *temp;
  float error;
  
  bool converged = false;
  int i, j = 0;
  
  /* Initial guess for Newton's method. */
  for (i = 0; i < n_s; i++) {
    new_energy[i] = energy[i];
  }
  
  /* Do not recalculate this in every Newton iteration. */
  constant_part[0] = energy[0] - 0.5 * (dt / ((xs[1] - xs[0]) * (xs[1] - xs[0]))) *
    (beta(energy[1]) - beta(energy[0]));
		      
  for (i = 1; i < (n_s - 1); i++) {
    constant_part[i] = energy[i] - (dt / (xs[i + 1] - xs[i - 1])) *
      (((beta(energy[i + 1]) - beta(energy[i])) / (xs[i + 1] - xs[i])) -
      ((beta(energy[i]) - beta(energy[i - 1])) / (xs[i] - xs[i - 1])));
  }
  
  constant_part[n_s - 1] = energy[n_s - 1] - 0.5 *
    (dt / ((xs[n_s - 1] - xs[n_s - 2]) * (xs[n_s - 1] - xs[n_s - 2]))) *
    (beta(energy[n_s - 2]) - beta(energy[n_s - 1]));
  
  while (!converged) {
    diff[0] = constant_part[0] - 0.5 * (dt / ((xs[1] - xs[0]) * (xs[1] - xs[0]))) *
      (beta(new_energy[1]) - beta(new_energy[0])) - new_energy[0];
      
    jac_diag[0] = 1.0 - 0.5 * dt * beta_pr(new_energy[0]) /
      ((xs[1] - xs[0]) * (xs[1] - xs[0]));
    jac_sup_diag[0] = 0.5 * dt * beta_pr(new_energy[1]) /
      ((xs[1] - xs[0]) * (xs[1] - xs[0]));
    
    for (i = 1; i < (n_s - 1); i++) {
      diff[i] = constant_part[i] - (dt / (xs[i + 1] - xs[i - 1])) *
        (((beta(new_energy[i + 1]) - beta(new_energy[i])) / (xs[i + 1] - xs[i])) -
	((beta(new_energy[i]) - beta(new_energy[i - 1])) / (xs[i] - xs[i - 1]))) -
	new_energy[i];
	
      jac_sub_diag[i - 1] = -1.0 * dt * beta_pr(new_energy[i - 1]) *
        (-1.0 / ((xs[i + 1] - xs[i - 1]) * (xs[i] - xs[i - 1])));
      jac_diag[i] = 1.0 - (dt * beta_pr(new_energy[i]) / (xs[i + 1] - xs[i - 1])) *
        ((1.0 / (xs[i + 1] - xs[i])) + (1.0 / (xs[i] - xs[i - 1])));
      jac_sup_diag[i] = -1.0 * dt * beta_pr(new_energy[i + 1]) *
        (-1.0 / ((xs[i + 1] - xs[i - 1]) * (xs[i + 1] - xs[i])));
    }
    
    diff[n_s - 1] = constant_part[n_s - 1] - 0.5 *
      (dt / ((xs[n_s - 1] - xs[n_s - 2]) * (xs[n_s - 1] - xs[n_s - 2]))) *
      (beta(new_energy[n_s - 2]) - beta(new_energy[n_s - 1])) - new_energy[n_s - 1];
      
    jac_sub_diag[n_s - 2] = 0.5 * dt * beta_pr(new_energy[n_s - 2]) /
      ((xs[n_s - 1] - xs[n_s - 2]) * (xs[n_s - 1] - xs[n_s - 2]));
    jac_diag[n_s - 1] = 1.0 - 0.5 * dt * beta_pr(new_energy[n_s - 1]) /
      ((xs[n_s - 1] - xs[n_s - 2]) * (xs[n_s - 1] - xs[n_s - 2]));
    
    tridiag(n_s, jac_sub_diag, jac_diag, jac_sup_diag, diff);
    error = 0.0;
    j++;
    
    for (i = 0; i < n_s; i++) {
      error += diff[i] * diff[i];
      new_energy[i] += diff[i];
    }

    if (((error / n_s) < SMALL_ERROR) || (j == TOO_MANY_LOOPS)) converged = true;
  }
  
  /* Enforce Neumann boundary conditions. */
  new_energy[0] = new_energy[1];
  new_energy[n_s - 1] = new_energy[n_s - 2];
  
  temp = new_energy;
  new_energy = energy;
  energy = temp;
}

void simulate_crank_nicolson() {
  char *filename = malloc(10 * sizeof(char));
  FILE *fp;
  
  bool exit_loop = false;
  int i, j = 0;
  
  float t = 0.0;
  float time_since_write = max_time;
  
  /* If a certain simulation is proving very difficult, we might want to perform
   * heuristics like increasing the timestep as time goes on and falling back to
   * a backup copy of the energy if we accidentally make it too big.
   */
  while (!exit_loop) {
    step_crank_nicolson();
    
    if (time_since_write > 10000.0 * dt) {
      sprintf(filename, "f%d.dat", j);
      fp = fopen(filename, "w");
      fprintf(fp, "# t = %f\n", t);
      
      for (i = 0; i < n_s; i++) {
        fprintf(fp, "%f\t%f\n", log(1.0 + xs[i]), log(1.0 + energy[i]));
      }
      
      j++;
      time_since_write = 0.0;
      fclose(fp);
    }
    
    /* We quit if time is up or the peak has reached the Hagedorn energy. */
    if ((energy[0] < 1.0) || (t > max_time)) exit_loop = true;
  
    t += dt;
    time_since_write += dt;
  }
  
  printf("Final time: %f\n", t);
  free(filename);
}

int main(int argc, char **argv) {
  int i;
  float x = 0.0;
  
  float read_number;
  bool reading = false;
  FILE *fp;
  
  /* Normally, the initial condition is hard coded. However, if a file
   * is specified on the command line, the file will be read in order
   * to determine the initial condition. This is ideal if the program
   * runs for awhile and then crashes. The last successfully generated
   * file can be passed to make the program pick up where it left off.
   */
  if (argc > 1) {
    reading = true;
    fp = fopen(argv[1], "r");
    fscanf(fp, "# t = %f", &read_number);
  }
  
  /* Choose a, b, c such that ax^2 + bx + c agrees with the value and
   * slope of x^(-1/10) when x = 0.5. It should also have a slope of
   * -0.01 at x = 1.0.
   */
  c0 = -0.01 + 0.1 * pow(2.0, -1.1);
  c1 = -0.01 - 2.0 * c0;
  c2 = pow(2.0, 0.1) - 0.5 * c1 - 0.25 * c0;
  /* Now choose a, b such that ax + b agrees with the value and slope
   * of the above at x = 1.0.
   */
  c3 = -0.01;
  c4 = c0 + c1 + c2 - c3;
  /* Choose a, b, c such that ax^2 + bx + c agrees with the value and
   * slope of the above at x = 0.9 * EF. It should also have half that
   * value at x = EF.
   */
  c5 = (-0.5 * (c3 * (0.9 * EF) + c4) - c3 * 0.1 * EF) / (0.1 * EF * 0.1 * EF);
  c6 = c3 - 2.0 * c5 * (0.9 * EF);
  c7 = 0.5 * (c3 * (0.9 * EF) + c4) - c6 * EF - c5 * EF * EF;
  /* Finally choose a, b such that a(x - b)^(-1/4) agrees with the value
   * and slope of the above at x = EF.
   */
  c8 = EF + 0.25 * (c5 * EF * EF + c6 * EF + c7) / (2.0 * c5 * EF + c6);
  c9 = (c5 * EF * EF + c6 * EF + c7) * pow(EF - c8, 0.25);
  
  /* Makes the total mass enough to fit twice below the Hagedorn interface. */
  length = PMAX * exp(gamma(0.5) + gamma(EXP - 0.5) - gamma (EXP));
  
  /* Two things we pick are a "dt" that is "reasonably small" and a "dx" that
   * fits 10 times between the inflection points. If we want the "dx" to vary
   * with position, the above is really the minimum "dx" value. This example
   * code is for performing a convergence test after a short amount of time so
   * we want the grid spacing to be uniform. The code will be rerun with "dx"
   * multiplied by (1/2), (1/4), (1/8), etc. When we run this code to generate
   * data, we will make the grid non-uniform and simulate for MUCH longer.
   */
  dx_min = 0.1 * sqrt(1.0 / (1.0 + 2.0 * EXP));
  dx_max = dx_min;
  //dx_max = 0.01 * length;
  dt = 0.001;
  max_time = 100000.0 * dt;
  //max_time = 999999.9;
  
  n_s = 0;
  
  /* The space step should vary linearly. */
  while (x < length) {
    n_s++;
    x += dx_min + ((dx_max - dx_min) / length) * x;
  }
  
  /* Now that we know the number of sites, we can allocate the arrays that will
   * be needed in the numerical solution. The numerical solution at a given time
   * is "energy" and we have a copy to advance forward.
   */
  energy = malloc(n_s * sizeof(float));
  new_energy = malloc(n_s * sizeof(float));
  
  /* The numbers in the Jacobian matrix used by Newton's method. */
  jac_sub_diag = malloc((n_s - 1) * sizeof(float));
  jac_diag = malloc(n_s * sizeof(float));
  jac_sup_diag = malloc((n_s - 1) * sizeof(float));
  
  /* Other things that should be stored for Newton's method. */
  constant_part = malloc(n_s * sizeof(float));
  diff = malloc(n_s * sizeof(float));
  
  /* This stores where our grid points are. */
  xs = malloc(n_s * sizeof(float));
  
  x = 0.0;
  
  /* This sets up the initial condition and assigns the grid points. */
  for (i = 0; i < n_s; i++) {
    if (reading) {
      fscanf(fp, "%f", &read_number);
      fscanf(fp, "%f", &read_number);
      energy[i] = exp(read_number) - 1.0;
    } else {
      energy[i] = PMAX * pow(1.0 / (1.0 + x * x), EXP);
    }
    
    x += dx_min + ((dx_max - dx_min) / length) * x;
    xs[i] = x;
  }
  
  if (reading) fclose(fp);
  
  /* Everything is set up so we can call the main loop. */
  printf("L: %f, dx: %f-%f, dt: %f, sites: %d\n", length, dx_min, dx_max, dt, n_s);
  simulate_crank_nicolson();
  
  free(energy);
  free(new_energy);
  free(jac_sup_diag);
  free(jac_diag);
  free(jac_sub_diag);
  free(constant_part);
  free(diff);
  free(xs);
}
\end{verbatim}

\end{document}